\documentclass[11pt,a4]{report}
\usepackage{amssymb}

\usepackage{graphicx}
\usepackage{amsmath}
\usepackage{makeidx}
\usepackage{indentfirst}

\newcounter{resultnum}[section]

\newcounter{conclusionnum}[section]

\newcounter{conditionnum}[section]

\newcounter{conjecturenum}[section]
\newtheorem{example}{Example}[section]

\newcounter{examplenum}[section]

\newcounter{exercisenum}[section]
\newtheorem{lemma}{Lemma}[section]

\newcounter{lemmanum}[section]

\newcounter{notationnum}[section]
\newtheorem{theorem}{Theorem}[section]

\newcounter{theoremnum}[section]
\newtheorem{definition}{Definition}[section]

\newcounter{definitionnum}[section]
\newtheorem{corollary}{Corollary}[section]

\newcounter{corollarynum}[section]
\newtheorem{remark}{Remark}[section]

\newcounter{remarknum}[section]
\newtheorem{proposition}{Proposition}[section]

\newcounter{propositionnum}[section]

\newcounter{acknowledgementnum}[section]

\newcounter{algorithmnum}[section]

\newcounter{axiomnum}[section]

\newcounter{casenum}[section]
\newtheorem{claim}{Claim}[section]

\newcounter{claimnum}[section]

\newcounter{summarynum}[section]

\newcounter{problemnum}[section]
\newenvironment{proof}[1][]{\textbf{Proof.} }{}
\newcommand{ \R} {\mbox{\rm I$\!$R}}
\newcommand{ \C} {\mbox{\rm I$\!$C}}

\begin{document}

\title{Nonholonomic Clifford and Finsler Structures,\\
Non--Commutative Ricci Flows,\\ and Mathematical Relativity}
\date{{\Large \textbf{\vskip3cm Habilitation Thesis}}\\
\vskip 0.5cm CNATDCU, Romania, 2012 \vskip 0.5cm}
\author{{\Large \textsf{Sergiu I. Vacaru\thanks{%
All Rights Reserved \copyright \ 2012 \ Sergiu I. Vacaru, \newline
sergiu.vacaru@uaic.ro,\ http://www.scribd.com/people/view/1455460-sergiu
\newline
{\ } \newline
Performed following a Guide on Habilitation Theses, see \newline
http://www.cnatdcu.ro/wp-content/uploads/2011/11/Ghid-de-abilitare-2012.pdf
\newline
{\ } \newline
\textbf{key directions:}\ mathematical physics, geometric methods in physics,%
\newline
general relativity and modified gravity theories, applied mathematics }}} \\
{\quad} \\
{\ {\textsl{\ Science Department, University "Al. I. Cuza" Ia\c si},} }\\
{\ {\textsl{\ 54 Lascar Catargi street, 700107, Ia\c si, Romania}} }}
\maketitle
\tableofcontents


\section{Summary of the Habilitation Thesis}

\ It is outlined applicant's 18 years research and pluralistic pedagogical
activity on mathematical physics, geometric methods in particle physics and
gravity, modifications and applications (after defending his  PhD thesis in
1994). Ten most relevant publications are structured conventionally into
three "strategic directions":\ 1) \textsf{nonholonomic geometric flows
evolutions and exact solutions} for Ricci solitons and field equations in
(modified) gravity theories; 2) \textsf{\ geometric methods in quantization}
of models with nonlinear dynamics and anisotropic field interactions; 3)
\textsf{(non) commutative geometry, almost K\"{a}hler and Clifford
structures,} Dirac operators and effective Lagrange--Hamilton and
Riemann--Finsler spaces.

The applicant was involved in more than 15 high level multi-disciplinary
international and national research programs, NATO and UNESCO, and
visiting/sabatical professor fellowships and grants in USA, Germany, Canada,
Spain, Portugal, Romania etc. He got support from organizers for more than
100 short visits with lectures and talks at International Conferences and
Seminars.

Both in relation to above mentioned strategic directions 1)--3) and in
"extension", he contributed with almost 60 scientific works published and
cited in high influence score journals (individually, almost 50 \%, and in
collaboration with senior, 30 \%, and yang, 20 \%, researchers).
Applicant's papers are devoted to various subjects (by 15 main directions)
in noncommutative geometry and gravity theories; deformation, A--brane,
gauge like and covariant anisotropic quantization; strings and brane
physics; geometry of curved flows and associated solitonic hierarchies with
hidden symmetries; noncommutative, quantum and/or supersymmetric
generalizations of Finsler and Lagrange--Hamilton geometry and gravity;
algebroids, gerbes, spinors and Clifford and almost K\"{a}hler structures;
fractional calculus, differential geometry and physics; off--diagonal exact
solutions for Einstein - Yang - Mills - Higgs - Dirac systems; geometric
mechanics, nonlinear evolution and diffusion processes, kinetics and
thermodynamics; locally anisotropic black holes/ellipsoids / wormholes and
cosmological solutions in Einstein and modified gravity theories;
applications of above listed results and methods in modern cosmology and
astrophysics and developments in standard particle physics and/or modified
gravity.

Beginning June 2009, the applicant holds time limited senior research
positions (CS 1) at the University Alexandru Ioan Cuza (UAIC) at Ia\c si
University, Romania. With such affiliations, he published by 25 articles in
top ISI journals; more than half of such papers won the so--called "red,
yellow and blue" excellence, respectively, (4,4, 6), in the competition of
articles by Romanian authors. During
2009-2011, he communicated his results at almost 30 International Conference
and Seminars having support from hosts in UK, Germany, France, Italy, Spain,
Belgium, Norway, Turkey and Romania.

For the Commission of Mathematics for habilitation of university professors
and senior researchers of grade 1, it is computed [for relevant publications
in absolute high influence score journals] this conventional ''eligibility
triple'': (points for all articles; articles last 7 years; number of
citations) = (55.9; 28.74; 53) which is higher than the minimal standards (5;
2.5; 12). Taking into account the multi-disciplinary character of research
on mathematical physics, there are provided similar data for the Commission
of Physics:\ (59.19; 31.6; 61) which is also higher than the corresponding
minimal standards (5; 5; 40).

Future research and pedagogical perspectives are positively related to the
fact that the applicant won recently a three years Grant IDEI,
PN-II-ID-PCE-2011-3-0256. This allows him to organize a computer--macros
basis for research and studies on mathematical and computational physics and
supervise a team of senior and young researches on "nonlinear dynamics and
gravity".




\newpage

\section{\ Sinteza tezei de abilitare (in Romanian)}

Este trecut\v a \^in revist\v a activitatea de 18 ani de cercetare \c si
didactic\v a prin cumul a aplicantului \^ in domenii legate de fizica
matematic\v a, metode geometrice \^ in fizica particulelor \c si gravita\c
tie, modific\v ari \c si aplica\c tii (dup\v a sus\c tinerea tezei de
doctorat in 1994). Zece cele mai relevante publica\c tii sunt structurate
conven\c tional in trei "direc\c tii strategice": 1) \textsf{\ evolu\c tii
neolonome "geometric flows" si solu\c tii exacte} pentru solitoni Ricci \c
si ecua\c tii de c\^ amp \^ in teorii de gravita\c tie (modificate);\ 2)
\textsf{\ metode geometrice \^ in cuantificarea } modelelor cu dinamic\v a
nelinear\v a \c si interac\c tiuni anisotrope de c\^ amp;\ 3) \textsf{\
geometrie (ne) comutativ\v a, structuri aproape K\"{a}hler \c si Clifford},
operatori Dirac si spa\c tii efective Lagrange--Hamilton \c si
Riemann--Finsler.

Aplicantul a fost implicat \^ in peste 15 programe interna\c tionale \c si
na\c tio\-nale de cercetare de nivel \^ inalt, multi--disciplinare, OTAN \c
si UNESCO, \c si granturi pentru profesor \^ in vizit\v a sau sabatic \^ in SUA,
Germania, Canada, Spania, Portugalia, Rom\^ ania etc. A op\c tinut suport de
la organizatori pentru peste 100 vizite scurte cu lec\c tii \c si comunic\v
ari la conferin\c te \c si seminare interna\c tionale.

Cu privire la direc\c tiile 1) -- 3) men\c tionate mai sus, c\^ at \c si \^
in extensio, aplicantul a contribuit cu circa 60 lucr\v ari \c stiin\c
tifice publicate \c si citate \^ in reviste cu punctaj \^ inalt de
influien\c t\v a (individual, circa 50 \%, \c si \^ in colaborare cu
cercet\v atori seniori, 30 \%, \c si tineri, 20 \%). Lucr\v arile
aplicantului sunt consacrate diferitor subiecte (circa 15 direc\c tii
principale) \^ in geometrie necomutativ\v a \c si teorii de gravita\c tie;\
cuantificare de deformare, A--brane, similar gauge \c si covariant
anisotrop\v a;\ fizic\v a string \c si brane;\ geometrii "curved flows" \c
si ierarhii solitonice asociate cu simetrii ascunse;\ generaliz\v ari
necomutative, cuantice \c si/ sau supersimetrice ale geometriilor Finsler \c
si Lagrange--Hamilton \c si gravita\c tie;\ algebroizi, gerbe, spinori \c si
structuri aproape K\"{a}hler;\ calculus frac\c tional, geometrie diferen\c
tial\v a \c si fizic\v a;\ solu\c tii ne--diagonale exacte pentru sisteme
Einstein - Yang - Mills - Higgs - Dirac;\ geometrie mecanic\v a, evolu\c tie
nelinear\v a \c si procese de difuzie;\ goluri negre / elipsoizi /
wormholuri \c si solu\c tii cosmologice in teorii de gravita\c tie Einstein
\c si modificate;\ aplica\c tii ale rezultatelor men\c tionate mai sus in
cosmologie modern\v a \c si astrofizic\v a \c si dezvolt\v ari \^ in fizica
particulelor standard\v a \c si / sau gravita\c tie modificat\v a.

\^ Incep\^ and cu iunie 2009, aplicantul are pozitii de cercet\v ator \c
stiin\c tific superior, CS 1, cu termen limitat, la Universitatea Alexandru Ioan Cuza
(UAIC) din Ia\c si, Rom\^ ania. Avand astfel afilieri,  a publicat peste 25
articole in reviste top ISI;\ peste o jum\v atate din  articole au
fost c\^ a\c stig\v atoare de tipul "ro\c su, galben \c si albastru",
respectiv, (4,4,6), in competi\c tia de lucr\v ari ale autorilor care activeaza \^ in  Rom\^ ania. \^ In decursul 2009--2011, el a comunicat
rezultatele sale la circa 30 conferin\c te \c si seminare interna\c tionale
av\^ ind suport de la gazde \c stiin\c tifice in MB, Germania, Fran\c ta,
Italia, Spania, Belgia, Norvegia, Turcia \c si Rom\^ ania.

Pentru comisia de matematic\v a pentru abilitarea profesorilor universitari
\c si a cercet\v atorilor superiori de gradul 1 este calculat [publica\c tii relevante \^in reviste cu scor \^ inalt de influien\c t\v a]
acest conven\c tional "triplu de eligibilitate": (puncte pentru toate
articolele; articole \^ in ultimii 7 ani; num\v arul cit\v arilor) =(55.9; 28.74; 53) ce depa\c se\c ste standardele minimale (5; 2.5; 12). Lu\^ and \^
in considera\c tie caracterul multi--disciplinar al cercet\v arii din
fizic\v a matematic\v a, sunt prezentate date similare pentru comisia de
fizic\v a:\ (59.19; 31.6; 61), ce tot dep\v a\c sesc standardele  minimale (5; 5; 40).

Perspective reale de cercetare \c si activitate pedagogic\v a pentru viitor
sunt legate de faptul c\v a applicantul a c\^ a\c stigat recent, pentru trei
ani, un Grant IDEI, PN-II-ID-PCE-2011-3-0256. Aceasta \^ ii permite s\v a
organizeze o baz\v a computer--macros pentru cercetare \c si studii \^ in
matematic\v a \c si fizic\v a compu\-ta\c tional\v a \c si conducere a unei
echipe de cercet\v atori seniori \c si tineri \^ in "dinamic\v a nelinear\v
a \c si gravita\c tie". 



\chapter{Achievements}

\label{ch1} Today, various directions in modern geometry and physics are so
interrelated and complex that it is often very difficult to master them as
separated subjects. Research and pedagogical activities on mathematical
physics, geometry and physics, relativity and high energy physics etc play a
multi- and/or inter--disciplinary character with various applications and
connections to advanced computer methods and graphics, modern technology and
engineering etc. There is a need of research teams of mathematicians skilled
both in geometric and analytic methods and oriented to fundamental and
experimental physics and/or, inversely, theoretical and mathematical physics
researches with a rigorous education and research experience in differential
geometry, nonlinear analysis, differential equations and computer methods.

The author of this Habilitation Thesis was involved in various
multi--disciplinary research projects and pluralistic pedagogical activity
on mathematics and physics\footnote{%
see Chapter \ref{ch3} with a relevant Bibliography and CV and publication
list (included in the File for this Habilitation Thesis) and, for instance,
reviews in MathSciNet and Web of Science and Webpage
http://www.scribd.com/people/view/1455460-sergiu with details on NATO and
UNESCO and visiting/sabatical professor fellowships and grants for various
long/short terms visits (respectively, almost 15/ 40) and research activity
in USA, UK, Germany, France, Italy, Spain, Portugal, Greece, Norway, Canada,
Turkey, Romania etc} after he got his PhD on theoretical physics, in 1994,
at the University Alexandru Ioan Cuza (UAIC) at Ia\c{s}i, Romania. The PhD
thesis was elaborated almost individually at the Department of Physics of
''M. Lomonosov'' State University (Moscow, Russia) and Institute of Applied
Physics, Academy of Sciences of Moldova (Chi\c sin\v au, Republic Moldova)
during 1984-1992,\footnote{%
in former URSS, there were some options for performing/defending equivalents
of PhD and Habilitation Thesis as an individual applicant} and finalized at
UAIC (1992-1994). That research on geometric and twistor methods in
classical field theory, gravity and condensed matter physics was developed
and extended to new directions in modern geometry, mathematics and physics
which are summarized and concluded in this Chapter (see section \ref{ssvp}
on various directions of research, related papers and comments).

The goal of Chapter \ref{ch1} is to present applicant's research,
professional and academic achievements in relevant (multi/ inter--)
disciplinary directions providing necessary proofs and references.\footnote{%
This review type work is performed following CNATDCU Guide on Habilitation
Theses, see link
http://www.cnatdcu.ro/wp-content/uploads/2011/11/Ghid-de-abilitare-2012.pdf
. In order to facilitate readers and/or experts from Countries with possible
different standards on habilitation, we summarize in English some most
important requests stated there in Romanian. Following points (a) (i) in the
Guide for such a thesis, the limits are between 150.000--300.000 characters
for a Chapter based on most important contributions and selected (maximum
10) most important and relevant author's articles \cite{v1}--\cite{v10}
attached to the "Habilitation File" (with necessary documents presented to the Commission).} His
original results are emphasized in a context of present International and
National matter of state of science and education. It is used a selection of
works and monographs from S. Vacaru's Publication List, see Refs. \cite{v1}--%
\cite{vcp8}, and a list of his last 7 years participations/talks at
International Conferences and Seminars, see Refs. \cite{vcs1}--\cite{vcs40}
(see details in Chapter \ref{ch3}). Taking into account the length limits
for such a thesis, there are included in the Bibliography only a part of
applicant's publications and recent talks; necessary references and
comments on "other" author's papers, and "other" authors, can be found in
the cited works.

Section \ref{ssvp} is devoted to scientific visibility and prestige of
applicant's activity. There are briefly outlined and commented the strategic and main
research results and relevant author's references, listed a series of
examples and contributions for International Scientific Conferences and
Seminars and mentioned most important grants and temporary positions.
Comments and conclusions on "main stream and other" important issues and
publications and an analysis of eligibility and minimal standards are
provided. Section \ref{sgsmimr} contains a review for experts on
differential geometry, mathematical physics and gravity theories based on a
selection of results from 10 most relevant and important author's articles %
\cite{v1}--\cite{v10}. The main goal of this section is to show some most
important examples of original research with an advanced level of mathematical methods and
possible applications in physics and geometric mechanics which can be found
in applicant's works.

\section{Scientific Visibility \& Prestige}

\label{ssvp} Conventionally, applicant's research activity correlated to 10
most relevant works \cite{v1}--\cite{v10} (see a survey in Section \ref%
{sgsmimr}) can be structured into \textbf{three strategic directions}:

\begin{enumerate}
\item \textsf{Nonholonomic commutative and noncommutative geometric flows
evolutions and exact solutions} for Ricci solitons and field equations in
(modified) gravity theories \cite{v4,v6,v3,v9};

\item \textsf{Geometric methods in quantization} of models with nonlinear
dynamics and anisotropic field interactions \cite{v7,v5,v10,v9};

\item \textsf{(Non) commutative geometry, almost K\"{a}hler and Clifford
structures,} Dirac operators, effective Lagrange--Hamilton and
Riemann--Finsler spaces and analogous/ modified gravity \cite{v2,v1,v8,v6}.
\end{enumerate}

In a more general context, including other partner works (inter--related),
one can be considered \textbf{15 main research directions}.

\subsection{Comments on  strategic and  main directions}

\label{ssmrd} It is presented a synopsis of related ISI works \cite{v1}--%
\cite{ve49}.\footnote{%
One should be mentioned here some series of ''not less'' important
contributions containing various preliminary, or alternative, ideas and
results published in Romania \cite{vl1}--\cite{vl5} and Republic of Moldova %
\cite{vrm1}-\cite{vrm5}; monographs \cite{vb1}-\cite{vb3}, chapters and
sections in collections of works \cite{vcb1}-\cite{vcb3}, reviews in
journals and encyclopedia \cite{vb4,vb5}, articles in Proceedings of
Conferences \cite{vcp1}-\cite{vcp8} and some recent electronic preprints
with reviews, computation details and proofs, see \cite{veprep1,veprep2}.}
There are outlined motivations, original ideas and most important results in
15 main directions.

\begin{enumerate}
\item \emph{(Non) commutative gauge theories of gravity, anisotropic
generalizations, and perturbative methods of quantization} \cite%
{v2,ve19,ve1,v8,ve17,ve17,vcp4,vb1,vb2,vb3,vrm2,vrm3,ve34,ve38}.

\begin{enumerate}
\item Affine and de Sitter models of gauge gravity.

\item Gauge like models of Lagrange--Finsler gravity.

\item Locally anisotropic gauge theories and perturbative quantization.

\item Noncommutative gauge gravity.
\end{enumerate}

\textit{Comments:} This direction was elaborated as a natural development of
some chapters and sections in author's PhD thesis (1994), where affine and
de Sitter gauge like models were considered for the twistor--gauge
formulation of gravity. The first publications on anisotropic gauge gravity
theories were in R. Moldova \cite{vrm2,vrm3} (1994-1996); see also a paper
together with a graduate student, Yu. Goncharenko \cite{ve1}, when authors
were allowed to present their results in a Western Journal. The main
constructions were based on the idea that the Einstein equations can be
equivalently reformulated as some Yang--Mills equations for the affine
and/or de Sitter frame bundles, with nonlinear realizations of corresponding
gauge groups and well defined projections on base spacetime manifolds (we
used the Popov--Dikhin approach, 1976, and A. Tseytlin generalization, 1982;
see references in above cited papers\footnote{%
in this thesis with explicit limits on length, there are provided references
only on applicant's works; contributions by other authors are cited exactly
in the mentioned references and/or 10 most relevant articles}).

In order to formulate (non) commutative and/or supersymmetric gauge theories
of Lagrange--Finsler gravity, we used the Cartan connection in the affine
and/or de Sitter bundles on Finsler (super) manifolds and various
anisotropic generalizations, including higher order tangent/vector bundles.
Such results are contained in some chapters of monographs \cite{vb1,vb2,vb3}
and presented at a NATO workshop in 2001, see \cite{vcp4}.

Formal re--definitions of Einstein gravity and generalizations as gauge like
models allowed the applicant to perform one of the most cited his works %
\cite{v2} (included as the second one in the list of most relevant
applicant's 10 articles). That paper was devoted to the Seiberg--Witten
transforms and noncommutative generalizations of Einstein and gauge gravity.
The corresponding gravitational equations with noncommutative deformations
can be integrated in very general off--diagonal forms \cite{ve19}, see Ref. %
\cite{v8} on noncommutative Finsler black hole solutions.

It should be mentioned here a collaboration with Prof. H. Dehnen (Konstanz
University, Germany, 2000-2003) on higher order Finsler--gauge theories,
nearly autoparallel maps and conservation laws, see \cite{ve16,ve17} and a
recent approach to two--connection perturbative quantization of gauge
gravity models \cite{ve34,ve38}.

\item \emph{Clifford structures and spinors on nonholonomic manifolds and
generalized Lagrange-Finsler and Hamilton-Cartan spaces} \cite%
{v1,vb2,vb1,ve21a,ve23,vcp8,ve36,vrm5,ve4,ve3,vcp5,vcb3,ve9,ve11,vl1,vl2,ve22d,vl5,ve20,v6}%
.

\begin{enumerate}
\item Definition of spinors and Dirac operators on generalized Finsler
spaces.

\item Clifford structures with nonlinear connections and nonholonomic
manifolds.

\item Spinors and field interactions in higher order anisotropic spaces.

\item Solutions for nonholonomic Einstein--Dirac systems and extra dimension
gravity.

\item Nonholonomic gerbes, index theorems, and Clifford--Finsler geometry.

\item Nonholonomic Clifford and Lagrange--Finsler algebroids.
\end{enumerate}

\textit{Comments:} A nonholonomic manifold/bundle space is by definition
enabled with a nonholonomic (equivalently, anholonomic, or
non--inte\-grab\-le) distribution, see main concepts and definitions in
''preliminaries'' of the section \ref{ssprelim} and references therein. For
various important geometric and physical models, it is enough to consider
spaces with nonholonomic splitting (as a Whitney sum) into conventional
horizontal (h) and vertical (v) subspaces.\footnote{%
Typical examples are 2+2 frame decompositions in general relativity and
vector/tangent bundles enabled with nonlinear connection structure (in
brief, N--connection, which can be defined as a non--integrable
h-v--splitting of the tangent bundle to a manifold, or to a tangent/vector
bundle), for instance, in a model of Finsler geometry.} One could be
conceptual and technical difficulties in adapting the geometric and physical
constructions on certain spaces enabled with N--connection structure. For
instance, the problem of definition of spinors and Dirac operators on
nonholonomic manifolds and/or Finsler--Lagrange spaces was not solved during
almost 60 years after first E. Cartan's monographs on spinors in curved
spaces and Finsler geometry (during 1932--1935). The applicant proposed
rigorous geometric definitions of Finsler spinors \cite{v1} and, in general,
of spinors and Dirac operators on nonholonomic manifolds/bundle spaces \cite%
{ve4,ve3,vb1,vb2,vb3}, and developed the so--called nonholonomic Clifford
geometry in a numbers of his and co--author works during 1995 -- present.

There were some attempts to define two dimensional spinor bundles on Finsler
spaces and generalizations in the 70th-80th of previous century (Takano and
Ono, in Japan, and Stavrinos, in Greece; see main references and historical
remarks in \cite{v1,vb2,vb1,ve21a}). Nevertheless, there were not provided
in those works any self--consistent definitions of spinors and Dirac
operators for Finsler spaces which would relate a Clifford algebra
structure, and spin operators, to Finsler metrics and connections. The
problem of definition of ''Finsler spinors'' is very important in
fundamental physics and mechanics if there are considered dependencies of
physical objects on "velocity/momentum" variables. For instance, such models
of Finsler spacetimes are elaborated for quantum gravity and modern
cosmology, see details and critical remarks in \cite{ve23,vcp8,ve36} and
Introduction to \cite{vb1}. Without spinors/fermions, it is not clear how to
construct ''viable'' physical models with dependence on some
''velocity/momentum'' type variables. Similar problems have to be solved for
generic off--diagonal solutions in Einstein gravity with spinors, and
modifications, and nontrivial N--connection structure and conventional
spacetime splitting.

In 1994--1995, the applicant became interested in the problem of
elaborating theories of gravitational and matter field interactions on
generalized Finsler spaces (in a more general context, in the sense of G.
Vr\v anceanu's definition of nonholonomic manifolds, 1926-1927). It was a
special research grant Romania--R. Moldova affiliated to the school on
generalized Finsler-Lagrange-Hamilton geometry at Ia\c si supervised by
Acad. R. Miron. The paper \cite{vrm5} (submitted in 1994 before establishing
that collaboration and published in 1996 in R. Moldova) contains the first
self--consistent definition of Clifford structures and spinors for Finsler
spaces and generalizations. Such results formulated in a more rigorous form,
with developments for complex and real spinor Lagrange--Finsler structures
and Dirac operators adapted to N--connections, were published also in J.
Math. Physics. (1996), see \cite{v1}.

Having defined nonolonomic Clifford bundles, it was possible to construct
geometric models of gravitational and field interactions on (super) spaces
with higher order anisotropy \cite{ve4,ve3}. There were obtained some new
results in differential spinor geometry and supergeometry with possible
applications in high energy physics (see more details in point 4b). It was
possible to involve in such activities two professors from Greece (P.
Stavrinos and G. Tsagas, see monographs \cite{vb2,vb1} and paper \cite{vcp5}%
) and some young researchers\footnote{%
at that time under-graduate and post-graduate students in R. Moldova (N.
Vicol and I. Chiosa, see papers \cite{ve21a,vcp5,vcb3}) and Romania (F.\ C.
Popa and O. \c Tin\c t\v areanu-Mircea, see \cite{ve9,ve11,vl1,vl2})}.
Together with some sections in monograph \cite{vb3}, such works contain a
series of new results on Dirac spinor waves and solitons, spinning
particles, in Taub NUT anisotropic spaces, solutions for Einstein--Dirac
systems in nonholonomic higher dimension gravity, supergravity and Finsler
modificaions of gravity.

There were elaborated three another directions related to nonholonomic
(Finsler) spinors and Dirac operators: For instance, papers \cite{ve22d,vl5}
(the first one, in collaboration with J. F. Gonzalez--Hernandez, in 2005, a
student from Madrid, Spain) are devoted to nonholonomic gerbes,
Clifford--Finsler structures and index theorems. Article \cite{ve20}
contains definitions and examples of nonholonomic Clifford and
Finsler--Clifford algebroids with theorems on main properties of indices of
connections in such spaces.

Finally, in this point, it should be noted that the constructions for the
nonholonomic Diract operators were applied for definition of noncommutative
Finsler spaces and Ricci flows in A. Connes sense, see details in Refs. \cite%
{v6} (the 6th most relevant and important applicant's paper) and in Part
III of monograph \cite{vb1} (there are connections to points 4b and 14e).

\item \emph{Nearly autoparallel maps, nonlinear connections, twistors and
conservation laws in Lagrange and Finsler spaces} \cite%
{ve2,vcb2,vcb1,ve16,ve17,vb3}.

\textit{Comments:} The geometry of nearly autoparallel maps (various
examples were studied by H. Weyl, A. Z. Petrov and summarized in a monograph published in Russian
by N. Sinyukov in 1979) generalizes various models with geodesic and
conformal transforms. Some chapters of applicant's PhD thesis were devoted
to such transforms and definition of corresponding invariants and
conservation laws for spaces with nontrivial torsion, endowed with
spinor/twistor structure etc.

This is an open direction for further research. For instance, the geometry
of twistors for curved spaces was studied in Ref. \cite{ve2} using nearly
autoparllel maps. Local twistors were defined on conformally flat spaces and
mapped via generalized transforms to more general (pseudo) Riemannian and
Einstein spaces. The key result was that even the twistor equations are not
integrable on general curved spaces such couples of spinors structures can
be defined via nonholonomic deformations and generalize nearly autoparallel
maps. Following this approach, we can consider analogs of Thomas invariants
and Weyl tensors (in certain generalized forms, with corresponding
symmetries and conservation laws). The constructions were generalized for
Lagrange and Finsler spaces \cite{vcb2} -- it was a collaboration with a
former applicant's student, S. Ostaf.

There are relevant certain results from Refs. \cite{vcb1,vcb2} (a
collaboration with the former PhD superviser in Romania, Prof. I. Gottlieb)
when the A. Moor's tensor integral was considered, see paper \cite{vcb2}. It
was also an article on tensor integration and conservation laws on
nonholonomic spaces published by applicant individually in R. Moldova, see %
\cite{vrm4}.

One should be mentioned again the articles \cite{ve16,ve17}, in
collaboration with Prof. H. Dehnen, where generalized geodesic and conformal
maps were considered in (higher order) models of Finsler gravity and in
gauge and Einstein gravity.

A part of results in this direction was partially summarized (also in
supersymmetric form) in two chapters of monograph \cite{vb3}.

\item \emph{Locally anisotropic gravity in low energy limits of string/
brane theories; geometry of super--Finsler space } \cite%
{ve3a,ve3,vb6,vcb3,vcp4,vl1,vl2,vb3}.

\begin{enumerate}
\item Nonholonomic background methods and locally anisotropic string
configurations.

\item Supersymmetric generalizations of Lagrange--Finsler spaces.
\end{enumerate}

\textit{Comments:} If some  Lagrange--Finsler geometry models are  related
to real physics, such configurations have to be derived in some low energy
limits of (super) string theory. Papers \cite{ve3a,ve3} published in very
influent score journals, Annals Phys. (NY) and Nucl. Phys. B; 1997), were
devoted to supersymmetric generalizations of theories with local
anisotropies and nonholonomic structures (the concept of
superspace/superbundle involves a special class of nonholonomic complex
distributions). The applicant is the author of Supersymmetry Encyclopedia
term ''super-Finsler space'' \cite{vb6}. We note that prof.\ A. Bejancu
introduced nonlinear connections with ''super-fiber'' indices in some his
preprints at Vest University of Timi\c soara and in a monograph on Finsler
geometry and applications, in 1990. Applicant's idea was to formulate
a rigorous approach to the geometry of N--connections in superspaces via
nonolonomic distributions taking as bases of superbundles various classes of
supermanifolds.

One of the main problems in such a research on supergeometry and
supergravity is that there is not a generally accepted definition of
''supermanifolds'' and ''superspaces'' - the existing ones differ for global
constructions. Via nonholonomic distributions, the concept of nonlinear
connection can be introduced for all considered concepts of superspace which
allow to elaborate corresponding models of supersymmetric Lagrange--Finsler
geometry. Following the background field method with supersymmetric and
N--adapted derivatives, and a correspondingly adapted variational principle,
locally anisotropic configurations can derived in low energy limits of
string theory.

A series of works on supersymmetric models of noholonomic superspaces and
supergravity was elaborated in R. Moldova and Romania and communicated at
International Conferences \cite{vcb3,vcp4,vl1,vl2} (in collaboration with
former applicant's students, N. Vicol, I.\ Chiosa, and with young
researchers from Bucharest-Magurele, F.\ C. Popa and O. T\^ in\c t\^
areanu-Mircea). Part I of monograph \cite{vb3} is devoted to the geometry of
nonholonomic supermanifolds and possible applications in physics. Here we
note that a series of papers on string, brane and quantum gravity were
published during last 15 years, in certain alternative ways, by Prof. N.
Mavromatos and co-authors from King's College of London.

\item \emph{Anisotropic Taub--NUT spaces and Dirac spin waves and solitonic
solutions} \cite{ve9,ve11,ve22c,vl4}.

\textit{Comments:} The applicant found a series of applications of his
anholonomic deformation method of constructing exact solutions (related to
anisotropic generalizations, exact solutions and physical models of
Taub--NUT spaces, with Dirac waves, solitons, spinning particles and
supersymmetric configurations) after he got some temporary positions at the
Institute of Space Sciences, Bucharest--Magurele, Romania, in 2001. It was a
collaboration with PhD students F. C. Popa and O. \c T\^ in\c t\v
areanu-Mircea, see articles \cite{ve9,ve11} published in high influence
score journals (Classical and Quantum Gravity and Nuclear Physics B). The
direction was latter, in 2006, extended to Ricci flow solutions related to
Taub NUT \cite{ve22c,vl4} (in collaboration with Prof. M. Vi\c sinescu).

The works cited in this point contain a number of examples of exact
solutions constructed in extra dimension and Einstein gravity theories using
the N--connection formalism and nonholonomic frame deformations which
originated from Finsler geometry and nonholonomic mechanics. Such results
are related to those outlined below in points 8,9 and 11d, 11e.

\item \emph{\ Nonholonomic anisotropic diffusion, kinetic and
thermodynamical processes in gravity and geometric mechanics} \cite%
{ve5,ve6,vrm1,vcp1,vb3,v4,v6,ve31}.

\begin{enumerate}
\item Stochastic processes, diffusion and thermodynamics on  nonholonomic
curved spaces (super) bundles.

\item Locally anisotropic kinetic processes and thermodynamics in cur\-ved
spaces.
\end{enumerate}

\textit{Comments:} A program of research with applications of Finsler
metrics and stochastic processes in biophysics was performed in the 90th of
previous century by professors P. Antonelli and T. Zastavniak in Canada. In
order to study diffusion processes on locally anisotropic spaces, it was
important to define Laplace operators for Finsler spaces (such constructions
were proposed by Prof. M. Anastasiei by 1992--1994, who sheared certain
information with, at that time a young researcher, S. Vacaru). That researcher, and present applicant,
proposed his definitions of Laplace operator using the canonical
distinguished connection and the Cartan distinguished connection and
corresponding \^{I}to and Stratonovich types of anisotropic stochastic
calculus on generalized Finsler space during Ia\c si Academic days in
October 1994. Those results with a study of stochastic and diffusion
processes on Finsler--Lagrange spaces and vector bundles enabled with
nonlinear connection structure were published latter (1995-1996) in R.
Moldova and Proceedings of a Conference in Greece, see Refs. \cite{vrm1,vcp1}%
. Independently, similar results were published in parallel by P. Antonelli,
T. Zastawniak and D. Hrimiuc with applications in biology and biophysics (1995-2004).

Applicant's research was oriented to exploration of locally anisotropic
diffusion processes with possible applications in modern physics and
cosmology. In 2001, he was able to publish two his papers in Annals of
Physics (Leipzig) and Annals of Physics (New York) on stochastic processes
and anisotropic thermodynamics in general relativity and, respectively, on
locally anisotropic kinetic processes and non--equilibrium thermodynamics
with some applications in cosmology, see Refs. \cite{ve5,ve6}. The main
results in those directions were based on the fact that anisotropic
processes with additional nonolonomic constraints, in general, with
velocity/momentum variables can be adapted to nonholonomic distributions
using metric compatible distinguished connections like in Finsler geometry.
Here, it should be mentioned that a Russian physicist, A. A. Vlasov,
published in 1966 a book on ''statistical distribution functions'' where for
the theory of kinetics in curved spaces certain classes of Finsler metrics
and connections were considered. Applicant's idea was to generalize the
results for Lagrange and Hamilton geometries and their higher order
anisotropic (including supersymmetric) extensions. It was shown that the N--
connection formalism and adapted  frames play a substantial role in definition of
anisotropic nonholonomic stochastic and diffusion processes and similarly in
kinetics and geometric thermodynamics of constrained physical systems etc.
Such constructions were summarized in two chapters of monograph \cite{vb3}.

Perhaps, there is a perspective direction for future investigations related
to above mentioned ''diffusion geometry'' and analogous thermodynamics. In
papers \cite{v4,v6,ve31}, there are considered generalizations of Grisha
Perelman's entropy and thermodynamical functionals for nonholonomic Ricci
flows and Lagrange--Finsler evolutions. In equilibrium, such processes can
described as certain Ricci solitonic systems or effective Einstein spaces
with nonholonomic constraints. Various classes of solutions of such
evolution and effective field equations can be described by stochastic
generating functions. To relate the thermodynamical values for Ricci flows
to some analogous diffusion processes and ''standard'' kinetic and
thermodynamic theory, or to black hole thermodynamic processes, is a
difficult mathematical physics problem with less known implications in
modern physics.

\item \emph{Differential fractional derivative geometry, gravity and
geometric mechanics, and deformation quantization} \cite%
{veprep1,vcs1,ve40,ve41,ve42,ve43,ve45,ve46}.

\textit{Comments:} This is a very recent direction of applicant's research
papers during 2010--2011. The problem of constructing ''fractional
derivatives'' was studied in a series of classical works by Leibnitz,
Riemann and other prominent mathematicians (fractional derivatives should
be not confused with "fractals" and fractional dimensions). At present, there is an
increasing number of publications with applications in modern engineering,
economics etc. For instance, there is a well known  series of works with
fractional derivative diffusion by F. Mainardi (last 30 years) and a
self--consistent reformulation of physical theories on flat spaces to
fractional derivatives was proposed by V. E. Tarasov (beginning 2005).

The main problems in elaborating geometric and gravitational models with
fractional derivatives were related to certain very cumbersome
integro--differential relations present in the Riemann--Liouville integral
operators. Such fractional derivatives acting on scalars do not result in
zero. In papers \cite{veprep1,vcs1}, there were elaborated models of Ricci
flows and gravity theories using the so--called Caputo's fractional
derivative transforming scalar values in zero. Such constructions can be
re--defined for the Riemann--Liouville fractional derivatives via
corresponding nonholonomic integro--differential transforms.

In a series of works \cite{ve40,ve41,ve42,ve43}, in collaboration with Prof.
D. Baleanu (from Ankara, Turkey, and Bucharest--Magurele, Romania), there
were elaborated fractional models of almost K\"{a}hler -- Lagrange geometry,
constructed exact solutions in gravity and geometric mechanics, with
solitonic hierarchies and deformation quantization of such theories. The
results were published in Proceedings of two International Conferences \cite%
{ve45,ve46} and a seminar in Italy \cite{vcs1}.

Finally (in this point), we note that there are not standard and unique ways
for constructing geometric and physical models with fractional derivatives.
For instance, a series of papers by G. Calcagni (2011) is based on a quite
different approach with the aim to unify fractional dimensions, fractional
derivatives, noncommutative and diffusion processes. The geometric formalism
and related fractional partial derivatives depend on certain assumptions on
the types of nonlocal and "memory" nonlinear effects we try to study, for
instance, in theories of condensed matter or quantum gravity.

\item \emph{Geometric methods of constructing generic off-diagonal solutions
for Ricci solitons, nonholonomic Einstein spaces and in modified theories of
gravity} \cite%
{v4,ve34,ve35,veprep2,vcp2,vcp7,ve19,ve37,ve39,ve25,ve26,ve27,vb1,ve48},\cite%
{ve6}--\cite{ve15}.

\begin{enumerate}
\item Decoupling property of (generalized) Einstein equations and
integrability for (modified) theories with commutative and noncommutative
variables.

\item Generating exact solutions with ellipsoidal, solitonic and pp--wave
configurations, possible cosmological solutions.

\item Generic off--diagonal Einstein--Yang--Mills--Higgs configurations.
\end{enumerate}

\textit{Comments:} The gravitational field equations in Einstein gravity and
modifications consist very sophisticate systems of nonlinear partial
differential equation (PDE) which can be solved in general form only for
some special ansatz (for instance, with diagonal metrics depending on 1-2
variables). A surprising and very important decoupling property of such PDE,
and generalizations to geometric flow evolution equations, was found with
respect to certain classes of nonholonomic frames with associated
N--connection structure. Such frames can be naturally defined, for instance,
for a class of nonholonomic splitting 2+2 splitting in general relativity
and any 2, or 3 + 2 + 2+2+.... decomposition with formal fibred structure,
up to corresponding frame transforms and deformation of connections, in
various modified gravity (with noncommutative, almost K\"{a}hler, Finsler
type \ variables etc). In result of such a decoupling, one obtains such sub--systems of PDE
which  can be integrated, i.e. solved, in very general forms,
for various classes of generic off--diagonal metrics (which can not be
diagonalized via frame transforms)\ and generalized connections with
nontrivial torsion, see details in section \ref{ssdecoupl} and Refs. \cite%
{v4,ve34,ve35,veprep2}. Imposing additional constraints, we can construct
very general classes of solutions for the torsionless and metric compatible
Levi--Civita connection.

The idea of general decoupling of gravitational field equations in Einstein,
string and Finsler gravity was communicated in 1998 at a conference in
Poland \cite{vcp2}, see also a more rigorous mathematical approach in \cite%
{vcp7}. The first examples of different classes of solutions were presented
in high influence score journals Annals of Physics (NY) and JHEP journals,
see \cite{ve6,ve7}. A number of new classes and possible physically
important off--diagonal solutions with ellipsoid/ toroidal symmetries and/or
wormhole, solitons, Dirac waves and nontrivial Einstein--Yang--Mills--Higgs
configurations, cosmological solutions etc were studied in Refs. \cite%
{ve14,ve15,ve37,ve39}.

The so--called anholonomic deformation method of constructing exact
solutions in commutative and noncommutative gravity and Ricci evolution
theories is perhaps the most general one for ''geometric'' generating of
exact solutions, see a number of additional examples in Refs. \cite%
{ve25,ve26,ve27},\cite{ve10}--\cite{ve14},\cite{ve19,v8,v9}. Parts I and II
in collection of works \cite{vb1} contain both geometric details and
examples for solutions in generalize metric--affine and
Lagrange--Finsler--affine gravity theories, noncommutative gravity, extra
dimension models etc. The possibility to derive off--diagonal solutions with
anisotropic scaling, off--diagonal parametric evolution, dependence on
generating and integration functions and parameters seem to be very important
in elaborating new models of covariant renormalizable theories of quantum
gravity \cite{v10,ve48}.

\item \emph{Warped off--diagonal wormhole configurations, flux tubes and
propagation of black holes in extra--dimensions }\cite%
{ve10,ve11,ve12,ve13,ve14,vb3}.

\textit{Comments:} The geometric methods of constructing solitonic and
pp--wave solutions on off--diagonal generalizations of such spacetimes were
applied also in a collaboration with Prof. D. Singleton, from California
State University at Fresno, USA, and some students from R. Moldova, (2001),
see papers \cite{ve10,ve11,ve12,ve13,ve14} and Parts I and II in monograph %
\cite{vb3}. This direction of research is related to that outlined above in
point 5 and provided explicit examples of application of the methods
mentioned in point 8.

Such results were cited in a series of works on brane gravity because the
applicant and co--authors were able to provide explicit applications of the
anholonomic deformation method of constructing exact solutions with
nonlinear off--diagonal warped interactions, non--compactified extra
dimensions and locally anisotropic gravitational configurations.

\item \emph{Solitonic gravitational hierarchies in Einstein and Finsler
gravity} \cite{ve32,ve30,ve7,ve9,ve12,ve13,ve21,ve41,vb3}.

\textit{Comments:} It was a collaboration with prof. S. Anco from Brock
University, Ontario, Canada, during applicant's visiting professor position
in \ 2006. It was known that the geometry of curve flows on spacetimes with
constrant curvature coefficients encode as bi--Hamilton systems various data
for solitonic hierarchies and corresponding sine--Gordon,
Kadomtzev--Petviashvili and other type solitonic equations. For more general
classes of geometries, such a program was considered less realistic because
of general dependence of Riemann curvature, Ricci and (possible) tensors on
spacetime coordinates.

The applicant used his expertise in generalized Finsler geometry and
nonholonomic deformations of geometric strucutres. The main idea was to
construct from a prescribed Finsler fundamental generating function, i.e.
metric, via corresponding N--connection splitting and frame transform,
following a well defined geometric structure, an auxiliary connection for
which the curvature tensor is determined by constant coefficients with
respect certain classes of N--adapted frames. In such cases, the geometric
data for Finser geometry \ (and various generalizations) can be encoded into
solitonic hierarchies, see Ref. \cite{ve30} (together with Prof. S. Anco).

The conventional N--connection splitting can be considered on (pseudo)
Riemanian (in particular, Einstein) spaces which also allows us to redefine
equivalently the geometric/physical data in terms of necessary type auxiliar
connections. Solitonic hierarchies can be derived similarly as in
Lagrange--Finsler geometry but mimicking such structures on nonholonomic
(pseudo)\ Riemann and effective Einstein--Cartan manifolds completely
determined by the metric structure, see Ref. \cite{ve32}. Such an approach
provides us with a new scheme of solitonic classification of very general
classes of exact solutions in Einstein, Einstein--Finsler and nonholonomic
Ricci flow equations.

This direction is related to series of works with solitonic configurations
in pp--wave spacetimes, solitonic propagation of black holes in extra
dimensions and in modified theories, solitonic wormholes and metric--affine
and/or noncommutative models of solitons in gravity and string/brane models,
fractional solitonic hierarchies etc, see a number of examples in Refs. \cite%
{ve7,ve9,ve12,ve13,ve21,ve41} and Parts I and II in monograph \cite{vb3}.

\item \emph{Principles of Einstein--Finsler gravity and applications} \cite%
{ve23,ve36,vcp8,vb3,vrm3,ve1,ve24,vb1,vb2,v8,v9,ve48,ve49}.

\begin{enumerate}
\item Classification of Lagrange--Finsler-affine spaces.

\item Critical remarks on Finsler gravity theories.

\item On axiomatics of Einstein--Finsler gravity.

\item Exact solutions in (non) commutative Finsler gravity and applications.

\item (Non) commutative Finsler black holes and branes, black rings,
ellipsoids and cosmological solutions.
\end{enumerate}

\textit{Comments:} There were many attempts to develop Finsler
generalizations of special and general relativity theories, see reviews of
results in Ref. \cite{ve23} and Introduction to monograph \cite{vb3}. Here
we note certain constructions by Profs. M. Matsumoto and Y. Takano (Japan)
and J. Horvath (Hungary) who proposed in the 70-80th of previous century
certain analogs of Einstein equations using Finsler connections (for
instance, using the Cartan distinguished connection, d--connection). In
Romania, such approaches were studied on vector/tangent bundles, including
generalized Lagrange spaces, by Acad. R. Miron and Profs. M. Anastasiei, G.
Atanasiu, A. Bejancu and others during 1980-1995.

There were unsolved, for instance, three very important issues which would
prove viability and relation to standard theories of physical models with
Finsler like metrics and N--connection and d--connection structure: \ 1) to
derive exact solutions for Finsler like gravity theories (for instance, what
would be some analogs of Finsler black holes, what kind of cosmological
solutions can be derived and considered for further research); 2) how to
define Finsler spinors; 3) how commutative and noncommutative models of
Finsler gravity can be related to string/brane and noncommutative
geometry/gravity theories. In points 1-4 above, it is sketched how solutions
of such problems were performed in applicant's works, see also references %
\cite{vrm3,ve1} on Finsler -- gauge formulations of gravity, with analogous
Yang--Mills equations for gravity.

In Part I of monograph \cite{vb3}, an important classification of Finsler
spaces and generalizations depending on compatibility of fundamental
geometric structures was elaborated. There were considered various classes
of metric compatible and noncompatible Finsler d--connections with general
nonvanishing torsion structure. Using the anholonommic deformation method
(see point 10 above), there were constructed explicit examples of exact
solutions in generalized Lagrange-Finsler--affine gravity and analyzed
possible physical implications. Here we note Ref. \cite{ve24} for extensions
of nonholonomic gravity and Finsler like theories to nonsymmetric metrics,
see also monographs \cite{vb1,vb2} on supersymmetric/spinor and
noncommutative Finsler modifications of gravity.

In Ref. \cite{ve23}, it was concluded that most closed to standard theories
of physics are the Finsler models with metric compatible d--connections (for
instance, the Cartan, or canonical, d--connection) constructed on tangent
bundle to Lorentz manifolds. Such theories allows us to define spinor and
fermions in form similar to general relativity but on nonholonomic
manifolds/bundles. Finsler--Ricci evolution models can be introduced via
nonholonomic deformations of the (pseudo) Riemannian ones \cite{ve49}.

Last five years, a series of new Finsler gravity papers (by a number of
authors: P. Stavrinos, A. Kouretsis, N. Mavromatos, J. Skakala, F. Girelli,
S. Liberati, L. Sindoni, C. L\" ammerzahl, V. Perlik, G. W. Gibbons and
others) where published in relation to expected Lorentz violations in
quantum gravity, anisotropic effects in modern cosmology etc. In a series of
papers by Zhe Chang and Xin Li (2009-2010), authors proposed that the Chern
d--connection has certain ''unique'' fundamental properties for
generalizations of the Einstein gravity theory. Such constructions were
considered to be less adequate for scenarios related to standard physics
because of generic nonmetricity in Chern's and Berwald's models of Finsler
geometry and gravity, see critical remarks \cite{ve36}.

Applicant's conclusions where that using the canonical d--connection and/or
Cartan's d--connection it is possible to construct Einstein -- Finsler like
theories of gravity on tangent/vector bundles, or on nonholonomic manifolds,
following  the same principles as in general relativity and  the
Ehlers--Pirani--Schild (EPS) axiomatics, see references in \cite{vcp8}.
Various important issues on modified dispersion relations, Finsler branes
and noncommutative black holes, models of quantum gravity etc are considered
in Refs. \cite{v8,v9,ve48}.

\item \emph{Stability of nonholonomic gravity and geometric flows with
nonsymmetric metrics and generalized connection structures} \cite%
{ve24,ve27,ve28}.

\textit{Comments:} The applicant extended his research activity to
geometries and physical models with nonsymmetric metrics after two his
visits to Perimeter Insitute, Canada (in 2007-2008, hosted by Prof. J. W.
Moffat, an expert in such directions, beginning 70th). Such theories were
orginally proposed by A. Einstein and L. P. Eisenhardt (1925-1945 and
1951-1952). There were pulished two papers with critical remarks on
perspectives in physics for such a direction (by T. Damour, S. Deser and J.
McCarthy, 1993, and T. Prokopec and W. Valkenburg, 2006) because of
un--physical modes and un--stability of some models. It should be noted that
in 1995 an improved model with nonintegrable constants was elaborated by J. L%
\'{e}gar\'{e} and J. W. Moffat. Nevertheless, questions on stability had to
be solved.

The applicant addressed the problem of nonsymmetric metrics in gravity from
view point of nonholomic geometric flows characterized by nonsymmetric Ricci
tensors \cite{ve27}. In such cases, under evolution, nonsymmetric components
of metrics appear naturally which results also in nonsymmetric Ricci soliton
configurations as certain equilibrium states. There were constructed
explicit classes of exact solutions with ''nonsymmetric'' ellipsoids which
are stable as deformations of black hole solutions \cite{ve28}. That was
possible by adapting the constructions to certain nonholonomic frames with
N--connection structure. In a more general context, such theories can be
re--written in almost K\"{a}hler and/or Lagrange--Finsler variables \cite%
{ve24} which allows us to study various geometric evolution models with
symmetric and nonsymmetric metrics and connections and perform deformation
quantization, see next point.

\item \emph{Deformation, A-brane and two-connection and gauge like
quantization of almost K\"{a}hler models of Einstein gravity and
modifications } \cite{v7,ve22,ve22b,vl3,ve18,ve32,v5,ve29,ve34,ve38,ve40}.

\begin{enumerate}
\item Almost K\"{a}hler and Lagrange--Finsler variables in geometric
mechanics and gravity theories.

\item Deformation quantization of generalized Lagrange--Finsler and
Hamilton--Cartan theories.

\item Fedosov quantization of Einstein gravity and modifications.

\item A--brane quantization of gravity.

\item Two--connection quantization of Einstein, loops, and gauge gravity
theories.
\end{enumerate}

\textit{Comments:} It is of primary importance in modern physics to
formulate a viable model of quantum gravity (QG). Various ideas, approaches
and techniques were proposed but up till present it is far to say that we
could overcome the problems arising in each quantization scheme. Gravity is
a generic nonlinear theory; not having a well defined mathematical branch of
nonlinear functional analysis, it is not possible to formulate a unique and
rigorous scheme; we still have to search for new experimental data and
relate the constructions to phenomenological models in modern cosmology and
high energy physics.

During last 7 years, the applicant published in high influence score
journals a series of papers on geometric methods in quantum gravity:\ The
first direction he addressed was that on deformation (Fedosov) quantization
of Lagrange--Finsler and gravity theories with nonholonomic variables \cite%
{ve22b,vl3}. The main idea was to use some very important results (due to A.
Karabegov and M. Schlichenmeier, 2001) on deformation quantization (DQ) of
almost K\"{a}hler geometries. Reformulating Lagrange--Finsler geometries in
almost symplectic/ K\"{a}hler variables, the scheme of DQ can be naturally
extended to various spaces admitting formal such parametrizations. In Refs. %
\cite{v7,ve22}, the approach was extended to gravitational theories by
prescribing a corresponding N--connection structure which allows to define
some effective almost K\"{a}ehler variables. So, the Fedosov method, in
nonholonomic variables, can be applied to quantize the Einstein and modified
theories in a sense of the DQ paradigm.

A series of results on DQ of Lagrange and Hamilton--Cartan geometries were
obtained in collaboration with Prof. F. Etayo and Dr. R. Santamaria
(University of Cantabria, Santander, Spain; 2005), see \cite{ve18}. When the
applicant, being at Fields Institute at Toronto (Canada), got also an
associated professor position at UAIC he performed a common research with
Prof.\ M. Anastasiei \cite{ve32}. It should be noted here that for Hamilton
configurations on co--tangent bundle, the geometry of phase space posses
additional simplectic symmetries which result in a very complex structure of
induced N--connections and linear connections. The DQ scheme has to be
applied in a quite different form for Lagrange spaces, i.e. on tangent
bundles, and for Hamilton spaces, or any other geometries on co--tangent
bundles. In the last case, a more advanced geometric techniques adapted to
Legandre transforms and almost simplectic structure had to be elaborated.
Recently, the DQ formalism was generalized fractional derivative geometries
and fractional mechanics and gravity, see \cite{ve40}.

Nevertheless, the DQ scheme is still not considered as a generally accepted
procedure with perturbative limits for operators acting on Hilbert spaces
etc. For instance, E. Witten and S. Gukov (2007) elaborated an alternative
formalism (the so--called brane quantization with A-model complexification).
In \cite{v5}, it was proved that the Einstein gravity in almost K\"{a}hler
variables can be quantized following the A--model method. Possible
connections to other approaches were analyzed in \cite{ve29} (for loop
gravity\ with Ashtekar--Barbero variables determined by Finsler like
connections) and in \cite{ve34,ve38} for the so--called bi--connection
formalism and perturbative quantization of gauge gravity models.

\item \emph{Covariant renormalizable anisotropic theories and exact
solutions in gravity }\cite{v10,ve48,ve34,ve38}.

\begin{enumerate}
\item Modified dispersions, generalized pseudo--Finsler structures and Ho%
\v{r}ava--Lifshitz theories on tangent bundles.

\item Covariant renormalizable models for generic off--diagonal spacetimes
and anisotropically modified gravity.
\end{enumerate}

\textit{Comments:} The Newton gravitational constant for four dimensional
interactions results in a generic non--renormalizability of the general
relativity theory. In the pervious point, we considered various schemes of
geometric, non--perturbative and/or gauge like quantization but those
constructions do not solve the problem of constructing a viable model of QG
with a perturbative scheme without divergences from the ultraviolet region
in momentum space (such methods are requested by phenomenology particle
physics and analysis of possible implications in modern cosmolgoy). A recent
approach to QG (the so--called Ho\v{r}ava--Lifshitz models, 2009) is
developed with nonholmogeneous anisotroic scaling of space and time like
variables which allow to develop certain covariant renormalization schemes
(in \cite{v10}, we followed certain ideas due to S. Odintsov, S. Nojiri etc,
2010).

Various models of QG, including those with anisotropic configurations, are
with modified dispersion relations which, in their turn, can be associated
with certain classes of Finsler fundamental generating functions. In Ref. %
\cite{ve48}, we developed a formalism for perturbative quantization of such
Ho\v{r}ava--Finsler models. In both cases, for constructions from the last
two mentioned papers, a crucial role in the quantization procedure is played
by the type of nonholonomic constraints, generating functions and parameters
which are involved in some families of generic off--diagonal solutions of
Einstein equations and generalizations (see point 8 above). In \cite%
{v10,ve48} and \cite{ve34,ve38}, we proved that the nonlinear gravitational
dynamics and corresponding nonholonomic constraints can such way
parametrized when certain ''remormalizable'' configurations survive in an
anisotropic form for which a covariant Ho\v{r}ava--Lifshitz \ quantization
formalism can be applied.

\item \emph{Nonholonomic Ricci flows evolution, thermodynamical
characteristics in geometric mechanics and (analogous) gravity, and
noncommutative geometry } \cite%
{v4,ve6,ve27,ve31,veprep1,ve21,ve22c,vl4,ve25,ve27}.

\begin{enumerate}
\item Generalization of Perelman's functionals and Hamilton's equations for
nonholonomic Ricci flows.

\item Analogous statistical and thermodynamic values for evolutions of
Lagrange--Finsler geometries and analogous gravity theories.

\item Nonholonomic Ricci solitons, exact solutions in gravity, and symmetric
and nonsymmetric metrics.

\item Geometric evolution of pp--wave and Taub NUT spaces.

\item Nonholonomic Dirac operators, distinguished spectral triples and
evolution of models of noncommutative geometry and gravity theories.
\end{enumerate}

\textit{Comments:} One of the most remarkable results in modern mathematics,
and physics, is the proof of the Poincar\'{e} conjecture by Grisha Perelman
(2002-2003) following methods of the theory of Ricci flows (1982). Those
constructions were originally considered for evolution of Riemannian and/or K%
\"{a}hler metrics using the Levi--Civita connection.

The applicant became interested in geometric analysis and possible
applications in physics beginning 2005 when he was with a sabbatical
professor position in Madrid, Spain. His idea was to consider additional
nonholonomic constraints on Ricci flows of/on (pseudo) Riemannian and/or
vector bundles and study geometric evolution of systems with a more complex
geometric structure, as well related modifications of physically important
models \cite{v4}. Such constructions allow us to study evolution, for
instance, of a (pseudo) Riemannian geometry into commutative and
noncommutative geometries \cite{v6}, with symmetric and nonsymmetric metrics
and connections \cite{ve27}, Lagrange--Finsler geometries \cite{ve31},
fractional derivative geometric evolution \cite{veprep1}. It is an important
task for further research to study subjects related to geometric flows and
renormalizations, noncommutative and supersymmetric models of evolution,
exact solutions for stationary Ricci solition configurations and modified
gravity theories, possible applications in modern cosmology and astrophysics
etc.

In the theory of nonholonomic Ricci flows, the key constructions are related
to scenarios of adapting the evolution to N--connection structure in a form
preserving certain \ important geometric/physical values and properties. For
instance (in Refs. \cite{ve21,ve22c,vl4,ve25,ve27}), there were analyzed
various classes of solutions for geometric flows of three and four
dimensional Taub NUT spaces, pp--wave and solitonic deformations of the
Schwarzschild solution. Such configurations, even in geometric mechanics are
characterized by analogous thermodynamics values derived from nonholonomic
versions of Perelman's functionals and associated entropy.
\end{enumerate}

\subsection{Visibility of scientific contributions}

Beginning 1994, he published in above mentioned strategic and main 15
directions more than 60 scientific articles in high influence score, and top
ISI journals, and three monographs with positive reviews in MathSciNet
and/or Zentralblatt, see Refs.\cite{v1}--\cite{ve49} and, additionally, \cite%
{vl1}--\cite{vcp8} in Chapter \ref{ch3}. Totally, there are found in
arXiv.org and inspirehep.net more than 120 scientific works and preprints
with details of computations and alternative ideas and constructions. There
are mentioned in Web of Science more than 100 citations (by 60, there \ are
listed in eligibility files attached to this Thesis).

The bulk of most important applicant's publications are in mathematical
physics journals: Journal of Mathematical Physics (10 papers), Int. J. Geom.
Meth. Mod. Phys. (6 papers), J. Geom. Phys. (2 papers) etc, and
theoretical/particle physics journals: Class. Quant. Grav. (5 papers), Nucl.
Phys. B (2 papers), JHEP (2 papes), Annals Phys. NY (2 papers), Phys. Lett.
A and B (4 papers), Int. J. Theor. Phys. (8 papers) etc.

\begin{enumerate}
\item As results of International Competitions the applicant got:

\begin{itemize}
\item three NATO/DAAD senior researcher fellowships for Portugal and
Germany, 2001-2004

\item four visiting professor fellowships in Greece, USA and Canada (2001,
2002, 2005-2006)

\item a sabbatical professor fellowship in Spain, 2004-2005

\item a research grant of R. Moldova government, 2000-2001

\item a three years Romanian Government Grant IDEI,
PN-II-ID-PCE-2011-3-0256, 2011--2014
\end{itemize}

\item Two visiting researcher positions related to ''scholar at risk
status'' at Fields and Perimeter Institute, Canada (2006-2008) and other
Universities and Research Institutes in different Western Countries; the
applicant had a specific research activity derived from his claims of
political refugee status from the ''communist R. Moldova'' during 2001-2009.
Here it should be noted some important visits at ICTP, Trieste, Italy
(1999), "I. Newton" Mathematical Institute at University of Cambridge, UK
(1999) and a recent visit at Albert Einstein Institute, Max Plank Institute,
Potsdam, Germany - October, 2010.

\item He got support (in the bulk complete, for travel, accommodations,
honorary etc) as an invited lecturer and talks for more than 100 conferences
and visits in USA, UK, Germany, Italy, France, Spain, Portugal, Greece,
Belgium, Austria, Luxembourg, Norway, Turkey, Poland, Romania etc (certain
relevant details are presented in Publication List for the file related to
this Habilitation Thesis). We also attach a list of last seven years
conferences and typical proceedings at the end of Chapter \ref{ch3}, see
respectively \cite{vcs1}--\cite{vcs40} and \cite{vcb1}--\cite{vcp8}.

\item \textit{Competitions of Articles:} During 2009-2011, CNCSIS accepted
as the best by 14 author's articles with grants about 900 E (''red'' 4
articles, \cite{v8,v9,v10,ve36}) and 450 E (''yellow'' 4 articles, \cite%
{v5,ve35,ve38,ve42}) and 110 E (''blue'' 6 articles, in 2009, \cite%
{v6,v7,ve28,ve30,ve31,ve32}).\footnote{%
For instance, see the list for 2011, numbers 1250-1252,\newline
http://uefiscdi.gov.ro/userfiles/file/PREMIERE\_ARTICOLE/ articole\%202011/%
\newline
evaluare/ REZUTATE \%20noiembrie\%20ACTUALIZAT\%2022\%20DECEMBRIE.pdf}
\end{enumerate}

Finally, it should be noted that the applicant's mobility was very
important and necessary for his research and collaborations.

\subsection{Eligibility, minimal standards and recent activity}

The applicant's research activity and main publications can be considered
by the Commission of Mathematics (a similar mathematical one evaluated
positively applicant's application for a Grant IDEI, in 2011), or by the
Commission of Physics, at CNATDCU, Romania. It should be taken into account
the multi-disciplinary character of research on mathematical physics. There
are a bit different standards for eligibility and evaluation of minimal
standards for such Commissions. For instance, it is not allowed to include
for consideration by mathematicians the publications in Int. J. Theor.
Phys., Rep. Phys. and other journals with less than 0.5 absolute influence
score. There are requested at least 12 citations in allowed journals. For
physicists, a series of journals with score higher than 0.3 became admissible but there are requested
more than 40 citations in an extended class of allowed journals (for
experimental and phenomenological physics journals, the number of
co--authors the number of publications per year are much higher then similar
ones in mathematics and applications and this give rise "statistically" to a grater number of citations).

We note here that for the Commission of Mathematics for habilitation of
university professors and senior researchers of grade 1, it might be
computed (for publications in relevant ''absolute influence score''
journals) this conventional ''eligibility triple'' with corresponding
(points for all articles; articles last 7 years; number of citations) =
(55.9; 28.74; 53) which is higher than respective minimal standards (5; 2.5;
12) - there are considered 45 published articles. Such details, explanations
and calculus are given in the requested evaluation files. As a matter of
principle, the applicant became eligible to compete for the most higher
positions of university professor/ senior researcher CS 1, in Romania, by 1997-1998. Similar data for the Commission of Physics, for 55 articles, can be computed \ (59.19; 31.6; 61), which is also higher than the corresponding
minimal eligibility standards (5; 5; 40).

All evaluated (and the bulk cited in this thesis) articles got positive
reviews in MathSciNet and Zentralblatt (one of them, on nonholonomic Ricci
flows, was appreciated in ''Nature'' being listed in World of Science,
Scopus with PDFs dubbed in inspirehep.net and arxiv.org, where a number of
citations can be found and checked.\footnote{%
The link to Nature Physics, vol. 4., issue 5, pp. 343 (2008) is \newline
http://www.nature.com/nphys/journal/v4/n5/full/nphys948.html\#Constant-flow
\par
\vskip3pt For conveniences, it is presented here the text:
\par
\vskip5pt \textit{Research Highlights}: Nature Physics 4, 343 (2008),
doi:10.1038/nphys948 \newline
J. Math. Phys. 49, 043504 (2008)
\par
\vskip3pt Only once, apparently, did Gregorio Ricci-Curbastro publish under
the name Ricci. That was in 1900, but the paper --- entitled Methodes de
calcul diff\'{e}rentiel absolu et leurs applications, and co-authored with
his former student Tullio Levi-Civita --- became the pioneering work on the
calculus of tensors, a calculus also used by Albert Einstein in his theory
of general relativity.
\par
Ricci-Curbastro's short name stuck, and Ricci flow' is the name given to one
of the mathematical tools arising from his work. That tool has become known
to a wider audience as a central element in Grigori Perelman's proof of the
Poincar\'{e} conjecture.
\par
Sergiu Vacaru now takes Perelmans work further, going beyond geometrical
objects and into the domain of physics with a generalized form of the
Ricci-flow theory. In the second paper of a series devoted to these
so-called non-holonomic Ricci flows, Vacaru shows how the theory may be
applied in tackling physical problems, such as in einsteinian gravity and
lagrangian mechanics.}

During 2009--2011, with affiliation at University Alexandru Ioan Cuza at Ia%
\c{s}i, Romania, he published almost 25 top ISI papers on mathematics and
physics (more than a half of them being in the ''red/yelow/ blue'' category
for Competition of Articles) and got financial support from organizers for
short term visits and invited lectures and talks (more than 30 ones). This
would allow the applicant to extend and develop his experience on research
and teaching in North America and Western Europe, Romania and former URSS,
on supervision PhD and master theses, elaborating monographs and textbooks
for university students and delivering lectures and seminars in English,
Romanian and Russian. \newpage

\section{A "Geometric" Survey of Selected Results}

\label{sgsmimr} The goal of this section is to provide a selection of
results from 10 most relevant applicant's publications \cite{v1}--\cite{v10}
containing explicit definitions, theorems and main formulas.\footnote{%
Cumbersome proofs and references to other authors are omitted. Nevertheless,
we shall provide a series of "simplest" examples in order to familiarize
readers with such geometric methods. Some "overlap" in denotations and
formulas will be possible because they exist in the original published
works. In abstract form, this is used for simplifying proofs, for instance,
in some models of commutative and noncommutative geometry.} Such a brief review
is oriented to advanced researchers and experts on mathematical physics and
geometric methods in physics.

\subsection{Nonholonomic Ricci evolution}

\label{ssprelim}

Currently a set of most important and fascinating problems in modern
geometry and physics involves the task to find canonical (optimal) metric
and connection structures on manifolds, state possible topological
configurations and analyze related physical implications. In the past almost
three decades, the Ricci flow theory has addressed such issues for
Riemannian manifolds. How to formulate and generalize the constructions for
non--Riemannian manifolds and physical theories, it is a challenging topic
in mathematics and physics. The typical examples come from string/brane
gravity containing nontrivial torsion fields and from modern mechanics and
field theory geometrized in terms of symplectic and/or generalized Finsler
(Lagrange or Hamilton) structures.

The goal of this subsection is to investigate the geometry of evolution
equations under non--integrable (equivalently, nonholonmic/ anholonomic)
constraints resulting in nonholonomic Riemann--Cartan and generalized
Finsler--Lagrange configurations.

\subsubsection{Preliminaries: \ nonholonomic manifolds and bundles}

A nonholonomic manifold is defined as a pair $\mathbf{V=}(M,\mathcal{D}),$
where $M$ is a manifold\footnote{%
we assume that the geometric/physical spaces are smooth and orientable
manifolds} and $\mathcal{D}$ is a non-integrable distribution on $M.$ For
certain important geometric and physical cases, one considers N--anholonomic
manifolds when the nonholonomic structure of $\mathbf{V}$ is established by
a nonlinear connection (N--connection), equivalently, a Whitney
decomposition of the tangent space into conventional horizontal (h)
subspace, $\left( h\mathbf{V}\right) ,$ and vertical (v) subspace, $\left( v%
\mathbf{V}\right) ,$\footnote{%
Usually, we consider a $(n+m)$--dimensional manifold $\mathbf{V,}$ with $%
n\geq 2$ and $m\geq 1$ (equivalently called to be a physical and/or
geometric space). In a particular case, $\mathbf{V=}TM,$ with $n=m$ (i.e. a
tangent bundle), or $\mathbf{V=E}=(E,M),$ $\dim M=n,$ is a vector bundle on $%
M,$ with total space $E.$ We suppose that a manifold $\mathbf{V}$ may be
provided with a local fibred structure into conventional ''horizontal'' and
''vertical'' directions. The local coordinates on $\mathbf{V}$ are denoted
in the form $u=(x,y),$ or $u^{\alpha }=\left( x^{i},y^{a}\right) ,$ where
the ''horizontal'' indices run the values $i,j,k,\ldots =1,2,\ldots ,n$ and
the ''vertical'' indices run the values $a,b,c,\ldots =n+1,n+2,\ldots ,n+m.$
\ }
\begin{equation}
T\mathbf{V}=h\mathbf{V}\oplus v\mathbf{V}.  \label{2whitney}
\end{equation}%
Locally, a N--connection $\mathbf{N}$ is defined by its coefficients $%
N_{i}^{a}(u),$%
\begin{equation}
\mathbf{N}=N_{i}^{a}(u)dx^{i}\otimes \frac{\partial }{\partial y^{a}},
\label{2coeffnc}
\end{equation}%
and states a preferred frame (vielbein) structure
\begin{equation}
\mathbf{e}_{\nu }=\left( \mathbf{e}_{i}=\frac{\partial }{\partial x^{i}}%
-N_{i}^{a}(u)\frac{\partial }{\partial y^{a}},e_{a}=\frac{\partial }{%
\partial y^{a}}\right) ,  \label{2dder}
\end{equation}%
and a dual frame (coframe) structure
\begin{equation}
\mathbf{e}^{\mu }=\left( e^{i}=dx^{i},\mathbf{e}%
^{a}=dy^{a}+N_{i}^{a}(u)dx^{i}\right) .  \label{2ddif}
\end{equation}%
The vielbeins (\ref{2ddif}) satisfy the nonholonomy relations
\begin{equation}
\lbrack \mathbf{e}_{\alpha },\mathbf{e}_{\beta }]=\mathbf{e}_{\alpha }%
\mathbf{e}_{\beta }-\mathbf{e}_{\beta }\mathbf{e}_{\alpha }=W_{\alpha \beta
}^{\gamma }\mathbf{e}_{\gamma }  \label{2anhrel}
\end{equation}%
with (antisymmetric) nontrivial anholonomy coefficients $W_{ia}^{b}=\partial
_{a}N_{i}^{b}$ and $W_{ji}^{a}=\Omega _{ij}^{a},$ where $\Omega _{ij}^{a}=%
\mathbf{e}_{j}\left( N_{i}^{a}\right) -\mathbf{e}_{i}\left( N_{j}^{a}\right)$
are the coefficients of N--connection curvature. The particular holonomic/
integrable case is selected by the integrability conditions $W_{\alpha \beta
}^{\gamma }=0.$

In N--adapted form, the tensor coefficients are defined with respect to
tensor products of vielbeins (\ref{2dder}) and (\ref{2ddif}). They are
called respectively distinguished tensors/ vectors /forms, in brief,
d--tensors, d--vectors, d--forms.

A distinguished connection (d--connection) $\mathbf{D}$ on a
N--anho\-lo\-no\-mic manifold $\mathbf{V}$ is a linear connection conserving
under parallelism the Whitney sum (\ref{2whitney}). In local form, a
d--connection $\mathbf{D}$ is given by its coefficients $\mathbf{\Gamma }_{\
\alpha \beta }^{\gamma }=\left(
L_{jk}^{i},L_{bk}^{a},C_{jc}^{i},C_{bc}^{a}\right) ,$ where $\
^{h}D=(L_{jk}^{i},L_{bk}^{a})$ and $\ ^{v}D=(C_{jc}^{i},C_{bc}^{a})$ are
respectively the covariant h-- and v--derivatives.\footnote{%
We shall use both the coordinate free and local coordinate formulas which is
convenient both to introduce compact denotations and sketch some proofs. The
left up/lower indices will be considered as labels of geometrical objects.
The boldfaced letters will point that the objects (spaces) are adapted
(provided) to (with) N--connection structure.}

The torsion of a d--connection $\mathbf{D=}(\ ^{h}D,\ \ ^{v}D)\mathbf{,}$
for any d--vectors $\mathbf{X=}h\mathbf{X}+v\mathbf{X=\ }^{h}\mathbf{X}+\
^{v}\mathbf{X}$ and $\mathbf{Y=}h\mathbf{Y}+v\mathbf{Y,}$ is defined by the
d--tensor field
\begin{equation}
\mathbf{T(X,Y)\doteqdot \mathbf{D}_{\mathbf{X}}Y-D}_{\mathbf{Y}}\mathbf{%
X-[X,Y],}  \label{2tors1}
\end{equation}%
with a corresponding N--adapted decomposition into {\small
\begin{eqnarray}
\mathbf{T(X,Y)} &=&\{h\mathbf{T}(h\mathbf{X},h\mathbf{Y}),h\mathbf{T}(h%
\mathbf{X},v\mathbf{Y}),h\mathbf{T}(v\mathbf{X},h\mathbf{Y}),h\mathbf{T}(v%
\mathbf{X},v\mathbf{Y}),  \notag \\
&&v\mathbf{T}(h\mathbf{X},h\mathbf{Y}),v\mathbf{T}(h\mathbf{X},v\mathbf{Y}),v%
\mathbf{T}(v\mathbf{X},h\mathbf{Y}),v\mathbf{T}v\mathbf{X},v\mathbf{Y})\}.
\label{2tors2}
\end{eqnarray}%
} The nontrivial N--adapted coefficients $\mathbf{T}=\{\mathbf{T}_{~\beta
\gamma }^{\alpha }=-\mathbf{T}_{~\gamma \beta }^{\alpha
}=(T_{~jk}^{i},T_{~ja}^{i},$ $T_{~jk}^{a},T_{~ja}^{b},T_{~ca}^{b})$ are
given in Refs. \cite{v3,v4}.\footnote{%
We omit repeating of cumbersome local formulas but emphasize the h-- and
v--decomposition of geometrical objects which is important for our further
constructions.}

The curvature of a d--connection $\mathbf{D}$ is defined
\begin{equation}
\mathbf{R(X,Y)\doteqdot \mathbf{D}_{\mathbf{X}}\mathbf{D}_{\mathbf{Y}}-D}_{%
\mathbf{Y}}\mathbf{D}_{\mathbf{X}}\mathbf{-D}_{\mathbf{[X,Y]}},
\label{2curv1}
\end{equation}%
with N--adapted decomposition%
\begin{eqnarray}
\mathbf{R(X,Y)Z} &=&\{\mathbf{R}(h\mathbf{X},h\mathbf{Y})h\mathbf{Z,R}(h%
\mathbf{X},v\mathbf{Y})h\mathbf{Z,R}(v\mathbf{X},h\mathbf{Y})h\mathbf{Z},
\notag \\
&&\mathbf{R}(v\mathbf{X},v\mathbf{Y})h\mathbf{Z},\mathbf{R}(h\mathbf{X},h%
\mathbf{Y})v\mathbf{Z,R}(h\mathbf{X},v\mathbf{Y})v\mathbf{Z},  \notag \\
&&\mathbf{R}(v\mathbf{X},h\mathbf{Y})v\mathbf{Z},\mathbf{R}(v\mathbf{X},v%
\mathbf{Y})v\mathbf{Z}\}.  \label{2curv2}
\end{eqnarray}%
The formulas for local N--adapted components and their symmetries, of the
d--torsion and d--curvature, can be computed by introducing $\mathbf{X}=%
\mathbf{e}_{\alpha },$ $\mathbf{Y}=\mathbf{e}_{\beta }$ and $\mathbf{Z}=%
\mathbf{e}_{\gamma }$ in (\ref{2curv2}). The formulas for nontrivial
N--adapted coefficients
\begin{equation*}
\mathbf{R=\{\mathbf{R}_{\ \beta \gamma \delta }^{\alpha }=}\left( R_{\
hjk}^{i}\mathbf{,}R_{\ bjk}^{a}\mathbf{,}R_{\ hja}^{i}\mathbf{,}R_{\ bja}^{c}%
\mathbf{,}R_{\ hba}^{i}\mathbf{,}R_{\ bea}^{c}\right) \mathbf{\}}
\end{equation*}%
are given in \cite{v3,v4}. Contracting the first and forth indices $\mathbf{%
\mathbf{R}_{\ \beta \gamma }=\mathbf{R}_{\ \beta \gamma \alpha }^{\alpha }}$%
, one gets the N--adapted coefficients for the Ricci tensor%
\begin{equation}
\mathbf{Ric\doteqdot \{\mathbf{R}_{\beta \gamma }=}\left(
R_{ij},R_{ia},R_{ai},R_{ab}\right) \mathbf{\}.}  \label{2dricci}
\end{equation}

A distinguished metric (in brief, d--metric) on a N--anholo\-nom\-ic
manifold $\mathbf{V}$ is a second rank symmetric tensor $\mathbf{g}$ which
in N--adapted form is written
\begin{equation}
\mathbf{g}=\ g_{ij}(x,y)\ e^{i}\otimes e^{j}+\ g_{ab}(x,y)\ \mathbf{e}%
^{a}\otimes \mathbf{e}^{b}.  \label{2m1}
\end{equation}%
In brief, we write $\mathbf{g=}hg\mathbf{\oplus _{N}}vg=[\ ^{h}g,\ ^{v}g].$
With respect to coordinate co--frames, the metric $\mathbf{g}$ can be
written in the form
\begin{equation}
\ \mathbf{g}=\underline{g}_{\alpha \beta }\left( u\right) du^{\alpha
}\otimes du^{\beta }  \label{2metr}
\end{equation}%
where%
\begin{equation}
\underline{g}_{\alpha \beta }=\left[
\begin{array}{cc}
g_{ij}+N_{i}^{a}N_{j}^{b}h_{ab} & N_{j}^{e}g_{ae} \\
N_{i}^{e}g_{be} & g_{ab}%
\end{array}%
\right] .  \label{2ansatz}
\end{equation}%
A d--connection $\mathbf{D}$ is compatible to a metric $\mathbf{g}$ if $%
\mathbf{Dg=}0.$

There are two classes of preferred linear connections defined by the
coefficients $\{\underline{g}_{\alpha \beta }\}$ of a metric structure $%
\mathbf{g}$ (equivalently, by the coefficients of corresponding d--metric $%
\left( g_{ij},\ h_{ab}\right) $ and N--connection $N_{i}^{a}:$ we shall
emphasize the functional dependence on such coefficients in some formulas):

\begin{itemize}
\item The unique metric compatible and torsionless Levi Civita connection $%
\nabla =\{\ _{\shortmid }\Gamma _{\ \alpha \beta }^{\gamma }(g_{ij},\
h_{ab},N_{i}^{a})\},$ for which$\ _{\shortmid }T_{~\beta \gamma }^{\alpha
}=0 $ and $\nabla \mathbf{g=}0.$ This is not a d--connection because it does
not preserve under parallelism the N--connec\-tion splitting (\ref{2whitney}%
). The curvature and Ricci tensors of $\ \nabla ,$ denoted $\ _{\shortmid }R%
\mathbf{_{\ \beta \gamma \delta }^{\alpha }}$ and $\ _{\shortmid }R\mathbf{%
_{\ \beta \gamma },}$ are computed respectively by formulas (\ref{2curv1})
and (\ref{2dricci}) when $\mathbf{D\rightarrow \nabla .}$

\item The unique metric canonical d--connection $\widehat{\mathbf{D}}$ $=\{%
\widehat{\mathbf{\Gamma }}_{\ \alpha \beta }^{\gamma }(g_{ij},\
h_{ab},N_{i}^{a})\}$ is defined by the conditions $\widehat{\mathbf{D}}%
\mathbf{g=}0$ and $h\widehat{\mathbf{T}}(hX,$ $hY)=0$ and $\mathbf{\ }v%
\widehat{\mathbf{T}}(vX,$ $\mathbf{\ }vY)$ $=0.$ The N--adapted coefficients
$\widehat{\mathbf{\Gamma }}_{\ \alpha \beta }^{\gamma }=\left( \widehat{L}%
_{jk}^{i},\widehat{L}_{bk}^{a},\widehat{C}_{jc}^{i},\widehat{C}%
_{bc}^{a}\right) $ and the deformation tensor $\ \ _{\shortmid }Z_{\ \alpha
\beta }^{\gamma },$ \ when $\nabla =\widehat{\mathbf{D}}+\widehat{Z},$ $\ $%
\begin{equation*}
_{\shortmid }\Gamma _{\ \alpha \beta }^{\gamma }(g_{ij},\ g_{ab},N_{i}^{a})=%
\widehat{\mathbf{\Gamma }}_{\ \alpha \beta }^{\gamma }(g_{ij},\
g_{ab},N_{i}^{a})+\ _{\shortmid }Z_{\ \alpha \beta }^{\gamma }(g_{ij},\
g_{ab},N_{i}^{a})
\end{equation*}%
for
\begin{eqnarray}
\widehat{L}_{jk}^{i} &=&\frac{1}{2}g^{ir}\left(
e_{k}g_{jr}+e_{j}g_{kr}-e_{r}g_{jk}\right) ,  \label{2candcon} \\
\widehat{L}_{bk}^{a} &=&e_{b}(N_{k}^{a})+\frac{1}{2}g^{ac}\left(
e_{k}g_{bc}-g_{dc}\ e_{b}N_{k}^{d}-g_{db}\ e_{c}N_{k}^{d}\right) ,  \notag \\
\widehat{C}_{jc}^{i} &=&\frac{1}{2}g^{ik}e_{c}g_{jk},\ \widehat{C}_{bc}^{a}=%
\frac{1}{2}g^{ad}\left( e_{c}g_{bd}+e_{c}g_{cd}-e_{d}g_{bc}\right) .  \notag
\end{eqnarray}%
and $\widehat{Z}=\{\ _{\shortmid }Z_{\ \alpha \beta }^{\gamma }\}$ given in %
\cite{v3,v4}.
\end{itemize}

We shall underline symbols or indices of geometrical objects in order to
emphasize that the components/formulas/equations are written with respect to
a local coordinate basis, for instance, $\underline{g}_{\alpha \beta }=g_{%
\underline{\alpha }\underline{\beta }},$ $\underline{\widehat{\mathbf{\Gamma
}}}_{\ \alpha \beta }^{\gamma }=\widehat{\mathbf{\Gamma }}_{\ \underline{%
\alpha }\underline{\beta }}^{\underline{\gamma }},$ $\ _{\shortmid }%
\underline{\Gamma }_{\ \alpha \beta }^{\gamma }=$ $\ _{\shortmid }\Gamma _{\
\underline{\alpha }\underline{\beta }}^{\underline{\gamma }},$ $\underline{%
\widehat{\mathbf{\mathbf{R}}}}\mathbf{_{\ \beta \gamma }=}\widehat{\mathbf{%
\mathbf{R}}}\mathbf{_{\ \underline{\beta }\underline{\gamma }},...}$

Having prescribed a nonholonomic $n+m$ splitting with coefficients $%
N_{i}^{a} $ on a (semi) Riemannian manifold $\mathbf{V}$ provided with
metric structure $\underline{g}_{\alpha \beta }$ (\ref{2metr}), we can work
with N--adapted frames (\ref{2dder}) and (\ref{2ddif}) and the equivalent
d--metric structure $\left( g_{ij},\ g_{ab}\right) $ (\ref{2m1}). On $%
\mathbf{V,}$ one can be introduced two (equivalent) canonical metric
compatible (both defined by the same metric structure, equivalently, by the
same d--metric and N--connection) linear connections: the Levi Civita
connection $\nabla $ and the canonical d--connection $\widehat{\mathbf{D}}.$
In order to perform geometric constructions in N--adapted form, we have to
work with the connection $\widehat{\mathbf{D}}$ which contains nontrivial
torsion coefficients $\widehat{T}_{~ja}^{i},\widehat{T}_{~jk}^{a},\widehat{T}%
_{~ja}^{b}$ induced by the ''off diagonal'' metric / N--connection
coefficients $N_{i}^{a}$ and their derivatives.

We conclude that the geometry of a N--aholonomic manifold $\mathbf{V}$ can be
described by data $\left\{ g_{ij},\ g_{ab},N_{i}^{a},\nabla \right\} $ or,
equivalently, by data $\left\{ g_{ij},\ g_{ab},N_{i}^{a},\widehat{\mathbf{D}}%
\right\} .$ \ Of course, two different linear connections, even defined by
the same metric structure, are characterized by different Ricci and Riemann
curvatures tensors and curvature scalars. In this works, we shall prefer
N--adapted constructions with $\widehat{\mathbf{D}}$ but also apply $\nabla $
if the proofs for $\widehat{\mathbf{D}}$ will be cumbersome. The idea is
that if a geometric Ricci flow construction is well defined for one of the
connections, $\nabla $ or $\widehat{\mathbf{D}},$ it can be equivalently
redefined for the second one by considering the distorsion tensor $\
_{\shortmid }Z_{\ \alpha \beta }^{\gamma }.$

\subsubsection{On nonholonomic evolution equations}

The Ricci flow equations were introduced by R. Hamilton as evolution
equations
\begin{equation}
\frac{\partial \underline{g}_{\alpha \beta }(\chi )}{\partial \chi }=-2\
_{\shortmid }\underline{R}_{\alpha \beta }(\chi )  \label{2heq1}
\end{equation}%
for a set of Riemannian metrics $\underline{g}_{\alpha \beta }(\chi )$ and
corresponding Ricci tensors $\ _{\shortmid }\underline{R}_{\alpha \beta
}(\chi )$ parametrized by a real $\chi .$

The normalized (holonomic) Ricci flows, with respect to the coordinate base $%
\partial _{\underline{\alpha }}=\partial /\partial u^{\underline{\alpha }},$
are described by the equations
\begin{equation}
\frac{\partial }{\partial \chi }g_{\underline{\alpha }\underline{\beta }%
}=-2\ _{\shortmid }R_{\underline{\alpha }\underline{\beta }}+\frac{2r}{5}g_{%
\underline{\alpha }\underline{\beta }},  \label{2feq}
\end{equation}%
where the normalizing factor $r=\int \ _{\shortmid }RdV/dV$ is introduced in
order to preserve the volume $V.$ For N--anholonomic Ricci flows, the
coefficients $g_{\underline{\alpha }\underline{\beta }}$ are parametrized in
the form (\ref{2ansatz}).

With respect to the N--adapted frames (\ref{2dder}) and (\ref{2ddif}), when
\begin{equation*}
\mathbf{e}_{\alpha }(\chi )=\mathbf{e}_{\alpha }^{\ \underline{\alpha }%
}(\chi )\ \partial _{\underline{\alpha }}\mbox{\ and \ }\mathbf{e}^{\alpha
}(\chi )=\mathbf{e}_{\ \underline{\alpha }}^{\alpha }(\chi )du^{\underline{%
\alpha }},
\end{equation*}%
the frame transforms are respectively parametrized in the form%
\begin{eqnarray}
\mathbf{e}_{\alpha }^{\ \underline{\alpha }}(\chi ) &=&\left[
\begin{array}{cc}
e_{i}^{\ \underline{i}}=\delta _{i}^{\underline{i}} & e_{i}^{\ \underline{a}%
}=N_{i}^{b}(\chi )\ \delta _{b}^{\underline{a}} \\
e_{a}^{\ \underline{i}}=0 & e_{a}^{\ \underline{a}}=\delta _{a}^{\underline{a%
}}%
\end{array}%
\right] ,  \label{2ft} \\
\mathbf{e}_{\ \underline{\alpha }}^{\alpha }(\chi ) &=&\left[
\begin{array}{cc}
e_{\ \underline{i}}^{i}=\delta _{\underline{i}}^{i} & e_{\ \underline{i}%
}^{b}=-N_{k}^{b}(\chi )\ \delta _{\underline{i}}^{k} \\
e_{\ \underline{a}}^{i}=0 & e_{\ \underline{a}}^{a}=\delta _{\underline{a}%
}^{a}%
\end{array}%
\right] ,  \notag
\end{eqnarray}%
where $\delta _{\underline{i}}^{i}$ is the Kronecher symbol.

\begin{definition}
\label{ntr} Nonholonomic deformations of geometric objects (and related
systems of equations) on a N--anholonomic manifold $\mathbf{V}$ are defined
for the same metric structure $\mathbf{g}$ by a set of transforms of
arbitrary frames into N--adapted ones and of the Levi Civita connection $%
\nabla $ into the canonical d--connection $\widehat{\mathbf{D}},$ locally
parametrized in the form
\begin{equation*}
\partial _{\underline{\alpha }}=(\partial _{\underline{i}},\partial _{%
\underline{a}})\rightarrow \mathbf{e}_{\alpha }=(\mathbf{e}_{i},e_{a});\ g_{%
\underline{\alpha }\underline{\beta }}\rightarrow \lbrack
g_{ij},g_{ab},N_{i}^{a}];\ _{\shortmid }\Gamma _{\ \alpha \beta }^{\gamma
}\rightarrow \widehat{\mathbf{\Gamma }}_{\ \alpha \beta }^{\gamma }.
\end{equation*}
\end{definition}

It should be noted that the heuristic arguments presented in this section do
not provide a rigorous proof of evolution equations with $\widehat{\mathbf{D}%
}$ and $\widehat{\mathbf{R}}_{\alpha \beta }$ all defined with respect to
N--adapted frames (\ref{2dder}) and (\ref{2ddif}).\footnote{%
The tensor $\widehat{\mathbf{R}}_{\alpha \beta }$ is not symmetric which
results, in general, in Ricci flows of nonsymmetric metrics.} A rigorous
proof for nonholonomic evolution equations is possible following a
N--adapted variational calculus for the Perelman's functionals presented
(below) for Theorems \ref{2theq1} and \ref{2theveq}.

\subsubsection{Generalized Perelman's functionals}

Following G. Perelman's ideas, the Ricci flow equations can be derived as
gradient flows for some functionals defined by the Levi Civita connection $%
\nabla .$ The functionals are written in the form (we use our system of
denotations)
\begin{eqnarray}
\ _{\shortmid }\mathcal{F}(\mathbf{g},\nabla ,f) &=&\int\limits_{\mathbf{V}%
}\left( \ _{\shortmid }R+\left| \nabla f\right| ^{2}\right) e^{-f}\ dV,
\label{2pfrs} \\
\ _{\shortmid }\mathcal{W}(\mathbf{g},\nabla ,f,\tau ) &=&\int\limits_{%
\mathbf{V}}\left[ \tau \left( \ _{\shortmid }R+\left| \nabla f\right|
\right) ^{2}+f-(n+m)\right] \mu \ dV,  \notag
\end{eqnarray}%
where $dV$ is the volume form of $\ \mathbf{g,}$ integration is taken over
compact $\mathbf{V}$ and $\ _{\shortmid }R$ is the scalar curvature computed
for $\nabla .$ For a parameter $\tau >0,$ we have $\int\nolimits_{\mathbf{V}%
}\mu dV=1$ when $\mu =\left( 4\pi \tau \right) ^{-(n+m)/2}e^{-f}.$ $\ $%
Following this approach, the Ricci flow is considered as a dynamical system
on the space of Riemannian metrics and the functionals $\ _{\shortmid }%
\mathcal{F}$ and $\ _{\shortmid }\mathcal{W}$ are of Lyapunov type. Ricci
flat configurations are defined as ''fixed'' on $\tau $ points of the
corresponding dynamical systems.

The functionals (\ref{2pfrs}) can be also re--defined in equivalent form for
the canonical d--connection, in the case of Lagrange--Finsler spaces. In
this section, we show that the constructions can be generalized for
arbitrary N--anholonomic manifolds, when the gradient flow is constrained to
be adapted to the corresponding N--connection structure.

\begin{claim}
For a set of N--anholonomic manifolds of dimension $n+m,$ the Perelman's
functionals for the canonical d--connection $\widehat{\mathbf{D}}$ are
defined
{\small
\begin{eqnarray}
&&\widehat{\mathcal{F}}(\mathbf{g},\widehat{\mathbf{D}},\widehat{f})
=\int\limits_{\mathbf{V}}\left( \ ^{h}\widehat{R}+\ ^{v}\widehat{R}+\left|
\widehat{\mathbf{D}}\widehat{f}\right| ^{2}\right) e^{-\widehat{f}}\ dV,
\label{2npf1} \\
&&\widehat{\mathcal{W}}(\mathbf{g},\widehat{\mathbf{D}},\widehat{f},\widehat{%
\tau })=\int\limits_{\mathbf{V}}[\widehat{\tau } (\ ^{h}\widehat{R}%
+\ ^{v}\widehat{R}+| ^{h}D\widehat{f}| +| ^{v}D\widehat{f}| ) ^{2}  +\widehat{f}-(n+m)] \widehat{\mu }\ dV],   \notag
\end{eqnarray}
}%
where $dV$ is the volume form of $\ ^{L}\mathbf{g;}$ $R$ and $S$ are
respectively the h- and v--components of the curvature scalar of $\ \widehat{%
\mathbf{D}}$ $\ $when $^{s}\widehat{\mathbf{R}}\doteqdot \mathbf{g}^{\alpha
\beta }\widehat{\mathbf{R}}_{\alpha \beta }=\ ^{h}\widehat{R}+\ ^{v}\widehat{%
R},$ for $\ \widehat{\mathbf{D}}_{\alpha }=(D_{i},D_{a}),$ or $\widehat{%
\mathbf{D}}=(\ ^{h}D,\ ^{v}D)$ when $\left| \widehat{\mathbf{D}}\widehat{f}%
\right| ^{2}=\left| ^{h}D\widehat{f}\right| ^{2}+\left| ^{v}D\widehat{f}%
\right| ^{2},$ and $\widehat{f}$ satisfies $\int\nolimits_{\mathbf{V}}%
\widehat{\mu }dV=1$ for $\widehat{\mu }=\left( 4\pi \tau \right)
^{-(n+m)/2}e^{-\widehat{f}}$ and $\widehat{\tau }>0.$
\end{claim}

Elaborating an N--adapted variational calculus, we shall consider both
variations in the so--called h-- and v--subspaces stated by decompositions (%
\ref{2whitney}). For simplicity, we consider the h--variation $^{h}\delta
g_{ij}=v_{ij},$ the v--variation $^{v}\delta g_{ab}=v_{ab},$ for a fixed
N--connection structure in (\ref{2m1}), and $^{h}\delta \widehat{f}=\ ^{h}f,$
$^{v}\delta \widehat{f}=\ ^{v}f.$

A number of important results in Riemannian geometry can be proved by using
normal coordinates in a point $u_{0}$ and its vicinity. Such constructions
can be performed on a N--anholonomic manifold $\mathbf{V.}$

\begin{proposition}
For any point $u_{0}\in \mathbf{V,}$ there is a system of N--adapted
coordinates for which $\widehat{\mathbf{\Gamma }}_{\ \alpha \beta }^{\gamma
}(u_{0})=0.$
\end{proposition}

\begin{proof}
In the system of normal coordinates in $u_{0},$ for the Levi Civita
connection, when $_{\shortmid }\Gamma _{\ \alpha \beta }^{\gamma }(u_{0})=0,$
we chose $\mathbf{e}_{\alpha }\mathbf{g}_{\beta \gamma }\mid _{u_{0}}=0.$
Following formulas (\ref{2candcon}), for a d--metric (\ref{2m1}),
equivalently (\ref{2metr}), we get $\widehat{\mathbf{\Gamma }}_{\ \alpha
\beta }^{\gamma }(u_{0})=0.\square $
\end{proof}

We generalize for arbitrary N--anholonomic manifolds (see proof in \cite{v4}%
):

\begin{lemma}
\label{2lem1}The first N--adapted variations of (\ref{2npf1}) are given by
\begin{eqnarray}
&&\delta \widehat{\mathcal{F}}(v_{ij},v_{ab},\ ^{h}f,\ ^{v}f)=
\label{2vnpf1} \\
&&\int\limits_{\mathbf{V}}\{[-v^{ij}(\widehat{R}_{ij}+\widehat{D}_{i}%
\widehat{D}_{j}\widehat{f})+(\frac{\ ^{h}v}{2}-\ ^{h}f)\left( 2\ ^{h}\Delta
\widehat{f}-|\ ^{h}D\ \widehat{f}|^{2}\right) +\ ^{h}\widehat{R}]  \notag \\
&&+[-v^{ab}(\widehat{R}_{ab}+\widehat{D}_{a}\widehat{D}_{b}\widehat{f})+(%
\frac{\ ^{v}v}{2}-\ ^{v}f)\left( 2\ ^{v}\Delta \widehat{f}-|\ ^{v}D\
\widehat{f}|^{2}\right) +\ ^{v}\widehat{R}]\}e^{-\widehat{f}}dV  \notag
\end{eqnarray}%
where $^{h}\Delta =\widehat{D}_{i}\widehat{D}^{i}$ and $^{v}\Delta =\widehat{%
D}_{a}\widehat{D}^{a},$ for $\widehat{\Delta }=$ $\ ^{h}\Delta +\ ^{v}\Delta
,$ and $\ ^{h}v=g^{ij}v_{ij},\ ^{v}v=g^{ab}v_{ab}.$
\end{lemma}

\begin{definition}
A d--metric $\ \mathbf{g}$ (\ref{2m1}) evolving by the (nonholonomic) Ricci
flow is called a (nonholonomic) breather if for some $\chi _{1}<\chi _{2}$
and $\alpha >0$ the metrics $\alpha \ \mathbf{g(}\chi _{1}\mathbf{)}$ and $%
\alpha \ \mathbf{g(}\chi _{2}\mathbf{)}$ differ only by a diffeomorphism (in
the N--anholonomic case, preserving the Whitney sum (\ref{2whitney}). The
cases $\alpha $ ($=,<,)>1$ define correspondingly the (steady, shrinking)
expanding breathers.
\end{definition}

The breather properties depend on the type of connections which are used for
definition of Ricci flows. For N--anholonomic manifolds, one can be the
situation when, for instance, the h--component of metric is steady but the
v--component is shrinking.

\subsubsection{Main theorems on nonholonomic Ricci flows}

Following a N--adapted variational calculus for $\widehat{\mathcal{F}}(%
\mathbf{g},\widehat{f}),$ see Lemma \ref{2lem1}, with Laplacian $\widehat{%
\Delta }$ and h- and v--components of the Ricci tensor, $\widehat{R}_{ij}$
and $\widehat{R}_{ab},$ defined by $\widehat{\mathbf{D}}$ and considering
parameter $\tau (\chi ),$ $\partial \tau /\partial \chi =-1$ (for
simplicity, we shall not consider the normalized term and put $\lambda =0),$
one holds

\begin{theorem}
\label{2theq1}The Ricci flows of d--metrics are characterized by evolution
equations
\begin{eqnarray*}
\frac{\partial g_{ij}}{\partial \chi } &=&-2\widehat{R}_{ij},\ \frac{%
\partial \underline{g}_{ab}}{\partial \chi }=-2\widehat{R}_{ab}, \\
\ \frac{\partial \widehat{f}}{\partial \chi } &=&-\widehat{\Delta }\widehat{f%
}+\left| \widehat{\mathbf{D}}\widehat{f}\right| ^{2}-\ ^{h}\widehat{R}-\ ^{v}%
\widehat{R}
\end{eqnarray*}%
and the property that
\begin{equation*}
\frac{\partial }{\partial \chi }\widehat{\mathcal{F}}(\mathbf{g(\chi ),}%
\widehat{f}(\chi ))=2\int\limits_{\mathbf{V}}\left[ |\widehat{R}_{ij}+%
\widehat{D}_{i}\widehat{D}_{j}\widehat{f}|^{2}+|\widehat{R}_{ab}+\widehat{D}%
_{a}\widehat{D}_{b}\widehat{f}|^{2}\right] e^{-\widehat{f}}dV
\end{equation*}%
and $\int\limits_{\mathbf{V}}e^{-\widehat{f}}dV$ is constant. The functional
$\widehat{\mathcal{F}}(\mathbf{g(\chi ),}\widehat{f}(\chi ))$ is
nondecreasing in time and the monotonicity is strict unless we are on a
steady d--gradient solution.
\end{theorem}

On N--anholonomic manifolds, we define the associated d--energy%
\begin{equation}
\widehat{\lambda }(\mathbf{g},\widehat{\mathbf{D}})\doteqdot \inf \{\widehat{%
\mathcal{F}}(\mathbf{g(\chi ),}\widehat{f}(\chi ))|\widehat{f}\in C^{\infty
}(\mathbf{V}),\int\limits_{\mathbf{V}}e^{-\widehat{f}}dV=1\}.  \label{asef}
\end{equation}%
This value contains information on nonholonomic structure on $\mathbf{V.}$
It is also possible to introduce the associated energy defined by $\
_{\shortmid }\mathcal{F}(\mathbf{g},\nabla ,f)$ from (\ref{2pfrs}), $\lambda
(\mathbf{g},\mathbf{\nabla })\doteqdot \inf \{\ _{\shortmid }\mathcal{F}(%
\mathbf{g(\chi ),}f(\chi ))|\ f\in C^{\infty }(\mathbf{V}),\int\limits_{%
\mathbf{V}}e^{-f}dV=1\}$. Both values $\widehat{\lambda }$ and $\lambda $
are defined by the same sets of metric structures $\mathbf{g}(\chi )$ but,
respectively, for different sets of linear connections, $\widehat{\mathbf{D}}%
(\chi )$ and $\mathbf{\nabla }(\chi )\mathbf{.}$ One holds also the property
that $\lambda $ is invariant under diffeomorphisms but $\widehat{\lambda }$
possesses only N--adapted diffeomorphism invariance. In this section, we
state the main properties of $\ \widehat{\lambda }.$

\begin{proposition}
\label{pasdeg}There are canonical N--adapted decompositions, to splitting (%
\ref{2whitney}), of the functional $\widehat{\mathcal{F}}$ and associated
d--energy $\widehat{\lambda }.$
\end{proposition}

From this Proposition, one follows

\begin{corollary}
\label{corden}The d--energy (respectively, h--energy or v--energy) has the
property:

\begin{itemize}
\item $\widehat{\lambda }$ (respectively, $\ ^{h}\widehat{\lambda }$ or$\
^{v}\widehat{\lambda })$ is nondecreasing along the N--anholonomic Ricci
flow and the monotonicity is strict unless we are on a steady distinguished
(respectively, horizontal or vertical) gradient soliton;

\item a steady distinguished (horizontal or vertical) breather is
necessarily a steady distinguished (respectively, horizontal or vertical)
gradient solution.
\end{itemize}
\end{corollary}

For any positive numbers $\ ^{h}a$ and $\ ^{v}a,$ $\widehat{a}=\ ^{h}a+\
^{v}a,$ and N--adapted diffeomorphisms on $\mathbf{V},$ denoted $\widehat{%
\varphi }=(\ ^{h}\varphi ,\ ^{v}\varphi ),$ we have
\begin{equation*}
\ \widehat{\mathcal{W}}(\ ^{h}a\ ^{h}\varphi ^{\ast }g_{ij},\ ^{v}a\
^{v}\varphi ^{\ast }g_{ab},\widehat{\varphi }^{\ast }\widehat{\mathbf{D}},%
\widehat{\varphi }^{\ast }\widehat{f},\widehat{a}\widehat{\tau })=\widehat{%
\mathcal{W}}(\mathbf{g},\widehat{\mathbf{D}},\widehat{f},\widehat{\tau })
\end{equation*}%
which mean that the functional $\widehat{\mathcal{W}}$  is
invariant under N--adapted parabol\-ic scaling, i.e. under respective
scaling of $\widehat{\tau }$ and $\mathbf{g}_{\alpha \beta }=\left(
g_{ij},g_{ab}\right) .$ For simplicity, we can restrict our considerations
to evolutions defined by d--metric coefficients $\mathbf{g}_{\alpha \beta }(%
\widehat{\tau })$ with not depending on $\widehat{\tau }$ values $%
N_{i}^{a}(u^{\beta }).$ In a similar form to Lemma \ref{2lem1}, we get the
following first N--adapted variation formula for $\widehat{\mathcal{W}}:$

\begin{lemma}
The first N--adapted variations of $\widehat{\mathcal{W}}$ are given by {\small
\begin{eqnarray*}
&&\delta \widehat{\mathcal{W}}(v_{ij},v_{ab},\ ^{h}f,\ ^{v}f,\widehat{\tau }%
)= \\
&&\int\limits_{\mathbf{V}}\{\widehat{\tau }[-v^{ij}(\widehat{R}_{ij}+%
\widehat{D}_{i}\widehat{D}_{j}\widehat{f}-\frac{g_{ij}}{2\widehat{\tau }}%
)-v^{ab}(\widehat{R}_{ab}+\widehat{D}_{a}\widehat{D}_{b}\widehat{f}-\frac{%
g_{ab}}{2\widehat{\tau }})] \\
&&+(\frac{\ ^{h}v}{2}-\ ^{h}f-\frac{n}{2\widehat{\tau }}\widehat{\eta })[%
\widehat{\tau }\left( \ ^{h}\widehat{R}+2\ ^{h}\Delta \widehat{f}-|\ ^{h}D\
\widehat{f}|^{2}\right) +\ ^{h}f-n-1] \\
&&+(\frac{\ ^{v}v}{2}-\ ^{v}f-\frac{m}{2\widehat{\tau }}\widehat{\eta })[%
\widehat{\tau }\left( \ ^{v}\widehat{R}+2\ ^{v}\Delta \widehat{f}-|\ ^{v}D\
\widehat{f}|^{2}\right) +\ ^{v}f-m-1] \\
&&+\widehat{\eta }\left( \ ^{h}\widehat{R}+\ ^{v}\widehat{R}+|\ ^{h}D\
\widehat{f}|^{2}+|\ ^{v}D\ \widehat{f}|^{2}-\frac{n+m}{2\widehat{\tau }}%
\right) \}(4\pi \widehat{\tau })^{-(n+m)/2}e^{-\widehat{f}}dV,
\end{eqnarray*}%
} where $\widehat{\eta }=\delta \widehat{\tau }.$
\end{lemma}

For the functional $\widehat{\mathcal{W}},$ one holds a result which is
analogous to Theorem \ref{2theq1}:

\begin{theorem}
\label{2theveq}If a d--metric $\mathbf{g}(\chi )$ (\ref{2m1}) and functions $%
\widehat{f}(\chi )$ and $\widehat{\tau }(\chi )$ evolve according the system
of equations%
\begin{eqnarray*}
\frac{\partial g_{ij}}{\partial \chi } &=&-2\widehat{R}_{ij},\ \frac{%
\partial g_{ab}}{\partial \chi }=-2\widehat{R}_{ab}, \\
\ \frac{\partial \widehat{f}}{\partial \chi } &=&-\widehat{\Delta }\widehat{f%
}+\left| \widehat{\mathbf{D}}\widehat{f}\right| ^{2}-\ ^{h}\widehat{R}-\ ^{v}%
\widehat{R}+\frac{n+m}{\widehat{\tau }},\ \frac{\partial \widehat{\tau }}{%
\partial \chi } =-1
\end{eqnarray*}%
and the property that
\begin{eqnarray*}
\frac{\partial }{\partial \chi }\widehat{\mathcal{W}}(\mathbf{g}(\chi )%
\mathbf{,}\widehat{f}(\chi ),\widehat{\tau }(\chi )) &=&2\int\limits_{%
\mathbf{V}}\widehat{\tau }[|\widehat{R}_{ij}+D_{i}D_{j}\widehat{f}-\frac{1}{2%
\widehat{\tau }}g_{ij}|^{2}+ \\
&&|\widehat{R}_{ab}+D_{a}D_{b}\widehat{f}-\frac{1}{2\widehat{\tau }}%
g_{ab}|^{2}](4\pi \widehat{\tau })^{-(n+m)/2}e^{-\widehat{f}}dV
\end{eqnarray*}%
and $\int\limits_{\mathbf{V}}(4\pi \widehat{\tau })^{-(n+m)/2}e^{-\widehat{f}%
}dV$ is constant. The functional $\widehat{\mathcal{W}}$ is h--
(v--)nondecre\-as\-ing in time and the monotonicity is strict unless we are
on a shrinking h-- (v--) gradient soliton. This functional is N--adapted
nondecreasing if it is both h-- and v--nondecreasing.
\end{theorem}

In this work, for Theorem \ref{2theveq}, the evolution equations are written
with respect to N--adapted frames. If the N--connection structure is fixed
in ''time'' $\chi ,$ or $\widehat{\tau },$ we do not have to consider
evolution equations for the N--anholoonomic frame structure. For more
general cases, the evolution of preferred N--adapted frames (\ref{2ft}) is
stated by:

\begin{corollary}
The evolution, for all time $\tau \in \lbrack 0,\tau _{0}),$ of preferred
frames on a N--anholonomic manifold $\ \mathbf{e}_{\alpha }(\tau )=\ \mathbf{%
e}_{\alpha }^{\ \underline{\alpha }}(\tau ,u)\partial _{\underline{\alpha }}$
is defined by
\begin{eqnarray*}
\ \mathbf{e}_{\alpha }^{\ \underline{\alpha }}(\tau ,u) &=&\left[
\begin{array}{cc}
\ e_{i}^{\ \underline{i}}(\tau ,u) & ~N_{i}^{b}(\tau ,u)\ e_{b}^{\
\underline{a}}(\tau ,u) \\
0 & \ e_{a}^{\ \underline{a}}(\tau ,u)%
\end{array}%
\right] ,\  \\
\mathbf{e}_{\ \underline{\alpha }}^{\alpha }(\tau ,u)\ &=&\left[
\begin{array}{cc}
e_{\ \underline{i}}^{i}=\delta _{\underline{i}}^{i} & e_{\ \underline{i}%
}^{b}=-N_{k}^{b}(\tau ,u)\ \ \delta _{\underline{i}}^{k} \\
e_{\ \underline{a}}^{i}=0 & e_{\ \underline{a}}^{a}=\delta _{\underline{a}%
}^{a}%
\end{array}%
\right],
\end{eqnarray*}%
with $\ g_{ij}(\tau )=\ e_{i}^{\ \underline{i}}(\tau ,u)\ e_{j}^{\
\underline{j}}(\tau ,u)\eta _{\underline{i}\underline{j}}$ and $g_{ab}(\tau
)=\ e_{a}^{\ \underline{a}}(\tau ,u)\ e_{b}^{\ \underline{b}}(\tau ,u)\eta _{%
\underline{a}\underline{b}}$, where $\eta _{\underline{i}\underline{j}%
}=diag[\pm 1,...\pm 1]$ and $\eta _{\underline{a}\underline{b}}=diag[\pm
1,...\pm 1]$ establish the signature of $\ \mathbf{g}_{\alpha \beta
}^{[0]}(u),$ is given by equations
 $\frac{\partial }{\partial \tau }\mathbf{e}_{\ \underline{\alpha }}^{\alpha
}\ =\ \mathbf{g}^{\alpha \beta }~\widehat{\mathbf{R}}_{\beta \gamma }~\
\mathbf{e}_{\ \underline{\alpha }}^{\gamma }$
if we prescribe that the geometric constructions are derived by the
canonical d--connection.
\end{corollary}

It should be noted that $\mathbf{g}^{\alpha \beta }~\widehat{\mathbf{R}}%
_{\beta \gamma }=g^{ij}\widehat{R}_{ij}+g^{ab}\widehat{R}_{ab}$  selects for evolution only the symmetric components of the Ricci
d--tensor for the canonical d--connection.

\subsection{Clifford structures adapted to nonlinear connections}

In this section, we outline some key results from Refs. \cite{v1,v6} on
spinors and Dirac operators for nonholonomic manifolds and generalized
Finsler spaces.

The spinor bundle on a manifold $M,\ dimM=n,$ is constructed on the tangent
bundle $TM$ by substituting the group $SO(n)$ by its universal covering $%
Spin(n).$ If a horizontal quadratic form $\ ^{h}g_{ij}(x,y)$ is defined on $%
T_{x}h\mathbf{V}$ we can consider h--spinor spaces in every point $x\in h%
\mathbf{V}$ with fixed $y^{a}.$ The constructions can be completed on $T%
\mathbf{V}$ by using the d--metric $\mathbf{g}$. In this case, the group $%
SO(n+m)$ is not only substituted by $Spin(n+m)$ but with respect to
N--adapted frames there are emphasized decompositions to $Spin(n)\oplus
Spin(m).$\footnote{%
It should be noted here that spin bundles may not exist for general
holonomic or nonholonomic manifolds. For simplicity, we do not provide such
topological considerations in this paper. We state that we shall work only
with N--anholonomic manifolds for which certain spinor structures can be
defined both for the h- and v--splitting; the existence of a well defined
decomposition $Spin(n)\oplus Spin(m)$ follows from N--connection splitting.}

\subsubsection{Clifford N--adapted modules (d--modules)}

A Clifford d--algebra is a\ $\ \wedge V^{n+m}$ algebra endowed with a
product
 $
\mathbf{u}\mathbf{v}+\mathbf{v}\mathbf{u}=2\mathbf{g}(\mathbf{u},\mathbf{v}%
)\ \mathbb{I}$
distinguished into h--, v--products
$$\ ^{h}u~^{h}v+~^{h}v~^{h}u = 2~^{h}g(u,v)\ ~^{h}\mathbb{I} ,\ \ ^{v}u\
\ ~^{v}v+~^{v}v\ \ ~^{v}u = 2\ \ ~^{v}h(~^{v}u,\ \ ~^{v}v)~^{v}%
\mathbb{I},$$
for any $\mathbf{u}=(~^{h}u,\ ~^{v}u),\ \mathbf{v}=(~^{h}v,\ ~^{v}v)\in
V^{n+m},$ where $\mathbb{I},$ $\ ~^{h}\mathbb{I}\ $\ and $^{v}\mathbb{I}\ \ $%
\ are unity matrices of corresponding dimensions $(n+m)\times (n+m),$ or $%
n\times n$ and $m\times m.$

A metric $^{h}g$ on $h\mathbf{V}$ is defined by sections of the tangent
space $T~h\mathbf{V}$ provided with a bilinear symmetric form on continuous
sections $\Gamma (T~h\mathbf{V}).$\footnote{%
for simplicity, we shall consider only ''horizontal'' geometric
constructions if they are similar to ''vertical'' ones} This allows us to
define Clifford h--algebras $~^{h}\mathcal{C}l(T_{x}h\mathbf{V}),$ in any
point $x\in T~h\mathbf{V},$ $\gamma _{i}\gamma _{j}+\gamma _{j}\gamma
_{i}=2\ g_{ij}~^{h}\mathbb{I}.$ For any point $x\in h\mathbf{V}$ and fixed $%
y=y_{0},$ there exists a standard complexification, $T_{x}h\mathbf{V}^{%
\mathbb{C}}\doteq T_{x}h\mathbf{V}+iT_{x}h\mathbf{V},$ which can be used for
definition of the 'involution' operator on sections of $T_{x}h\mathbf{V}^{%
\mathbb{C}},$
\begin{equation*}
~^{h}\sigma _{1}~^{h}\sigma _{2}(x)\doteq ~^{h}\sigma _{2}(x)~^{h}\sigma
_{1}(x),\ ~^{h}\sigma ^{\ast }(x)\doteq ~^{h}\sigma (x)^{\ast },\forall x\in
h\mathbf{V},
\end{equation*}%
where ''*'' denotes the involution on every $~^{h}\mathcal{C}l(T_{x}h\mathbf{%
V}).$

\begin{definition}
A Clifford d--space on a nonholonomic manifold $~\mathbf{V}$ enabled with a
d--metric $\ \mathbf{g}(x,y)$ and a N--connection $\ \mathbf{N}$ is defined
as a Clifford bundle $~\mathcal{C}l(\mathbf{V})=~^{h}\mathcal{C}l(h\mathbf{V}%
)\oplus ~^{v}\mathcal{C}l(v\mathbf{V}),$ for the Clifford h--space $~^{h}%
\mathcal{C}l(h\mathbf{V})\doteq ~^{h}\mathcal{C}l(T^{\ast }h\mathbf{V})$ and
Clifford v--space $^{v}\mathcal{C}l(v\mathbf{V})\doteq ~^{v}\mathcal{C}%
l(T^{\ast }v\mathbf{V}).$
\end{definition}

For a fixed N--connection structure, a Clifford N--anholonomic bundle on $%
\mathbf{V}$ is defined $\ ~^{N}\mathcal{C}l(\mathbf{V})\doteq ~^{N}\mathcal{C%
}l(T^{\ast }\mathbf{V}).$ Let us consider a complex vector bundle $~^{E}\pi
:\ E\rightarrow \mathbf{V}$ on an N--anholonomic space $\mathbf{V}$ when the
N--connection structure is given for the base manifold. The Clifford
d--module of a vector bundle ${E}$ is defined by the $C(\mathbf{V})$--module
$\Gamma ({E})$ of continuous sections in ${E},$ 
$c:\ \Gamma (~^{N}\mathcal{C}l(\mathbf{V}))\rightarrow End(\Gamma ({E})).$

In general, a vector bundle on a N--anholonomic manifold may be not adapted
to the N--connection structure on base space.

\subsubsection{h--spinors, v--spinors and d--spinors}

Let us consider a vector space $V^{n}$ provided with Clifford structure. We
denote such a space $^{h}V^{n}$ in order to emphasize that its tangent space
is provided with a quadratic form $\ ^{h}g.$ We also write $~^{h}\mathcal{C}%
l(V^{n})\equiv \mathcal{C}l(~^{h}V^{n})$ and use subgroup $SO(\
^{h}V^{n})\subset O(\ ^{h}V^{n}).$

\begin{definition}
The space of complex h--spins is defined by the subgroup $\ $%
\begin{equation*}
^{h}Spin^{c}(n)\equiv Spin^{c}(\ ^{h}V^{n})\equiv \
^{h}Spin^{c}(V^{n})\subset \mathcal{C}l(\ ^{h}V^{n}),
\end{equation*}
determined by the products of pairs of vectors $w\in \ ^{h}V^{\mathbb{C}}$
when $w\doteq \lambda u$ where $\lambda $ is a complex number of module 1
and $u$ is of unity length in $\ ^{h}V^{n}.$
\end{definition}

Similar constructions can be performed for the v--subspace $~^{v}V^{m},$
which allows us to define similarly the group of real v--spins.

A usual spinor is a section of a vector bundle $S$ on a manifold $M$ when an
irreducible representation of the group $Spin(M)\doteq Spin(T_{x}^{\ast }M)$
is defined on the typical fiber. The set of sections $\Gamma (S)$ is a
irreducible Clifford module. If the base manifold is of type $h\mathbf{V},$
or is a general N--anholonomic manifold $\mathbf{V},$ we have to define
spinors on such spaces in a form adapted to the respective N--connection
structure.

\begin{definition}
A h--spinor bundle $\ ^{h}S$ on a h--space $h\mathbf{V}$ is a complex vector
bundle with both defined action of the h--spin group $\ ^{h}Spin(V^{n})$ on
the typical fiber and an irreducible representation of the group $\ ^{h}Spin(%
\mathbf{V})\equiv Spin(h\mathbf{V})\doteq Spin(T_{x}^{\ast }h\mathbf{V}).$
The set of sections $\Gamma (\ ^{h}S)$ defines an irreducible Clifford
h--module.
\end{definition}

The concept of ''d--spinors'' has been introduced for the spaces provided
with N--connection structure \cite{v1}:

\begin{definition}
\label{ddsp} A distinguished spinor (d--spinor) bundle $\mathbf{S}\doteq (\
\ ^{h}S,\ \ ^{v}S)$ on a N--anho\-lo\-nom\-ic manifold $\mathbf{V},$ $\ dim%
\mathbf{V}=n+m,$ is a complex vector bundle with a defined action of the
spin d--group $Spin\ \mathbf{V}\doteq Spin(V^{n})\oplus Spin(V^{m})$ with
the splitting adapted to the N--connection structure which results in an
irreducible representation $Spin(\mathbf{V})\doteq Spin(T^{\ast }\mathbf{V}%
). $ The set of sections $\Gamma (\mathbf{S})=\Gamma \ (\ \ ^{h}{S})\oplus
\Gamma (\ \ \ ^{v}{S})$ is an irreducible Clifford d--module.
\end{definition}

If we study algebras through theirs representations, we also have to
consider various algebras related by the Morita equivalence.\footnote{%
The Morita equivalence can be analyzed by applying in N--adapted form, both
on the base and fiber spaces, the consequences of the Plymen's theorem (in
this work, we omit details of such considerations).}

The possibility to distinguish the $Spin(n)$ (or, correspondingly $Spin(h%
\mathbf{V}),$\ $Spin(V^{n})\oplus Spin(V^{m}))$ allows us to define an
antilinear bijection $\ ^{h}J:\ \ ^{h}S\ \rightarrow \ \ ^{h}S$ (or $\ \
^{v}J:\ \ \ ^{v}S\ \rightarrow \ \ ^{v}\ S$ and $\mathbf{J}:\ \mathbf{S}\
\rightarrow \ \mathbf{S})$ with properties
\begin{eqnarray}
~^{h}J(~^{h}a\psi ) &=&~^{h}\chi (~^{h}a)~^{h}J~^{h}\psi ,\mbox{ for }%
~^{h}a\in \Gamma ^{\infty }(\mathcal{C}l(h\mathbf{V}));  \notag \\
(~^{h}J~^{h}\phi |~^{h}J~^{h}\psi ) &=&(~^{h}\psi |~^{h}\phi )\mbox{ for }%
~^{h}\phi ,~^{h}\psi \in ~^{h}S.  \label{jeq}
\end{eqnarray}

The considerations presented in this Section consists the proof of:\

\begin{theorem}
\label{mr1}Any d--metric and N--con\-nec\-ti\-on structure defines naturally
the fundamental geometric objects and structures (such as the Clifford
h--module, v--module and Clifford d--modules,or the h--spin, v--spin
structures and d--spinors) for the corresponding nonholonomic spin manifold
and/or N--anholo\-nom\-ic spinor (d--spinor) manifold.
\end{theorem}

\subsection{N--anholonomic Dirac operators}

The geometric constructions depend on the type of linear connections
considered for definition of such Dirac operators. They are metric
compatible and N--adapted if the canonical d--connection is used (similar
constructions can be performed for any deformation which results in a metric
compatible d--connection).

\subsubsection{Noholonomic vielbeins and spin d--connections}

Let us consider a Hilbert space of finite dimension. For a local dual
coordinate basis $e^{\underline{i}}\doteq dx^{\underline{i}}$ on $h\mathbf{V}%
,$\ we may respectively introduce certain classes of orthonormalized
vielbeins and the N--adapted vielbeins, $\ e^{\hat{\imath}}\doteq e_{\
\underline{i}}^{\hat{\imath}}(x,y)\ e^{\underline{i}}$ and $e^{i}\doteq e_{\
\underline{i}}^{i}(x,y)\ e^{\underline{i}},$when $g^{\underline{i}\underline{%
j}}\ e_{\ \underline{i}}^{\hat{\imath}}e_{\ \underline{j}}^{\hat{\jmath}%
}=\delta ^{\hat{\imath}\hat{\jmath}}$ and $g^{\underline{i}\underline{j}}\
e_{\ \underline{i}}^{i}e_{\ \underline{j}}^{j}=g^{ij}.$

We define the algebra of Dirac's gamma horizontal matrices (in brief, gamma
h--matrices defined by self--adjoint matrices $M_{k}(\mathbb{C})$ where $%
k=2^{n/2}$ is the dimension of the irreducible representation of $\mathcal{C}%
l(h\mathbf{V})$ from relation $\ \gamma ^{\hat{\imath}}\gamma ^{\hat{\jmath}%
}+\gamma ^{\hat{\jmath}}\gamma ^{\hat{\imath}}=2\delta ^{\hat{\imath}\hat{%
\jmath}}\ ^{h}\mathbb{I}.$ \ The action of $dx^{i}\in \mathcal{C}l(h\mathbf{V%
})$ on a spinor $\ \ ^{h}\psi \in \ ^{h}S$ is given by formulas
\begin{equation}
\ \ \ ^{h}c(dx^{\hat{\imath}})\doteq \gamma ^{\hat{\imath}}\mbox{ and }\ \ \
^{h}c(dx^{i})\ \ ^{h}\psi \doteq \gamma ^{i}\ \ ^{h}\psi \equiv e_{\ \hat{%
\imath}}^{i}\ \gamma ^{\hat{\imath}}\ \ ^{h}\psi .  \label{gamfibb}
\end{equation}

Similarly, we can define the algebra of Dirac's gamma vertical matrices
related to a typical fiber $F$ (in brief, gamma v--matrices defined by
self--adjoint matrices $M_{k}^{\prime }(\mathbb{C}),$ where $k^{\prime
}=2^{m/2}$ is the dimension of the irreducible representation of $\mathcal{C}%
l(F)$) from relation $\gamma ^{\hat{a}}\gamma ^{\hat{b}}+\gamma ^{\hat{b}%
}\gamma ^{\hat{a}}=2\delta ^{\hat{a}\hat{b}}\ ^{v}\mathbb{I}.$ The action of
$dy^{a}\in \mathcal{C}l(F)$ on a spinor $\ ^{v}\psi \in \ ^{v}S$ is
 $\ ^{v}c(dy^{\hat{a}}):= \gamma ^{\hat{a}}$ and $\ ^{v}c(dy^{a})\
^{v}\psi \doteq \gamma ^{a}\ ^{v}\psi \equiv e_{\ \hat{a}}^{a}\ \gamma ^{%
\hat{a}}\ ^{v}\psi$.

A more general gamma matrix calculus with distinguished gamma matrices (in
brief, gamma d--matrices\footnote{%
in some our works we wrote $\sigma $ instead of $\gamma $}) can be
elaborated for N--anholonomic manifolds $\mathbf{V}$ provided with d--metric
structure $\mathbf{g}=~^{h}g\oplus ~^{v}h]$ and for d--spinors $\breve{\psi}%
\doteq (~^{h}\psi ,\ ~^{v}\psi )\in \mathbf{S}\doteq (~^{h}S,\ ~^{v}S).$ In
this case, we consider d--gamma matrix relations $\ \gamma ^{\hat{\alpha}%
}\gamma ^{\hat{\beta}}+\gamma ^{\hat{\beta}}\gamma ^{\hat{\alpha}}=2\delta ^{%
\hat{\alpha}\hat{\beta}}\mathbb{\ I},$ with the action of $du^{\alpha }\in
\mathcal{C}l(\mathbf{V})$ on a d--spinor $\breve{\psi}\in \ \mathbf{S}$
resulting in distinguished irreducible representations
\begin{equation}
\mathbf{c}(du^{\hat{\alpha}})\doteq \gamma ^{\hat{\alpha}}\mbox{ and }%
\mathbf{c}=(du^{\alpha })\ \breve{\psi}\doteq \gamma ^{\alpha }\ \breve{\psi}%
\equiv e_{\ \hat{\alpha}}^{\alpha }\ \gamma ^{\hat{\alpha}}\ \breve{\psi}
\label{gamfibd}
\end{equation}%
which allows us to write 
$\gamma ^{\alpha }(u)\gamma ^{\beta }(u)+\gamma ^{\beta }(u)\gamma ^{\alpha
}(u)=2g^{\alpha \beta }(u)\ \mathbb{I}.$ 

In the canonical representation, we have the irreducible form $\breve{\gamma}%
\doteq ~^{h}\gamma \oplus \ ~^{v}\gamma $ and $\breve{\psi}\doteq ~^{h}\psi
\oplus ~^{v}\psi ,$ for instance, by using block type of h-- and
v--matrices. We can also write such formulas as couples of gamma and/or h--
and v--spinor objects written in N--adapted form, $\gamma ^{\alpha }\doteq
(~^{h}\gamma ^{i},~^{v}\gamma ^{a})$ and $\breve{\psi}\doteq (~^{h}\psi ,\
~^{v}\psi ).$

The spin connection $~_{S}\nabla $ for Riemannian manifolds is induced by
the Levi--Civita connection $\ ~\Gamma ,$ $~_{S}\nabla \doteq d-\frac{1}{4}\
\Gamma _{\ jk}^{i}\gamma _{i}\gamma ^{j}\ dx^{k}.$ On N--anholonomic
manifolds, spin d--connection operators $~_{\mathbf{S}}\mathbf{\nabla }$ can
be similarly constructed from any metric compatible d--connection ${\mathbf{%
\Gamma }}_{\ \beta \mu }^{\alpha }$ using the N--adapted absolute
differential $\delta $ acting, for instance, on a scalar function $f(x,y)$
in the form
 $\delta f=\left( \mathbf{e}_{\nu }f\right) \delta u^{\nu }=\left( \mathbf{e}%
_{i}f\right) dx^{i}+\left( e_{a}f\right) \delta y^{a}$,
for $\delta u^{\nu }=\mathbf{e}^{\nu },$ see N--elongated operators.

\begin{definition}
The canonical spin d--connection is defined by the canonical d--connection,
\begin{equation}
~_{\mathbf{S}}\widehat{\nabla }\doteq \delta -\frac{1}{4}\ \widehat{\mathbf{%
\Gamma }}_{\ \beta \mu }^{\alpha }\gamma _{\alpha }\gamma ^{\beta }\delta
u^{\mu },  \label{csdc}
\end{equation}%
where the N--adapted coefficients $\widehat{\mathbf{\Gamma }}_{\ \beta \mu
}^{\alpha }$ are given by formulas (\ref{2candcon}).
\end{definition}

We note that the canonical spin d--connection $~_{\mathbf{S}}\widehat{\nabla
}$ is metric compatible and contains nontrivial d--torsion coefficients
induced by the N--anholonomy relations.

\subsubsection{Dirac d--operators}

We consider a vector bundle $\mathbf{E}$ on a N--anholonomic manifold $%
\mathbf{V}$ (with two compatible N--connections defined as h-- and
v--splitting of $T\mathbf{E}$ and $T\mathbf{V}$)). A d--connection $\mathcal{%
D}:\ \Gamma ^{\infty }(\mathbf{E})\rightarrow \Gamma ^{\infty }(\mathbf{E}%
)\otimes \Omega ^{1}(\mathbf{V})$ preserves by parallelism splitting of the
tangent total and base spaces and satisfy the Leibniz condition $\mathcal{D}%
(f\sigma )=f(\mathcal{D}\sigma )+\delta f\otimes \sigma ,$ for any $f\in
C^{\infty }(\mathbf{V}),$ and $\sigma \in \Gamma ^{\infty }(\mathbf{E})$ and
$\delta $ defining an N--adapted exterior calculus by using N--elongated
operators which emphasize d--forms instead of usual forms on $\mathbf{V},$
with the coefficients taking values in $\mathbf{E}.$

The metricity and Leibniz conditions for $\mathcal{D}$ are written
respectively
\begin{equation}
\mathbf{g}(\mathcal{D}\mathbf{X},\mathbf{Y})+\mathbf{g}(\mathbf{X},\mathcal{D%
}\mathbf{Y})=\delta \lbrack \mathbf{g}(\mathbf{X},\mathbf{Y})],  \label{mc1}
\end{equation}%
for any $\mathbf{X},\ \mathbf{Y}\in \chi (\mathbf{V}),$ and
 $\mathcal{D}(\sigma \beta )\doteq \mathcal{D}(\sigma )\beta +\sigma \mathcal{D%
}(\beta )$,
for any $\sigma ,\beta \in \Gamma ^{\infty }(\mathbf{E}).$ For local
computations, we may define the corresponding coefficients of the geometric
d--objects and write
\begin{equation*}
\mathcal{D}\sigma _{\acute{\beta}}\doteq {\mathbf{\Gamma }}_{\ {\acute{\beta}%
}\mu }^{\acute{\alpha}}\ \sigma _{\acute{\alpha}}\otimes \delta u^{\mu }={%
\mathbf{\Gamma }}_{\ {\acute{\beta}}i}^{\acute{\alpha}}\ \sigma _{\acute{%
\alpha}}\otimes dx^{i}+{\mathbf{\Gamma }}_{\ {\acute{\beta}}a}^{\acute{\alpha%
}}\ \sigma _{\acute{\alpha}}\otimes \delta y^{a},
\end{equation*}%
where fiber ''acute'' indices are considered as spinor ones.

The respective actions of the Clifford d--algebra and Clifford h--algebra
can be transformed into maps $\Gamma ^{\infty }(\mathbf{S})\otimes \Gamma (%
\mathcal{C}l(\mathbf{V}))$ and $\Gamma ^{\infty }(~^{h}S)\otimes \Gamma (%
\mathcal{C}l(~^{h}V)$ to $\Gamma ^{\infty }(\mathbf{S})$ and, respectively, $%
\Gamma ^{\infty }(~^{h}S)$ by considering maps of type (\ref{gamfibb}) and (%
\ref{gamfibd}),
 $\widehat{\mathbf{c}}(\breve{\psi}\otimes \mathbf{a})\doteq \mathbf{c}(%
\mathbf{a})\breve{\psi}\mbox{\ and\ }~^{h}\widehat{c}(~^{h}{\psi }\otimes
~^{h}{a})\doteq ~^{h}{c}(~^{h}{a})~^{h}{\psi }$.

\begin{definition}
\label{dddo} The Dirac d--operator (Dirac h--operator, or v--operant) on a
spin N--anholonomic manifold $(\mathbf{V},\mathbf{S},J)$ (on a h--spin
manifold\newline
$(h\mathbf{V},~^{h}S,~^{h}J),$ or on a v--spin manifold $(v\mathbf{V}%
,~^{v}S,~^{v}J))$ is defined {\small
\begin{equation}
\mathbb{D} :=-i\ (\widehat{\mathbf{c}}\circ ~_{\mathbf{S}}\nabla ) =\left( \
^{h}\mathbb{D}=-i\ (\ ^{h}\widehat{{c}}\circ \ _{\mathbf{S}}^{h}\nabla ),\ \
^{v}\mathbb{D}=-i\ (\ ^{v}\widehat{{c}}\circ \ _{\mathbf{S}}^{v}\nabla
)\right)  \label{ddo}
\end{equation}%
} Such N--adapted Dirac d--operators are called canonical and denoted $%
\widehat{\mathbb{D}}=(\ ^{h}\widehat{\mathbb{D}},\ ^{v}\widehat{\mathbb{D}}\
)$\ if they are defined for the canonical d--connection (\ref{2candcon}) and
respective spin d--connection (\ref{csdc}).
\end{definition}

We formulate:

\begin{theorem}
\label{mr2} Let $(\mathbf{V},\mathbf{S},\mathbf{J})$ (\ $(h\mathbf{V},\
^{h}S,\ ^{h}J))$ be a spin N--anholonomic manifold ( h--spin space). There
is the canonical Dirac d--operator (Dirac h--operator) defined by the almost
Hermitian spin d--operator
  $\ _{\mathbf{S}}\widehat{\nabla }:\ \Gamma ^{\infty }(\mathbf{S})\rightarrow
\Gamma ^{\infty }(\mathbf{S})\otimes \Omega ^{1}(\mathbf{V})$
(spin h--operator $\ ~_{\mathbf{S}}^{h}\widehat{\nabla }:\ \Gamma ^{\infty
}(\ ^{h}{S})\rightarrow \Gamma ^{\infty }(\ ^{h}S)\otimes \Omega ^{1}(h%
\mathbf{V})\ )$ commuting with $\mathbf{J}$ ($\ ^{h}J$), see (\ref{jeq}),
and satisfying the conditions
\begin{equation*}
(~_{\mathbf{S}}\widehat{\nabla }\breve{\psi}\ |\ \breve{\phi})\ +(\breve{\psi%
}\ |~_{\mathbf{S}}\widehat{\nabla }\breve{\phi})\ =\delta (\breve{\psi}\ |\
\breve{\phi}) \mbox{ and } \ _{\mathbf{S}}\widehat{\nabla }(\mathbf{c}(%
\mathbf{a})\breve{\psi})\ =\mathbf{c}(\widehat{\mathbf{D}}\mathbf{a})\breve{%
\psi}+\mathbf{c}(\mathbf{a}) _{\mathbf{S}}\widehat{\nabla }\breve{\psi}
\end{equation*}%
for $\mathbf{a}\in \mathcal{C}l(\mathbf{V})$ and $\breve{\psi}\in \Gamma
^{\infty }(\mathbf{S})$, $(\ (~_{\mathbf{S}}^{h}\widehat{\nabla }\ ^{h}{\psi
}|\ \ ^{h}\phi )\ +(~^{h}{\psi }\ |\ ~_{\mathbf{S}}^{h}\widehat{\nabla }\ \
^{h}\phi )\ =\ ^{h}\delta (~^{h}{\psi }\ |\ \ ^{h}\phi )$ and $\ ~_{\mathbf{S%
}}^{h}\widehat{\nabla }(~^{h}c(~^{h}a)~^{h}{\psi })\ =~^{h}c(\ \ ^{h}%
\widehat{D}~^{h}a)~^{h}{\psi }+~^{h}c(~^{h}a)~_{\mathbf{S}}^{h}\widehat{%
\nabla }~^{h}{\psi }$ for $~^{h}a\in \mathcal{C}l(h\mathbf{V})$ and $\breve{%
\psi}\in \Gamma ^{\infty }(\ ^{h}{S)}$ )\ determined by the metricity (\ref%
{mc1}) and Leibnitz (\ref{lc1}) conditions.
\end{theorem}

The geometric information of a spin manifold (in particular, the metric) is
contained in the Dirac operator. For nonholonomic manifolds, the canonical
Dirac d--operator has h-- and v--irreducible parts related to off--diagonal
metric terms and nonholonomic frames with associated structure. In a more
special case, the canonical Dirac d--operator is defined by the canonical
d--connection. Nonholonomic Dirac d--operators contain more information than
the usual, holonomic, ones.

\begin{proposition}
\label{pdohv} If $\widehat{\mathbb{D}}=(\ ^{h}\widehat{\mathbb{D}},\ ^{v}%
\widehat{\mathbb{D}}\ )$\ is the canonical Dirac d--operator then $\left[
\widehat{\mathbb{D}},\ f\right] = i\mathbf{c}(\delta f)$, equivalently,
 $\left[\ ^{h}\widehat{\mathbb{D}},\ f\right] +\left[\ ^{v}\widehat{\mathbb{D}}%
,\ f\right] = i\ ~^{h}c(dx^{i}\frac{\delta f}{\partial x^{i}})+i~\
^{v}c(\delta y^{a}\frac{\partial f}{\partial y^{a}})$,
 for all $f\in C^{\infty }(\mathbf{V}).$
\end{proposition}

\begin{proof}
It is a straightforward computation following from Definition \ref{dddo}. $%
\Box $
\end{proof}

The canonical Dirac d--operator and its h-- and v--components have all the
properties of the usual Dirac operators (for instance, they are
self--adjoint but unbounded). It is possible to define a scalar product on $%
\Gamma ^{\infty }(\mathbf{S})$,
\begin{equation}
<\breve{\psi},\breve{\phi}>\doteq \int_{\mathbf{V}}(\breve{\psi}|\breve{\phi}%
)|\nu _{\mathbf{g}}|  \label{scprod}
\end{equation}%
where $\ \nu _{\mathbf{g}}=\sqrt{det|g|}\ \sqrt{det|h|}\ dx^{1}...dx^{n}\
dy^{n+1}...dy^{n+m}$ is the volume d--form on the N--anholonomic manifold $%
\mathbf{V}.$

\subsubsection{N--adapted spectral triples and distance in d--spinor spaces}

We denote $\ ^{N}\mathcal{H}\doteq L_{2}(\mathbf{V},\mathbf{S})=\left[ \ ^{h}%
\mathcal{H}=L_{2}(h\mathbf{V},\ ^{h}S),\ ^{v}\mathcal{H}=L_{2}(v\mathbf{V},\
^{v}S)\right] $ the Hilbert d--space obtained by completing $\Gamma ^{\infty
}(\mathbf{S})$ with the norm defined by the scalar product (\ref{scprod}). \
Similarly to the holonomic spaces, by using formulas (\ref{ddo}) and (\ref%
{csdc}), one may prove that there is a self--adjoint unitary endomorphism $%
~_{[cr]}\Gamma $ of $~^{N}\mathcal{H},$ called ''chirality'', being a $%
\mathbb{Z}_{2}$ graduation of $~^{N}\mathcal{H},$\footnote{%
we use the label\ $[cr]$\ in order to avoid misunderstanding with the symbol
$\Gamma $ used for linear connections} which satisfies the condition $%
\widehat{\mathbb{D}}\ ~_{[cr]}\Gamma =-~_{[cr]}\Gamma \ \widehat{\mathbb{D}}%
. $ 
Such conditions can be written also for the irreducible components $\ ^{h}%
\widehat{\mathbb{D}}$ and$\ ^{v}\widehat{\mathbb{D}}\ .$

\begin{definition}
A distinguished canonical spectral triple (canonical spectral d--triple) $%
(~^{N}\mathcal{A},~^{N}\mathcal{H},\ \widehat{\mathbb{D}})$ for a d--algebra
$~^{N}\mathcal{A}$ is defined by a Hilbert d--space $~^{N}\mathcal{H},$ a
representation of $~^{N}\mathcal{A}$ in the algebra $~^{N}\mathcal{B}(~^{N}%
\mathcal{H})$ of d--operators bounded on $~^{N}\mathcal{H},$ and by a
self--adjoint d--operator $~^{N}\mathcal{H},$ of compact resolution,%
\footnote{%
An operator $D$ is of compact resolution if for any $\lambda \in sp(D)$ the
operator $(D-\lambda \mathbb{I})^{-1}$ is compact.} such that $[~^{N}%
\mathcal{H},a]\in ~^{N}\mathcal{B}(~^{N}\mathcal{H})$ for any $a\in ~^{N}%
\mathcal{A}.$
\end{definition}

Every canonical spectral d--triple is defined by two usual spectral triples
which in our case corresponds to certain h-- and v--components induced by
the corresponding h-- and v--components of the Dirac d--operator. For such
spectral h(v)--triples we, can define the notion of $KR^{n}$--cycle and $%
KR^{m}$--cycle and consider respective Hochschild complexes. To define a
noncommutative geometry the h-- and v-- components of a canonical spectral
d--triples must satisfy certain well defined seven conditions:\ the spectral
dimensions are of order $1/n$ and $1/m,$ respectively, for h-- and
v--components of the canonical Dirac d--operator; there are satisfied the
criteria of regularity, finiteness and reality; representations are of 1st
order; there is orientability and Poincar\'{e} duality holds true. Such
conditions can be satisfied by any Dirac operators and canonical Dirac
d--operators (in the last case we have to work with d--objects). \footnote{%
We omit in this paper the details on axiomatics and related proofs for such
considerations.}

\begin{definition}
\label{dncdg} A spectral d--triple is a real one satisfying the above
mentioned seven conditions for the h-- and v--irreversible components and
defining a (d--spinor) N--anholonomic noncommutative geometry stated by the
data $(\ ^{N}\mathcal{A},$ $\ ^{N}\mathcal{H},\ \widehat{\mathbb{D}},\
\mathbf{J},\ ~_{[cr]}\Gamma )$ and derived for the Dirac d--operator (\ref%
{ddo}).
\end{definition}

For N--adapted constructions, we can consider d--algebras $~^{N}\mathcal{A}%
=\ ^{h}\mathcal{A\oplus }$ $\ ^{v}\mathcal{A}$. We generate N--anholonomic
commutative geometries if we take $~^{N}\mathcal{A}\doteq C^{\infty }(%
\mathbf{V}),$ or $\mathcal{~}^{h}\mathcal{A}\doteq C^{\infty }(h\mathbf{V}).$

Let us show how it is possible to compute distance in a d--spinor space:
\begin{theorem}
\label{mr3} Let $(~^{N}\mathcal{A},~^{N}\mathcal{H},\ \widehat{\mathbb{D}},%
\mathbf{J},\ _{[cr]}\Gamma )$ defines a noncommutative geometry being
irreducible for $~^{N}\mathcal{A}\doteq C^{\infty }(\mathbf{V}),$ where $%
\mathbf{V}$ is a compact, connected and oriented manifold without
boundaries, of spectral dimension $dim\ \mathbf{V}=n+m.$ In this case, there
are satisfied the conditions:

\begin{enumerate}
\item There is a unique  $\mathbf{g}(\widehat{\mathbb{D}}%
)=(~^{h}g,\ ^{v}g)$ with the ''nonlinear'' geodesic distance on $\mathbf{V}$
defined by
 $d(u_{1},u_{2})=\sup_{f\in C(\mathbf{V})}\left\{ f(u_{1},u_{2})/\parallel
\lbrack \widehat{\mathbb{D}},f]\parallel \leq 1\right\}$,
 for any smooth function $f\in C(\mathbf{V}).$

\item A N--anholonomic manifold $\mathbf{V}$\ is a spin N--anholono\-mic
space, for which the operators $\widehat{\mathbb{D}}^{\prime }$ satisfying
the condition $\mathbf{g}(\widehat{\mathbb{D}}^{\prime })=\mathbf{g}(%
\widehat{\mathbb{D}})$ define an union of affine spaces identified by the
d--spinor structures on $\mathbf{V.}$

\item The functional $\mathcal{S}(\widehat{\mathbb{D}})\doteq \int |\widehat{%
\mathbb{D}}|^{-n-m+2}$ defines a quadratic d--form with $(n+m)$--splitting
for every affine space which is minimal for $\widehat{\mathbb{D}}=%
\overleftarrow{\mathbb{D}}$ as the canonical Dirac d--operator corresponding
to the d--spin structure with the minimum proportional to the
Einstein--Hilbert action constructed for the canonical d--connection with
 d--scalar curvature $\ ^{s}\mathbf{R}$,
\begin{equation*}
\mathcal{S}(\overleftarrow{\mathbb{D}})=-\frac{n+m-2}{24}\ \int_{\mathbf{V}%
}\ ^{s}\mathbf{R}\ \sqrt{~^{h}g}\ \sqrt{~^{v}h}\ dx^{1}...dx^{n}\ \delta
y^{n+1}...\delta y^{n+m}.
\end{equation*}
\end{enumerate}
\end{theorem}

The existence of a canonical d--connection structure which is metric
compatible and constructed from the coefficients of the d--metric and
N--connecti\-on structure is of crucial importance allowing the formulation
and proofs of the main results of this work. As a matter of principle, we
can consider any splitting of connections and compute a unique distance like
we stated in the above Theorem \ref{mr3}, but for a ''non--canonical'' Dirac
d--operator. This holds true for any noncommutative geometry induced by a
metric compatible d--connection supposed to be uniquely induced by a metric
tensor.

In more general cases, we can consider any metric compatible d--connecti\-on
with arbitrary d--torsion. Such constructions can be also elaborated in
N--adapted form by preserving the respective h- and v--irreducible
decompositions. For the Dirac d--operators, we have to start with the
Proposition \ref{pdohv} and then to repeat all constructions both on h-- and
v--subspaces. In this article, we do not analyze (non) commutative
geometries enabled with general torsions but consider only nonholonomic
deformations when distorsions are induced by a metric structure.

Finally, we note that Theorem \ref{mr3} allows us to extract from a
canonical nonholonomic model of noncommutative geometry various types of
commutative geometries (holonomic and N--anholonomic Riemannian spaces,
Finsler--Lagrange spaces and generalizations) for corresponding nonholonomic
Dirac operators.

\subsection{Noncommutative geometry and Ricci flows}
\label{ssncrf}

The Ricci flow equations and Perelman functionals can be re--defined with
respect to moving frames subjected to nonholonomic constraints.\footnote{%
there are used also some other equivalent terms like anholonomic, or
non--integrable, restrictions/ constraints; we emphasize that in classical
and quantum physics the field and evolution equations play a fundamental
role but together with certain types of constraints and broken symmetries; a
rigorous mathematical approach to modern physical theories can be elaborated
only following geometric methods from 'nonholonomic field theory and
mechanics'} Considering models of evolution of geometric objects in a form
adapted to certain classes of nonholonomic constraints, we proved that
metrics and connections defining (pseudo) Riemannian spaces may flow into
similar nonholonomically deformed values modelling generalized Finsler and
Lagrange configurations, with symmetric and nonsymmetric metrics, or
possessing noncommutative symmetries.

The original Hamilton--Perelman constructions were for unconstrained flows
of metrics evolving only on (pseudo) Riemannian manifolds. There were proved
a set of fundamental results in mathematics and physics (for instance, the
Thurston and Poincar\'{e} conjectures, related to spacetime topological
properties, Ricci flow running of physical constants and fields etc).
Nevertheless, a number of important problems in geometry and physics are
considered in the framework of classical and quantum field theories with
constraints (for instance, the Lagrange and Hamilton mechanics, Dirac
quantization of constrained systems, gauge theories with broken symmetries
etc). With respect to the Ricci flow theory, to impose constraints on
evolution equations is to extend the research programs on manifolds enabled
with nonholonomic distributions, i.e. to study flows of fundamental
geometric structures on nonholonomic manifolds.

Imposing certain noncommutative conditions on physical variables and
coordinates in an evolution theory, we transfer the constructions and
methods into the field of noncommutative geometric analysis on nonholonomic
manifolds. This also leads naturally to various problems related to
noncommutative generalizations of the Ricci flow theory and possible
applications in modern physics. In this work, \ we follow the approach to
noncommutative geometry when the spectral action paradigm, with spectral
triples and Dirac operators, gives us a very elegant formulation of the
standard model in physics.

Following the spectral action paradigm, all details of the standard models
of particle interactions and gravity can be ''extracted'' from a
noncommutative geometry generated by a spectral triple $(\mathcal{A},%
\mathcal{H},\mathcal{D})$ by postulating the action%
\begin{equation}
Tr~~f(\mathcal{D}^{2}/\Lambda ^{2})+<\Psi |\mathcal{D}|\Psi >,  \label{tract}
\end{equation}%
where ''spectral'' is in the sense that the action depends only on the
spectrum of the Dirac operator $\mathcal{D}$ on a certain noncommutative
space defined by a noncommutative associative algebra $\mathcal{A}=C^{\infty
}(V)\otimes ~^{P}\mathcal{A}.$ In formula (\ref{tract}), $Tr$ is the trace
in operator algebra and $\Psi $ is a spinor, all defined for a Hilbert space
$\mathcal{H}$, $\Lambda $ is a cutoff scale and $~f$ \ is a positive
function. For a number of physical applications, $~^{P}\mathcal{A}$ is a
finite dimensional algebra and $C^{\infty }(V)$ is the algebra of complex
valued and smooth functions over a ''space'' $V,$ a topological manifold,
which for different purposes can be enabled with various necessary geometric
structures. The spectral geometry of $\mathcal{A}$ is given by the product
rule $\mathcal{H}=L^{2}(V,S)\otimes ~^{P}\mathcal{H},$ where $L^{2}(V,S)$ is
the Hilbert space of $L^{2}$ spinors and $~^{P}\mathcal{H}$ is the Hilbert
space of quarks and leptons fixing the choice of the Dirac operator $~^{P}D$
and the action $~^{P}\mathcal{A}$ for fundamental particles. Usually, the
Dirac operator from (\ref{tract}) is parametrized $\mathcal{D}=~^{V}D\otimes
1+\gamma _{5}\otimes ~^{P}D,$ where $~^{V}D$ is the Dirac operator of the
Levi--Civita spin connection on $V.$\footnote{%
in this work, we shall use left ''up'' and ''low'' abstract labels which
should not be considered as tensor or spinor indices written in the right
side of symbols for geometrical objects}

In order to construct exact solutions with noncommutative symmetries and
noncommutative gauge models of gravity and include dilaton fields, one has
to use instead of $~^{V}D$ certain generalized types of Dirac operators
defined by nonholonomic and/or conformal deformations of the 'primary'
Levi--Civita spin connection. In a more general context, the problem of
constructing well defined geometrically and physically motivated
nonholonomic Dirac operators is related to the problem of definition of
spinors and Dirac operators on Finsler--Lagrange spaces and generalizations.

\subsubsection{ Spectral Functionals and Ricci Flows}

\label{s4}The goal of this section is to prove that the Perelman's
functionals and their generalizations for nonholonomic Ricci flows in can be
extracted from flows of a generalized Dirac operator $~^{N}\mathcal{D}(\chi
)=~\mathbb{D}(\chi )\otimes 1$ included in spectral functionals of type
\begin{equation}
Tr~~^{b}f(~^{N}\mathcal{D}^{2}(\chi )/\Lambda ^{2}),  \label{trperfunct}
\end{equation}%
where $~^{b}f(\chi )$ are testing functions labelled by $b=1,2,3$ and
depending on a real flow parameter $\chi ,$ which in the commutative variant
of the Ricci flow theory corresponds to that for R. Hamilton's equations.
For simplicity, we shall use one cutoff parameter $\Lambda $ and suppose
that operators under flows act on the same algebra $\mathcal{A}$ and Hilbert
space $\mathcal{H},$ i.e. we consider families of spectral triples of type $(%
\mathcal{A},\mathcal{H},~^{N}\mathcal{D}(\chi )).$\footnote{%
we shall omit in this section the left label ''N'' for algebras and Hilbert
spaces if that will not result in ambiguities}

\begin{definition}
The normalized Ricci flow equations (R. Hamilton's equations) generalized on
nonholonomic manifolds are defined in the form
\begin{equation}
\frac{\partial \mathbf{g}_{\alpha \beta }(\chi )}{\partial \chi }=-2~^{N}%
\mathbf{R}_{\alpha \beta }(\chi )+\frac{2r}{5}\mathbf{g}_{\alpha \beta
}(\chi ),  \label{normcomrf}
\end{equation}%
where $\mathbf{g}_{\alpha \beta }(\chi )$ defines a family of d--metrics
parametrized in the form (\ref{m1}) on a N-anholonomic manifold $\mathbf{V}$
enabled with a family of N--connections $N_{i}^{a}(\chi ).$
\end{definition}

The effective ''cosmological'' constant $2r/5$ in (\ref{normcomrf}) with
normalizing factor $r=\int_{v}~_{s}^{N}\mathbf{R}dv/v$ is introduced with
the aim to preserve a volume $v$ on $\mathbf{V},\,$\ where $~_{s}^{N}\mathbf{%
R}$ is the scalar curvature.\footnote{%
We note that in our work there used two mutually related flow parameters $%
\chi $ and $\tau ;$ for simplicity, in this work we write only $\chi $ even,
in general, such parameters should be rescaled for different geometric
analysis constructions.}

The corresponding family of Ricci tensors $~^{N}\mathbf{R}_{\alpha \beta
}(\chi ),$ in (\ref{normcomrf}), and nonholonomic Dirac operators $%
~^{N}D(\chi ),$ in (\ref{tract}), are defined for any value of $\chi $ by a
general metric compatible linear connection $~^{N}\mathbf{\Gamma }$ adapted
to a N--connection structure. In a particular case, we can consider the
Levi--Civita connection $~\Gamma ,$ which is used in standard geometric
approaches to physical theories. Nevertheless, for various purposes in
modelling evolution of off--diagonal Einstein metrics, constrained physical
systems, effective Finsler and Lagrange geometries, Fedosov quantization of
field theories and gravity etc\footnote{%
the coefficients of corresponding N--connection structures being defined
respectively by the generic off--diagonal metric terms, anholonomy frame
coefficients, Finsler and Lagrange fundamental functions etc}, \ it is
convenient to work with a ''N--adapted'' linear connection $~^{N}\mathbf{%
\Gamma (g).}$ If such a connection is also uniquely defined by a metric
structure $\mathbf{g},$ we are able to re--define the constructions in an
equivalent form for the corresponding Levi--Civita connection.

In noncommutative geometry, all physical information on generalized Ricci
flows can be encoded into a corresponding family of nonholonomic Dirac
operators $~^{N}\mathcal{D}(\chi ).$ For simplicity, in this work, we shall
consider that $~^{P}D=0,$ i.e. we shall not involve into the
(non)commutative Ricci flow theory the particle physics. Perhaps a
''comprehensive'' noncommutative Ricci flow theory should include as a
stationary case the ''complete'' spectral action (\ref{tract}) parametrized
for the standard models of gravity and particle physics.

\subsubsection{Spectral flows and Perelman functionals}

Let us consider a family of generalized d--operators{\small
\begin{equation}
\mathcal{D}^{2}(\chi )=-\left\{ \frac{\mathbb{I}}{2}~\mathbf{g}^{\alpha
\beta }(\chi )\left[ \mathbf{e}_{\alpha }(\chi )\mathbf{e}_{\beta }(\chi )+%
\mathbf{e}_{\beta }(\chi )\mathbf{e}_{\alpha }(\chi )\right] +\mathbf{A}%
^{\nu }(\chi )\mathbf{e}_{\nu }(\chi )+\mathbf{B}(\chi )\right\} ,
\label{oper1}
\end{equation}%
} where the real flow parameter $\chi \in \lbrack 0,\chi _{0})$ and, for any
fixed values of this parameter, the matrices $\mathbf{A}^{\nu }(\chi )$ and $%
\mathbf{B}(\chi )$ are determined by a N--anholonomic Dirac operator $%
\mathbb{D}$ induced by a metric compatible d--connection $\mathbf{D,}$ see
and Definition \ref{dddo}; for the canonical d--connection, we have to put
''hats'' on symbols and write $\widehat{\mathcal{D}}^{2},\widehat{\mathbf{A}}%
^{\nu }$ and $\widehat{\mathbf{B}}\mathbf{.}$ We introduce two functionals $%
\mathcal{F}$ and $\mathcal{W}$ depending on $\chi ,$%
\begin{equation}
\mathcal{F}=Tr~\left[ ~^{1}f~(\chi )(~^{^{1}\phi }\mathcal{D}^{2}(\chi
)/\Lambda ^{2})\right] \simeq \sum\limits_{k\geq 0}~^{1}f_{(k)}(\chi
)~~^{1}a_{(k)}(~^{^{1}\phi }\mathcal{D}^{2}(\chi )/\Lambda ^{2})
\label{ncpf1a}
\end{equation}%
{\small
\begin{eqnarray}
\mbox{ and } \mathcal{W} &=&~^{2}\mathcal{W+}~^{3}\mathcal{W},
\label{ncpf2a} \\
\mbox{\qquad for }\ ^{e}\mathcal{W} &=&Tr\left[ \ ^{e}f(\chi )(\ ^{^{e}\phi }%
\mathcal{D}^{2}(\chi )/\Lambda ^{2})\right] = \sum\limits_{k\geq 0}\
^{e}f_{(k)}(\chi )\ ^{e}a_{(k)}(\ ^{^{e}\phi }\mathcal{D}^{2}(\chi )/\Lambda
^{2}),  \notag
\end{eqnarray}%
} where we consider a cutting parameter $\Lambda ^{2}$ for both cases $%
e=2,3. $ Functions $~^{b}f,$ with label $b$ taking values $1,2,3,$ have to
be chosen in a form which insure that for a fixed $\chi $ we get certain
compatibility with gravity and particle physics and result in positive
average energy and entropy for Ricci flows of geometrical objects. For such
testing functions, ones hold true the formulas
\begin{eqnarray}
~^{~^{b}}f_{(0)}(\chi ) &=&\int\limits_{0}^{\infty }~^{b}f(\chi
,u)u~du,~^{b}f_{(2)}(\chi )=\int\limits_{0}^{\infty }~^{b}f(\chi ,u)~du,
\notag \\
~^{b}f_{(2k+4)}(\chi ) &=&(-1)^{k}~~^{b}f^{(k)}(\chi ,0),\quad k\geq 0.
\label{ffunct}
\end{eqnarray}%
We will comment the end of this subsection on dependence on $\chi $ of such
functions.

The coefficients $~^{b}a_{(k)}$ can be computed as the Seeley -- de Witt
coefficients (we chose such notations when in the holonomic case the scalar
curvature is negative for spheres and the space is locally Euclidean). In
functionals (\ref{ncpf1a}) and (\ref{ncpf2a}), we consider dynamical scaling
factors of type $~^{b}\rho =\Lambda \exp (~^{~^{b}}\phi ),$ when, for
instance,
\begin{eqnarray}
~^{^{1}\phi }\mathcal{D}^{2} &=&~e^{-~^{1}\phi }~\mathcal{D}^{2}e^{~^{1}\phi
}  \label{ddiracscale} \\
&=&-\left\{ \frac{\mathbb{I}}{2}~~^{^{1}\phi }\mathbf{g}^{\alpha \beta }%
\left[ ~^{^{1}\phi }\mathbf{e}_{\alpha }~^{^{1}\phi }\mathbf{e}_{\beta
}+~^{^{1}\phi }\mathbf{e}_{\beta }~^{^{1}\phi }\mathbf{e}_{\alpha }\right]
+~^{^{1}\phi }\mathbf{A}^{\nu }~^{^{1}\phi }\mathbf{e}_{\nu }+~^{^{1}\phi }%
\mathbf{B}\right\} ,  \notag
\end{eqnarray}%
\begin{eqnarray*}
\mbox{\ for \ }~^{^{1}\phi }\mathbf{A}^{\nu } &=&~e^{-2~^{1}\phi }\times
\mathbf{A}^{\nu }-2~^{^{1}\phi }\mathbf{g}^{\nu \mu }\times ~^{^{1}\phi }%
\mathbf{e}_{\beta }(^{1}\phi ), \\
~^{^{1}\phi }\mathbf{B} &\mathbf{=}&~e^{-2~^{1}\phi }\times \left( \mathbf{B-%
\mathbf{A}^{\nu }~}^{^{1}\phi }\mathbf{e}_{\beta }(^{1}\phi )\right) \mathbf{%
+~}^{^{1}\phi }\mathbf{g}^{\nu \mu }\times \mathbf{~}^{^{1}\phi }W_{\nu \mu
}^{\gamma }\mathbf{~}^{^{1}\phi }\mathbf{e}_{\gamma },
\end{eqnarray*}%
for re--scaled d--metric $~^{^{1}\phi }\mathbf{g}_{\alpha \beta
}=~e^{2~^{1}\phi }\times \mathbf{g}_{\alpha \beta }$ and N--adapted frames $%
~^{^{1}\phi }\mathbf{e}_{\alpha }=\ e^{~^{1}\phi }\times \mathbf{e}_{\alpha
} $ satisfying anholonomy relations, with re--scaled
nonholonomy coefficients $^{^{1}\phi }W_{\nu \mu }^{\gamma }.$ We emphasize
that similar formulas can be written by substituting respectively the labels
and scaling factors containing $\ ^{1}\phi $ with $^{2}\phi $ and $^{3}\phi
. $ For simplicity, we shall omit left labels $1,2,3$ for $\phi $ and $f,a$
if that will not result in ambiguities.

Let us denote by $~^{s}\mathbf{R}(\mathbf{g}_{\mu \nu })$ and $\mathbf{C}%
_{\mu \nu \lambda \gamma }(\mathbf{g}_{\mu \nu }),$ correspondingly, the
scalar curvature and conformal Weyl d--tensor \footnote{%
for any metric compatible d--connection $\mathbf{D,}$ the Weyl d--tensor can
be computed by formulas similar to those for the Levi--Civita connection $%
\nabla ;$ here we note that if a Weyl d--tensor is zero, in general, the
Weyl tensor for $\nabla $ does not vanish (and inversely)}%
\begin{eqnarray*}
\mathbf{C}_{\mu \nu \lambda \gamma } &=&\mathbf{R}_{\mu \nu \lambda \gamma }+%
\frac{1}{2}\left( \mathbf{R}_{\mu \lambda }\mathbf{g}_{\nu \gamma }-\mathbf{R%
}_{\nu \lambda }\mathbf{g}_{\mu \gamma }-\mathbf{R}_{\mu \gamma }\mathbf{g}%
_{\nu \lambda }+\mathbf{R}_{\nu \gamma }\mathbf{g}_{\mu \lambda }\right) \\
&&-\frac{1}{6}\left( \mathbf{g}_{\mu \lambda }\mathbf{g}_{\nu \gamma }-%
\mathbf{g}_{\nu \lambda }\mathbf{g}_{\mu \gamma }\right) ~\ _{s}\mathbf{R,}\
\end{eqnarray*}%
defined by a d--metric $\mathbf{g}_{\mu \nu }$ and a metric compatible
d--connection $\mathbf{D}$ (in our approach, $\mathbf{D}$ can be any
d--connection constructed in a unique form from $\mathbf{g}_{\mu \nu }$ and $%
\mathbf{N}_{i}^{a}$ following a well defined geometric principle). For
simplicity, we shall work on a four dimensional space and use values {\small
\begin{eqnarray*}
&&\int d^{4}u~\sqrt{\det |e^{~2\phi }\mathbf{g}_{\mu \nu }|}\mathbf{R}%
(e^{2~\phi }\mathbf{g}_{\mu \nu })^{\ast ~}\mathbf{R}^{\ast }(e^{2~\phi }%
\mathbf{g}_{\mu \nu })= \\
&&\int d^{4}u~\sqrt{\det |\mathbf{g}_{\mu \nu }|}\mathbf{R}(\mathbf{g}_{\mu
\nu })^{\ast ~}\mathbf{R}^{\ast }(\mathbf{g}_{\mu \nu })= \frac{1}{4}\int
d^{4}u~\frac {\epsilon ^{\mu \nu \alpha \beta }\epsilon _{\rho \sigma \gamma
\delta }}{\left( \sqrt{\det |\mathbf{g}_{\mu \nu }|}\right)} \mathbf{R}%
_{\quad \mu \nu }^{\rho \sigma }\mathbf{R}_{\quad \alpha \beta }^{\gamma
\delta },
\end{eqnarray*}%
} for the curvature d--tensor $\mathbf{R}_{\quad \mu \nu }^{\rho \sigma }$,
where sub--integral values are defined by Chern-Gauss--Bonnet terms $\mathbf{%
R}^{\ast }\ \mathbf{R}^{\ast }\equiv \frac{1}{4\sqrt{\det |\mathbf{g}_{\mu
\nu }|}}\epsilon ^{\mu \nu \alpha \beta }\epsilon _{\rho \sigma \gamma
\delta }\mathbf{R}_{\mu \nu \quad }^{\ \ \rho \sigma }\mathbf{R}_{\alpha
\beta \quad }^{\ \ \gamma \delta }$.

One has the four dimensional approximation {\small
\begin{eqnarray}
&&Tr~\left[ ~f~(\chi )(~^{\phi }\mathcal{D}^{2}(\chi )/\Lambda ^{2})\right]
\simeq \frac{45}{4\pi ^{2}}~f_{(0)}\int \delta ^{4}u~e^{2\phi }\sqrt{\det |%
\mathbf{g}_{\mu \nu }|}+ \frac{15}{16\pi ^{2}} \times  \label{4dapr} \\
&& f_{(2)}\int \delta ^{4}u~e^{2\phi }\sqrt{\det |\mathbf{g}_{\mu \nu }|}
\left( ~_{s}\mathbf{R}(e^{2\phi }\mathbf{g}_{\mu \nu })+3e^{-2\phi }\mathbf{g%
}^{\alpha \beta }(\mathbf{e}_{\alpha }\phi ~\mathbf{e}_{\beta }\phi +\mathbf{%
e}_{\beta }\phi ~\mathbf{e}_{\alpha }\phi )\right)  \notag \\
&&+\frac{1}{128\pi ^{2}}~f_{(4)}\int \delta ^{4}u~e^{2\phi }\sqrt{\det |%
\mathbf{g}_{\mu \nu }|}\times  \notag \\
&&\left( 11~\mathbf{R}^{\ast }(e^{2\phi }\mathbf{g}_{\mu \nu })\mathbf{R}%
^{\ast }(e^{2\phi }\mathbf{g}_{\mu \nu })-18\mathbf{C}_{\mu \nu \lambda
\gamma }(e^{2\phi }\mathbf{g}_{\mu \nu })\mathbf{C}^{\mu \nu \lambda \gamma
}(e^{2\phi }\mathbf{g}_{\mu \nu })\right) .  \notag
\end{eqnarray}
} Let us state some additional hypotheses which will be used for proofs of
the theorems in this section: Hereafter we shall consider a four dimensional
compact N--anholonomic manifold $\mathbf{V,}$ with volume forms $\delta V=%
\sqrt{\det |\mathbf{g}_{\mu \nu }|}\delta ^{4}u$ and normalization $%
\int\nolimits_{\mathbf{V}}\delta V~\mu =1$ for $\mu =e^{-f}(4\pi \chi
)^{-(n+m)/2}$ with $f$ being a scalar function $f(\chi ,u)$ and $\chi >0.$

Now, we are able to formulate the main results of this section:
\begin{theorem}
\label{thmr1}For the scaling factor $~^{1}\phi =-f/2,$ the spectral
functional (\ref{ncpf1a}) can be approximated $\mathcal{F}=\ ^{P}\mathcal{F}(%
\mathbf{g,D,}f),$ where the first Perelman functional (in our case for
N--anholonomic Ricci flows) is
\begin{equation*}
\ ^{P}\mathcal{F}=\int\nolimits_{\mathbf{V}}\delta V~e^{-f}\left[ ~_{s}%
\mathbf{R}(e^{-f}\mathbf{g}_{\mu \nu })+\frac{3}{2}e^{f}\mathbf{g}^{\alpha
\beta }(\mathbf{e}_{\alpha }f~\mathbf{e}_{\beta }f+\mathbf{e}_{\beta }f~%
\mathbf{e}_{\alpha }f)\right] .
\end{equation*}
\end{theorem}

There are some important remarks.

\begin{remark}
\label{r15}For nonholonomic Ricci flows of (non)commutative geometries, we
have to adapt the evolution to certain N--connection structures (i.e.
nonholonomic constraints). This results in additional possibilities to
re--scale coefficients and parameters in spectral functionals and their
commutative limits:

\begin{enumerate}
\item The evolution parameter $\chi ,$ scaling factors $\ ^{b}f$ and
nonholonomic constraints and coordinates can be re--scaled/ redefined (for
instance, $\chi \rightarrow \check{\chi}$ and $\ ^{b}f\rightarrow \ ^{b}%
\check{f})$ such a way that the spectral functionals have limits to some
'standard' nonholonomic versions of Perelman functionals (with prescribed
types of coefficients).

\item Using additional dependencies on $\chi $ and freedom in choosing
scaling factors $\ ^{b}f(\chi ),$ we can prescribe such nonholonomic
constraints/ configurations on evolution equations (for instance, with $^{1}%
\check{f}_{(2)}=16\pi ^{2}/15$ and $~~^{1}\check{f}_{(0)}=~^{1}\check{f}%
_{(4)}=0)$ when the spectral functionals result exactly in necessary types
of effective Perelman \ functionals (with are commutative, but, in general,
nonholonomic).

\item For simplicity, we shall write in brief only $\chi $ and $f$
considering that we have chosen such scales, parametrizations of coordinates
and N--adapted frames and flow parameters when coefficients in spectral
functionals and resulting evolution equations maximally correspond to
certain generally accepted commutative physical actions/ functionals.

\item For nonholonomic Ricci flow models (commutative or noncommutative
ones) with a fixed evolution parameter $\chi ,$ we can construct certain
effective nonholonomic evolution models with induced noncommutative
corrections for coefficients.

\item Deriving effective nonholonomic evolution models from spectral
functionals, we can use the technique of ''extracting'' physical models from
spectral actions. For commutative and/or noncommutative geometric/ physical
models of nonholonomic Ricci flows, we have to generalize the approach to
include spectral functionals and N--adapted evolution equations depending on
the type of nonholonomic constraints, normalizations and re--scalings of
constants and effective conformal factors.
\end{enumerate}
\end{remark}

\vskip4ptWe ''extract'' from the second spectral functional (\ref{ncpf2a})
another very important physical value:

\begin{theorem}
\label{thmr2} The functional (\ref{ncpf2a}) is approximated\\ $\mathcal{W}=\ ^{P}%
\mathcal{W}(\mathbf{g,D,}f,\chi),$ where the second Perelman functional is%
{\small
\begin{equation*}
\ ^{P}\mathcal{W}=\int\nolimits_{\mathbf{V}}\delta V~\mu \times [ \chi
\left( ~_{s}\mathbf{R}(e^{-f}\mathbf{g}_{\mu \nu })+\frac{3}{2}e^{f}\mathbf{g%
}^{\alpha \beta }(\mathbf{e}_{\alpha }f~\mathbf{e}_{\beta }f+\mathbf{e}%
_{\beta }f~\mathbf{e}_{\alpha }f)\right) +f-(n+m)],
\end{equation*}
} for scaling $~^{2}\phi =-f/2$ in $~^{2}\mathcal{W}$ \ and $~^{3}\phi =(\ln
|f-(n+m)|-f)/2$ in $~^{3}\mathcal{W},$ \ from (\ref{ncpf2a}).
\end{theorem}

\vskip4pt The nonholonomic version of Hamilton equations (\ref{normcomrf})
can be derived from commutative Perelman functionals $\ ^{P}\mathcal{F}$ and
$\ ^{P}\mathcal{W}$. The original Hamil\-ton--Perelman Ricci flows
constructions can be generated for $\mathbf{D}=\nabla .$ The surprising
result is that even we start with a Levi--Civita linear connection, the
nonholonomic evolution will result almost sure in generalized geometric
configurations with various $\mathbf{N}$ and $\mathbf{D}$ structures.

\subsubsection{Spectral functionals for thermodynamical values}

Certain important thermodynamical values such as the average energy and
entropy can be derived directly from noncommutative spectral functionals as
respective commutative configurations of spectral functionals of type (\ref%
{ncpf1a}) and (\ref{ncpf2a}) but with different testing functions than in
Theorems \ref{thmr1} and \ref{thmr2}.

\begin{theorem}
\label{thae}Using a scaling factor of type $~^{1}\phi =-f/2,$ we extract
from the spectral functional (\ref{ncpf1a}) a nonholonomic version of
average energy, $\mathcal{F}\rightarrow <\mathcal{E}>,$ where {\small
\begin{equation}
\left\langle \mathcal{E}\right\rangle =-\chi ^{2}\int\nolimits_{\mathbf{V}%
}\delta V~\mu \left[ _{s}\mathbf{R}(e^{-f}\mathbf{g}_{\mu \nu })+\frac{3}{2}%
\mathbf{g}^{\alpha \beta }(\mathbf{e}_{\alpha }f~\mathbf{e}_{\beta }f+%
\mathbf{e}_{\beta }f~\mathbf{e}_{\alpha }f)-\frac{n+m}{2\chi }\right]
\label{naen}
\end{equation}%
} if the testing function is chosen to satisfy the conditions $%
~^{1}f_{(0)}(\chi )=4\pi ^{2}(n+m)\chi /45(4\pi \chi )^{(n+m)/2},$ $%
~^{1}f_{(2)}(\chi )=16\pi ^{2}\chi ^{2}/15(4\pi \chi )^{(n+m)/2}$ and $%
~^{1}f_{(4)}(\chi )=0.$
\end{theorem}

Similarly to Theorem \ref{thmr2} (inverting the sign of nontrivial
coefficients of the testing function) we prove:

\begin{theorem}
\label{thnhs}We extract a nonholonomic version of entropy of nonholonomic
Ricci flows from the functional (\ref{ncpf2a}), $\mathcal{W}\rightarrow ~%
\mathcal{S},$ where {\small
\begin{equation*}
\mathcal{S}=-\int\nolimits_{\mathbf{V}}\delta V~\mu [\chi \left( ~_{s}%
\mathbf{R}(e^{-f}\mathbf{g}_{\mu \nu })-\frac{3}{2}e^{f}\mathbf{g}^{\alpha
\beta }(\mathbf{e}_{\alpha }f~\mathbf{e}_{\beta }f+\mathbf{e}_{\beta }f~%
\mathbf{e}_{\alpha }f)\right) +f-(n+m)],
\end{equation*}
} if we introduce $\delta V=\delta ^{4}u$ and $\mu =e^{-f}(4\pi \chi
)^{-(n+m)/2}$ into formula (\ref{4dapr}), for $\chi >0$ and $\int\nolimits_{%
\mathbf{V}}dV~\mu =1$ in (\ref{4dapr}), for scaling $~^{2}\phi =-f/2$ in $%
~^{2}\mathcal{W}$ \ and $~^{3}\phi =(\ln |f-(n+m)|-f)/2$ in $~^{3}\mathcal{W}%
,$ \ from (\ref{ncpf2a}).
\end{theorem}

We can formulate and prove a Theorem alternative to Theorem \ref{thae} and
get the formula (\ref{naen}) from the spectral functional $~^{2}\mathcal{W+}%
~^{3}\mathcal{W}.$ Such a proof is similar to that for Theorem \ref{thmr2},
but with corresponding nontrivial coefficients for two testing functions $%
~~^{2}f(\chi )$ and $~~^{3}f(\chi ).$ The main difference is that for
Theorem \ref{thae} it is enough to use only one testing function. We do not
present such computations in this work.

It is not surprising that certain 'commutative' thermodynamical physical
values can be derived alternatively from different spectral functionals
because such type 'commutative' thermodynamical values can be generated by a
partition function
\begin{equation}
\widehat{Z}=\exp \left\{ \int\nolimits_{\mathbf{V}}\delta V~\mu \left[ -f+%
\frac{n+m}{2}\right] \right\} ,  \label{nhpf}
\end{equation}%
associated to any $Z=\int \exp (-\beta E)d\omega (E)$ being the partition
function for a canonical ensemble at temperature $\beta ^{-1},$ which in it
turn is defined by the measure taken to be the density of states $\omega
(E). $ In this case, we can compute the average energy, $\ \left\langle
E\right\rangle =-\partial \log Z/\partial \beta ,$ the entropy $S=\beta
\left\langle E\right\rangle +\log Z$ and the fluctuation $\sigma
=\left\langle (E-\left\langle E\right\rangle )^{2}\right\rangle =\partial
^{2}\log Z/\partial \beta ^{2}.$

\begin{remark}
 Following a straightforward computation for (\ref{nhpf}) we prove that
{\small
\begin{equation}
\widehat{\sigma }=2\chi ^{2}\int\nolimits_{\mathbf{V}}\delta V~\mu \left[
\left| R_{ij}+D_{i}D_{j}f-\frac{1}{2\chi }g_{ij}\right| ^{2}+\left|
R_{ab}+D_{a}D_{b}f-\frac{1}{2\chi }g_{ab}\right| ^{2}\right] .
\label{nhfluct}
\end{equation}
}
\end{remark}

Using formula $\mathbf{R}_{\mu \nu }^{2}\mathbf{=}\frac{1}{2}\mathbf{C}_{\mu
\nu \rho \sigma }^{2}-\frac{1}{2}\mathbf{R}^{\ast ~}\mathbf{R}^{\ast }+\frac{%
1}{3}\ _{s}\mathbf{R}^{2}$ (it holds true for any metric compatible
d--connections, similarly to the formula for the Levi--Civita connection, we
expect that the formula for fluctuations (\ref{nhfluct}) can be generated
directly, by corresponding re-scalings, from a spectral action with
nontrivial coefficients for testing functions when $~f_{(4)}\neq 0,$ see
formula (\ref{4dapr}). Here we note that in the original Perelman's
functionals there were not introduced terms being quadratic on curvature/
Weyl / Ricci tensors. For nonzero $~f_{(4)},$ such terms may be treated as
certain noncommutative / quantum contributions to the classical commutative
Ricci flow theory. For simplicity, we omit such considerations in this work.

The framework of Perelman's functionals and generalizations to corresponding
spectral functionals can be positively applied for developing statistical
analogies of (non) commutative Ricci flows. For instance, the functional $%
\mathcal{W}$ is the ''opposite sign'' entropy, see formulas from Theorems %
\ref{thmr2} and \ref{thnhs}. Such constructions may be considered for a
study of optimal ''topological'' configurations and evolution of both
commutative and noncommutative geometries and relevant theories of physical
interactions.

\subsection{(Non) commutative gauge gravity}
We consider main results of Refs. \cite{v1,ve19} concerning noncommutative
gauge models of gravity:

The basic idea was to use a geometrical result due to D. A. Popov and I. I.
Dikhin (1976) that the Einstein gravity can be equivalently represented as a
gauge theory with a Cartan type connection in the bundle of affine frames.
Such gauge theories are with nonsemisimple structure gauge groups, i. e.
with degenerated metrics in the total spaces. Using an auxiliary symmetric
form for the typical fiber, any such model can be transformed into a
variational one. There is an alternative way to construct geometrically a
usual Yang--Mills theory by applying a corresponding set of absolute
derivations and dualities defined by the Hodge operator. For both
approaches, there is a projection formalism reducing the geometric field
equations on the base space to be exactly the Einstein equations from the
general relativity theory.

For more general purposes, it was suggested to consider also extensions to a
nonlinear realization with the (anti) de Sitter gauge structural group (A.
Tseytlin, 1982). The constructions with nonlinear group realizations are
very important because they prescribe a consistent approach of
distinguishing the frame indices and coordinate indices subjected to
different rules of transformation. This approach to gauge gravity (of
course, after a corresponding generalizations of the Seiberg--Witten map)
may include, in general, quadratic on curvature and torsion terms.

\subsubsection{Nonlinear gauge models for the (anti) de Sitter group}

We introduce vielbein decompositions of (in general) complex metrics
\begin{equation*}
\widehat{g}_{\alpha \beta }(u) = e_{\alpha }^{\ \alpha ^{\prime }}\left(
u\right) e_{\beta }^{\ \beta ^{\prime }}\left( u\right) \eta _{\alpha
^{\prime }\beta ^{\prime }},\ e_{\alpha }^{\ \alpha ^{\prime }}e_{\ \alpha
^{\prime }}^{\beta } =\delta _{\alpha }^{\beta }\mbox{ and }e_{\alpha }^{\
\alpha ^{\prime }}e_{\ \beta ^{\prime }}^{\alpha }=\delta _{\beta ^{\prime
}}^{\alpha ^{\prime }},
\end{equation*}
where $\eta _{\alpha ^{\prime }\beta ^{\prime }}$ is a constant diagonal
matrix (for real spacetimes we can consider it as the flat Minkowski metric,
for instance, $\eta _{\alpha ^{\prime }\beta ^{\prime }}=diag\left(
-1,+1,...,+1\right) $) and $\delta _{\alpha }^{\beta }$ and $\delta _{\beta
^{\prime }}^{\alpha ^{\prime }}$ are Kronecker's delta symbols. The
vielbiens with an associated N--connection structure \ $N_{i}^{a}\left(
x^{j},y^{a}\right) ,$ being real or complex valued functions, have a special
parametrization%
\begin{equation}
e_{\alpha }^{\ \alpha ^{\prime }}(u)=\left[
\begin{array}{cc}
e_{i}^{\ i^{\prime }}\left( x^{j}\right) & N_{i}^{c}\left(
x^{j},y^{a}\right) \ e_{c}^{\ b^{\prime }}\left( x^{j},y^{a}\right) \\
0 & e_{e}^{\ e^{\prime }}\left( x^{j},y^{a}\right)%
\end{array}%
\right]  \label{viel1}
\end{equation}%
and
\begin{equation}
e_{\ \alpha ^{\prime }}^{\alpha \ }(u)=\left[
\begin{array}{cc}
e_{\ i^{\prime }}^{i}\left( x^{j}\right) & -N_{i}^{c}\left(
x^{j},y^{a}\right) \ e_{\ i^{\prime }}^{i\ }\left( x^{j}\right) \\
0 & e_{\ c^{\prime }}^{c}\left( x^{j},y^{a}\right)%
\end{array}%
\right]  \label{viel2}
\end{equation}%
with $e_{i}^{\ i^{\prime }}\left( x^{j}\right) $ and $e_{c}^{\ b^{\prime
}}\left( x^{j},y^{a}\right) $ generating the coefficients of a metric
defined with respect to anholonmic frames, {\small
\begin{equation}
g_{ij}\left( x^{j}\right) =e_{i}^{\ i^{\prime }}\left( x^{j}\right) e_{j}^{\
j^{\prime }}\left( x^{j}\right) \eta _{i^{\prime }j^{\prime }}\mbox{ and }%
h_{ab}\left( x^{j},y^{c}\right) =e_{a}^{\ a^{\prime }}\left(
x^{j},y^{c}\right) e_{b}^{\ b^{\prime }}\left( x^{j},y^{c}\right) \eta
_{a^{\prime }b^{\prime }}.  \label{metr01}
\end{equation}%
} By using vielbeins and metrics of type (\ref{viel1}) and (\ref{viel2})
and, respectively, (\ref{metr01}), we can model in a unified manner various
types of (pseudo) Riemannian, Einstein--Cartan, Riemann--Finsler and vector/
covector bundle nonlinear connection commutative and noncommutative
geometries in effective gauge and string theories (it depends on the
parametrization of $e_{i}^{\ i^{\prime }},e_{c}^{\ b^{\prime }}$ and $%
N_{i}^{c}$ on coordinates and anholonomy relations).

We consider the de Sitter space $\Sigma ^{4}$ as a hypersurface defined by
the equations $\eta _{AB}u^{A}u^{B}=-l^{2}$ in the four dimensional flat
space enabled with diagonal metric $\eta _{AB},\eta _{AA}=\pm 1$ (in this
section $A,B,C,...=1,2,...,5),$ where $\{u^{A}\}$ are global Cartesian
coordinates in $\R^{5};l>0$ is the curvature of de Sitter space (for
simplicity, we consider here only the de Sitter case; the anti--de Sitter
configuration is to be stated by a hypersurface $\eta
_{AB}u^{A}u^{B}=l^{2}). $ The de Sitter group $S_{\left( \eta \right)
}=SO_{\left( \eta \right) }\left( 5\right) $ is the isometry group of $%
\Sigma ^{5}$--space with $6$ generators of Lie algebra ${\mathit{s}o}%
_{\left( \eta \right) }\left( 5\right) $ satisfying the commutation
relations
\begin{equation}
\left[ M_{AB},M_{CD}\right] =\eta _{AC}M_{BD}-\eta _{BC}M_{AD}-\eta
_{AD}M_{BC}+\eta _{BD}M_{AC}.  \label{dsc}
\end{equation}

We can decompose the capital indices $A,B,...$ as $A=\left( \alpha ^{\prime
},5\right) ,B=\left( \beta ^{\prime },5\right) ,...,$ and the metric $\eta
_{AB}$ as $\eta _{AB}=\left( \eta _{\alpha ^{\prime }\beta ^{\prime }},\eta
_{55}\right) .$ The operators (\ref{dsc}) $M_{AB}$ can be decomposed as $%
M_{\alpha ^{\prime }\beta ^{\prime }}=\mathcal{F}_{\alpha ^{\prime }\beta
^{\prime }}$ and $P_{\alpha ^{\prime }}=l^{-1}M_{5\alpha ^{\prime }}$
written as
\begin{eqnarray}
\left[ \mathcal{F}_{\alpha ^{\prime }\beta ^{\prime }},\mathcal{F}_{\gamma
^{\prime }\delta ^{\prime }}\right] &=&\eta _{\alpha ^{\prime }\gamma
^{\prime }}\mathcal{F}_{\beta ^{\prime }\delta ^{\prime }}-\eta _{\beta
^{\prime }\gamma ^{\prime }}\mathcal{F}_{\alpha ^{\prime }\delta ^{\prime
}}+\eta _{\beta ^{\prime }\delta ^{\prime }}\mathcal{F}_{\alpha ^{\prime
}\gamma ^{\prime }}-\eta _{\alpha ^{\prime }\delta ^{\prime }}\mathcal{F}%
_{\beta ^{\prime }\gamma ^{\prime }},  \notag \\
\left[ P_{\alpha ^{\prime }},P_{\beta ^{\prime }}\right] &=&-l^{-2}\mathcal{F%
}_{\alpha ^{\prime }\beta ^{\prime }},\
\left[ P_{\alpha ^{\prime }},\mathcal{F}_{\beta ^{\prime }\gamma ^{\prime }}%
\right] =\eta _{\alpha ^{\prime }\beta ^{\prime }}P_{\underline{\gamma }%
}-\eta _{\alpha ^{\prime }\gamma ^{\prime }}P_{\beta ^{\prime }}, \label{dsca}
\end{eqnarray}%
where the Lie algebra ${\mathit{s}o}_{\left( \eta \right) }\left( 5\right) $
is split into a direct sum, ${\mathit{s}o}_{\left( \eta \right) }\left(
5\right) ={\mathit{s}o}_{\left( \eta \right) }(4)\oplus V_{4}$ with $V_{4}$
being the vector space stretched on vectors $P_{\underline{\alpha }}.$ We
remark that $\Sigma ^{4}=S_{\left( \eta \right) }/L_{\left( \eta \right) },$
where $L_{\left( \eta \right) }=SO_{\left( \eta \right) }\left( 4\right) .$
For $\eta _{AB}=diag\left( -1,+1,+1,+1\right) $ and $S_{10}=SO\left(
1,4\right) ,L_{6}=SO\left( 1,3\right) $ is the group of Lorentz rotations.

The generators $I^{\underline{a}}$ and structure constants $f_{~\underline{t}%
}^{\underline{s}\underline{p}}$ of the de Sitter Lie group can be
paramet\-riz\-ed in a form distinguishing the de Sitter generators and
commutations (\ref{dsca}). The action of the group $S_{\left( \eta \right) }$
may be realized by using $4\times 4$ matrices with a parametrization
distinguishing the subgroup $L_{\left( \eta \right) }:\
 B=bB_{L}$,
where $B_{L}=\left(
\begin{array}{cc}
L & 0 \\
0 & 1%
\end{array}%
\right),$ $L\in L_{\left( \eta \right) }$ is the de Sitter bust matrix
transforming the vector $\left( 0,0,...,\rho \right) \in {\R}^{5}$ into the
arbitrary point $\left( V^{1},V^{2},...,V^{5}\right) \in \Sigma _{\rho
}^{5}\subset \mathcal{R}^{5}$ with curvature $\rho ,$ $(V_{A}V^{A}=-\rho
^{2},V^{A}=\tau ^{A}\rho ),$ and the matrix $b$ is expressed $b=\left(
\begin{array}{cc}
\delta _{\quad \beta ^{\prime }}^{\alpha ^{\prime }}+\frac{\tau ^{\alpha
^{\prime }}\tau _{\beta ^{\prime }}}{\left( 1+\tau ^{5}\right) } & \tau
^{\alpha ^{\prime }} \\
\tau _{\beta ^{\prime }} & \tau ^{5}%
\end{array}%
\right)$. The de Sitter gauge field is associated with a ${\mathit{s}o}_{\left( \eta
\right) }\left( 5\right) $--valued connection 1--form
\begin{equation}
\widetilde{\Omega }=\left(
\begin{array}{cc}
\omega _{\quad \beta ^{\prime }}^{\alpha ^{\prime }} & \widetilde{\theta }%
^{\alpha ^{\prime }} \\
\widetilde{\theta }_{\beta ^{\prime }} & 0%
\end{array}%
\right) ,  \label{dspot}
\end{equation}%
where $\omega _{\quad \beta ^{\prime }}^{\alpha ^{\prime }}\in so(4)_{\left(
\eta \right) },$ $\widetilde{\theta }^{\alpha ^{\prime }}\in \mathcal{R}^{4},%
\widetilde{\theta }_{\beta ^{\prime }}\in \eta _{\beta ^{\prime }\alpha
^{\prime }}\widetilde{\theta }^{\alpha ^{\prime }}.$

The actions of $S_{\left( \eta \right) }$ mix the components of the matrix $%
\omega _{\quad \beta ^{\prime }}^{\alpha ^{\prime }}$ and $\widetilde{\theta
}^{\alpha ^{\prime }}$ fields in (\ref{dspot}). Because the introduced
para\-met\-ri\-za\-ti\-on is invariant on action on $SO_{\left( \eta \right)
}\left( 4\right) $ group, we cannot identify $\omega _{\quad \beta ^{\prime
}}^{\alpha ^{\prime }}$ and $\widetilde{\theta }^{\alpha ^{\prime }},$
respectively, with the connection $\Gamma ^{\lbrack c]}$ and the 1--form $%
e^{\alpha }$ defined by a N--connection structure with the coefficients
chosen as in (\ref{viel1}) and (\ref{viel2}). To avoid this difficulty we \
can consider nonlinear gauge realizations of the de Sitter group $S_{\left(
\eta \right) }$ by introducing the nonlinear gauge field
\begin{equation}
\Gamma =b^{-1}{\widetilde{\Omega }}b+b^{-1}db=\left(
\begin{array}{cc}
\Gamma _{~\beta ^{\prime }}^{\alpha ^{\prime }} & \theta ^{\alpha ^{\prime }}
\\
\theta _{\beta ^{\prime }} & 0%
\end{array}%
\right) ,  \label{npot}
\end{equation}%
\begin{eqnarray}
\mbox{ where } \Gamma _{\quad \beta ^{\prime }}^{\alpha ^{\prime }}
&=&\omega _{\quad \beta ^{\prime }}^{\alpha ^{\prime }}-\left( \tau ^{\alpha
^{\prime }}D\tau _{\beta ^{\prime }}-\tau _{\beta ^{\prime }}D\tau ^{\alpha
^{\prime }}\right) /\left( 1+\tau ^{5}\right) ,  \notag \\
\theta ^{\alpha ^{\prime }} &=&\tau ^{5}\widetilde{\theta }^{\alpha ^{\prime
}}+D\tau ^{\alpha ^{\prime }}-\tau ^{\alpha ^{\prime }}\left( d\tau ^{5}+%
\widetilde{\theta }_{\gamma ^{\prime }}\tau ^{\gamma ^{\prime }}\right)
/\left( 1+\tau ^{5}\right) ,  \notag \\
D\tau ^{\alpha ^{\prime }} &=&d\tau ^{\alpha ^{\prime }}+\omega _{\quad
\beta ^{\prime }}^{\alpha ^{\prime }}\tau ^{\beta ^{\prime }}.  \notag
\end{eqnarray}

The action of the group $S\left( \eta \right) $ is nonlinear, yielding the
transformation rules $\Gamma ^{\prime }=L^{\prime }\Gamma \left( L^{\prime
}\right) ^{-1}+L^{\prime }d\left( L^{\prime }\right) ^{-1},~\theta ^{\prime
}=L\theta $, where the nonlinear matrix--valued function $L^{\prime
}=L^{\prime }\left( \tau ^{\alpha },b,B_{T}\right)$ is defined from $%
B_{b}=b^{\prime }B_{L^{\prime }}$. The de Sitter 'nonlinear' algebra is defined by generators (\ref{dsca})
and nonlinear gauge transforms of type (\ref{npot}).

\subsubsection{De Sitter Nonlinear Gauge Gravity and General Relativity}

We generalize the constructions from Refs \cite{v1} to the case when the de
Sitter nonlinear gauge gravitational connection (\ref{npot}) is defined by
the viebeins (\ref{viel1}) and (\ref{viel2}) and the linear connection $%
\Gamma _{\quad \beta \mu }^{[c]\alpha }=\{\Gamma _{\quad \beta \mu }^{\alpha
}\},$
\begin{equation}
\Gamma =\left(
\begin{array}{cc}
\Gamma _{\quad \beta ^{\prime }}^{\alpha ^{\prime }} & l_{0}^{-1}e^{\alpha
^{\prime }} \\
l_{0}^{-1}e_{\beta ^{\prime }} & 0%
\end{array}%
\right)  \label{conds}
\end{equation}%
where
\begin{eqnarray}
\Gamma _{\quad \beta ^{\prime }}^{\alpha ^{\prime }}&=&\Gamma _{\quad \beta
^{\prime }\mu }^{\alpha ^{\prime }}\delta u^{\mu },  \label{condsc} \\
 \mbox{ for \ }
\Gamma _{\quad \beta ^{\prime }\mu }^{\alpha ^{\prime }} &=&e_{\alpha
}^{~\alpha ^{\prime }}e_{\quad \beta ^{\prime }}^{\beta }\Gamma _{\quad
\beta \mu }^{\alpha }+e_{\alpha }^{~\alpha ^{\prime }}\delta _{\mu }e_{\quad
\beta ^{\prime }}^{\alpha },\  e^{\alpha ^{\prime }} = e_{\mu }^{~\alpha ^{\prime }}\delta u^{\mu },  \notag
\end{eqnarray}%
and $l_{0}$ being a dimensional constant.

The matrix components of the curvature of the connection (\ref{conds}),
\begin{equation*}
\mathcal{R}^{(\Gamma )}=d\Gamma +\Gamma \wedge \Gamma ,
\end{equation*}%
can be written
\begin{equation}
\mathcal{R}^{(\Gamma )}=\left(
\begin{array}{cc}
\mathcal{R}_{\quad \beta ^{\prime }}^{\alpha ^{\prime }}+l_{0}^{-1}\pi
_{\beta ^{\prime }}^{\alpha ^{\prime }} & l_{0}^{-1}T^{\alpha ^{\prime }} \\
l_{0}^{-1}T^{\beta ^{\prime }} & 0%
\end{array}%
\right) ,  \label{curvs}
\end{equation}%
for $\pi _{\beta ^{\prime }}^{\alpha ^{\prime }}=e^{\alpha ^{\prime }}\wedge
e_{\beta ^{\prime }},~\mathcal{R}_{\quad \beta ^{\prime }}^{\alpha ^{\prime
}}=\frac{1}{2}\mathcal{R}_{\quad \beta ^{\prime }\mu \nu }^{\alpha ^{\prime
}}\delta u^{\mu }\wedge \delta u^{\nu },$ $\mathcal{R}_{\quad \beta ^{\prime
}\mu \nu }^{\alpha ^{\prime }}=e_{~\beta ^{\prime }}^{\beta }e_{\alpha
}^{\quad \alpha ^{\prime }}R_{\quad \beta _{\mu \nu }}^{\alpha },$ with the
coefficients $R_{\quad \beta {\mu \nu }}^{\alpha }$ defined with
h--v--invariant components.

The de Sitter gauge group is semisimple: we are able to construct a
variational gauge gravitational theory with the Lagrangian
\begin{equation}
L=L_{\left( g\right) }+L_{\left( m\right) }  \label{lagrangc}
\end{equation}%
where the gauge gravitational Lagrangian is defined
\begin{equation*}
L_{\left( g\right) }=\frac{1}{4\pi }Tr\left( \mathcal{R}^{(\Gamma )}\wedge
\ast _{G}\mathcal{R}^{(\Gamma )}\right) =\mathcal{L}_{\left( G\right)
}\left| g\right| ^{1/2}\delta ^{4}u,
\end{equation*}%
for $\mathcal{L}_{\left( g\right) }=\frac{1}{2l^{2}}T_{\quad \mu \nu
}^{\alpha ^{\prime }}T_{\alpha ^{\prime }}^{\quad \mu \nu }+\frac{1}{%
8\lambda }\mathcal{R}_{\quad \beta ^{\prime }\mu \nu }^{\alpha ^{\prime }}%
\mathcal{R}_{\quad \alpha ^{\prime }}^{\beta ^{\prime }\quad \mu \nu }{}-%
\frac{1}{l^{2}}\left( {\overleftarrow{R}}\left( \Gamma \right) -2\lambda
_{1}\right),$ with $\delta ^{4}u$ being the volume element, $\left| g\right|
$ is the determinant computed the metric coefficients stated with respect to
N--elongated frames, $T_{\quad \mu \nu }^{\alpha ^{\prime }}=e_{\quad \alpha
}^{\alpha ^{\prime }}T_{\quad \mu \nu }^{\alpha }$ (the gravitational
constant $l^{2}$ satisfies the relations $l^{2}=2l_{0}^{2}\lambda ,\lambda
_{1}=-3/l_{0}),\quad Tr$ denotes the trace on $\alpha ^{\prime },\beta
^{\prime }$ indices. The matter field Lagrangian from (\ref{lagrangc}) is
defined
\begin{equation*}
L_{\left( m\right) }=-\frac{1}{2}Tr\left( \Gamma \wedge \ast _{g}\mathcal{I}%
\right) =\mathcal{L}_{\left( m\right) }\left| g\right| ^{1/2}\delta ^{n}u,
\end{equation*}%
with the Hodge operator derived by $\left| g\right| $ and $\left| h\right| $
where
\begin{equation*}
\mathcal{L}_{\left( m\right) }=\frac{1}{2}\Gamma _{\quad \beta ^{\prime }\mu
}^{\alpha ^{\prime }}S_{\quad \alpha }^{\beta ^{\prime }\quad \mu }-t_{\quad
\alpha ^{\prime }}^{\mu }l_{\quad \mu }^{\alpha ^{\prime }}.
\end{equation*}%
The matter field source $\mathcal{J}$ is obtained as a variational
derivation of $\mathcal{L}_{\left( m\right) }$ on $\Gamma $ and is
paramet\-riz\-ed in the form $\mathcal{J}=\left(
\begin{array}{cc}
S_{\quad \underline{\beta }}^{\alpha ^{\prime }} & -l_{0}\tau ^{\alpha
^{\prime }} \\
-l_{0}\tau _{\beta ^{\prime }} & 0%
\end{array}%
\right),$ with $\tau ^{\alpha ^{\prime }}=\tau _{\quad \mu }^{\alpha
^{\prime }}\delta u^{\mu }$ and $S_{\quad \beta ^{\prime }}^{\alpha ^{\prime
}}=S_{\quad \beta ^{\prime }\mu }^{\alpha ^{\prime }}\delta u^{\mu }$ being
respectively the canonical tensors of energy--momentum and spin density.

Varying the action $S=\int \delta ^{4}u\left( \mathcal{L}_{\left( g\right) }+%
\mathcal{L}_{\left( m\right) }\right)$ on the $\Gamma $--variables (\ref%
{conds}), we obtain the gau\-ge--gra\-vi\-ta\-ti\-on\-al field equations:%
\begin{equation}
d\left( \ast \mathcal{R}^{(\Gamma )}\right) +\Gamma \wedge \left( \ast
\mathcal{R}^{(\Gamma )}\right) -\left( \ast \mathcal{R}^{(\Gamma )}\right)
\wedge \Gamma =-\lambda \left( \ast \mathcal{J}\right) ,  \label{eqs}
\end{equation}%
were the Hodge operator $\ast $ is used. This equations can be alternatively
derived in geometric form by applying the absolute derivation and dual
operators.

Distinguishing the variations on $\Gamma $ and $e$--variables, we rewrite (%
\ref{eqs})
\begin{eqnarray*}
\widehat{\mathcal{D}}\left( \ast \mathcal{R}^{(\Gamma )}\right) +\frac{%
2\lambda }{l^{2}}(\widehat{\mathcal{D}}\left( \ast \pi \right) +e\wedge
\left( \ast T^{T}\right) -\left( \ast T\right) \wedge e^{T}) &=&-\lambda
\left( \ast S\right) , \\
\widehat{\mathcal{D}}\left( \ast T\right) -\left( \ast \mathcal{R}^{(\Gamma
)}\right) \wedge e-\frac{2\lambda }{l^{2}}\left( \ast \pi \right) \wedge e
&=&\frac{l^{2}}{2}\left( \ast t+\frac{1}{\lambda }\ast \varsigma \right) ,
\end{eqnarray*}%
$e^{T}$ being the transposition of \ $e,$ where
\begin{eqnarray}
T^{t} &=&\{T_{\alpha ^{\prime }}=\eta _{\alpha ^{\prime }\beta ^{\prime
}}T^{\beta ^{\prime }},~T^{\beta ^{\prime }}=\frac{1}{2}T_{\quad \mu \nu
}^{\beta ^{\prime }}\delta u^{\mu }\wedge \delta u^{\nu }\},  \notag \\
e^{T} &=&\{e_{\alpha ^{\prime }}=\eta _{\alpha ^{\prime }\beta ^{\prime
}}e^{\beta ^{\prime }},~e^{\beta ^{\prime }}=e_{\quad \mu }^{\beta ^{\prime
}}\delta u^{\mu }\},\qquad \widehat{\mathcal{D}}=\delta +\widehat{\Gamma },
\notag
\end{eqnarray}%
($\widehat{\Gamma }$ acts as $\Gamma _{\quad \beta ^{\prime }\mu }^{\alpha
^{\prime }}$ on indices $\gamma ^{\prime },\delta ^{\prime },...$ and as $%
\Gamma _{\quad \beta \mu }^{\alpha }$ on indices $\gamma ,\delta ,...).$ The
value $\varsigma $ defines the energy--momentum tensor of the gauge
gravitational field $\widehat{\Gamma }:\
 \varsigma _{\mu \nu }\left( \widehat{\Gamma }\right) =\frac{1}{2}Tr\left(
\mathcal{R}_{\mu \alpha }\mathcal{R}_{\quad \nu }^{\alpha }-\frac{1}{4}%
\mathcal{R}_{\alpha \beta }\mathcal{R}^{\alpha \beta }G_{\mu \nu }\right).$

Equations (\ref{eqs}) make up the complete system of variational field
equations for the nonlinear de Sitter gauge gravity. We note that we can
obtain a nonvariational Poincar\' e gauge gravitational theory if we
consider the contraction of the gauge potential (\ref{conds}) to a potential
$\Gamma ^{\lbrack P]}$ with values in the Poincar\' e Lie algebra
\begin{equation}
\Gamma =\left(
\begin{array}{cc}
\Gamma _{\quad \beta ^{\prime }}^{\alpha ^{\prime }} & l_{0}^{-1}e^{\alpha
^{\prime }} \\
l_{0}^{-1}e_{\beta ^{\prime }} & 0%
\end{array}%
\right) \rightarrow \Gamma ^{\lbrack P]}=\left(
\begin{array}{cc}
\Gamma _{\quad \beta ^{\prime }}^{\alpha ^{\prime }} & l_{0}^{-1}e^{\alpha
^{\prime }} \\
0 & 0%
\end{array}%
\right) .  \label{poinc}
\end{equation}%
A similar gauge potential was considered in the formalism of linear and
affine frame bundles on curved spacetimes by D. Popov and I. Dikhin. They
considered the gauge potential (\ref{poinc}) to be just the Cartan
connection form in the affine gauge like gravity and proved that the
Yang--Mills equations of their theory are equivalent, after projection on
the base, to the Einstein equations.

\subsubsection{Enveloping algebras for gauge gravity connections}

We define the gauge fields on a noncommutative space as elements of an
algebra $\mathcal{A}_{u}$ that form a representation of the generator $I$%
--algebra for the de Sitter gauge group and the noncommutative space is
modelled as the associative algebra of $\C.$\ This algebra is freely
generated by the coordinates modulo ideal $\mathcal{R}$ generated by the
relations (one accepts formal power series)\ $\mathcal{A}_{u}=%
\C[[{\hat
u}^1,...,{\hat u}^N]]/\mathcal{R}.$ A variational gauge gravitational theory
can be formulated by using a minimal extension of the affine structural
group ${\mathcal{A}f}_{3+1}\left( {\R}\right) $ to the de Sitter gauge group
$S_{10}=SO\left( 4+1\right) $ acting on ${\R}^{4+1}$.

The gauge fields are elements of the algebra $\widehat{\psi }\in \mathcal{A}%
_{I}^{(dS)}$ that form the nonlinear representation of the de Sitter algebra
${\mathit{s}o}_{\left( \eta \right) }\left( 5\right) $ (the whole algebra is
denoted $\mathcal{A}_{z}^{(dS)}).$ The elements transform $\delta \widehat{%
\psi }=i\widehat{\gamma }\widehat{\psi },\widehat{\psi }\in \mathcal{A}_{u},%
\widehat{\gamma }\in \mathcal{A}_{z}^{(dS)},$ under a nonlinear de Sitter
transformation. The action of the generators (\ref{dsca}) on $\widehat{\psi }
$ is defined as the resulting element will form a nonlinear representation
of $\mathcal{A}_{I}^{(dS)}$ and, in consequence, $\delta \widehat{\psi }\in
\mathcal{A}_{u}$ despite $\widehat{\gamma }\in \mathcal{A}_{z}^{(dS)}.$ We
emphasize that for any representation the object $\widehat{\gamma }$ takes
values in enveloping de Sitter algebra but not in a Lie algebra as would be
for commuting spaces. We introduce a connection $\widehat{\Gamma }^{\nu }\in
\mathcal{A} _{z}^{(dS)} $ in order to define covariant coordinates, $%
\widehat{U}^{\nu }=\widehat{u}^{v}+\widehat{\Gamma }^{\nu }$. The values $%
\widehat{U}^{\nu }\widehat{\psi }$ transform covariantly, i. e. $\delta
\widehat{U}^{\nu }\widehat{\psi }=i\widehat{\gamma }\widehat{U}^{\nu }%
\widehat{\psi },$ if and only if the connection $\widehat{\Gamma }^{\nu }$
satisfies the transformation law of the enveloping nonlinear realized de
Sitter algebra, $\delta \widehat{\Gamma }^{\nu }\widehat{\psi }=-i[\widehat{u%
}^{v},\widehat{\gamma }]+i[\widehat{\gamma },\widehat{\Gamma }^{\nu }]$,
where $\delta \widehat{\Gamma }^{\nu }\in \mathcal{A}_{z}^{(dS)}.$

The enveloping algebra--valued connection has infinitely many component
fields. Nevertheless, all component fields can be induced from a Lie
algebra--valued connection by a Seiberg--Witten map for $SO(n)$ and $Sp(n)).$
Here, we show that similar constructions can be performed for nonlinear
realizations of de Sitter algebra when the transformation of the connection
is considered$\delta \widehat{\Gamma }^{\nu }=-i[u^{\nu },^{\ast }~\widehat{%
\gamma }]+i[\widehat{\gamma },^{\ast }~\widehat{\Gamma }^{\nu }]$. We treat
in more detail the canonical case with the star product. The first term in
the variation $\delta \widehat{\Gamma }^{\nu }$ gives $-i[u^{\nu },^{\ast }~%
\widehat{\gamma }]=\theta ^{\nu \mu }\frac{\partial }{\partial u^{\mu }}%
\gamma$. Assuming that the variation of $\widehat{\Gamma }^{\nu }=\theta
^{\nu \mu }Q_{\mu }$ starts with a linear term in $\theta ,$ we have $\delta
\widehat{\Gamma }^{\nu }=\theta ^{\nu \mu }\delta Q_{\mu },~\delta Q_{\mu }=%
\frac{\partial }{\partial u^{\mu }}\gamma +i[\widehat{\gamma },^{\ast
}~Q_{\mu }]$. We expand the star product in $\theta $ but not in $g_{a}$ and
find up to first order in $\theta $ that
\begin{equation}
\gamma =\gamma _{\underline{a}}^{1}I^{\underline{a}}+\gamma _{\underline{a}%
\underline{b}}^{1}I^{\underline{a}}I^{\underline{b}}+...,Q_{\mu }=q_{\mu ,%
\underline{a}}^{1}I^{\underline{a}}+q_{\mu ,\underline{a}\underline{b}%
}^{2}I^{\underline{a}}I^{\underline{b}}+...  \label{seriesa}
\end{equation}%
where $\gamma _{\underline{a}}^{1}$ and $q_{\mu ,\underline{a}}^{1}$ are of
order zero in $\theta $ and $\gamma _{\underline{a}\underline{b}}^{1}$ and $%
q_{\mu ,\underline{a}\underline{b}}^{2}$ are of second order in $\theta .$
The expansion in $I^{\underline{b}}$ leads to an expansion in $g_{a}$ of the
$\ast $--product because the higher order $I^{\underline{b}}$--derivatives
vanish. For de Sitter case, we take the generators $I^{\underline{b}}$ (\ref%
{dsca}), with the corresponding de Sitter structure constants $f_{~%
\underline{d}}^{\underline{b}\underline{c}}\simeq f_{~\underline{\beta }}^{%
\underline{\alpha }\underline{\beta }}$ (in our further identifications with
spacetime objects like frames and connections we shall use Greek indices).
The result of calculation of variations of (\ref{seriesa}), by using $g_{a}$%
, is
\begin{eqnarray}
\delta q_{\mu ,\underline{a}}^{1} &=&\frac{\partial \gamma _{\underline{a}%
}^{1}}{\partial u^{\mu }}-f_{~\underline{a}}^{\underline{b}\underline{c}%
}\gamma _{\underline{b}}^{1}q_{\mu ,\underline{c}}^{1},\ \delta Q_{\tau
}=\theta ^{\mu \nu }\partial _{\mu }\gamma _{\underline{a}}^{1}\partial
_{\nu }q_{\tau ,\underline{b}}^{1}I^{\underline{a}}I^{\underline{b}}+...,
\notag \\
\delta q_{\mu ,\underline{a}\underline{b}}^{2} &=&\partial _{\mu }\gamma _{%
\underline{a}\underline{b}}^{2}-\theta ^{\nu \tau }\partial _{\nu }\gamma _{%
\underline{a}}^{1}\partial _{\tau }q_{\mu ,\underline{b}}^{1}-2f_{~%
\underline{a}}^{\underline{b}\underline{c}}\{\gamma _{\underline{b}%
}^{1}q_{\mu ,\underline{c}\underline{d}}^{2}+\gamma _{\underline{b}%
\underline{d}}^{2}q_{\mu ,\underline{c}}^{1}\}.  \notag
\end{eqnarray}

Let us introduce the objects $\varepsilon ,$ taking the values in de Sitter
Lie algebra and $W_{\mu },$ taking values in the enveloping de Sitter
algebra, i. e. $\varepsilon =\gamma _{\underline{a}}^{1}I^{\underline{a}}$
and $W_{\mu }=q_{\mu ,\underline{a}\underline{b}}^{2}I^{\underline{a}}I^{%
\underline{b}}$, with the variation $\delta W_{\mu }$ satisfying the
equation
\begin{equation*}
\delta W_{\mu }=\partial _{\mu }(\gamma _{\underline{a}\underline{b}}^{2}I^{%
\underline{a}}I^{\underline{b}})-\frac{1}{2}\theta ^{\tau \lambda
}\{\partial _{\tau }\varepsilon ,\partial _{\lambda }q_{\mu
}\}+{}i[\varepsilon ,W_{\mu }]+i[(\gamma _{\underline{a}\underline{b}}^{2}I^{%
\underline{a}}I^{\underline{b}}),q_{\nu }].
\end{equation*}%
This equation can be solved in the form%
\begin{equation}
\gamma _{\underline{a}\underline{b}}^{2}=\frac{1}{2}\theta ^{\nu \mu
}(\partial _{\nu }\gamma _{\underline{a}}^{1})q_{\mu ,\underline{b}%
}^{1},~q_{\mu ,\underline{a}\underline{b}}^{2}=-\frac{1}{2}\theta ^{\nu \tau
}q_{\nu ,\underline{a}}^{1}\left( \partial _{\tau }q_{\mu ,\underline{b}%
}^{1}+R_{\tau \mu ,\underline{b}}^{1}\right) .  \notag
\end{equation}%
The values $R_{\tau \mu ,\underline{b}}^{1}=\partial _{\tau }q_{\mu ,%
\underline{b}}^{1}-\partial _{\mu }q_{\tau ,\underline{b}}^{1}+f_{~%
\underline{d}}^{\underline{e}\underline{c}}q_{\tau ,\underline{e}}^{1}q_{\mu
,\underline{e}}^{1}$ could be identified with the coefficients $\mathcal{R}%
_{\quad \underline{\beta }\mu \nu }^{\underline{\alpha }}$ of de Sitter
nonlinear gauge gravity curvature (see formula (\ref{curvs})) if in the
commutative limit $q_{\mu ,\underline{b}}^{1}\simeq \left(
\begin{array}{cc}
\Gamma _{\quad \underline{\beta }}^{\underline{\alpha }} & l_{0}^{-1}\chi ^{%
\underline{\alpha }} \\
l_{0}^{-1}\chi _{\underline{\beta }} & 0%
\end{array}%
\right) $ (see (\ref{conds})).

We note that the below presented procedure can be generalized to all the
higher powers of $\theta .$ As an example, we compute the first order
corrections to the gravitational curvature:

\subsubsection{Noncommutative covariant gauge gravity dynamics}

The constructions from the previous subsection can be summarized by a
conclusion that the de Sitter algebra valued object $\varepsilon =\gamma _{%
\underline{a}}^{1}\left( u\right) I^{\underline{a}}$ determines all the
terms in the enveloping algebra $\gamma =\gamma _{\underline{a}}^{1}I^{%
\underline{a}}+\frac{1}{4}\theta ^{\nu \mu }\partial _{\nu }\gamma _{%
\underline{a}}^{1}\ q_{\mu ,\underline{b}}^{1}\left( I^{\underline{a}}I^{%
\underline{b}}+I^{\underline{b}}I^{\underline{a}}\right)$ $+...$ and the
gauge transformations are defined by $\gamma _{\underline{a}}^{1}\left(
u\right) $ and $q_{\mu ,\underline{b}}^{1}(u),$ when $\delta _{\gamma
^{1}}\psi =i\gamma \left( \gamma ^{1},q_{\mu }^{1}\right) \ast \psi$.
 We compute
\begin{eqnarray*}
\lbrack \gamma ,^{\ast }\zeta ] &=&i\gamma _{\underline{a}}^{1}\zeta _{%
\underline{b}}^{1}f_{~\underline{c}}^{\underline{a}\underline{b}}I^{%
\underline{c}}+\frac{i}{2}\theta ^{\nu \mu }\{\partial _{v}\left( \gamma _{%
\underline{a}}^{1}\zeta _{\underline{b}}^{1}f_{~\underline{c}}^{\underline{a}%
\underline{b}}\right) q_{\mu ,\underline{c}} \\
&&+{}\left( \gamma _{\underline{a}}^{1}\partial _{v}\zeta _{\underline{b}%
}^{1}-\zeta _{\underline{a}}^{1}\partial _{v}\gamma _{\underline{b}%
}^{1}\right) q_{\mu ,\underline{b}}f_{~\underline{c}}^{\underline{a}%
\underline{b}}+2\partial _{v}\gamma _{\underline{a}}^{1}\partial _{\mu
}\zeta _{\underline{b}}^{1}\}I^{\underline{d}}I^{\underline{c}},
\end{eqnarray*}%
where we used the properties that, for the de Sitter enveloping algebras,
one holds the general formula for compositions of two transformations $%
\delta _{\gamma }\delta _{\varsigma }-\delta _{\varsigma }\delta _{\gamma
}=\delta _{i(\varsigma \ast \gamma -\gamma \ast \varsigma )}$. This is also
true for the restricted transformations defined by $\gamma ^{1}, $$\delta
_{\gamma ^{1}}\delta _{\varsigma ^{1}}-\delta _{\varsigma ^{1}}\delta
_{\gamma ^{1}}=\delta _{i(\varsigma ^{1}\ast \gamma ^{1}-\gamma ^{1}\ast
\varsigma ^{1})}$.

Such commutators could be used for definition of tensors
\begin{equation}
\widehat{S}^{\mu \nu }=[\widehat{U}^{\mu },\widehat{U}^{\nu }]-i\widehat{%
\theta }^{\mu \nu },  \label{tensor1}
\end{equation}%
where $\widehat{\theta }^{\mu \nu }$ is respectively stated for the
canonical, Lie and quantum plane structures. Under the general enveloping
algebra one holds the transform $\delta \widehat{S}^{\mu \nu }=i[\widehat{%
\gamma },\widehat{S}^{\mu \nu }]$. For instance, the canonical case is
characterized by%
\begin{eqnarray}
S^{\mu \nu } &=&i\theta ^{\mu \tau }\partial _{\tau }\Gamma ^{\nu }-i\theta
^{\nu \tau }\partial _{\tau }\Gamma ^{\mu }+\Gamma ^{\mu }\ast \Gamma ^{\nu
}-\Gamma ^{\nu }\ast \Gamma ^{\mu }  \notag \\
&=&\theta ^{\mu \tau }\theta ^{\nu \lambda }\{\partial _{\tau }Q_{\lambda
}-\partial _{\lambda }Q_{\tau }+Q_{\tau }\ast Q_{\lambda }-Q_{\lambda }\ast
Q_{\tau }\}.  \notag
\end{eqnarray}

We introduce the gravitational gauge strength (curvature)
\begin{equation}
R_{\tau \lambda }=\partial _{\tau }Q_{\lambda }-\partial _{\lambda }Q_{\tau
}+Q_{\tau }\ast Q_{\lambda }-Q_{\lambda }\ast Q_{\tau },  \label{qcurv}
\end{equation}%
which could be treated as a noncommutative extension of de Sitter nonlinear
gauge gravitational curvature (\ref{curvs}), and calculate
\begin{equation}
R_{\tau \lambda ,\underline{a}}=R_{\tau \lambda ,\underline{a}}^{1}+\theta
^{\mu \nu }\{R_{\tau \mu ,\underline{a}}^{1}R_{\lambda \nu ,\underline{b}%
}^{1}{}-\frac{1}{2}q_{\mu ,\underline{a}}^{1}\left[ (D_{\nu }R_{\tau \lambda
,\underline{b}}^{1})+\partial _{\nu }R_{\tau \lambda ,\underline{b}}^{1}%
\right] \}I^{\underline{b}},  \notag
\end{equation}%
where the gauge gravitation covariant derivative is introduced,%
\begin{equation*}
(D_{\nu }R_{\tau \lambda ,\underline{b}}^{1})=\partial _{\nu }R_{\tau
\lambda ,\underline{b}}^{1}+q_{\nu ,\underline{c}}R_{\tau \lambda ,%
\underline{d}}^{1}f_{~\underline{b}}^{\underline{c}\underline{d}}.
\end{equation*}%
Following the gauge transformation laws for $\gamma $ and $q^{1}$ we find $%
\delta _{\gamma ^{1}}R_{\tau \lambda }^{1}=i\left[ \gamma ,^{\ast }R_{\tau
\lambda }^{1}\right]$ with the restricted form of $\gamma .$

One can be formulated a gauge covariant gravitational dynamics of
noncommutative spaces following the nonlinear realization of de Sitter
algebra and the $\ast $--formalism and introducing derivatives in such a way
that one does not obtain new relations for the coordinates. In this case, a
Leibniz rule can be defined that $\widehat{\partial }_{\mu }\widehat{u}^{\nu
}=\delta _{\mu }^{\nu }+d_{\mu \sigma }^{\nu \tau }\ \widehat{u}^{\sigma }\
\widehat{\partial }_{\tau }$, where the coefficients $d_{\mu \sigma }^{\nu
\tau }=\delta _{\sigma }^{\nu }\delta _{\mu }^{\tau }$ are chosen to have
not new relations when $\widehat{\partial }_{\mu }$ acts again to the right
hand side. One holds the $\ast $--derivative formulas
\begin{equation}
{\partial }_{\tau }\ast f=\frac{\partial }{\partial u^{\tau }}f+f\ast {%
\partial }_{\tau },~[{\partial }_{l},{{}^{\ast }}(f\ast g)]=([{\partial }%
_{l},{{}^{\ast }}f])\ast g+f\ast ([{\partial }_{l},{}^{\ast }g])  \notag
\end{equation}%
and the Stokes theorem$\int [\partial _{l},f]=\int d^{N}u[\partial
_{l},^{\ast }f]=\int d^{N}u\frac{\partial }{\partial u^{l}}f=0,$ where, for
the canonical structure, the integral is defined, $\int \widehat{f}=\int
d^{N}uf\left( u^{1},...,u^{N}\right)$.

An action can be introduced by using such integrals. For instance, for a
tensor of type (\ref{tensor1}), when $\delta \widehat{L}=i\left[ \widehat{%
\gamma },\widehat{L}\right]$, we can define a gauge invariant action%
 $W=\int d^Nu\ Tr\widehat{L},~\delta W=0$,
 were the trace has to be taken for the group generators.
For the nonlinear de Sitter gauge gravity a proper action is $L=\frac{1}{4}%
R_{\tau \lambda }R^{\tau \lambda }$, where $R_{\tau \lambda }$ is defined by
(\ref{qcurv}) (in the commutative limit we shall obtain the connection (\ref%
{conds})). In this case the dynamic of noncommutative space is entirely
formulated in the framework of quantum field theory of gauge fields. In
general, we are dealing with anisotropic gauge gravitational interactions.
The method works for matter fields as well to restrictions to the general
relativity theory.

\subsubsection{Noncommutative symmetries and star product deformations}

The aim of this subsection is to prove that there are possible extensions of
exact solutions from the Einstein and gauge gravity possessing hidden
noncommutative symmetries without introducing new fields. For simplicity, we
present the formulas including decompositions up to the second order on
noncommutative parameter $\theta ^{\alpha \beta }$ for vielbeins,
connections and curvatures which can be arranged to result in different
models of noncommutative gravity. We give the data for the $SU\left(
1,n+m-1\right) $ and $SO\left( 1,n+m-1\right) $ gauge models containing, in
general, complex N--elongated frames, modelling some exact solutions. All
data can be considered for extensions with nonlinear realizations into a
bundle of affine/or de Sitter frames (in this case, one generates
noncommutative gauge theories of type \cite{v1}) or to impose certain
constraints and broking symmetries.

The standard approaches to noncommutative geometry also contain certain
noncommutative relations for coordinates,
\begin{equation}
\lbrack u^{\alpha },u^{\beta }]=u^{\alpha }u^{\beta }-u^{\beta }u^{\alpha
}=i\theta ^{\alpha \beta }(u^{\gamma })  \label{coordnc}
\end{equation}%
were, in the simplest models, the commutator $[u^{\alpha },u^{\beta }]$ is
approximated to be constant, but there were elaborated approaches for
general manifolds with the noncommutative parameter $\theta ^{\alpha \beta }$
treated as functions on $u^{\gamma }$. We define the star (Moyal) product to
include possible N--elongated partial derivatives and a quantum constant $%
\hbar $,
\begin{eqnarray}
f\ast \varphi &=&f\varphi +\frac{\hbar }{2}B^{\overline{\alpha }\overline{%
\beta }}\left( \delta _{\overline{\alpha }}f\delta _{\overline{\beta }%
}\varphi +\delta _{\overline{\beta }}f\delta _{\overline{\alpha }}\varphi
\right) +\hbar ^{2}B^{\overline{\alpha }\overline{\beta }}B^{\overline{%
\gamma }\overline{\mu }}\left[ \delta _{(\overline{\alpha }}\delta _{%
\overline{\gamma })}f\right] \left[ \delta _{(\overline{\beta }}\delta _{%
\overline{\mu })}\varphi \right]  \notag \\
&&+\frac{2}{3}\hbar ^{2}B^{\overline{\alpha }\overline{\beta }}\delta _{%
\overline{\beta }}B^{\overline{\gamma }\overline{\mu }}\{~\left[ \delta _{(%
\overline{\alpha }}\delta _{\overline{\gamma })}f\right] \delta _{\overline{%
\mu }}\varphi +[\delta _{(\overline{\alpha }}\delta _{\overline{\gamma }%
)}\varphi ]\delta _{\overline{\mu }}f\}+O\left( \hbar ^{3}\right) ,
\label{form01}
\end{eqnarray}%
where, for instance, $\delta _{(\mu }\delta _{\nu )}=(1/2)(\delta _{\mu
}\delta _{\nu }+\delta _{\nu }\delta _{\mu })$,
\begin{equation}
B^{\overline{\alpha }\overline{\beta }}=\frac{\theta ^{\alpha \beta }}{2}%
\left( \delta _{\alpha }u^{\overline{\alpha }}\delta _{\beta }u^{\overline{%
\beta }}+\delta _{\beta }u^{\overline{\alpha }}\delta _{\alpha }u^{\overline{%
\beta }}\right) +O\left( \hbar ^{3}\right)  \label{form02}
\end{equation}%
is defined for new coordinates $u^{\overline{\alpha }}=u^{\overline{\alpha }%
}\left( u^{\alpha }\right) $ inducing a suitable Poisson bi--vector field $%
B^{\overline{\alpha }\overline{\beta }}\left( \hbar \right) $ being related
to a quantum diagram formalism (we shall not consider details concerning
geometric quantization in this paper by investigating only classicassical
deformations related to any anholonomic frame and coordinate (\ref{coordnc})
noncommutativity origin). The formulas (\ref{form01}) and (\ref{form02})
transform into the usual ones with partial derivatives $\partial _{\alpha }$
and $\partial _{\overline{\alpha }}$ for vanishing anholonomy coefficients.
We can define a star product being invariant under diffeomorphism
transforms, $\ast \rightarrow \ast ^{\lbrack -]},$ adapted to the
N--connection structure ( in a vector bundle provided with N--connection
configuration, we use the label $[-]$ in order to emphasize the dependence
on coordinates $u^{\overline{\alpha }}$ with 'overlined' indices), by
introducing the transforms%
\begin{equation*}
f^{[-]}\left( \hbar \right) =\Theta f\left( \hbar \right),\ f^{[-]}\ast
^{\lbrack -]}\varphi ^{\lbrack -]}=\Theta \left( \Theta ^{-1}f^{[-]}\ast
\Theta ^{-1}\right) \varphi ^{\lbrack -]},
\end{equation*}
for $\Theta =1+\sum_{[k=1]}\hbar ^{k}\Theta _{\lbrack k]},$ for simplicity,
computed up to the squared order on $\hbar ,$ $\Theta =1-2\hbar ^{2}\theta
^{\mu \nu }\theta ^{\rho \sigma }[ \left[ \delta _{(\mu }\delta _{\nu )}u^{%
\overline{\alpha }}\right] \left[ \delta _{(\rho }\delta _{\sigma )}u^{%
\overline{\beta }}\right] \delta _{(\overline{\alpha }}\delta _{\overline{%
\beta })} +[\delta _{(\mu }\delta _{\rho )}u^{\overline{\alpha }}](\delta
_{\nu }u^{\overline{\beta }})(\delta _{\sigma }u^{\overline{\gamma }})$ $%
\left[ \delta _{(\overline{\alpha }}\delta _{\overline{\beta }}\delta _{%
\overline{\gamma })}\right]] +O\left( \hbar ^{4}\right)$, where $\delta _{(%
\overline{\alpha }}\delta _{\overline{\beta }}\delta _{\overline{\gamma })}=$
$(1/3!)(\delta _{\overline{\alpha }}\delta _{\overline{\beta }}\delta _{%
\overline{\gamma }}+$ all \textsl{symmetric permutations). }In our further
constructions we shall omit the constant $\hbar $ considering that $\theta
\sim \hbar $ is a small value by writing the necessary terms in the
approximation $O\left( \theta ^{3}\right) $ or $O\left( \theta ^{4}\right) .$

We consider a noncommutative gauge theory on a space with N--connecti\-on
structure stated by the gauge fields $\widehat{A}_{\mu }=\left( \widehat{A}%
_{i},\widehat{A}_{a}\right) $ when ''hats'' on symbols will be used for the
objects defined on spaces with coordinate noncommutativity. In general, the
gauge model can be with different types of structure groups like $SL\left( k,%
\C\right) ,$ $SU_{k},$ $U_{k},SO(k-1,1)$ and their nonlinear realizations.
For instance, for the $U\left( n+m\right) $ gauge fields there are satisfied
the conditions $\widehat{A}_{\mu }^{+}=-\widehat{A}_{\mu },$where $"+"$ is
the Hermitian conjugation. It is useful to present the basic geometric
constructions for a unitary structural group containing the $SO\left(
4,1\right) $ as a particular case if we wont to consider noncommutative
extensions of 4D exact solutions.

The noncommutative gauge transforms of potentials are defined by using the
star product $\widehat{A}_{\mu }^{[\varphi ]}=\widehat{\varphi }\ast
\widehat{A}_{\mu }\widehat{\varphi }_{[\ast ]}^{-1}-\widehat{\varphi }\ast
\delta _{\mu }\widehat{\varphi }_{[\ast ]}^{-1}$, where the N--elongated
partial derivatives are used and $\widehat{\varphi }\ast \widehat{\varphi }%
_{[\ast ]}^{-1}=1=$ $\widehat{\varphi }_{[\ast ]}^{-1}\ast \widehat{\varphi }%
.$ The matrix coefficients of fields will be distinguished by ''overlined''
indices, for instance, $\widehat{A}_{\mu }=\{\widehat{A}_{\mu }^{\underline{%
\alpha }\underline{\beta }}\},$ and for commutative values, $A_{\mu
}=\{A_{\mu }^{\underline{\alpha }\underline{\beta }}\}.$ Such fields are
subjected to the conditions $(\widehat{A}_{\mu }^{\underline{\alpha }%
\underline{\beta }})^{+}\left( u,\theta \right) =-\widehat{A}_{\mu }^{%
\underline{\beta }\underline{\alpha }}\left( u,\theta \right)$ and $\widehat{%
A}_{\mu }^{\underline{\alpha }\underline{\beta }}\left( u,-\theta \right) =-%
\widehat{A}_{\mu }^{\underline{\beta }\underline{\alpha }}\left( u,\theta
\right)$. There is a basic assumption that the noncommutative fields are
related to the commutative fields by the Seiberg--Witten map in a manner
that there are not new degrees of freedom being satisfied the equation%
\begin{equation}
\widehat{A}_{\mu }^{\underline{\alpha }\underline{\beta }}(A)+\Delta _{%
\widehat{\lambda }}\widehat{A}_{\mu }^{\underline{\alpha }\underline{\beta }%
}(A)=\widehat{A}_{\mu }^{\underline{\alpha }\underline{\beta }}(A+\Delta _{%
\widehat{\lambda }}A),  \label{sw1}
\end{equation}%
where $\widehat{A}_{\mu }^{\underline{\alpha }\underline{\beta }}(A)$
denotes a functional dependence on commutative field $A_{\mu }^{\underline{%
\alpha }\underline{\beta }}$, $\widehat{\varphi }=\exp \widehat{\lambda }$
and the infinitesimal deformations $\widehat{A}_{\mu }^{\underline{\alpha }%
\underline{\beta }}(A)$ and of $A_{\mu }^{\underline{\alpha }\underline{%
\beta }}$ are
\begin{eqnarray*}
\Delta _{\widehat{\lambda }}\widehat{A}_{\mu }^{\underline{\alpha }%
\underline{\beta }}&=&\delta _{\mu }\widehat{\lambda }^{\underline{\alpha }%
\underline{\beta }}+\widehat{A}_{\mu }^{\underline{\alpha }\underline{\gamma
}}\ast \widehat{\lambda }^{\underline{\gamma }\underline{\beta }}-\widehat{%
\lambda }^{\underline{\alpha }\underline{\gamma }}\ast \widehat{A}_{\mu }^{%
\underline{\gamma }\underline{\beta }} \\
\mbox{ and\  } \Delta _{\lambda }A_{\mu }^{\underline{\alpha }\underline{%
\beta }}&=&\delta _{\mu }\lambda ^{\underline{\alpha }\underline{\beta }%
}+A_{\mu }^{\underline{\alpha }\underline{\gamma }}\ast \lambda ^{\underline{%
\gamma }\underline{\beta }}-\lambda ^{\underline{\alpha }\underline{\gamma }%
}\ast A_{\mu }^{\underline{\gamma }\underline{\beta }},
\end{eqnarray*}%
where instead of partial derivatives $\partial _{\mu }$ we use the
N--elongated ones, $\delta _{\mu }$ and sum on index $\underline{\gamma }.$

Solutions of the Seiberg--Witten equations for models of gauge gravity are
considered, for instance, in Ref. \cite{v1} (there are discussed procedures
of deriving expressions on $\theta $ to all orders). Here we present only
the first order on $\theta $ for the coefficients $\widehat{\lambda }^{%
\underline{\alpha }\underline{\beta }}$ and the first and second orders for $%
\widehat{A}_{\mu }^{\underline{\alpha }\underline{\beta }}$ including
anholonomy relations and not depending on model considerations,%
\begin{equation*}
\widehat{\lambda }^{\underline{\alpha }\underline{\beta }}=\lambda ^{%
\underline{\alpha }\underline{\beta }}+\frac{i}{4}\theta ^{\nu \tau
}\{(\delta _{\nu }\lambda ^{\underline{\alpha }\underline{\gamma }})A_{\mu
}^{\underline{\gamma }\underline{\beta }}+A_{\mu }^{\underline{\alpha }%
\underline{\gamma }}(\delta _{\nu }\lambda ^{\underline{\gamma }\underline{%
\beta }})\}+O\left( \theta ^{2}\right)
\end{equation*}%
{\small
\begin{eqnarray}
&& \mbox{and \ }  \widehat{A}_{\mu }^{\underline{\alpha }\underline{\beta }} =A_{\mu }^{%
\underline{\alpha }\underline{\beta }}-\frac{i}{4}\theta ^{\nu \tau
}\{A_{\mu }^{\underline{\alpha }\underline{\gamma }}\left( \delta _{\tau
}A_{\nu }^{\underline{\gamma }\underline{\beta }}+R_{\quad \tau \nu }^{%
\underline{\gamma }\underline{\beta }}\right) +\left( \delta _{\tau }A_{\mu
}^{\underline{\alpha }\underline{\gamma }}+R_{\quad \tau \mu }^{\underline{%
\alpha }\underline{\gamma }}\right) A_{\nu }^{\underline{\gamma }\underline{%
\beta }}\}  \notag \\
&&+\frac{1}{32}\theta ^{\nu \tau }\theta ^{\rho \sigma }\{[2A_{\rho }^{%
\underline{\alpha }\underline{\gamma }}(R_{\quad \sigma \nu }^{\underline{%
\gamma }\underline{\varepsilon }}R_{\quad \mu \tau }^{\underline{e}%
\underline{\beta }} +R_{\quad \mu \tau }^{\underline{\gamma }\underline{%
\varepsilon }}R_{\quad \sigma \nu }^{\underline{\varepsilon }\underline{%
\beta }})  \label{nnccon1} \\
&&+2(R_{\quad \sigma \nu }^{\underline{\alpha }\underline{\varepsilon }%
}R_{\quad \mu \tau }^{\underline{\varepsilon }\underline{\gamma }}+R_{\quad
\mu \tau }^{\underline{\alpha }\underline{\varepsilon }}R_{\quad \sigma \nu
}^{\underline{\varepsilon }\underline{\gamma }})A_{\rho }^{\underline{\gamma
}\underline{\beta }}]  \notag \\
&&-[A_{\nu }^{\underline{\alpha }\underline{\gamma }}\left( D_{\tau
}R_{\quad \sigma \mu }^{\underline{\gamma }\underline{\beta }}+\delta _{\tau
}R_{\quad \sigma \mu }^{\underline{\gamma }\underline{\beta }}\right)
+\left( D_{\tau }R_{\quad \sigma \mu }^{\underline{\alpha }\underline{\gamma
}}+\delta _{\tau }R_{\quad \sigma \mu }^{\underline{\alpha }\underline{%
\gamma }}\right) A_{\nu }^{\underline{\gamma }\underline{\beta }}]  \notag \\
&&- \delta _{\sigma }[A_{\nu }^{\underline{\alpha }\underline{\gamma }%
}\left( \delta _{\tau }A_{\mu }^{\underline{\gamma }\underline{\beta }%
}+R_{\quad \tau \mu }^{\underline{\gamma }\underline{\beta }}\right) +\left(
\delta _{\tau }A_{\mu }^{\underline{\alpha }\underline{\gamma }}+R_{\quad
\tau \mu }^{\underline{\alpha }\underline{\gamma }}\right) A_{\nu }^{%
\underline{\gamma }\underline{\beta }}]+  \notag \\
&&[(\delta _{\nu }A_{\rho }^{\underline{\alpha }\underline{\gamma }})\left(
2\delta _{(\tau }\delta _{\sigma )}A_{\mu }^{\underline{\gamma }\underline{%
\beta }}+\delta _{\tau }R_{\quad \sigma \mu }^{\underline{\gamma }\underline{%
\beta }}\right) +\left( 2\delta _{(\tau }\delta _{\sigma )}A_{\mu }^{%
\underline{\alpha }\underline{\gamma }}+\delta _{\tau }R_{\quad \sigma \mu
}^{\underline{\alpha }\underline{\gamma }}\right) (\delta _{\nu }A_{\rho }^{%
\underline{\gamma }\underline{\beta }})]- \notag  \\
&&[A_{\nu }^{\underline{\alpha }\underline{\varepsilon }}\left( \delta
_{\tau }A_{\rho }^{\underline{\varepsilon }\underline{\gamma }}+R_{\quad
\tau \rho }^{\underline{\varepsilon }\underline{\gamma }}\right) +\left(
\delta _{\tau }A_{\rho }^{\underline{\alpha }\underline{\varepsilon }%
}+R_{\quad \tau \rho }^{\underline{\alpha }\underline{\varepsilon }}\right)
A_{\nu }^{\underline{\varepsilon }\underline{\gamma }}]\left( \delta
_{\sigma }A_{\mu }^{\underline{\gamma }\underline{\beta }}+R_{\quad \sigma
\mu }^{\underline{\gamma }\underline{\beta }}\right) - \notag  \\
&&\left( \delta _{\sigma }A_{\mu }^{\underline{\alpha }\underline{\gamma }%
}+R_{\quad \sigma \mu }^{\underline{\alpha }\underline{\gamma }}\right)
[A_{\nu }^{\underline{\gamma }\underline{\varepsilon }}\left( \delta _{\tau
}A_{\rho }^{\underline{\varepsilon }\underline{\beta }}+R_{\quad \tau \rho
}^{\underline{\varepsilon }\underline{\beta }}\right) +\left( \delta _{\tau
}A_{\rho }^{\underline{\gamma }\underline{\varepsilon }}+R_{\quad \tau \rho
}^{\underline{\gamma }\underline{\varepsilon }}\right) A_{\nu }^{\underline{%
\varepsilon }\underline{\beta }}]+O\left( \theta ^{3}\right) , \notag
\end{eqnarray}%
}
where the curvature is defined $R_{\quad \tau \nu }^{\underline{\alpha }%
\underline{\beta }}=e_{\alpha }^{\underline{\alpha }}e^{\underline{\beta }%
\beta }R_{\beta \ \tau \nu }^{\ \alpha }$, when $\Gamma \rightarrow A,$ and
for the gauge model of gravity, see (\ref{curvs}) and (\ref{qcurv}). By
using the star product, we can write symbolically the solution (\ref{nnccon1}%
) in general form,%
\begin{equation*}
\Delta \widehat{A}_{\mu }^{\underline{\alpha }\underline{\beta }}\left(
\theta \right) =-\frac{i}{4}\theta ^{\nu \tau }\left[ \widehat{A}_{\mu }^{%
\underline{\alpha }\underline{\gamma }}\ast \left( \delta _{\tau }\widehat{A}%
_{\nu }^{\underline{\gamma }\underline{\beta }}+\widehat{R}_{\quad \tau \nu
}^{\underline{\gamma }\underline{\beta }}\right) +\left( \delta _{\tau }%
\widehat{A}_{\mu }^{\underline{\alpha }\underline{\gamma }}+\widehat{R}%
_{\quad \tau \mu }^{\underline{\alpha }\underline{\gamma }}\right) \ast
\widehat{A}_{\nu }^{\underline{\gamma }\underline{\beta }}\right],
\end{equation*}%
where $\widehat{R}_{\quad \tau \nu }^{\underline{\gamma }\underline{\beta }}$
is defined by the same formulas as $R_{\quad \tau \nu }^{\underline{\alpha }%
\underline{\beta }}$ but with the star products, like $AA\rightarrow A\ast
A. $

There is a problem how to determine the dependence of the noncommutative
vielbeins $\widehat{e}_{\alpha }^{\underline{\alpha }}$ on commutative ones $%
e_{\alpha }^{\underline{\alpha }}.$ If we consider the frame fields to be
included into a (anti) de Sitter gauge gravity model with the connection (%
\ref{conds}), the vielbein components should be treated as certain
coefficients of the gauge potential with specific nonlinear transforms for
which the results of Ref. \cite{v1} hold. The main difference (considered in
this work) is that the frames are in general with anholonomy induced by a
N--connection field. In order to derive in a such model the Einstein gravity
we have to analyze the reduction (\ref{poinc}) to a Poincar\'e gauge gravity.

An explicit calculus of the curvature of such gauge potential show that the
coefficients of the curvature of (\ref{poinc}), obtained as a reduction from
the $SO\left( 4,1\right) $ gauge group is given by the coefficients (\ref%
{curvs}) with vanishing torsion and constraints of type $\widehat{A}_{\nu }^{%
\underline{\gamma }\underline{5}}=\epsilon \widehat{e}_{\nu }^{\underline{%
\gamma }}$ and $\widehat{A}_{\nu }^{\underline{5}\underline{5}}=\epsilon
\widehat{\phi }_{\nu }$ with $\widehat{R}_{\quad \tau \nu }^{\underline{5}%
\underline{5}}\sim \epsilon $ vanishing in the limit $\epsilon \rightarrow 0$
(we obtain the same formulas for the vielbein and curvature components
derived for the inhomogeneous Lorentz group but generalized to N--elongated
derivatives and with distinguishing into h--v--components). The result for $%
\widehat{e}_{\mu }^{\underline{\mu }}$ in the limit $\epsilon \rightarrow 0$
generalized to the case of canonical connections defining the covariant
derivatives $D_{\tau }$ and corresponding curvatures is{\small
\begin{eqnarray}
&&\widehat{e}_{\mu }^{\underline{\mu }} =e_{\mu }^{\underline{\mu }}-\frac{i%
}{4}\theta ^{\nu \tau }\left[ A_{\nu }^{\underline{\mu }\underline{\gamma }%
}\delta _{\tau }e_{\mu }^{\underline{\gamma }}+\left( \delta _{\tau }A_{\mu
}^{\underline{\mu }\underline{\gamma }}+R_{\quad \tau \mu }^{\underline{\mu }%
\underline{\gamma }}\right) e_{\nu }^{\underline{\gamma }}\right] +
\label{qfr} \\
&&\frac{1}{32}\theta ^{\nu \tau }\theta ^{\beta \sigma }\{2(R_{\quad \sigma
\nu }^{\underline{\mu }\underline{\varepsilon }}R_{\quad \mu \tau }^{%
\underline{\varepsilon }\underline{\gamma }}+R_{\quad \mu \tau }^{\underline{%
\mu }\underline{\varepsilon }}R_{\quad \sigma \nu }^{\underline{\varepsilon }%
\underline{\gamma }})e_{\beta }^{\underline{\gamma }} -A_{\beta }^{\underline{\mu }\underline{\gamma }} (D_{\tau }R_{\quad
\sigma \mu }^{\underline{\gamma }\underline{\beta }}+\delta _{\tau }R_{\quad
\sigma \mu }^{\underline{\gamma }\underline{\beta }}) e_{\beta }^{%
\underline{\beta }}  \notag \\
&&-[A_{\nu }^{\underline{\mu }\underline{\gamma }}\left( D_{\tau }R_{\quad
\sigma \mu }^{\underline{\gamma }\underline{\beta }}+\delta _{\tau }R_{\quad
\sigma \mu }^{\underline{\gamma }\underline{\beta }}\right) +\left( D_{\tau
}R_{\quad \sigma \mu }^{\underline{\mu }\underline{\gamma }}+\delta _{\tau
}R_{\quad \sigma \mu }^{\underline{\mu }\underline{\gamma }}\right) A_{\nu
}^{\underline{\gamma }\underline{\beta }}]e_{\beta }^{\underline{\beta }}-
\notag \\
&&e_{\beta }^{\underline{\beta }}\delta _{\sigma }\left[ A_{\nu }^{%
\underline{\mu }\underline{\gamma }}\left( \delta _{\tau }A_{\mu }^{%
\underline{\gamma }\underline{\beta }}+R_{\quad \tau \mu }^{\underline{%
\gamma }\underline{\beta }}\right) +\left( \delta _{\tau }A_{\mu }^{%
\underline{\mu }\underline{\gamma }}+R_{\quad \tau \mu }^{\underline{\mu }%
\underline{\gamma }}\right) A_{\nu }^{\underline{\gamma }\underline{\beta }}%
\right] +2\left( \delta _{\nu }A_{\beta }^{\underline{\mu }\underline{\gamma
}}\right) \delta _{(\tau }\delta _{\sigma )}e_{\mu }^{\underline{\gamma }} \notag \\
&&-A_{\beta }^{\underline{\mu }\underline{\gamma }}\delta _{\sigma }\left[
A_{\nu }^{\underline{\gamma }\underline{\beta }}\delta _{\tau }e_{\mu }^{%
\underline{\beta }}+\left( \delta _{\tau }A_{\mu }^{\underline{\gamma }%
\underline{\beta }}+R_{\quad \tau \mu }^{\underline{\gamma }\underline{\beta
}}\right) e_{\nu }^{\underline{\beta }}\right] -\left( \delta _{\nu
}e_{\beta }^{\underline{\gamma }}\right) \delta _{\tau }\left( \delta
_{\sigma }A_{\mu }^{\underline{\mu }\underline{\gamma }}+R_{\quad \sigma \mu
}^{\underline{\mu }\underline{\gamma }}\right) - \notag \\
&&\left[ A_{\nu }^{\underline{\mu }\underline{\gamma }}\left( \delta _{\tau
}A_{\beta }^{\underline{\gamma }\underline{\beta }}+R_{\quad \tau \beta }^{%
\underline{\gamma }\underline{\beta }}\right) +\left( \delta _{\tau
}A_{\beta }^{\underline{\mu }\underline{\gamma }}+R_{\quad \tau \beta }^{%
\underline{\mu }\underline{\gamma }}\right) A_{\nu }^{\underline{\gamma }%
\underline{\beta }}\right] \delta _{\sigma }e_{\mu }^{\underline{\beta }}- \notag \\
&&\left( \delta _{\sigma }A_{\mu }^{\underline{\mu }\underline{\gamma }%
}+R_{\quad \sigma \mu }^{\underline{\mu }\underline{\gamma }}\right) \left[
A_{\mu }^{\underline{\gamma }\underline{\beta }}\left( \delta _{\nu
}e_{\beta }^{\underline{\beta }}\right) +e_{\nu }^{\underline{\beta }}\left(
\delta _{\sigma }A_{\mu }^{\underline{\gamma }\underline{\beta }}+R_{\quad
\sigma \mu }^{\underline{\gamma }\underline{\beta }}\right) \right]
\}+O\left( \theta ^{3}\right) . \notag
\end{eqnarray}
}

Having the decompositions (\ref{qfr}), we can define the inverse vielbein $%
\widehat{e}_{\ast \underline{\mu }}^{\mu }$ from the equation $\widehat{e}%
_{\ast \underline{\mu }}^{\mu }\ast \widehat{e}_{\mu }^{\underline{\nu }%
}=\delta _{\underline{\mu }}^{\underline{\nu }}$ and consequently compute $%
\theta $--deformations of connections, curvatures, torsions and any type of
actions and field equations (for simplicity, we omit such cumbersome
formulas).

\subsection{Exact solutions for (non)commutative Finsler branes}
\label{ssdecoupl} We show how the anholnomic deformation method can be
applied for generating Finsler like solutions, with nontrivial nonlinear
connection structure, in noncommutative gravity \cite{v8}.

\subsubsection{Nonholonomic Distributions and Noncommutative Gravity}

\label{s2} There exist many formulations of noncommutative geometry/gravity
based on nonlocal deformation of spacetime and field theories starting from
noncommutative relations of type
\begin{equation}
u^{\alpha }u^{\beta }-u^{\beta }u^{\alpha }=i\theta ^{\alpha \beta },
\label{fuzcond}
\end{equation}%
where $u^{\alpha }$ are local spacetime coordinates, $i$ is the imaginary
unity, $i^{2}=-1,$ and $\theta ^{\alpha \beta }$ is an anti--symmetric
second--rank tensor (which, for simplicity, for certain models, is taken to
be with constant coefficients). Following our unified approach to (pseudo)
Riemannian and Finsler--Lagrange spaces  (using the
geometry of nonholonomic manifolds) we consider that for $\theta ^{\alpha
\beta }\rightarrow 0$ the local coordinates $u^{\alpha }$ are on a four
dimensional (4-d) nonholonomic manifold $\mathbf{V}$ of necessary smooth
class. Such spacetimes can be enabled with a conventional $2+2$ splitting
(defined by a nonholonomic, equivalently, anholonomic/ non--integrable real
distribution), when local coordinates $u=(x,y)$ on an open region $U\subset
\mathbf{V}$ are labelled in the form $u^{\alpha }=(x^{i},y^{a}),$ with
indices of type $i,j,k,...=1,2$ and $a,b,c...=3,4.$ The coefficients of
tensor like objects on $\mathbf{V}$ can be computed with respect to a
general (non--coordinate) local basis $e_{\alpha }=(e_{i},e_{a}).$\footnote{%
If $\mathbf{V}=TM$ is the total space of a tangent bundle $\left( TM,\pi
,M\right) $ on a two dimensional (2--d) base manifold $M,$ the values $x^{i}$
and $y^{a}$ are respectively the base coordinates (on a low--dimensional
space/ spacetime) and fiber coordinates (velocity like). Alternatively, we
can consider that $\mathbf{V}=V$ is a 4--d nonholonomic manifold (in
particular, a pseudo--Riemannian one) with local fibered structure.}

On a commutative $\mathbf{V,}$ any (prime) metric $\mathbf{g=g}_{\alpha
\beta }\mathbf{e}^{a}\otimes \mathbf{e}^{\beta } $ (for instance, a
Schwarzschild, ellipsoid, ring or other type solution, their conformal
transforms and nonholonomic deformations which, in general, are not
solutions of the Einstein equations) can be parametrized in the form
\begin{eqnarray}
\mathbf{g} &=&g_{i}(u)dx^{i}\otimes dx^{i}+h_{a}(u)\mathbf{e}^{a}\otimes
\mathbf{e}^{a},  \label{prime} \\
\mathbf{e}^{\alpha } &=&\mathbf{e}_{\ \underline{\alpha }}^{\alpha }(u)du^{%
\underline{\alpha }}=\left( e^{i}=dx^{i},\mathbf{e}%
^{a}=dy^{a}+N_{i}^{a}dx^{i}\right) .  \notag
\end{eqnarray}
It is convenient to work with the so--called
canonical distinguished connection (in brief, canonical d--connection $%
\widehat{\mathbf{D}}=\{\widehat{\mathbf{\Gamma }}_{\ \alpha \beta }^{\gamma
}\})$ which is metric compatible, $\widehat{\mathbf{D}}\mathbf{g}=0,$ and
completely defined by the coefficients of a metric $\mathbf{g}$ (\ref{prime}%
) and a N--connection $\mathbf{N,}$ subjected to the condition that the
so--called $h$-- and $v$--components of torsion are zero.\footnote{%
by definition, a d--connection is a linear connection preserving under
parallelism a given N--connection splitting; in general, a d--connection has
a nontrivial torsion tensor but for the canonical d--connection the torsion
is induced by the anholonomy coefficients which in their turn are defined by
certain off--diagonal N--coefficients in the corresponding metric} Using
 formula $\Gamma _{\ \alpha \beta
}^{\gamma }=\widehat{\mathbf{\Gamma }}_{\ \alpha \beta }^{\gamma }+\ Z_{\
\alpha \beta }^{\gamma },$ where $\nabla =\{\ \Gamma _{\ \alpha \beta
}^{\gamma }\}$ is \ the \ Levi--Civita connection (this connection is metric
compatible, torsionless and completely defined by the coefficients of the
same metric structure $\mathbf{g}$), we can perform all geometric
constructions in two equivalent forms: applying the covariant derivative $%
\widehat{\mathbf{D}}$ and/or $\nabla .$ This is possible because all values $%
\ \Gamma ,$ $\widehat{\mathbf{\Gamma }}$ and $\ Z$ are
completely determined in unique forms by $\mathbf{g}$ for a prescribed
nonholonomic splitting.

There were considered different constructions of $\ ^{\theta }\mathcal{A}$
corresponding to different choices of the so--called ''symbols of
operators'' and the extended Weyl ordered symbol $\mathcal{W},$ to get an
algebra isomorphism with properties
 $\mathcal{W}[\ ^{1}f\star \ ^{2}f]\equiv \mathcal{W}[\ ^{1}f]\mathcal{W}[\
^{2}f]=\ ^{1}\hat{f}\ \ ^{2}\hat{f}$, %
for $\ ^{1}f,\ ^{2}f\in \mathcal{C}(\mathbf{V})$ and $\ ^{1}\hat{f}\ ,\ ^{2}%
\hat{f}\in \ ^{\theta }\mathcal{A}(\mathbf{V}),$ when the induced $\star $%
--product is associative and noncommutative. Such a product can be
introduced on nonholonomic manifolds using the N--elongated partial
derivatives,
\begin{equation}
\ ^{1}\hat{f}\star \ ^{2}\hat{f}=\sum\limits_{k=0}^{\infty }\frac{1}{k!}%
\left( \frac{i}{2}\right) ^{k}\theta ^{\alpha _{1}\beta _{1}}\ldots \theta
^{\alpha _{k}\beta _{k}}\mathbf{e}_{\alpha _{1}}\ldots \mathbf{e}_{\alpha
_{k}}\ ^{1}f(u)\ \mathbf{e}_{\beta _{1}}\ldots \mathbf{e}_{\beta _{k}}\
^{2}f(u).  \label{starpr}
\end{equation}

For a noncommutative nonholonomic spacetime model $\ ^{\theta }\mathbf{V}$
of a spacetime $\mathbf{V,}$ we can derive a N--adapted local frame
structure $\ \ ^{\theta }\mathbf{e}_{\alpha }=(\ \ ^{\theta }\mathbf{e}%
_{i},\ \ ^{\theta }\mathbf{e}_{a})$ which can be constructed by
noncommutative deformations of $\mathbf{e}_{\alpha },$%
\begin{eqnarray}
\ ^{\theta }\mathbf{e}_{\alpha \ }^{\ \underline{\alpha }} &=&\mathbf{e}%
_{\alpha \ }^{\ \underline{\alpha }}+i\theta ^{\alpha _{1}\beta _{1}}\mathbf{%
e}_{\alpha \ \alpha _{1}\beta _{1}}^{\ \underline{\alpha }}+\theta ^{\alpha
_{1}\beta _{1}}\theta ^{\alpha _{2}\beta _{2}}\mathbf{e}_{\alpha \ \alpha
_{1}\beta _{1}\alpha _{2}\beta _{2}}^{\ \underline{\alpha }}+\mathcal{O}%
(\theta ^{3}),  \label{ncfd} \\
\ ^{\theta }\mathbf{e}_{\ \star \underline{\alpha }}^{\alpha } &=&\mathbf{e}%
_{\ \underline{\alpha }}^{\alpha }+i\theta ^{\alpha _{1}\beta _{1}}\mathbf{e}%
_{\ \underline{\alpha }\alpha _{1}\beta _{1}}^{\alpha }+\theta ^{\alpha
_{1}\beta _{1}}\theta ^{\alpha _{2}\beta _{2}}\mathbf{e}_{\ \underline{%
\alpha }\alpha _{1}\beta _{1}\alpha _{2}\beta _{2}}^{\alpha }+\mathcal{O}%
(\theta ^{3}),  \notag
\end{eqnarray}%
subjected to the condition $\ ^{\theta }\mathbf{e}_{\ \star \underline{%
\alpha }}^{\alpha }\star \ ^{\theta }\mathbf{e}_{\alpha \ }^{\ \underline{%
\beta }}=\delta _{\underline{\alpha }}^{\ \underline{\beta }},$ for $\delta
_{\underline{\alpha }}^{\ \underline{\beta }}$ being the Kronecker tensor,
where $\mathbf{e}_{\alpha \ \alpha _{1}\beta _{1}}^{\ \underline{\alpha }}$
and $\mathbf{e}_{\alpha \ \alpha _{1}\beta _{1}\alpha _{2}\beta _{2}}^{\
\underline{\alpha }}$ can be written in terms of $\mathbf{e}_{\alpha \ }^{\
\underline{\alpha }},\theta ^{\alpha \beta }$ and the spin distinguished
connection corresponding to $\widehat{\mathbf{D}}.$

The noncommutative deformation of a metric (\ref{prime}), $\mathbf{g}$ $%
\rightarrow \ ^{\theta }\mathbf{g,}$ can be defined in the form%
\begin{equation}
\ ^{\theta }\mathbf{g}_{\alpha \beta }=\frac{1}{2}\eta _{\underline{\alpha }%
\underline{\beta }}\left[ \ ^{\theta }\mathbf{e}_{\alpha \ }^{\ \underline{%
\alpha }}\star \left( \ ^{\theta }\mathbf{e}_{\beta \ }^{\ \underline{\beta }%
}\right) ^{+}+\ ^{\theta }\mathbf{e}_{\beta \ }^{\ \underline{\beta }}\star
\left( \ ^{\theta }\mathbf{e}_{\alpha \ }^{\ \underline{\alpha }}\right) ^{+}%
\right] ,  \label{dmnc}
\end{equation}%
where $\left( \ldots \right) ^{+}$ denotes Hermitian conjugation and $\eta _{%
\underline{\alpha }\underline{\beta }}$ is the flat Minkowski space metric.
In N--adapted form, as nonholonomic deformations, such metrics were used for
constructing exact solutions in string/gauge/Einstein and Lagrange--Finsler
metric--affine and noncommutative gravity theories.

The target metrics resulting after noncommutative nonholonomic transforms
(to be investigated in this work) can \ be parametrized in general form
\begin{eqnarray}
\ ^{\theta }\mathbf{g} &=&\ ^{\theta }g_{i}(u,\theta )dx^{i}\otimes dx^{i}+\
^{\theta }h_{a}(u,\theta )\ ^{\theta }\mathbf{e}^{a}\otimes \ ^{\theta }%
\mathbf{e}^{a},  \label{target} \\
\ ^{\theta }\mathbf{e}^{\alpha } &=&\ ^{\theta }\mathbf{e}_{\ \underline{%
\alpha }}^{\alpha }(u,\theta )du^{\underline{\alpha }}=\left( e^{i}=dx^{i},\
^{\theta }\mathbf{e}^{a}=dy^{a}+\ ^{\theta }N_{i}^{a}(u,\theta
)dx^{i}\right) ,  \notag
\end{eqnarray}%
where it is convenient to consider conventional polarizations $\eta _{\ldots
}^{\ldots }$ when
\begin{equation}
\ ^{\theta }g_{i}=\check{\eta}_{i}(u,\theta )g_{i},\ \ ^{\theta }h_{a}=%
\check{\eta}_{a}(u,\theta )h_{a},\ ^{\theta }N_{i}^{a}(u,\theta )=\ \check{%
\eta}_{i}^{a}(u,\theta )N_{i}^{a},  \label{polf}
\end{equation}%
for $g_{i},h_{a},N_{i}^{a}$ given by a prime metric (\ref{prime}).

In this work, we shall analyze noncommutative deformations induced by (\ref{fuzcond})
for a class of four dimensional, 4--d, (pseudo) Riemannian
metrics (or 2--d (pseudo) Finsler metrics) defining (non) commutative
Finsler--Einstein spaces as exact solutions of the Einstein equations,
\begin{equation}
\ ^{\theta }\widehat{E}_{\ j}^{i}=\ _{h}^{\theta }\Upsilon (u)\delta _{\
j}^{i},\ \widehat{E}_{\ b}^{a}=\ _{v}^{\theta }\Upsilon (u)\delta _{\
b}^{a},\ ^{\theta }\widehat{E}_{ia}=\ \ ^{\theta }\widehat{E}_{ai}=0,
\label{eeqcdcc}
\end{equation}%
where $\ ^{\theta }\widehat{\mathbf{E}}_{\alpha \beta }=\{\ ^{\theta }%
\widehat{E}_{ij},\ ^{\theta }\widehat{E}_{ia},\ ^{\theta }\widehat{E}_{ai},\
^{\theta }\widehat{E}_{ab}\}$ are the components of the Einstein tensor
computed for the canonical distinguished connection (d--connection) $\ \
^{\theta }\widehat{\mathbf{D}}$. Functions $\ _{h}^{\theta }\Upsilon $ and $%
\ \ _{v}^{\theta }\Upsilon $ are considered to be defined by certain matter
fields in a corresponding model of (non) commutative gravity. The geometric
objects in (\ref{eeqcdcc}) must be computed using the $\star $--product (\ref%
{starpr}) and the coefficients may contain  the complex unity $i.$
Nevertheless, it is possible to prescribe such nonholonomic distributions on
the ''prime'' $\mathbf{V}$ when, for instance,
 $ \widehat{E}_{\ j}^{i}(u)\rightarrow \ \widehat{E}_{\ j}^{i}(u,\theta ),\
_{h}^{\theta }\Upsilon (u)\rightarrow \ _{h}\Upsilon (u,\theta ),\ldots$ %
and we get Lagrange--Finsler and/or (pseudo) Riemannian
geometries, and corresponding gravitational models, with parametric
dependencies of geometric objects on $\theta .$

Solutions of nonholonomic equations (\ref{eeqcdcc}) are typical ones for the
Finsler gravity with metric compatible d--connections\footnote{%
We emphasize that Finlser like coordinates can be considered on any
(pseudo), or complex Riemannian manifold and inversely. A real Finsler metric%
$\ \mathbf{f=\{f}$ $_{\alpha \beta }\}$ can be parametrized in the canonical
Sasaki form $ \mathbf{f}=\ f_{ij}dx^{i}\otimes dx^{j}+\ f_{ab}\ ^{c}\mathbf{e}%
^{a}\otimes \ ^{c}\mathbf{e}^{b}$, $\ ^{c}\mathbf{e}^{a}=dy^{a}+\
^{c}N_{i}^{a}dx^{i}$,
where the Finsler configuration is defied by 1) a fundamental real Finsler
(generating) function $F(u)=F(x,y)=F(x^{i},y^{a})>0$ if $y\neq 0$ and
homogeneous of type $F(x,\lambda y)=|\lambda |F(x,y),$ for any nonzero $%
\lambda \in \mathbb{R},$ with positively definite Hessian $\ f_{ab}=\frac{1}{%
2}\frac{\partial ^{2}F^{2}}{\partial y^{a}\partial y^{b}},$ when $\det |\
f_{ab}|\neq 0$. The Cartan canonical N--connection structure $\ ^{c}\mathbf{N%
}=\{\ ^{c}N_{i}^{a}\}$ is defined for an effective Lagrangian $L=F^{2}$ as $%
\ \ ^{c}N_{i}^{a}=\frac{\partial G^{a}}{\partial y^{2+i}}$ with $G^{a}=%
\frac{1}{4}\ f^{a\ 2+i}\left( \frac{\partial ^{2}L}{\partial y^{2+i}\partial
x^{k}}y^{2+k}-\frac{\partial L}{\partial x^{i}}\right) ,$ where $\ f^{ab}$
is inverse to $\ f_{ab}$ and respective contractions of horizontal (h) and
vertical (v) indices, $\ i,j,...$ and $a,b...,$ are performed following the
rule: we can write, for instance, an up $v$--index $a$ as $a=2+i$ and
contract it with a low index $i=1,2.$ In brief, we shall write $y^{i}$
instead of $y^{2+i},$ or $y^{a}$.} or in the so--called
Einsteing/string/brane/gauge gravity with nonholonomic/Finsler \ like
variables. In the standard approach to the Einstein gravity, when $\widehat{%
\mathbf{D}}\rightarrow \nabla ,$ the Einstein spaces are defined by metrics $%
\mathbf{g}$ as solutions of the equations
\begin{equation}
\ E_{\alpha \beta }=\Upsilon _{\alpha \beta },  \label{eeqlcc}
\end{equation}%
where $\ E_{\alpha \beta }$ is the Einstein tensor for $\nabla $ and $%
\Upsilon _{\alpha \beta }$ is proportional to the energy--momentum tensor of
matter in general relativity. Of course, for noncommutative gravity models
in (\ref{eeqlcc}), we must consider values of type $\ ^{\theta }\nabla ,\ \
^{\theta }E,\ \ ^{\theta }\Upsilon $ etc. Nevertheless, for certain general
classes of ansatz of primary metrics $\mathbf{g}$ on a $\mathbf{V}$ we can
reparametrize such a way the nonholonomic distributions on corresponding $\
^{\theta }\mathbf{V}$ that $\ ^{\theta }\mathbf{g}(u)=\mathbf{\tilde{g}}%
(u,\theta )$ are solutions of (\ref{eeqcdcc}) transformed into a system of
partial differential equations (with parametric dependence of coefficients
on $\theta )$ which after certain further restrictions on coefficients
determining the nonholonomic distribution can result in generic
off--diagonal solutions for general relativity.\footnote{%
the metrics for such spacetimes can not diagonalized by coordinate transforms%
}

\subsubsection{General solutions with noncommutative parameters}

\label{s3} A noncommutative deformation of coordinates of type (\ref{fuzcond}%
) defined by $\theta $ together with correspondingly stated nonholonomic
distributions on $\ ^{\theta }\mathbf{V}$ transform prime metrics $\mathbf{g}
$ (for instance, a Schwarzschild solution on $\mathbf{V}$) into respective
classes of target metrics $\ ^{\theta }\mathbf{g}=\mathbf{\tilde{g}}$ as
solutions of Finsler type gravitational field equations (\ref{eeqcdcc})
and/or standard Einstein equations (\ref{eeqlcc}) in general gravity. The
goal of this section is to show how such solutions and their
noncommutative/nonholonomic transforms can be constructed in general form
for vacuum and non--vacuum locally anisotropic configurations.

We parametrize the noncommutative and nonholonomic transform of a metric $%
\mathbf{g}$ (\ref{prime}) into a $\ ^{\theta }\mathbf{g}=\mathbf{\tilde{g}}$
(\ref{target}) resulting from formulas (\ref{ncfd}), and (\ref{dmnc}) and
expressing of polarizations in (\ref{polf}), as $\check{\eta}_{\alpha
}(u,\theta )=\grave{\eta}_{\alpha }(u)+\mathring{\eta}_{\alpha }(u)\theta
^{2}+\mathcal{O}(\theta ^{4})$, in the form%
\begin{eqnarray}
\ ^{\theta }g_{i} &=&\grave{g}_{i}(u)+\mathring{g}_{i}(u)\theta ^{2}+%
\mathcal{O}(\theta ^{4}),\ ^{\theta }h_{a}=\grave{h}_{a}(u)+\mathring{h}%
_{a}(u)\theta ^{2}+\mathcal{O}(\theta ^{4}),  \notag \\
\ ^{\theta }N_{i}^{3} &=&\ ^{\theta }w_{i}(u,\theta ),\ \ ^{\theta
}N_{i}^{4}=\ ^{\theta }n_{i}(u,\theta ),  \label{coefm}
\end{eqnarray}%
where $\grave{g}_{i}=g_{i}$ and $\grave{h}_{a}=h_{a}$ for $\grave{\eta}%
_{\alpha }=1,$;  for general $\grave{\eta}_{\alpha }(u)$ we get
nonholonomic deformations which do not depend on $\theta .$

The gravitational field equations (\ref{eeqcdcc}) for a metric (\ref{target}) with coefficients (\ref{coefm}) and sources of type
{\small
\begin{equation}
\ ^{\theta }\mathbf{\Upsilon }_{\beta }^{\alpha }=[\Upsilon
_{1}^{1}=\Upsilon _{2}(x^{i},v,\theta ),\Upsilon _{2}^{2}=\Upsilon
_{2}(x^{i},v,\theta ),\Upsilon _{3}^{3}=\Upsilon _{4}(x^{i},\theta
),\Upsilon _{4}^{4}=\Upsilon _{4}(x^{i},\theta )]  \label{sdiag}
\end{equation}%
} transform into this system of partial differential equations: {\small
\begin{eqnarray}
&&\ ^{\theta }\widehat{R}_{1}^{1}=\ ^{\theta }\widehat{R}_{2}^{2}=\frac{1}{%
2\ ^{\theta }g_{1}\ ^{\theta }g_{2}}\times  \label{ep1a} \\
&&\left[ \frac{\ ^{\theta }g_{1}^{\bullet }\ ^{\theta }g_{2}^{\bullet }}{2\
^{\theta }g_{1}}+\frac{(\ ^{\theta }g_{2}^{\bullet })^{2}}{2\ ^{\theta }g_{2}%
}-\ ^{\theta }g_{2}^{\bullet \bullet }+\frac{\ ^{\theta }g_{1}^{^{\prime }}\
^{\theta }g_{2}^{^{\prime }}}{2\ ^{\theta }g_{2}}+\frac{(\ ^{\theta
}g_{1}^{^{\prime }})^{2}}{2\ ^{\theta }g_{1}}-\ ^{\theta }g_{1}^{^{\prime
\prime }}\right] =-\Upsilon _{4}(x^{i},\theta ),  \notag \\
&&\ ^{\theta }\widehat{S}_{3}^{3}=\ ^{\theta }\widehat{S}_{4}^{4}=\frac{1}{%
2\ ^{\theta }h_{3}\ ^{\theta }h_{4}} [\ ^{\theta }h_{4}^{\ast } (\ln \sqrt{%
|\ ^{\theta }h_{3}\ ^{\theta }h_{4}|})^{\ast }-\ ^{\theta }h_{4}^{\ast \ast
}] =-\Upsilon _{2}(x^{i},v,\theta ),  \notag \\
&&\ ^{\theta }\widehat{R}_{3i}=-\ ^{\theta }w_{i}\frac{\beta }{2\ ^{\theta
}h_{4}}-\frac{\alpha _{i}}{2\ ^{\theta }h_{4}}=0, \ \ ^{\theta }\widehat{R}_{4i}=-\frac{\ ^{\theta }h_{3}}{2\ ^{\theta }h_{4}}%
\left[ \ ^{\theta }n_{i}^{\ast \ast }+\gamma \ ^{\theta }n_{i}^{\ast }\right]
=0,  \notag
\end{eqnarray}%
} where, for $\ ^{\theta }h_{3,4}^{\ast }\neq 0,$%
{\small
\begin{equation}
\alpha _{i} = \ ^{\theta }h_{4}^{\ast }\partial _{i}\phi ,\ \beta =\
^{\theta }h_{4}^{\ast }\ \phi ^{\ast },\ \gamma =\frac{3\ ^{\theta
}h_{4}^{\ast }}{2\ ^{\theta }h_{4}}-\frac{\ ^{\theta }h_{3}^{\ast }}{\
^{\theta }h_{3}}, \phi =\ln |\ ^{\theta }h_{3}^{\ast }/\sqrt{|\ ^{\theta }h_{3}\ ^{\theta
}h_{4}|}|,  \label{coefa}
\end{equation}%
}
when the necessary partial derivatives are written in the form \ $a^{\bullet
}=\partial a/\partial x^{1},$ $a^{\prime }=\partial a/\partial x^{2},$\ $%
a^{\ast }=\partial a/\partial v.$ In the vacuum case, we must consider $%
\Upsilon _{2,4}=0.$ Various classes of (non) holonomic Einstein,
Finsler--Einstein and generalized spaces can be generated if the \ sources (%
\ref{sdiag}) are taken $\Upsilon _{2,4}=\lambda ,$ where $\lambda $ is a
nonzero cosmological constant.

Let us express the coefficients of a target metric (\ref{target}), and
respective polarizations (\ref{polf}), in the form%
{\small
\begin{eqnarray}
\ ^{\theta }g_{k} &=&\epsilon _{k}e^{\psi (x^{i},\theta )},  \label{ansatz1}
\\
\ ^{\theta }h_{3} &=&\epsilon _{3}h_{0}^{2}(x^{i},\theta) \left[ f^{\ast
} (x^{i},v,\theta)\right] ^{2}|\varsigma (x^{i},v,\theta),\ ^{\theta }h_{4} = \epsilon _{4}\left[ f (x^{i},v,\theta)
-f_{0}(x^{i},\theta)\right] ^{2},  \notag \\
\ ^{\theta }N_{k}^{3} &=&w_{k}(x^{i},v,\theta ) ,\ ^{\theta
}N_{k}^{4}=n_{k}\left( x^{i},v,\theta \right) ,  \notag
\end{eqnarray}%
}
with arbitrary constants $\epsilon _{\alpha }=\pm 1,$ and $h_{3}^{\ast }\neq
0$ and $h_{4}^{\ast }\neq 0,$ when $f^{\ast }=0.$ By straightforward
verifications, we can prove that any off--diagonal metric
{\small
\begin{eqnarray}
&&\ \ _{\circ }^{\theta }\mathbf{g}=e^{\psi } \epsilon _{i}\
dx^{i}\otimes dx^{i} +\epsilon _{3}h_{0}^{2}\left[ f^{\ast }\right] ^{2}|\varsigma |\ \delta
v\otimes \delta v +\epsilon _{4}\left[ f-f_{0}\right] ^{2}\ \delta
y^{4}\otimes \delta y^{4},  \notag \\
&&\delta v=dv+w_{k}\left( x^{i},v,\theta \right) dx^{k},\ \delta
y^{4}=dy^{4}+n_{k}\left( x^{i},v,\theta \right) dx^{k},  \label{gensol1}
\end{eqnarray}%
}
defines an exact solution of the system of partial differential equations (%
\ref{ep1a}), i.e. of the Einstein equation for the canonical
d--connection (\ref{eeqcdcc}) for a metric of type (\ref{target}) with the
coefficients of form (\ref{ansatz1}), if there are satisfied the conditions%
\footnote{%
we put the left symbol ''$\circ $'' in order to emphasize that such a metric
is a solution of gravitational field equations}:
\begin{enumerate}
\item function $\psi $ is a solution of equation $\epsilon _{1}\psi
^{\bullet \bullet }+\epsilon _{2}\psi ^{^{\prime \prime }}=\Upsilon _{4};$

\item the value $\varsigma $ is computed following formula
\begin{equation*}
\varsigma \left( x^{i},v,\theta \right) =\varsigma _{\lbrack 0]}\left(
x^{i},\theta \right) -\frac{\epsilon _{3}}{8}h_{0}^{2}(x^{i},\theta )\int
\Upsilon _{2}f^{\ast }\left[ f-f_{0}\right] dv
\end{equation*}%
and taken $\varsigma =1$ for $\Upsilon _{2}=0;$

\item for a given source $\Upsilon _{4},$ the N--connection coefficients are
computed following the formulas
\begin{eqnarray*}
w_{i}\left( x^{k},v,\theta \right) &=&-\partial _{i}\varsigma /\varsigma
^{\ast },  \label{gensol1w} \\
n_{k}\left( x^{k},v,\theta \right) &=&\ ^{1}n_{k}\left( x^{i},\theta \right)
+\ ^{2}n_{k}\left( x^{i},\theta \right) \int \frac{\left[ f^{\ast }\right]
^{2}\varsigma dv}{\left[ f-f_{0}\right] ^{3}},  \label{gensol1n}
\end{eqnarray*}%
and $w_{i}\left( x^{k},v,\theta \right) $ are arbitrary functions if $%
\varsigma =1$ for $\Upsilon _{2}=0.$
\end{enumerate}

It should be emphasized that such solutions depend on arbitrary nontrivial
functions $f$ (with $f^{\ast }\neq 0),$ $f_{0},$ $h_{0},$ $\ \varsigma
_{\lbrack 0]},$ $\ ^{1}n_{k}$ and $\ \ ^{2}n_{k},$ and sources $\Upsilon
_{2} $ and $\Upsilon _{4}.$ Such values for the corresponding
quasi--classical limits of solutions to metrics of signatures $\epsilon
_{\alpha }=\pm 1$ have to be defined by certain boundary conditions and
physical considerations.

Ansatz of type (\ref{target}) for coefficients (\ref{ansatz1}) with $%
h_{3}^{\ast }=0$ but $h_{4}^{\ast }\neq 0$ (or, inversely, $h_{3}^{\ast
}\neq 0$ but $h_{4}^{\ast }=0)$ consist more special cases and request a bit
different method of constructing exact solutions.

\subsubsection{Off--diagonal solutions for the Levi--Civita connection}

The solutions for the gravitational field equations for the canonical
d--connection (which can be used for various models of noncommutative
Finsler gravity and generalizations) presented in the previous subsection
can be constrained additionally and transformed into solutions of the
Einstein equations for the Levi--Civita connection (\ref{eeqlcc}), all
depending, in general, on parameter $\theta .$ Such classes of metrics are
of type
\begin{eqnarray}
\ \ _{\circ }^{\theta }\mathbf{g} &=&e^{\psi (x^{i},\theta )}\left[ \epsilon
_{1}\ dx^{1}\otimes dx^{1}+\epsilon _{2}\ dx^{2}\otimes dx^{2}\right]
\label{eeqsol} \\
&&+h_{3}\left( x^{i},v,\theta \right) \ \delta v\otimes \delta v+h_{4}\left(
x^{i},v,\theta \right) \ \delta y^{4}\otimes \delta y^{4},  \notag \\
\delta v &=&dv+w_{1}\left( x^{i},v,\theta \right) dx^{1}+w_{2}\left(
x^{i},v,\theta \right) dx^{2},  \notag \\
\delta y^{4} &=&dy^{4}+n_{1}\left( x^{i},\theta \right) dx^{1}+n_{2}\left(
x^{i},\theta \right) dx^{2},  \notag
\end{eqnarray}%
with the coefficients restricted to satisfy the conditions
\begin{eqnarray*}
\epsilon _{1}\psi ^{\bullet \bullet }+\epsilon _{2}\psi ^{^{\prime \prime }}
&=&\Upsilon _{4},\ h_{4}^{\ast }\phi /h_{3}h_{4}=\Upsilon _{2},  \label{ep2b}
\\
w_{1}^{\prime }-w_{2}^{\bullet }+w_{2}w_{1}^{\ast }-w_{1}w_{2}^{\ast }
&=&0,\ n_{1}^{\prime }-n_{2}^{\bullet }=0,  \notag
\end{eqnarray*}%
for $w_{i}=\partial _{i}\phi /\phi ^{\ast },$ see (\ref{coefa}), for given
sources $\Upsilon _{4}(x^{k},\theta )$ and $\Upsilon _{2}(x^{k},v,\theta ).$

Even the ansatz (\ref{eeqsol}) depends on three coordinates $(x^{k},v)$ and
noncommutative parameter $\theta ,$ it allows us to construct more general
classes of solutions with dependence on four coordinates if such metrics can
be related by chains of nonholonomic transforms.

\subsubsection{Noncommutative deforms of the
Schwarz\-schild metric}

Solutions of type (\ref{gensol1}) and/or (\ref{eeqsol}) are very general
ones induced by noncommutative nonholonomic distributions and it is not
clear what type of physical interpretation can be associated to such
metrics. There are analyzed certain classes of nonholonomic
constraints which allows us to construct black hole solutions and
noncommutative corrections.

\paragraph{Vacuum noncommutative nonholonomic configurations:}

In the simplest case, we analyse a class of holonomic nocommutative
deformations, with $\ _{\shortmid }^{\theta }N_{i}^{a}=0,$\ of the
Schwarzschild metric%
\begin{eqnarray*}
~^{Sch}\mathbf{g} &=&\ _{\shortmid }g_{1}dr\otimes dr+\ _{\shortmid }g_{2}\
d\vartheta \otimes d\vartheta +\ _{\shortmid }h_{3}\ d\varphi \otimes
d\varphi +\ _{\shortmid }h_{4}\ dt\otimes \ dt, \\
\ _{\shortmid }g_{1} &=&-\left( 1-\frac{\alpha }{r}\right) ^{-1},\ \
_{\shortmid }g_{2}=-r^{2},\ \ _{\shortmid }h_{3}=-r^{2}\sin ^{2}\vartheta ,\
\ _{\shortmid }h_{4}=1-\frac{\alpha }{r},
\end{eqnarray*}%
written in spherical coordinates $u^{\alpha }=(x^{1}=\xi ,x^{2}=\vartheta
,y^{3}=\varphi ,y^{4}=t)$ for $\alpha =2G\mu _{0}/c^{2},$ correspondingly
defined by the Newton constant $G,$ a point mass $\mu _{0}$ and light speed $%
c.$ Taking
\begin{eqnarray}
\ _{\shortmid }\grave{g}_{i} &=&\ _{\shortmid }g_{i},\grave{h}_{a}=\
_{\shortmid }h_{a},\   \ _{\shortmid }\mathring{g}_{1} =-\frac{\alpha (4r-3\alpha )}{%
16r^{2}(r-\alpha )^{2}},\ _{\shortmid }\mathring{g}_{2}=-\frac{%
2r^{2}-17\alpha (r-\alpha )}{32r(r-\alpha )},  \notag \\
\ _{\shortmid }\mathring{h}_{3} &=&-\frac{(r^{2}+\alpha r-\alpha ^{2})\cos
\vartheta -\alpha (2r-\alpha )}{16r(r-\alpha )},\ _{\shortmid }\mathring{h}%
_{4}=-\frac{\alpha (8r-11\alpha )}{16r^{4}},  \label{defaux}
\end{eqnarray}%
for $\ _{\shortmid }^{\theta }g_{i}=\ _{\shortmid }\grave{g}_{i}+\
_{\shortmid }\mathring{g}_{i}\theta ^{2}+\mathcal{O}(\theta ^{4}),\ \ \
_{\shortmid }^{\theta }h_{a}=\ _{\shortmid }\grave{h}_{a}+\ _{\shortmid }%
\mathring{h}_{a}\theta ^{2}+\mathcal{O}(\theta ^{4})$, we get a
''degenerated'' case of solutions (\ref{gensol1}), because $\ _{\shortmid
}^{\theta }h_{a}^{\ast }=\partial \ _{\shortmid }^{\theta }h_{a}/\partial
\varphi =0$, which is related to the case of holonomic/ integrable
off--diagonal metrics.

A more general class of noncommutative deformations of the Schwarz\-schild
metric can be generated by nonholonomic transform of type (\ref{polf}) when
the metric coefficients polarizations, $\check{\eta}_{\alpha },$ and
N--connection coefficients, $\ \ _{\shortmid }^{\theta }N_{i}^{a},$ for $\
_{\shortparallel }^{\theta }g_{i} =\check{\eta}_{i}(r,\vartheta ,\theta )\
_{\shortmid }g_{i},\ \ _{\shortparallel }^{\theta }h_{a}=\check{\eta}%
_{a}(r,\vartheta ,\varphi ,\theta )\ _{\shortmid }h_{a},$ $\ \
_{\shortparallel }^{\theta }N_{i}^{3} =\ w_{i}(r,\vartheta ,\varphi ,\theta
),\ \ _{\shortparallel }^{\theta }N_{i}^{4}=\ n_{i}(r,\vartheta ,\varphi
,\theta )$, are constrained to define a metric (\ref{gensol1}) for $\Upsilon
_{4}=\Upsilon _{2}=0.$ The coefficients of such metrics, computed with
respect to N--adapted frames  defined by $\ _{\shortparallel
}^{\theta }N_{i}^{a},$ can be re--parametrized{\small
\begin{eqnarray}
\ \ \ _{\shortparallel }^{\theta }g_{k} &=&\epsilon _{k}e^{\psi (r,\vartheta
,\theta )}=\ _{\shortmid }\grave{g}_{k}+\delta \ _{\shortmid }\grave{g}%
_{k}+(\ _{\shortmid }\mathring{g}_{k}+\delta \ _{\shortmid }\mathring{g}%
_{k})\theta ^{2}+\mathcal{O}(\theta ^{4});  \label{ncfdm} \\
\ \ \ _{\shortparallel }^{\theta }h_{3} &=&\epsilon _{3}h_{0}^{2}\left[
f^{\ast }(r,\vartheta ,\varphi ,\theta )\right] ^{2} =  \notag \\
&& \left( \ _{\shortmid }\grave{h}_{3}+\delta \ _{\shortmid }\grave{h}%
_{3}\right) +\left( \ _{\shortmid }\mathring{h}_{3}+\delta \ _{\shortmid }%
\mathring{h}_{3}\right) \theta ^{2}+\mathcal{O}(\theta ^{4}),h_{0}=const\neq
0;  \notag \\
\ _{\shortparallel }^{\theta }h_{4} &=&\epsilon _{4}[ f(r,\vartheta ,\varphi
,\theta )-f_{0}(r,\vartheta ,\theta )] ^{2}= (\ _{\shortmid }\grave{h}%
_{4}+\delta \ _{\shortmid }\grave{h}_{4}) +(\ _{\shortmid }\mathring{h}%
_{4}+\delta \ _{\shortmid }\mathring{h}_{4}) \theta ^{2}+\mathcal{O}(\theta
^{4}),  \notag
\end{eqnarray}%
} where the nonholonomic deformations $\delta \ _{\shortmid }\grave{g}%
_{k},\delta \ _{\shortmid }\mathring{g}_{k},\delta \ _{\shortmid }\grave{h}%
_{a},\delta \ _{\shortmid }\mathring{h}_{a}$ are for correspondingly given
generating functions $\psi (r,\vartheta ,\theta )$ and $f(r,\vartheta
,\varphi ,\theta )$ expressed as series on $\theta ^{2k},$ for $k=1,2,... .$
Such coefficients define noncommutative Finsler type spacetimes being
solutions of the Einstein equations for the canonical d--connection. They
are determined by the (prime) Schwarzschild data $\ _{\shortmid }g_{i}$ and $%
\ _{\shortmid }h_{a}$ and certain classes on noncommutative nonholonomic
distributions defining off--diagonal gravitational interactions. In order to
get solutions for the Levi--Civita connection, we have to constrain (\ref%
{ncfdm}) additionally in a form to generate metrics of type (\ref{eeqsol})
with coefficients subjected to conditions (\ref{ep2b}) for zero sources $%
\Upsilon _{\alpha }.$

\paragraph{Noncommutative deformations with nontrivial sources:}

In the holonomic case, there are known such noncommutative generalizations
of the Schwarzschild metric when%
{\small
\begin{eqnarray}
~^{ncS}\mathbf{g} &=&\ _{\intercal }g_{1}dr\otimes dr+\ _{\intercal }g_{2}\
d\vartheta \otimes d\vartheta +\ _{\intercal }h_{3}\ d\varphi \otimes
d\varphi +\ _{\intercal }h_{4}\ dt\otimes \ dt,  \label{ncsch}  \\
\ \ _{\intercal }g_{1} &=&-\left( 1-\frac{\ 4\mu _{0}\gamma }{\sqrt{\pi }r}%
\right) ^{-1}, \ _{\intercal }g_{2}=-r^{2},\
 \ _{\intercal }h_{3} =  -r^{2}\sin ^{2}\vartheta , \ _{\intercal
}h_{4}=1-\frac{\ 4\mu _{0}\gamma }{\sqrt{\pi }r}, \notag
\end{eqnarray}%
}
for $\gamma $ being the so--called lower incomplete Gamma function $\gamma (%
\frac{3}{2},\frac{r^{2}}{4\theta }):=
\int\nolimits_{0}^{r^{2}}p^{1/2}e^{-p}dp,$ is the solution of noncommutative  Einstein equations $\ ^{\theta }E_{\alpha \beta
}=\frac{8\pi G}{c^{2}}\ \ ^{\theta }T_{\alpha \beta }$, where $\ ^{\theta
}E_{\alpha \beta }$ is formally left unchanged (i.e. is for the commutative
Levi--Civita connection in commutative coordinates) but
{\small
\begin{equation}
\ ^{\theta }T_{\ \beta }^{\alpha }=\left(
\begin{array}{cccc}
-p_{1} &  &  &  \\
& -p_{\perp } &  &  \\
&  & -p_{\perp } &  \\
&  &  & \rho _{\theta }%
\end{array}%
\right),  \label{ncs}
\end{equation}%
}
with $p_{1}=-\rho _{\theta }$ and $p_{\perp }=-\rho _{\theta }-\frac{r}{2}%
\partial _{r}\rho _{\theta }(r)$ is taken for a self--gravitating,
anisotropic fluid--type matter modeling noncommutativity.

Via nonholonomic deforms, we can generalize the solution (\ref{ncsch}) to
off--diagonal metrics of type {\small
\begin{eqnarray}
\ ~_{\theta }^{ncS}\mathbf{g} &=&-e^{\psi (r,\vartheta ,\theta )}\left[ \
dr\otimes dr+d\vartheta \otimes d\vartheta \right] -h_{0}^{2}\left[ f^{\ast
}(r,\vartheta ,\varphi ,\theta )\right] ^{2}|\varsigma (r,\vartheta ,\varphi
,\theta )|\   \notag \\
&& \delta \varphi \otimes \delta \varphi +\left[ f(r,\vartheta ,\varphi
,\theta )-f_{0}(r,\vartheta ,\theta )\right] ^{2}\ \delta t\otimes \delta t,
\label{ncsolsch} \\
\delta \varphi &=&d\varphi +w_{1}(r,\vartheta ,\varphi ,\theta
)dr+w_{2}(r,\vartheta ,\varphi ,\theta )d\vartheta ,  \notag \\
\delta t &=&dt+n_{1}(r,\vartheta ,\varphi ,\theta )dr+n_{2}(r,\vartheta
,\varphi ,\theta )d\vartheta ,  \notag
\end{eqnarray}%
} being exact solutions of the Einstein equation for the canonical
d--connection (\ref{eeqcdcc}) with locally anisotropically self--gravitating
source
\begin{equation*}
\ ^{\theta }\mathbf{\Upsilon }_{\beta }^{\alpha }=[\Upsilon
_{1}^{1}=\Upsilon _{2}^{2}=\Upsilon _{2}(r,\vartheta ,\varphi ,\theta
),\Upsilon _{3}^{3}=\Upsilon _{4}^{4}=\Upsilon _{4}(r,\vartheta ,\theta )].
\end{equation*}
Such sources should be taken with certain polarization coefficients when $%
\Upsilon \sim \eta T$ is constructed using the matter energy--momentum
tensor (\ref{ncs}).

The coefficients of metric (\ref{ncsolsch}) are computed to satisfy
correspondingly the conditions:

\begin{enumerate}
\item function $\psi (r,\vartheta ,\theta )$ is a solution of equation $\psi
^{\bullet \bullet }+\psi ^{^{\prime \prime }}=-\Upsilon _{4};$

\item for a nonzero constant $h_{0}^{2},$ and given $\Upsilon _{2},$
\begin{equation*}
\varsigma \left( r,\vartheta ,\varphi ,\theta \right) =\varsigma _{\lbrack
0]}\left( r,\vartheta ,\theta \right) +h_{0}^{2}\int \Upsilon _{2}f^{\ast }
\left[ f-f_{0}\right] d\varphi ;
\end{equation*}

\item the N--connection coefficients are
\begin{eqnarray*}
w_{i}\left( r,\vartheta ,\varphi ,\theta \right) &=&-\partial _{i}\varsigma
/\varsigma ^{\ast }, \\
n_{k}\left( r,\vartheta ,\varphi ,\theta \right) &=&\ ^{1}n_{k}\left(
r,\vartheta ,\theta \right) +\ ^{2}n_{k}\left( r,\vartheta ,\theta \right)
\int \frac{\left[ f^{\ast }\right] ^{2}\varsigma }{\left[ f-f_{0}\right] ^{3}%
}d\varphi .
\end{eqnarray*}
\end{enumerate}

The above presented class of metrics describes nonholonomic deformations of
the Schwarzschild metric into (pseudo) Finsler configurations induced by the
noncommutative parameter. Subjecting the coefficients of (\ref{ncsolsch}) to
additional constraints of type (\ref{ep2b}) with nonzero sources $\Upsilon
_{\alpha },$ we extract a subclass of solutions for noncommutative gravity
with effective Levi--Civita connection.

\paragraph{Noncommutative ellipsoidal deformations:}

In this section, we provide a method of extracting ellipsoidal
configurations from a general metric (\ref{ncsolsch}) with coefficients
constrained to generate solutions on the Einstein equations for the
canonical d--connection or Levi--Civita connection.

We consider a diagonal metric depending on noncommutative parameter $\theta$
(in general, such a metric may not solve  any gravitational field
equations)
\begin{equation}
~^{\theta }\mathbf{g}=-d\xi \otimes d\xi -r^{2}(\xi )\ d\vartheta \otimes
d\vartheta -r^{2}(\xi )\sin ^{2}\vartheta \ d\varphi \otimes d\varphi
+\varpi ^{2}(\xi )\ dt\otimes \ dt,  \label{5aux1}
\end{equation}%
where the local coordinates and nontrivial  coefficients of metric are
\begin{eqnarray}
x^{1} &=&\xi ,x^{2}=\vartheta ,y^{3}=\varphi ,y^{4}=t,  \label{5aux1p} \\
\check{g}_{1} &=&-1,\ \check{g}_{2}=-r^{2}(\xi ),\ \check{h}_{3}=-r^{2}(\xi
)\sin ^{2}\vartheta ,\ \check{h}_{4}=\varpi ^{2}(\xi ),  \notag
\end{eqnarray}%
for $\xi =\int dr\ \left| 1-\frac{2\mu _{0}}{r}+\frac{\theta }{r^{2}}\right|
^{1/2}\mbox{\ and\ }\varpi ^{2}(r)=1-\frac{2\mu _{0}}{r}+\frac{\theta }{r^{2}%
}$. For $\theta =0$ and variable $\xi (r),$ this metric is just the the
Schwarzschild solution written in spacetime spherical coordinates $%
(r,\vartheta ,\varphi ,t).$

Target metrics are generated by nonholonomic deforms with $g_{i}=\eta _{i}%
\check{g}_{i}$ and $h_{a}=\eta _{a}\check{h}_{a}$ and some nontrivial $%
w_{i},n_{i},$ where $(\check{g}_{i},\check{h}_{a})$ are given by data (\ref%
{5aux1p}) and parametrized by an ansatz of type (\ref{ncsolsch}),
\begin{eqnarray}
~_{\eta }^{\theta }\mathbf{g} &=&-\eta _{1}(\xi ,\vartheta ,\theta )d\xi
\otimes d\xi -\eta _{2}(\xi ,\vartheta ,\theta )r^{2}(\xi )\ d\vartheta
\otimes d\vartheta  \label{5sol1} \\
&&-\eta _{3}(\xi ,\vartheta ,\varphi ,\theta )r^{2}(\xi )\sin ^{2}\vartheta
\ \delta \varphi \otimes \delta \varphi +\eta _{4}(\xi ,\vartheta ,\varphi
,\theta )\varpi ^{2}(\xi )\ \delta t\otimes \delta t,  \notag \\
\delta \varphi &=&d\varphi +w_{1}(\xi ,\vartheta ,\varphi ,\theta )d\xi
+w_{2}(\xi ,\vartheta ,\varphi ,\theta )d\vartheta ,\   \notag \\
\delta t &=&dt+n_{1}(\xi ,\vartheta ,\theta )d\xi +n_{2}(\xi ,\vartheta
,\theta )d\vartheta ;  \notag
\end{eqnarray}
the coefficients of such metrics are constrained to be solutions of the
system of equations (\ref{ep1a}). Such equations for $\Upsilon _{2}=0$ state certain
relations between the coefficients of the vertical metric and polarization functions,%
{\small
\begin{equation}
h_{3} =-h_{0}^{2}(b^{\ast })^{2}=\eta _{3}(\xi ,\vartheta ,\varphi ,\theta
)r^{2}(\xi )\sin ^{2}\vartheta, \
h_{4} = b^{2}=\eta _{4}(\xi ,\vartheta ,\varphi ,\theta )\varpi ^{2}(\xi ),
 \label{aux41}
\end{equation}%
}
for $|\eta _{3}|=(h_{0})^{2}|\check{h}_{4}/\check{h}_{3}|\left[ \left( \sqrt{%
|\eta _{4}|}\right) ^{\ast }\right] ^{2}.$ In these formulas, we have to
chose $h_{0}=const$ (it must be $h_{0}=2$ in order to satisfy the condition (%
\ref{ep2b})), where $\eta _{4}$ can be any function satisfying the condition
$\eta _{4}^{\ast }\neq 0.$ We generate a class of solutions for any function
$b(\xi ,\vartheta ,\varphi ,\theta )$ with $b^{\ast }\neq 0.$ For classes of
solutions with nontrivial sources, it is more convenient to work directly
with $\eta _{4},$ for $\eta _{4}^{\ast }\neq 0$ but, for vacuum
configurations, we can chose as a generating function, for instance, $h_{4},
$ for $h_{4}^{\ast }\neq 0.$

It is possible to compute the polarizations $\eta _{1}$ and $\eta _{2},$
when $\eta _{1}=\eta _{2}r^{2}=e^{\psi (\xi ,\vartheta )},$ from (\ref{ep1a}%
) with $\Upsilon _{4}=0,$ i.e. from $\psi ^{\bullet \bullet }+\psi ^{\prime
\prime }=0.$

Putting the above defined values of coefficients in the ansatz (\ref{5sol1}%
), we find a class of exact vacuum solutions of the Einstein equations
defining stationary nonholonomic deformations of the Sch\-warz\-schild
metric, {\small
\begin{eqnarray}
~^{\varepsilon }\mathbf{g} &=&-e^{\psi (\xi ,\vartheta ,\theta )}\left( d\xi
\otimes d\xi +\ d\vartheta \otimes d\vartheta \right)  \label{5sol1a} \\
&&-4\left[ \left( \sqrt{|\eta _{4}(\xi ,\vartheta ,\varphi ,\theta )|}%
\right) ^{\ast }\right] ^{2}\varpi ^{2}(\xi )\ \delta \varphi \otimes \
\delta \varphi +\eta _{4}(\xi ,\vartheta ,\varphi ,\theta )\varpi ^{2}(\xi
)\ \delta t\otimes \delta t,  \notag \\
\delta \varphi &=&d\varphi +w_{1}(\xi ,\vartheta ,\varphi ,\theta )d\xi
+w_{2}(\xi ,\vartheta ,\varphi ,\theta )d\vartheta , \notag \\
\delta t&=& dt+\ ^{1}n_{1}(\xi ,\vartheta ,\theta )d\xi +\ ^{1}n_{2}(\xi
,\vartheta ,\theta )d\vartheta .  \notag
\end{eqnarray}%
} The N--connection coefficients $w_{i}$ and $\ ^{1}n_{i}$ in (\ref{5sol1a})
must satisfy the last two conditions\ from (\ref{ep2b}) in order to get
vacuum metrics in Einstein gravity. Such vacuum solutions are for
nonholonomic deformations of a static black hole metric into (non) holonomic
noncommutative Einstein spaces with locally anistoropic backgrounds (on
coordinate $\varphi )$ defined by an arbitrary function $\eta _{4}(\xi
,\vartheta ,\varphi ,\theta )$ with $\partial _{\varphi }\eta _{4}\neq 0,$
an arbitrary $\psi (\xi ,\vartheta ,\theta )$ solving the 2--d Laplace
equation and certain integration functions $\ ^{1}w_{i}(\xi ,\vartheta
,\varphi ,\theta )$ and $\ ^{1}n_{i}(\xi ,\vartheta ,\theta ).$ The
nonholonomic structure of such spaces depends parametrically on
noncommutative parameter(s) $\theta .$

In general, the solutions from the target set of metrics (\ref{5sol1}), or (%
\ref{5sol1a}), do not define black holes and do not describe obvious
physical situations. Nevertheless, they preserve the singular character of
the coefficient $\varpi ^{2}(\xi )$ vanishing on the horizon of a
Schwarzschild black hole if we take only smooth integration functions for
some small noncommutative parameters $\theta .$ We can also consider a
prescribed physical situation when, for instance, $\eta _{4}$ mimics 3--d,
or 2--d, solitonic polarizations on coordinates $\xi ,\vartheta ,\varphi ,$
or on $\xi ,\varphi .$

\subsubsection{Extracting black hole and rotoid configurations}

\label{s5} From a class of metrics (\ref{5sol1a}) defining nonholonomic
noncommutative deformations of the Schwarzschild solution depending on
parameter $\theta ,$ it is possible to select locally anisotropic
configurations with possible physical interpretation of gravitational vacuum
configurations with spherical and/or rotoid (ellipsoid) symmetry.

\paragraph{Linear parametric noncommutative polarizations:}

Let us consider generating functions of type
 $b^{2}=q(\xi ,\vartheta ,\varphi )+ \bar{\theta} s(\xi ,\vartheta ,\varphi )$
and, for simplicity, restrict our analysis only with linear decompositions
on a small dimensionless parameter $\bar{\theta}\sim \theta ,$ with $0<\bar{%
\theta}<<1.$ This way, we shall construct off--diagonal exact solutions of
the Einstein equations depending on $\ \bar{\theta} $ which for rotoid
configurations can be considered as a small eccentricity.\footnote{%
From a formal point of view, we can summarize on all orders $\ \left(\bar{%
\theta}\right) ^{2},$ $\left(\bar{\theta}\right) ^{3}...$ stating such
recurrent formulas for coefficients when get convergent series to some
functions depending both on spacetime coordinates and a parameter $\bar{%
\theta}$.} For $b$, we get $\left( b^{\ast }\right) ^{2}=%
\left[ (\sqrt{|q|})^{\ast }\right] ^{2} [ 1+\bar{\theta}\frac{1}{(\sqrt{%
|q|})^{\ast }}\left( \frac{s}{\sqrt{|q|}}\right) ^{\ast }]$, which
allows us to compute the vertical coefficients of d--metric (\ref{5sol1a})
(i.e $h_{3}$ and $h_{4}$ and corresponding polarizations $\eta _{3} $ and $%
\eta _{4})$ using formulas (\ref{aux41}).
On should emphasize that nonholonomic deformations are not obligatory
related to noncommutative ones. For instance, in a particular case, we can
generate nonholonomic deformations of the Schwarzschild solution not
depending on $\ \bar{\theta}:$ we have to put $\bar{\theta}=0$ in the above
formulas\ and consider $b^{2}=q$ and $\left( b^{\ast }\right) ^{2}= [(%
\sqrt{|q|})^{\ast }] ^{2}.$

Nonholonomic deforms to rotoid configurations can be generated for
\begin{equation}
q=1-\frac{2\mu (\xi ,\vartheta ,\varphi )}{r}\mbox{ and }s=\frac{q_{0}(r)}{%
4\mu ^{2}}\sin (\omega _{0}\varphi +\varphi _{0}),  \label{aux42}
\end{equation}%
with $\mu (\xi ,\vartheta ,\varphi )=\mu _{0}+\bar{\theta}\mu _{1}(\xi
,\vartheta ,\varphi )$ (anisotropically polarized mass) with certain
constants $\mu ,\omega _{0}$ and $\varphi _{0}$ and arbitrary
functions/polarizations $\mu _{1}(\xi ,\vartheta ,\varphi )$ and $q_{0}(r)$
to be determined from some boundary conditions, with $\ \bar{\theta}$
treated as the eccentricity of an ellipsoid.\footnote{%
we can relate $\bar{\theta}$ to an eccentricity because the coefficient $%
h_{4}=b^{2}=\eta _{4}(\xi ,\vartheta ,\varphi ,\ \bar{\theta} )$ $\varpi
^{2}(\xi )$ becomes zero for data (\ref{aux42}) if $r_{+}\simeq {2\mu _{0}}/[%
{1+\bar{\theta} \frac{q_{0}(r)}{4\mu ^{2}}\sin (\omega _{0}\varphi +\varphi
_{0})}],$ which is the ''parametric'' equation for an ellipse $r_{+}(\varphi
)$ for any fixed values $\frac{q_{0}(r)}{4\mu ^{2}},\omega _{0},\varphi _{0}$
and $\mu _{0}$} Such a noncommutative nonholonomic configuration determines
a small deformation of the Schwarzschild spherical horizon into a rotoid configuration with eccentricity $\bar{\theta}.$

We provide the general solution for noncommutative ellipsoidal black holes
determined by nonholonomic h--components of metric and N--connecti\-on
coefficients which ''survive'' in the limit $\ \bar{\theta}\rightarrow 0,$
i.e. such values do not depend on noncommutative parameter. Dependence \ on
noncommutativity is contained in v--components of metric. This class of
stationary rotoid type solutions is parametrized in the form {\small
\begin{eqnarray*}
~_{\theta }^{rot}\mathbf{g} &=&-e^{\psi }\left( d\xi \otimes d\xi +\
d\vartheta \otimes d\vartheta \right) -4\left[ (\sqrt{|q|})^{\ast }\right]
^{2} [1+\bar{\theta}\frac{1}{(\sqrt{|q|})^{\ast }}\left( \frac{s}{\sqrt{%
|q|}}\right) ^{\ast }]  \notag \\
&& \delta \varphi \otimes \ \delta \varphi +\left( q+\bar{\theta}s\right) \
\delta t\otimes \delta t,  \notag  \label{rotoidm} \\
\delta \varphi &=&d\varphi +w_{1}d\xi +w_{2}d\vartheta ,\ \delta t=dt+\
^{1}n_{1}d\xi +\ ^{1}n_{2}d\vartheta ,  \notag
\end{eqnarray*}%
} with functions $q(\xi ,\vartheta ,\varphi )$ and $s(\xi ,\vartheta
,\varphi ) $ given by formulas (\ref{aux42}) and N--connec\-ti\-on
coefficients $w_{i}(\xi ,\vartheta ,\varphi )$ and $\ n_{i}=$ $\
^{1}n_{i}(\xi ,\vartheta ) $ subjected to conditions
$w_{1}w_{2}\left( \ln |\frac{w_{1}}{w_{2}}|\right) ^{\ast } =w_{2}^{\bullet
}-w_{1}^{\prime },\quad w_{i}^{\ast }\neq 0;$
 or  $w_{2}^{\bullet }-w_{1}^{\prime } = 0,\quad w_{i}^{\ast }=0;\
^{1}n_{1}^{\prime }(\xi ,\vartheta )-\ ^{1}n_{2}^{\bullet }(\xi ,\vartheta
)=0$
and $\psi (\xi ,\vartheta )$ being any function for which $\psi ^{\bullet
\bullet }+\psi ^{\prime \prime }=0.$

\paragraph{Rotoids and noncommutative solitonic distributions:}

There are sta\-tic three dimensional solitonic distributions $\eta (\xi
,\vartheta ,\varphi ,\theta ),$ defined as solutions of a solitonic equation%
\footnote{%
a function $\eta $ can be a solution of any three dimensional solitonic and/
or other nonlinear wave equations} $\eta ^{\bullet \bullet }+\epsilon (\eta
^{\prime }+6\eta \ \eta ^{\ast }+\eta ^{\ast \ast \ast })^{\ast }=0,\
\epsilon =\pm 1$, resulting in stationary black ellipsoid--solitonic
noncommutative spacetimes $^{\theta }\mathbf{V}$ \ generated as further
deformations of a metric $~_{\theta }^{rot}\mathbf{g}$ (\ref{rotoidm}). Such
metrics are of type {\small
\begin{eqnarray}
~_{sol\theta }^{rot}\mathbf{g} &=&-e^{\psi }\left( d\xi \otimes d\xi +\
d\vartheta \otimes d\vartheta \right)  \label{solrot} \\
&&-4\left[ (\sqrt{|\eta q|})^{\ast }\right] ^{2}\left[ 1+\bar{\theta}\frac{1%
}{(\sqrt{|\eta q|})^{\ast }}\left( \frac{s}{\sqrt{|\eta q|}}\right) ^{\ast }%
\right] \ \delta \varphi \otimes \ \delta \varphi  \notag \\
&&+\eta \left( q+\bar{\theta}s\right) \ \delta t\otimes \delta t,  \notag \\
\delta \varphi &=&d\varphi +w_{1}d\xi +w_{2}d\vartheta ,\ \delta t=dt+\
^{1}n_{1}d\xi +\ ^{1}n_{2}d\vartheta.  \notag
\end{eqnarray}%
} For small values of $\bar{\theta},$ a possible spacetime noncommutativity
determines nonholonomic embedding of the Schwarzschild solution into a
solitonic vacuum. In the limit of small polarizations, when $|\eta |\sim 1,$
it is preserved the black hole character of metrics and the solitonic
distribution can be considered as on a Schwarzschild background. It is also
possible to take such parameters of $\eta $ when a black hole is
nonholonomically placed on a ''gravitational hill'' defined by a soliton
induced by spacetime noncommutativity.

A vacuum metric (\ref{solrot}) can be generalized for (pseudo) Finsler
spaces with canonical d--connection as a solution of equations $\widehat{%
\mathbf{R}}_{\alpha \beta }=0$ (\ref{eeqcdcc}) if the metric is generalized
to a subclass of (\ref{5sol1}) with stationary coefficients subjected to
conditions
\begin{eqnarray*}
&&\psi ^{\bullet \bullet }(\xi ,\vartheta ,\bar{\theta})+\psi ^{^{\prime
\prime }}(\xi ,\vartheta ,\bar{\theta})=0; \\
h_{3} &=&\pm e^{-2\ ^{0}\phi }\frac{\left( h_{4}^{\ast }\right) ^{2}}{h_{4}}%
\mbox{ for  given }h_{4}(\xi ,\vartheta ,\varphi ,\bar{\theta}),\ \phi =\
^{0}\phi =const; \\
w_{i} &=&w_{i}(\xi ,\vartheta ,\varphi ,\bar{\theta})%
\mbox{ are any
functions  }; \\
n_{i} &=&\ \ ^{1}n_{i}(\xi ,\vartheta ,\bar{\theta})+\ ^{2}n_{i}(\xi
,\vartheta ,\bar{\theta})\int \left( h_{4}^{\ast }\right)
^{2}|h_{4}|^{-5/2}dv,\ n_{i}^{\ast }\neq 0; \\
&=&\ ^{1}n_{i}(\xi ,\vartheta ,\bar{\theta}),n_{i}^{\ast }=0,
\end{eqnarray*}%
for $h_{4}=\eta (\xi ,\vartheta ,\varphi ,\bar{\theta})\left[ q(\xi
,\vartheta ,\varphi )+\bar{\theta}s(\xi ,\vartheta ,\varphi )\right] .$ In
the limit $\bar{\theta}\rightarrow 0,$ we get a Schwarzschild configuration
mapped nonholonomically on a N--anholonomic (pseudo) Riemannian spacetime
with a prescribed nontrivial N--connection structure.

\subsubsection{Noncommutative gravity and (pseudo) Finsler variables}
We summarize the main steps of such noncommutative complex Finsler --
(pseudo) Riemannian transform:
\begin{enumerate}
\item Let us consider a solution for (non)holonomic noncommutative
generalized Einstein gravity with a metric\footnote{%
we shall omit the left label $\theta $ in this section if this will not
result in ambiguities}
\begin{eqnarray*}
\ ^{\theta }\mathbf{\mathring{g}} &=&\mathring{g}_{i}dx^{i}\otimes dx^{i}+%
\mathring{h}_{a}(dy^{a}+\mathring{N}_{j}^{a}dx^{j})\otimes (dy^{a}+\mathring{%
N}_{i}^{a}dx^{i}) \\
&=&\mathring{g}_{i}e^{i}\otimes e^{i}+\mathring{h}_{a}\mathbf{\mathring{e}}%
^{a}\otimes \mathbf{\mathring{e}}^{a}=\mathring{g}_{i^{\prime \prime
}j^{\prime \prime }}e^{i^{\prime \prime }}\otimes e^{j^{\prime \prime }}+%
\mathring{h}_{a^{\prime \prime }b^{\prime \prime }}\mathbf{\mathring{e}}%
^{a^{\prime \prime }}\otimes \mathbf{\mathring{e}}^{b^{\prime \prime }}
\end{eqnarray*}%
related to an arbitrary (pseudo) Riemannian metric with transforms of type
 $\ ^{\theta }\mathbf{\mathring{g}}_{\alpha ^{\prime \prime }\beta ^{\prime
\prime }}=\ \mathbf{\mathring{e}}_{\ \alpha ^{\prime \prime }}^{\alpha
^{\prime }}\ \mathbf{\mathring{e}}_{\ \beta ^{\prime \prime }}^{\beta
^{\prime }}\ ^{\theta }\mathbf{g}_{\alpha ^{\prime }\beta ^{\prime }}$
 parametrized in the form%
\begin{equation*}
\mathring{g}_{i^{\prime \prime }j^{\prime \prime }}=g_{i^{\prime }j^{\prime
}}\mathbf{\mathring{e}}_{\ i^{\prime \prime }}^{i^{\prime }}\mathbf{%
\mathring{e}}_{\ j^{\prime \prime }}^{j^{\prime }}+h_{a^{\prime }b^{\prime }}%
\mathbf{\mathring{e}}_{\ i^{\prime \prime }}^{a^{\prime }}\mathbf{\mathring{e%
}}_{\ j^{\prime \prime }}^{b^{\prime }},\ \mathring{h}_{a^{\prime \prime
}b^{\prime \prime }}=g_{i^{\prime }j^{\prime }}\mathbf{\mathring{e}}_{\
a^{\prime \prime }}^{i^{\prime }}\mathbf{\mathring{e}}_{\ b^{\prime \prime
}}^{j^{\prime }}+h_{a^{\prime }b^{\prime }}\mathbf{\mathring{e}}_{\
a^{\prime \prime }}^{a^{\prime }}\mathbf{\mathring{e}}_{\ b^{\prime \prime
}}^{b^{\prime }}.
\end{equation*}
For $\mathbf{\mathring{e}}_{\ i^{\prime \prime }}^{i^{\prime }}=\delta _{\
i^{\prime \prime }}^{i^{\prime }},\mathbf{\mathring{e}}_{\ a^{\prime \prime
}}^{a^{\prime }}=\delta _{\ a^{\prime \prime }}^{a^{\prime }},$ we write
 $\mathring{g}_{i^{\prime \prime }}=g_{i^{\prime \prime }}+h_{a^{\prime
}}\left( \mathbf{\mathring{e}}_{\ i^{\prime \prime }}^{a^{\prime }}\right)
^{2},~\mathring{h}_{a^{\prime \prime }}=g_{i^{\prime }}\left( \mathbf{%
\mathring{e}}_{\ a^{\prime \prime }}^{i^{\prime }}\right) ^{2}+h_{a^{\prime
\prime }}$,
 i.e. in a form of four equations for eight unknown variables $\mathbf{%
\mathring{e}}_{\ i^{\prime \prime }}^{a^{\prime }}$ and $\mathbf{\mathring{e}%
}_{\ a^{\prime \prime }}^{i^{\prime }},$ and
 $\ \mathring{N}_{i^{\prime \prime }}^{a^{\prime \prime }}=\mathbf{\mathring{e}%
}_{i^{\prime \prime }}^{\ i^{\prime }}\ \mathbf{\mathring{e}}_{\ a^{\prime
}}^{a^{\prime \prime }}\ N_{i^{\prime }}^{a^{\prime }}=N_{i^{\prime \prime
}}^{a^{\prime \prime }}$.

\item We choose on $\ ^{\theta }\mathbf{V}$ a fundamental Finsler function $%
F=\ ^{3}F(x^{i},v,\theta )+\ ^{4}F(x^{i},y,\theta )$ inducing canonically a
d--metric of type
\begin{eqnarray*}
\ ^{\theta }\mathbf{f} &=&\ f_{i}dx^{i}\otimes dx^{i}+\ f_{a}(dy^{a}+\
^{c}N_{j}^{a}dx^{j})\otimes (dy^{a}+\ ^{c}N_{i}^{a}dx^{i}), \\
&=&\ f_{i}e^{i}\otimes e^{i}+\ f_{a}\ ^{c}\mathbf{e}^{a}\otimes \ ^{c}%
\mathbf{e}^{a}
\end{eqnarray*}%
determined by data $\ \ ^{\theta }\mathbf{f}_{\alpha \beta }=\left[ \
f_{i},\ f_{a},\ ^{c}N_{j}^{a}\right] $ in a canonical N--elongated base $\
^{c}\mathbf{e}^{\alpha }=(dx^{i},\ ^{c}\mathbf{e}^{a}=dy^{a}+\
^{c}N_{i}^{a}dx^{i}).$

\item We define $g_{i^{\prime }}=\ f_{i^{\prime }}\left( \frac{\mathring{w}%
_{i^{\prime }}}{\ ^{c}w_{i^{\prime }}}\right) ^{2}\frac{h_{3^{\prime }}}{\
f_{3^{\prime }}}\mbox{ \ and \ }\ g_{i^{\prime }}=\ f_{i^{\prime }}\left(
\frac{\mathring{n}_{i^{\prime }}}{\ ^{c}n_{i^{\prime }}}\right) ^{2}\frac{%
h_{4^{\prime }}}{\ f_{4^{\prime }}}$. Both formulas are compatible if $%
\mathring{w}_{i^{\prime }}$ and $\mathring{n}_{i^{\prime }}$ are constrained
to satisfy the conditions
 $\Theta _{1^{\prime }}=\Theta _{2^{\prime }}=\Theta$,
 where $\ \Theta _{i^{\prime }}=\left( \frac{\mathring{w}_{i^{\prime }}}{\
^{c}w_{i^{\prime }}}\right) ^{2}\left( \frac{\mathring{n}_{i^{\prime }}}{\
^{c}n_{i^{\prime }}}\right) ^{2},$ \ and $\ \Theta =\left( \frac{\mathring{w}%
_{1^{\prime }}}{\ ^{c}w_{1^{\prime }}}\right) ^{2}\left( \frac{\mathring{n}%
_{1^{\prime }}}{\ ^{c}n_{1^{\prime }}}\right) ^{2}=$ $\left( \frac{\mathring{w}_{2^{\prime }}}{\ ^{c}w_{2^{\prime }}}\right)
^{2}\left( \frac{\mathring{n}_{2^{\prime }}}{\ ^{c}n_{2^{\prime }}}\right)
^{2}.$ \ Using $\Theta ,$ we compute $g_{i^{\prime }}=\left( \frac{\mathring{%
w}_{i^{\prime }}}{\ ^{c}w_{i^{\prime }}}\right) ^{2}\frac{\ f_{i^{\prime }}}{%
\ f_{3^{\prime }}}$ and $h_{3^{\prime }}=h_{4^{\prime }}\Theta$, where (in
this case) there is not summing on indices. So, we constructed the data $%
g_{i^{\prime }},h_{a^{\prime }}$ and $w_{i^{\prime }},n_{j^{\prime }}.$

\item The values $\mathbf{\mathring{e}}_{\ i^{\prime \prime }}^{a^{\prime }}$
and $\mathbf{\mathring{e}}_{\ a^{\prime \prime }}^{i^{\prime }}$ are
determined as any nontrivial solutions of
\begin{equation*}
\mathring{g}_{i^{\prime \prime }}=g_{i^{\prime \prime }}+h_{a^{\prime
}}\left( \mathbf{\mathring{e}}_{\ i^{\prime \prime }}^{a^{\prime }}\right)
^{2},\ \mathring{h}_{a^{\prime \prime }}=g_{i^{\prime }}\left( \mathbf{%
\mathring{e}}_{\ a^{\prime \prime }}^{i^{\prime }}\right) ^{2}+h_{a^{\prime
\prime }},\ \mathring{N}_{i^{\prime \prime }}^{a^{\prime \prime
}}=N_{i^{\prime \prime }}^{a^{\prime \prime }}.
\end{equation*}%
For instance, we can choose and, respectively, express
\begin{eqnarray*}
\mathbf{\mathring{e}}_{\ 1^{\prime \prime }}^{3^{\prime }} &=&\pm \sqrt{%
\left| \left( \mathring{g}_{1^{\prime \prime }}-g_{1^{\prime \prime
}}\right) /h_{3^{\prime }}\right| },\mathbf{\mathring{e}}_{\ 2^{\prime
\prime }}^{3^{\prime }}=0,\mathbf{\mathring{e}}_{\ i^{\prime \prime
}}^{4^{\prime }}=0 \\
\mathbf{\mathring{e}}_{\ a^{\prime \prime }}^{1^{\prime }} &=&0,\mathbf{%
\mathring{e}}_{\ 3^{\prime \prime }}^{2^{\prime }}=0,\mathbf{\mathring{e}}%
_{\ 4^{\prime \prime }}^{2^{\prime }}=\pm \sqrt{\left| \left( \mathring{h}%
_{4^{\prime \prime }}-h_{4^{\prime \prime }}\right) /g_{2^{\prime }}\right| },
\end{eqnarray*}%
and
 $e_{\ 1}^{1^{\prime }}=\pm \sqrt{\left| \frac{\ f_{1}}{g_{1^{\prime }}}%
\right| },\ e_{\ 2}^{2^{\prime }}=\pm \sqrt{\left| \frac{\ f_{2}}{%
g_{2^{\prime }}}\right| },\ e_{\ 3}^{3^{\prime }}=\pm \sqrt{\left| \frac{\
f_{3}}{h_{3^{\prime }}}\right| },\ e_{\ 4}^{4^{\prime }}=\pm \sqrt{\left|
\frac{\ f_{4}}{h_{4^{\prime }}}\right| }$.
\end{enumerate}

Finally, in this section, we conclude that any model of noncommutative
nonhlonomic gravity with distributions of type (\ref{fuzcond}) and/or (\ref%
{whitney}) can be equivalently re--formulated as a Finsler gravity induced
by a generating function of type $F=\ ^{3}F+\ ^{4}F.$ In the limit $\theta
\rightarrow 0,$ for any solution $^{\theta }\mathbf{\mathring{g}},$ there is
a scheme of two nonholonomic transforms which allows us to rewrite the
Schwarzschild solution and its noncommutative/nonholonomic deformations as a
Finsler metric $\ ^{\theta }\mathbf{f.}$

\subsection{Geometric methods and quantum gravity}
Let us consider a real (pseudo) Riemann manifold $V^{2n}$ of necessary
smooth class; $\dim V^{2n}=2n,$ where the dimension $n\geq 2$ is fixed.%
\footnote{%
for constructions related to Einstein's gravity $2n=4$} We label the local
coordinates in the form $u^{\alpha }=(x^{i},y^{a}),$ or $u=(x,y),$ where
indices run values $i,j,...=1,2,...n$ and $a,b,...=n+1,n+2,...,n+n,$ and $%
x^{i}$ and $y^{a}$ are respectively the conventional horizontal / holonomic
(h) and vertical / nonholonomic coordinates (v). For the local Euclidean
signature, we consider that all local basis vectors are real but, for the
pseudo--Euclidean signature $(-,+,+,+),$ we introduce $e_{j=1}=\mathit{i}%
\partial /\partial x^{1},$ where $i$ is the imaginary unity, $\mathit{i}%
^{2}=-1,$ and the local coordinate basis vectors can be written in the form $%
e_{\alpha }=\partial /\partial u^{\alpha }=(\mathit{i}\partial /\partial
x^{1},\partial /\partial x^{2},...,\partial /\partial x^{n},\partial
/\partial y^{a}).$\footnote{%
for simplicity, we shall omit to write in explicit form the imaginary unity
considering that we can always distinguish the pseudo--Euclidean signature
by a corresponding metric form or a local system of coordinates with a
coordinate proportional to the imaginary unit} The Einstein's rule on
summing up/low indices will be applied unless indicated otherwise.

Any metric on $V^{2n}$ can be written as%
\begin{equation}
\mathbf{g}=g_{ij}(x,y)\ e^{i}\otimes e^{j}+h_{ab}(x,y)\ e^{a}\otimes \ e^{b},
\label{m1}
\end{equation}%
where the dual vielbeins (tetrads, in four dimensions) $e^{a}=(e^{i},e^{a})$
are parametrized
 $e^{i}=e_{\ \underline{i}}^{i}(u)dx^{\underline{i}}$ and $e^{a}=e_{\
\underline{i}}^{a}(u)dx^{\underline{i}}+e_{\ \underline{a}}^{a}(u)dy^{%
\underline{a}}$,
for $e_{\underline{\alpha }}=\partial /\partial u^{\underline{\alpha }}=(e_{%
\underline{i}}=\partial /\partial x^{\underline{i}},e_{\underline{a}%
}=\partial /\partial y^{\underline{a}})$ and $e^{\underline{\beta }}=du^{%
\underline{\beta }}=(e^{\underline{j}}=dx^{\underline{j}},dy^{\underline{b}%
}) $ being, respectively, any fixed local coordinate base and dual base.

\begin{proposition}
\label{pr01}Any metric $\mathbf{g}$ (\ref{m1}) can be expressed in the form
\begin{equation}
\mathbf{\check{g}}=\check{g}_{i^{\prime }j^{\prime }}(x,y)\ \check{e}%
^{i^{\prime }}\otimes \check{e}^{j^{\prime }}+\check{h}_{a^{\prime
}b^{\prime }}(x,y)\ \mathbf{\check{e}}^{a^{\prime }}\otimes \ \mathbf{%
\check{e}}^{b^{\prime }},  \label{hvmetr1}
\end{equation}%
where $\check{e}^{i^{\prime }}=\delta _{\underline{i}}^{i^{\prime }}dx^{%
\underline{i}}$ and $\mathbf{\check{e}}^{a^{\prime }}=\delta _{\ \underline{a%
}}^{a^{\prime }}(u)dy^{\underline{a}}+\check{N}_{\ \underline{i}}^{a^{\prime
}}(u)dx^{\underline{i}}$ for
\begin{eqnarray}
\check{h}_{a^{\prime }b^{\prime }}(u) &=&\frac{1}{2}\frac{\partial ^{2}%
\mathcal{L}(x^{i^{\prime }},y^{c^{\prime }})}{\partial y^{a^{\prime
}}\partial y^{b^{\prime }}},  \label{elf} \\
\check{N}_{\ \underline{i}}^{a^{\prime }}(u) &=&\frac{\partial G^{a^{\prime
}}(x,y)}{\partial y^{n+\underline{j}}},  \label{ncel}
\end{eqnarray}%
where $\delta _{\underline{i}}^{i^{\prime }}$ is the Kronecker symbol, $%
\check{g}_{i^{\prime }j^{\prime }}=\check{h}_{n+i^{\prime }\ n+j^{\prime }}$
and $\check{h}^{ab}$ is the inverse of $\check{h}_{a^{\prime }b^{\prime }},$
for $\det |\check{h}_{a^{\prime }b^{\prime }}|\neq 0$ and
\begin{equation}
2G^{a^{\prime }}(x,y)=\frac{1}{2}\ \check{h}^{a^{\prime }\ n+i}\left( \frac{%
\partial ^{2}\mathcal{L}}{\partial y^{i}\partial x^{k}}y^{n+k}-\frac{%
\partial \mathcal{L}}{\partial x^{i}}\right) ,  \label{sprlf}
\end{equation}%
where $i,k=1,2,...n.$
\end{proposition}

By a straightforward computation, we can prove
\begin{lemma}
\label{lem01}Considering $\mathcal{L}$ from (\ref{elf}) and (\ref{sprlf}) to
be a regular Lagrangian, we have that the Euler--Lagrange equations
\begin{equation}
\frac{d}{d\tau }\left( \frac{\partial \mathcal{L}}{\partial y^{i}}\right) -%
\frac{\partial \mathcal{L}}{\partial x^{i}}=0,  \label{eleq}
\end{equation}%
where $y^{i}=y^{n+i}=\frac{dx^{i}}{d\tau }$, for $x^{i}(\tau )$ depending on
the parameter $\tau$. These equations are equivalent to
the ``nonlinear'' geodesic equations
\begin{equation}
\frac{d^{2}x^{i}}{d\tau ^{2}}+2G^{i}(x^{k},\frac{dx^{j}}{d\tau })=0
\label{ngeq}
\end{equation}%
defining the paths of a canonical semispray $S=y^{i}\frac{\partial }{%
\partial x^{i}}-2G^{a}(x,y)\frac{\partial }{\partial y^{a}},$ for $G^{a}$
given by equations (\ref{sprlf}).
\end{lemma}

The Lemma motivates
\begin{definition}
A (pseudo) Riemannian space with metric $\mathbf{g}$ (\ref{m1}) is modelled
by a mechanical system with regular effective Lagrangian $\mathcal{L}$ if
there is a nontrivial frame transform defined by any $e_{\ \underline{i}%
}^{i},e_{\ \underline{i}}^{a}$ and $e_{\ \underline{a}}^{a}$ when $\mathbf{g}%
=\mathbf{\check{g}}$\ (\ref{hvmetr1}).
\end{definition}

Inversely, we say that a regular mechanical model with Lagrangian $\mathcal{L%
}$ and Euler--Lagrange equations (\ref{eleq}) is geometrized in terms of a
(pseudo) Riemannian geometry with metric $\mathbf{g}$ (\ref{m1}) if $%
\mathcal{L}$ is a generating function for (\ref{elf}), (\ref{sprlf}) and (%
\ref{ncel}), when $\mathbf{g}=\mathbf{\check{g}}$ (\ref{hvmetr1}) and the
nonlinear geodesic equations (\ref{ngeq}) are equivalent to (\ref{eleq}).

Any equivalent modelling of regular mechanical systems as (pseudo)
Riemannian spaces introduces additional geometric structures on  $%
V^{2n}.$

\begin{definition}
\label{defnc}A nonlinear connection (N--connection) $\mathbf{N}$ on $V^{2n}$
is defined by a Whitney sum (nonholonomic distribution)
\begin{equation}
T(V^{2n})=h(V^{2n})\oplus v(V^{2n}),  \label{whitney}
\end{equation}%
splitting globally the tangent bundle $T(V^{2n})$ into respective h-- and
v--subspac\-es, $h(V^{2n})$ and $v(V^{2n}),$ given locally by a set of
coefficients $N_{i}^{a}(x,y)$ where
 $\mathbf{N=}N_{i}^{a}(x,y)dx^{i}\otimes \frac{\partial }{\partial y^{a}}$.
\end{definition}

We note that a subclass of linear connections is defined by $%
N_{i}^{a}=\Gamma _{b}^{a}(x)y^{b}.$

We can perform N--adapted geometric constructions by defining the
coefficients of geometric objects (and associated equations) with respect to
N--adapted noholonomic frames of type (\ref{2dder}) and (\ref{2ddif}). The
N--adapted tensors, vectors, forms, etc., are called respectively
distinguished tensors, etc., (in brief, d--tensors, d--vectors, d--forms,
etc.). For instance, a vector field $\mathbf{X}\in T\mathbf{V}^{2n} $ is
expressed as $\mathbf{X}=(hX,\ vX),$ or $\mathbf{X}=X^{\alpha }\mathbf{e}%
_{\alpha }=X^{i}\mathbf{e}_{i}+X^{a}e_{a},$ where $hX=X^{i}\mathbf{e}_{i}$
and $vX=X^{a}e_{a}$ state, respectively, the horizontal (h) and vertical (v)
components of the vector adapted to the N--connection structure.

\begin{proposition}
\label{pr02}Any effective regular Lagrangian $\mathcal{L}$, prescribed on $%
\mathbf{V}^{2n}$, defines a canonical N--connection structure $\mathbf{%
\check{N}=\{}\check{N}_{\ \underline{i}}^{a^{\prime }}(u)\}$ (\ref{ncel})
and preferred frame structures $\mathbf{\check{e}}_{\nu }=(\mathbf{\check{e}}%
_{i},e_{a^{\prime }})$ and $\mathbf{\check{e}}^{\mu }=(e^{i},\mathbf{%
\check{e}}^{a^{\prime }}).$
\end{proposition}

\begin{proof}
The proposition can be proved by straightforward computations. The
coefficients $\check{N}_{\ \underline{i}}^{a^{\prime }}$ satisfy the
conditions of Definition \ref{defnc}. We define $\mathbf{\check{e}}_{\nu }=(%
\mathbf{\check{e}}_{i},e_{a})$ and $\mathbf{\check{e}}^{\mu }=(e^{i},\mathbf{%
\check{e}}^{a})$ in explicit form by introducing $\check{N}_{\ \underline{i}%
}^{a^{\prime }},$ respectively, in formulas (\ref{2dder}) and (\ref{2ddif}).$%
\square $
\end{proof}

\vskip3pt Similar constructions can be defined for $\mathcal{L=F}^{2}(x,y),$
where an effective Finsler metric $\mathcal{F}$ is a differentiable function
of class $C^{\infty }$ in any point $(x,y)$ with $y\neq 0$ and is continuous
in any point $(x,0);$ $\mathcal{F}(x,y)>0$ if $y\neq 0;$ it satisfies the
homogeneity condition $\mathcal{F}(x,\beta y)=|\beta |\mathcal{F}(x,y)$ for
any nonzero $\beta \in \mathbb{R}$ and the Hessian (\ref{elf}) computed for $%
\mathcal{L=F}^{2}$ is positive definite. In this case, we can say that a
(pseudo) Riemannian space with metric $\mathbf{g}$ is modeled by an
effective Finsler geometry and, inversely, a Finsler geometry is modeled on
a (pseudo) Riemannian space.

\begin{definition}
A (pseudo) Riemannian manifold $\mathbf{V}^{2n}$ is nonholonomic
(N-\--an\-holonomic) if it is provided with a nonholonomic distribution on $%
TV^{2n}$ (N--connection structure $\mathbf{N}$).
\end{definition}

We formulate the first main result in this paper:

\begin{theorem}
\label{mth1}Any (pseudo) Riemannian space can be transformed into a
N--anho\-lo\-nomic manifold $\mathbf{V}^{2n}$ modeling an effective Lagrange
(or Finsler) geometry by prescribing a generating Lagrange (or Finsler)
function $\mathcal{L}(x,y)$ (or $\mathcal{F}(x,y)).$
\end{theorem}

\begin{proof}
Such a proof follows from Propositions \ref{pr01} and \ref{pr02}\ and Lemma %
\ref{lem01}. It should be noted that, by corresponding vielbein transforms $%
e_{\ \underline{i}}^{i},e_{\ \underline{i}}^{a}$ and $e_{\ \underline{a}%
}^{a},$ any metric $\mathbf{g}$ with coefficients defined with respect to an
arbitrary co--frame $\mathbf{e}^{\mu },$ see (\ref{m1}), can be transformed
into canonical Lagrange (Finsler) ones, $\mathbf{\check{g}}$ (\ref{hvmetr1}%
). The $\mathbf{\check{g}}$ coefficients are computed with respect to $%
\mathbf{\check{e}}^{\mu }=(e^{i},\mathbf{\check{e}}^{a}),$ with the
associated N--connection structure $\check{N}_{\ \underline{i}}^{a^{\prime
}},$ all defined by a prescribed $\mathcal{L}(x,y)$ (or $\mathcal{F}%
(x,y)).\square $
\end{proof}

\vskip3pt

Finally, it should be noted that considering an arbitrary effective
Lagrangian $\mathcal{L}(x,y)$ on a four dimensional (pseudo) Riemannian
spacetime and defining a corresponding $2+2$ decomposition, local Lorentz
invariance is not violated. We can work in any reference frame and
coordinates, but the constructions adapted to the canonical N--connection
structure and an analogous mechanical modeling are more convenient for
developing a formalism of deformation quantization of gravity following the
appropriate methods for Lagrange--Finsler and almost K\"{a}hler spaces.

\subsubsection{Almost K\"{a}hler Models for (Pseudo) Riemannian and Lagrange
Spaces}

The goal of this section is to prove that for any (pseudo) Riemannian metric
and $n+n$ splitting we can define canonical almost symplectic structures.
The analogous mechanical modeling developed in previous sections is
important from two points of view: Firstly, it provides both geometric and
physical interpretations for the class of nonholonomic transforms with $n+n$
splitting and adapting to the N--connection. Secondly, such canonical
constructions can be equivalently redefined as a class of almost K\"{a}hler
geometries with associated N--connection when certain symplectic forms and
linear connection structures are canonically induced by the metric $\mathbf{g%
}(x,y)$ and effective Lagrangian $\mathcal{L}(x,y)$ on $\mathbf{V}^{2n}.$

Let $\mathbf{\check{e}}_{\alpha ^{\prime }}=(\mathbf{\check{e}}%
_{i},e_{b^{\prime }})$ and $\mathbf{\check{e}}^{\alpha ^{\prime }}=(e^{i},\
\mathbf{\check{e}}^{b^{\prime }})$ be defined respectively by (\ref{2dder})
and (\ref{2ddif}) for the canonical N--connection $\mathbf{\check{N}}$
stated by a metric structure $\mathbf{g}=\mathbf{\check{g}}$ on $\mathbf{V}%
^{2n}.$ We introduce a linear operator $\mathbf{\check{J}}$ acting on
tangent vectors to $\mathbf{V}^{2n}$ following formulas $\mathbf{\check{J}}(%
\mathbf{\check{e}}_{i})=-e_{n+i}$\ and \ $\mathbf{\check{J}}(e_{n+i})=%
\mathbf{\check{e}}_{i},$ where the index $a^{\prime }$ runs values $n+i$ for
$i=1,2,...n$ and $\mathbf{\check{J}\circ \check{J}=-I}$ for $\mathbf{I}$
being the unity matrix. Equivalently, we introduce a tensor field on $%
\mathbf{V}^{2n},${\small
\begin{eqnarray*}
\mathbf{\check{J}} &=&\mathbf{\check{J}}_{\ \beta }^{\alpha }\ e_{\alpha
}\otimes e^{\beta }=\mathbf{\check{J}}_{\ \underline{\beta }}^{\underline{%
\alpha }}\ \frac{\partial }{\partial u^{\underline{\alpha }}}\otimes du^{%
\underline{\beta }} =\mathbf{\check{J}}_{\ \beta ^{\prime }}^{\alpha
^{\prime }}\ \mathbf{\check{e}}_{\alpha ^{\prime }}\otimes \mathbf{\check{e}}%
^{\beta ^{\prime }}=\mathbf{-}e_{n+i}\otimes e^{i}+\mathbf{\check{e}}%
_{i}\otimes \ \mathbf{\check{e}}^{n+i} \\
&=&-\frac{\partial }{\partial y^{i}}\otimes dx^{i}+\left( \frac{\partial }{%
\partial x^{i}}-\check{N}_{i}^{n+j}\frac{\partial }{\partial y^{j}}\right)
\otimes \left( dy^{i}+\check{N}_{k}^{n+i}dx^{k}\right) .
\end{eqnarray*}%
} It is clear that $\mathbf{\check{J}}$ defines globally an almost complex
structure on\ $\mathbf{V}^{2n}$ completely determined by a fixed $\mathcal{L}%
(x,y).$

\begin{definition}
The Nijenhuis tensor field for any almost complex structure $\mathbf{J}$
determined by a N--connection (equivalently, the curvature of
N--connecti\-on) is defined as
\begin{equation}
\ ^{\mathbf{J}}\mathbf{\Omega (X,Y)=-[X,Y]+[JX,JY]-J[JX,Y]-J[X,JY],}
\label{neijt}
\end{equation}%
for any d--vectors $\mathbf{X}$ and $\mathbf{Y.}$
\end{definition}

With respect to N--adapted bases the Neijenhuis tensor $\ ^{\mathbf{J}}%
\mathbf{\Omega =\{}\Omega _{ij}^{a}\mathbf{\}}$ has the coefficients
\begin{equation}
\Omega _{ij}^{a}=\frac{\partial N_{i}^{a}}{\partial x^{j}}-\frac{\partial
N_{j}^{a}}{\partial x^{i}}+N_{i}^{b}\frac{\partial N_{j}^{a}}{\partial y^{b}}%
-N_{j}^{b}\frac{\partial N_{i}^{a}}{\partial y^{b}}.  \label{nccurv}
\end{equation}%
A N--anholonomic manifold $\mathbf{V}^{2n}$ is integrable if $\Omega
_{ij}^{a}=0.$ We get a complex structure if and only if both the h-- and
v--distributions are integrable, i.e., if and only if $\Omega _{ij}^{a}=0$
and $\frac{\partial N_{j}^{a}}{\partial y^{i}}-\frac{\partial N_{i}^{a}}{%
\partial y^{j}}=0.$

\begin{definition}
An almost symplectic structure on a manifold $V^{n+m},$ \newline
$\dim V^{n+m}=n+m,$ is defined by a nondegenerate 2--form
$\theta =\frac{1}{2}\theta _{\alpha \beta }(u)e^{\alpha }\wedge e^{\beta }.$
\end{definition}

We have

\begin{proposition}
For any $\theta $ on $V^{n+m},$ there is a unique N--connection $\mathbf{N}%
=\{N_{i}^{a}\}$ defined as a splitting $TV^{n+m}=hV^{n+m}\oplus vV^{n+m},$
where indices $i,j,..=1,2,...n$ and $a,b,...=n+1,n+1,...n+m$. The function $%
\theta $ satisfies the following conditions:%
\begin{equation}
\theta =(h\mathbf{X},v\mathbf{Y})=0\mbox{ and }\theta =h\theta +v\theta ,
\label{aux02}
\end{equation}%
for any $\mathbf{X}=h\mathbf{X}+v\mathbf{X,}$ $\mathbf{Y}=h\mathbf{Y}+v%
\mathbf{Y}$ and $h\theta (\mathbf{X,Y})\doteqdot \theta (h\mathbf{X,}h%
\mathbf{Y}),$\newline
$v\theta (\mathbf{X,Y})\doteqdot \theta (v\mathbf{X,}v\mathbf{Y}).$ Here the
symbol ''$\doteqdot $'' means ''by definition''.
\end{proposition}

\begin{proof}
For $\mathbf{X=e}_{\alpha }=(\mathbf{e}_{i},e_{a})$ and $\mathbf{Y=e}%
_{\beta}=(\mathbf{e}_{l},e_{b}),$ where $\mathbf{e}_{\alpha }$ is a
N--adapted basis of dimension $n+m,$ we write the first equation in (\ref%
{aux02}) as $\theta =\theta (\mathbf{e}_{i},e_{a})=\theta (\frac{\partial }{%
\partial x^{i}},\frac{\partial }{\partial y^{a}})-N_{i}^{b}\theta (\frac{%
\partial }{\partial y^{b}},\frac{\partial }{\partial y^{a}})=0.$ We can find
a unique solution form and define $N_{i}^{b}$ if $rank|\theta (\frac{%
\partial }{\partial y^{b}},\frac{\partial }{\partial y^{a}})|=m.$ Denoting
locally
\begin{equation}
\theta =\frac{1}{2}\theta _{ij}(u)e^{i}\wedge e^{j}+\frac{1}{2}\theta
_{ab}(u)\mathbf{e}^{a}\wedge \mathbf{e}^{b},  \label{aux03}
\end{equation}%
where the first term is for $h\theta $ and the second term is $v\theta ,$ we
get the second formula in (\ref{aux02}). We may consider the particular case
in which $n=m.\square $
\end{proof}

\begin{definition}
An almost Hermitian model of a (pseudo) Riemannian spa\-ce $\mathbf{V}^{2n}$
equipped with an N--connection structure $\mathbf{N}$ is defined by a triple
$\mathbf{H}^{2n}=(\mathbf{V}^{2n},\theta ,\mathbf{J}),$ where $\mathbf{%
\theta (X,Y)}\doteqdot \mathbf{g}\left( \mathbf{JX,Y}\right).$
\end{definition}

In addition, we have

\begin{definition}
A space $\mathbf{H}^{2n}$ is almost K\"{a}hler, denoted $\mathbf{K}^{2n},$
if and only if $d\mathbf{\theta }=0.$
\end{definition}

If a (pseudo) Riemannian space is modeled by a Lagrange--Finsler geometry,
the second main result of this paper follows

\begin{theorem}
Having chosen a generating function $\mathcal{L}(x,y)$ (or $\mathcal{F}%
(x,y)) $ on a (pseudo) Riemannian manifold $V^{n+n},$ we can model this
space as an almost K\"{a}hler geometry, i.e. $\mathbf{\check{H}}^{2n}=%
\mathbf{\check{K}}^{2n}.$
\end{theorem}

\begin{proof}
For $\mathbf{g}=\mathbf{\check{g}}$ (\ref{hvmetr1}) and structures $\mathbf{%
\check{N}}$ and $\mathbf{\check{J}}$ canonically defined by $\mathcal{L},$
we define $\mathbf{\check{\theta}(X,Y)}\doteqdot \mathbf{\check{J}}\left(
\mathbf{\check{F}X,Y}\right) $ for any d--vectors $\mathbf{X}$ and $\mathbf{%
Y.}$ In local N--adapted form form, we have
\begin{eqnarray}
\mathbf{\check{\theta}} &=&\frac{1}{2}\check{\theta}_{\alpha \beta
}(u)e^{\alpha }\wedge e^{\beta }=\frac{1}{2}\check{\theta}_{\underline{%
\alpha }\underline{\beta }}(u)du^{\underline{\alpha }}\wedge du^{\underline{%
\beta }}  \label{asymstr} \\
&=&\check{g}_{ij}(x,y)\check{e}^{n+i}\wedge dx^{j}=\check{g}%
_{ij}(x,y)(dy^{n+i}+\check{N}_{k}^{n+i}dx^{k})\wedge dx^{j}.  \notag
\end{eqnarray}%
Let us consider the form $\check{\omega}=\frac{1}{2}\frac{\partial \mathcal{L%
}}{\partial y^{n+i}}dx^{i}.$ A straightforward computation, using
Proposition \ref{pr02} and N--connection $\mathbf{\check{N}}$ (\ref{ncel}),
shows that $\mathbf{\check{\theta}}=d\check{\omega},$ which means that $d%
\mathbf{\check{\theta}}=dd\check{\omega}=0$ and that the canonical effective
Lagrange structures $\mathbf{g}=\mathbf{\check{g},\check{N}}$ and $\mathbf{%
\check{J}}$ induce an almost K\"{a}hler geometry. Instead of "Lagrangian
mechanics variables" we can introduce another type redefining $\mathbf{%
\check{\theta}}$ with respect to an arbitrary co--frame basis using
vielbeins $\mathbf{e}_{\ \underline{\alpha }}^{\alpha } $ and their duals $%
\mathbf{e}_{\alpha \ }^{\ \underline{\alpha }},$ defined by $e_{\ \underline{%
i}}^{i},e_{\ \underline{i}}^{a}$ and $e_{\ \underline{a}}^{a}$. So, we can
compute $\check{\theta}_{\alpha \beta }=\mathbf{e}_{\alpha \ }^{\ \underline{%
\alpha }}\mathbf{e}_{\beta \ }^{\ \underline{\beta }}\check{\theta}_{%
\underline{\alpha }\underline{\beta }}$ and express the 2--form (\ref%
{asymstr})as $\check{\theta}=\frac{1}{2}\check{\theta}_{ij}(u)e^{i}\wedge
e^{j}+\frac{1}{2}\check{\theta}_{ab}(u)\mathbf{\check{e}}^{a}\wedge \mathbf{%
\check{e}}^{b},$ see (\ref{aux03}). The coefficients $\check{\theta}_{ab}=%
\check{\theta}_{n+i\ n+j}$ above are equal, respectively, to the
coefficients $\check{\theta}_{ij}$ and the dual N--adapted basis $\mathbf{%
\check{e}}^{\alpha }=(e^{i},\mathbf{\check{e}}^{a})$ is elongated by $%
\check{N}_{j}^{a}$ (\ref{ncel}). It should be noted that for a general
2--form $\theta $ directly constructed from a metric $\mathbf{g}$ and almost
complex $\mathbf{J}$ structures on $V^{2n}$, we have that $d\theta \neq 0.$
For a $n+n$ splitting induced by an effective Lagrange (Finsler) generating
function, we have $d\mathbf{\check{\theta}}=0$ which results in a canonical
almost K\"{a}hler model completely defined by $\mathbf{g}=\mathbf{\check{g}}$
and chosen $\mathcal{L}(x,y)$ (or $\mathcal{F}(x,y)).$ $\square $
\end{proof}

\paragraph{N--adapted symplectic connections:}

In our approach, we work with nonholonomic (pseudo) Riemannian manifolds $%
\mathbf{V}^{2n}$ enabled with an effective N--connection and almost
symplectic structures defined canonically by the metric structure $\mathbf{g}%
=\mathbf{\check{g}}$ and a fixed $\mathcal{L}(x,y).$ In this section, we
analyze the class of linear connections that can be adapted to the
N--connection and/or symplectic structure and defined canonically if a
corresponding nonholonomic distribution is induced completely by $\mathcal{L}%
,$ or $\mathcal{F}.$

From the class of arbitrary affine connections on $\mathbf{V}^{2n},$ one
prefers to work with N--adapted linear connections, called distinguished
connections (d--connections).

\begin{definition}
A linear connection on $\mathbf{V}^{2n}$ is a d--connection
\begin{equation*}
\mathbf{D}=(hD;vD)=\{\mathbf{\Gamma }_{\beta \gamma }^{\alpha
}=(L_{jk}^{i},\ ^{v}L_{bk}^{a};C_{jc}^{i},\ ^{v}C_{bc}^{a})\},
\end{equation*}%
with local coefficients computed with respect to N--adapted (\ref{2dder})
and (\ref{2ddif}), which preserves the distribution (\ref{whitney}) under
parallel transports.
\end{definition}

For a d--connection $\mathbf{D,}$ we can define respectively the torsion and
curvature tensors,
\begin{eqnarray}
\mathbf{T}(\mathbf{X},\mathbf{Y}) &\doteqdot &\mathbf{D}_{\mathbf{X}}\mathbf{%
Y}-\mathbf{D}_{\mathbf{Y}}\mathbf{X}-[\mathbf{X},\mathbf{Y}],  \label{ators}
\\
\mathbf{R}(\mathbf{X},\mathbf{Y})\mathbf{Z} &\doteqdot &\mathbf{D}_{\mathbf{X%
}}\mathbf{D}_{\mathbf{Y}}\mathbf{Z}-\mathbf{D}_{\mathbf{Y}}\mathbf{D}_{%
\mathbf{X}}\mathbf{Z}-\mathbf{D}_{[\mathbf{X},\mathbf{Y}]}\mathbf{Z},
\label{acurv}
\end{eqnarray}%
where $[\mathbf{X},\mathbf{Y}]\doteqdot \mathbf{XY}-\mathbf{YX,}$ for any
d--vectors $\mathbf{X} $ and $\mathbf{Y}.$ The coefficients $\mathbf{T}=\{%
\mathbf{T}_{\ \beta \gamma }^{\alpha }\}$ and $\mathbf{R}=\{\mathbf{R}_{\
\beta \gamma \tau }^{\alpha }\}$ can be written in terms of $\mathbf{e}%
_{\alpha }$ and $\mathbf{e}^{\beta }$ by introducing $\mathbf{X}\rightarrow
\mathbf{e}_{\alpha },\mathbf{Y}\rightarrow \mathbf{e}_{\beta }\mathbf{,Z}%
\rightarrow \mathbf{e}_{\gamma }$ in (\ref{ators}) and (\ref{acurv}).

\begin{definition}
A d--connection $\mathbf{D}$\ is metric compatible with a d--metric $\mathbf{%
g}$ if $\mathbf{D}_{\mathbf{X}}\mathbf{g}=0$ for any d--vector field $%
\mathbf{X.}$
\end{definition}

If an almost symplectic structure is defined on a N--anholonomic manifold,
one considers:

\begin{definition}
\label{defasstr}An almost symplectic d--connection $\ _{\theta }\mathbf{D}$
on $\mathbf{V}^{2n},$ or (equivalently) a d--connection compatible with an
almost symplectic structure $\theta ,$ is defined such that $\ _{\theta }%
\mathbf{D}$ is N--adapted, i.e., it is a d--connection, and $\ _{\theta }%
\mathbf{D}_{\mathbf{X}}\theta =0,$ for any d--vector $\mathbf{X.}$
\end{definition}

We can always fix a d--connection $\ _{\circ }\mathbf{D}$ on $\mathbf{V}%
^{2n} $ and then construct an almost symplectic $\ _{\theta }\mathbf{D.}$

\begin{example}
Let us represent $\theta $ in N--adapted form (\ref{aux03}). Having chosen a
\begin{eqnarray*}
\ _{\circ }\mathbf{D} &=&\left\{ h\ _{\circ }D=(\ _{\circ }D_{k},\ \ _{\circ
}^{v}D_{k});v\ _{\circ }D=(\ _{\circ }D_{c},\ \ _{\circ }^{v}D_{c})\right\}
\\
&=&\{\ _{\circ }\mathbf{\Gamma }_{\beta \gamma }^{\alpha }=(\ _{\circ
}L_{jk}^{i},\ _{\circ }^{v}L_{bk}^{a};\ _{\circ }C_{jc}^{i},\ _{\circ
}^{v}C_{bc}^{a})\},
\end{eqnarray*}%
we can verify that
\begin{eqnarray*}
\ _{\theta }\mathbf{D} &=&\left\{ h\ _{\theta }D=(\ _{\theta }D_{k},\ \
_{\theta }^{v}D_{k});v\ _{\theta }D=(\ _{\theta }D_{c},\ \ _{\theta
}^{v}D_{c})\right\} \\
&=&\{\ _{\theta }\mathbf{\Gamma }_{\beta \gamma }^{\alpha }=(\ _{\theta
}L_{jk}^{i},\ _{\theta }^{v}L_{bk}^{a};\ _{\theta }C_{jc}^{i},\ _{\theta
}^{v}C_{bc}^{a})\},
\end{eqnarray*}%
with
\begin{eqnarray}
\ _{\theta }L_{jk}^{i} &=&\ _{\circ }L_{jk}^{i}+\frac{1}{2}\theta ^{ih}\
_{\circ }D_{k}\theta _{jh},\ \ _{\theta }^{v}L_{bk}^{a}=\ _{\circ
}^{v}L_{bk}^{a}+\frac{1}{2}\theta ^{ae}\ _{\circ }^{v}D_{k}\theta _{eb},
\label{2csdc} \\
\ _{\theta }C_{jc}^{i} &=&\ _{\theta }C_{jc}^{i}+\frac{1}{2}\theta ^{ih}\
_{\circ }D_{c}\theta _{jh},\ \ _{\theta }^{v}C_{bc}^{a}=\ _{\circ
}^{v}C_{bc}^{a}+\frac{1}{2}\theta ^{ae}\ _{\circ }^{v}D_{c}\theta _{eb},
\notag
\end{eqnarray}%
satisfies the conditions $\ _{\theta }D_{k}\theta _{jh}=0,\ \ _{\theta
}^{v}D_{k}\theta _{eb}=0,\ _{\theta }D_{c}\theta _{jh}=0,\ _{\theta
}^{v}D_{c}\theta _{eb}=0,$ which is equivalent to $\ _{\theta }\mathbf{D}_{%
\mathbf{X}}\theta =0$ from Definition \ref{defasstr}.
\end{example}

Let us introduce the operators
\begin{equation}
\Theta _{jk}^{hi}=\frac{1}{2}(\delta _{j}^{h}\delta _{k}^{i}-\theta
_{jk}\theta ^{ih})\mbox{ and }\Theta _{cd}^{ab}=\frac{1}{2}(\delta
_{c}^{a}\delta _{d}^{b}-\theta _{cd}\theta ^{ab}),  \label{thop}
\end{equation}%
with the coefficients computed with respect to N--adapted bases. By
straightforward computations, one proves the following theorem.

\begin{theorem}
The set of d--connections \newline
$\ _{s}\mathbf{\Gamma }_{\beta \gamma }^{\alpha }=(\ _{s}L_{jk}^{i},\
_{s}^{v}L_{bk}^{a};\ _{s}C_{jc}^{i},\ _{s}^{v}C_{bc}^{a})$ which are
compatible with an almost symplectic structure $\theta $ (\ref{aux03}), are
parametrized by
\begin{eqnarray}
\ _{s}L_{jk}^{i} &=&\ _{\theta }L_{jk}^{i}+\Theta _{jl}^{hi}\ Y_{hk}^{l},\
_{s}^{v}L_{bk}^{a}=\ _{\theta }^{v}L_{bk}^{a}+\Theta _{bd}^{ca}\ Y_{ck}^{d},
\label{fsdc} \\
\ _{s}C_{jc}^{i} &=&\ _{\theta }C_{jc}^{i}+\Theta _{jl}^{hi}\ Y_{hc}^{l},\
_{s}^{v}C_{bc}^{a}=\ _{\theta }^{v}C_{bc}^{a}+\Theta _{bd}^{ea}\ Y_{ec}^{d},
\notag
\end{eqnarray}%
where $\ _{\theta }\mathbf{\Gamma }_{\beta \gamma }^{\alpha }=(\ _{\theta
}L_{jk}^{i},\ _{\theta }^{v}L_{bk}^{a};\ _{\theta }C_{jc}^{i},\ _{\theta
}^{v}C_{bc}^{a})$ is given by (\ref{2csdc}), the $\Theta $--operators are
those from (\ref{thop}) and $\mathbf{Y}_{\beta \gamma }^{\alpha }=\left(
Y_{jk}^{i},Y_{bk}^{a},Y_{jc}^{i},Y_{bc}^{a}\right) $ are arbitrary d--tensor
fields.
\end{theorem}

From the set of metric and/or almost symplectic compatible d--connecti\-ons
on a (pseudo) Riemannian manifold $V^{2n},$ we can select those which are
completely defined by $\mathbf{g}$ and a prescribed effective Lagrange
structure $\mathcal{L}(x,y):$

\begin{theorem}
There is a unique normal d--connection
\begin{eqnarray}
\ \widehat{\mathbf{D}} &=&\left\{ h\widehat{D}=(\widehat{D}_{k},^{v}\widehat{%
D}_{k}=\widehat{D}_{k});v\widehat{D}=(\widehat{D}_{c},\ ^{v}\widehat{D}_{c}=%
\widehat{D}_{c})\right\}  \label{ndc} \\
&=&\{\widehat{\mathbf{\Gamma }}_{\beta \gamma }^{\alpha }=(\widehat{L}%
_{jk}^{i},\ ^{v}\widehat{L}_{n+j\ n+k}^{n+i}=\widehat{L}_{jk}^{i};\ \widehat{%
C}_{jc}^{i}=\ ^{v}\widehat{C}_{n+j\ c}^{n+i},\ ^{v}\widehat{C}_{bc}^{a}=%
\widehat{C}_{bc}^{a})\},  \notag
\end{eqnarray}%
which is metric compatible, $\widehat{D}_{k}\check{g}_{ij}=0$ and $\widehat{D%
}_{c}\check{g}_{ij}=0,$ and completely defined by $\mathbf{g}=\mathbf{\check{%
g}}$ and a fixed $\mathcal{L}(x,y).$
\end{theorem}

\begin{proof}
First, we note that if a normal d--connection exists, it is completely
defined by couples of h-- and v--components $\ \widehat{\mathbf{D}}_{\alpha
}=(\widehat{D}_{k},\widehat{D}_{c}),$ i.e. $\widehat{\mathbf{\Gamma }}%
_{\beta \gamma }^{\alpha }=(\widehat{L}_{jk}^{i},\ ^{v}\widehat{C}%
_{bc}^{a}). $ Choosing
\begin{equation}
\widehat{L}_{jk}^{i}=\frac{1}{2}\check{g}^{ih}\left( \mathbf{\check{e}}_{k}%
\check{g}_{jh}+\mathbf{\check{e}}_{j}\check{g}_{hk}-\mathbf{\check{e}}_{h}%
\check{g}_{jk}\right) ,\widehat{C}_{jk}^{i}=\frac{1}{2}\check{g}^{ih}\left(
\frac{\partial \check{g}_{jh}}{\partial y^{k}}+\frac{\partial \check{g}_{hk}%
}{\partial y^{j}}-\frac{\partial \check{g}_{jk}}{\partial y^{h}}\right) ,
\label{cdcc}
\end{equation}%
where $\mathbf{\check{e}}_{k}=\partial /\partial x^{k}+\check{N}%
_{k}^{a}\partial /\partial y^{a},$ $\check{N}_{k}^{a}$ and $\check{g}_{jk}=%
\check{h}_{n+i\ n+j}$ are defined by canonical values (\ref{elf}) and (\ref%
{ncel}) induced by a regular $\mathcal{L}(x,y),$ we can prove that this
d--connection is unique and satisfies the conditions of the theorem. \ Using
vielbeins $\mathbf{e}_{\ \underline{\alpha }}^{\alpha }$ and their duals $%
\mathbf{e}_{\alpha \ }^{\ \underline{\alpha }},$ defined by $e_{\ \underline{%
i}}^{i},e_{\ \underline{i}}^{a}$ and $e_{\ \underline{a}}^{a}$, we can
compute the coefficients of $\widehat{\mathbf{\Gamma }}_{\beta \gamma
}^{\alpha }$ (\ref{ndc}) with respect to arbitrary frame basis $e_{\alpha }$
and co--basis $e^{\alpha }$ on $V^{n+m}.\square $
\end{proof}

\vskip5pt Introducing the normal d--connection 1--form $\widehat{\mathbf{%
\Gamma }}_{j}^{i}=\widehat{L}_{jk}^{i}e^{k}+\widehat{C}_{jk}^{i}\mathbf{%
\check{e}}^{k},$ for $e^{k}=dx^{k}$ and $\mathbf{\check{e}}^{k}=dy^{k}+%
\check{N}_{i}^{k}dx^{k},$ we can prove that the Cartan structure equations
are satisfied,%
\begin{equation}
de^{k}-e^{j}\wedge \widehat{\mathbf{\Gamma }}_{j}^{k}=-\widehat{\mathcal{T}}%
^{i},\ d\mathbf{\check{e}}^{k}-\mathbf{\check{e}}^{j}\wedge \widehat{\mathbf{%
\Gamma }}_{j}^{k}=-\ ^{v}\widehat{\mathcal{T}}^{i},  \label{cart1}
\end{equation}%
and
\begin{equation}
d\widehat{\mathbf{\Gamma }}_{j}^{i}-\widehat{\mathbf{\Gamma }}_{j}^{h}\wedge
\widehat{\mathbf{\Gamma }}_{h}^{i}=-\widehat{\mathcal{R}}_{\ j}^{i}.
\label{cart2}
\end{equation}

The h-- and v--components of the torsion 2--form $\widehat{\mathcal{T}}%
^{\alpha }=\left( \widehat{\mathcal{T}}^{i},\ ^{v}\widehat{\mathcal{T}}%
^{i}\right) =\widehat{\mathbf{T}}_{\ \tau \beta}^{\alpha }\ \mathbf{\check{e}%
}^{\tau }\wedge \mathbf{\check{e}}^{\beta }$ and from (\ref{cart1}) the
components are computed
\begin{equation}
\widehat{\mathcal{T}}^{i}=\widehat{C}_{jk}^{i}e^{j}\wedge \mathbf{\check{e}}%
^{k},\ ^{v}\widehat{\mathcal{T}}^{i}=\frac{1}{2}\check{\Omega}%
_{kj}^{i}e^{k}\wedge e^{j}+(\frac{\partial \check{N}_{k}^{i}}{\partial y^{j}}%
-\widehat{L}_{\ kj}^{i})e^{k}\wedge \mathbf{\check{e}}^{j},  \label{tform}
\end{equation}%
where $\check{\Omega}_{kj}^{i}$ are coefficients of the curvature of the
canonical N--connection $\check{N}_{k}^{i}$ defined by formulas similar to (%
\ref{nccurv}). Such formulas also follow from (\ref{ators}) redefined for $%
\widehat{\mathbf{D}}_{\alpha }$ and $\mathbf{\check{e}}_{\alpha },$ when the
torsion $\widehat{\mathbf{T}}_{\beta \gamma }^{\alpha }$ is parametrized as
\begin{equation}
\widehat{T}_{jk}^{i}=0,\widehat{T}_{jc}^{i}=\widehat{C}_{\ jc}^{i},\widehat{T%
}_{ij}^{a}=\check{\Omega}_{ij}^{a},\widehat{T}_{ib}^{a}=e_{b}\check{N}%
_{i}^{a}-\widehat{L}_{\ bi}^{a},\widehat{T}_{bc}^{a}=0.  \label{cdtors}
\end{equation}%
It should be noted that $\widehat{\mathbf{T}}$ vanishes on h- and
v--subspaces, i.e. $\widehat{T}_{jk}^{i}=0$ and $\widehat{T}_{bc}^{a}=0,$
but certain nontrivial h--v--components induced by the nonholonomic
structure are defined canonically by $\mathbf{g}=\mathbf{\check{g}}$ and $%
\mathcal{L}.$

We can also compute the curvature 2--form from (\ref{cart2}),%
\begin{equation}
\widehat{\mathcal{R}}_{\ \gamma }^{\tau }=\widehat{\mathbf{R}}_{\ \gamma
\alpha \beta }^{\tau }\ \mathbf{\check{e}}^{\alpha }\wedge \ \mathbf{%
\check{e}}^{\beta }=\frac{1}{2}\widehat{R}_{\ jkh}^{i}e^{k}\wedge e^{h}+%
\widehat{P}_{\ jka}^{i}e^{k}\wedge \mathbf{\check{e}}^{a}+\frac{1}{2}\
\widehat{S}_{\ jcd}^{i}\mathbf{\check{e}}^{c}\wedge \mathbf{\check{e}}^{d},
\label{cform}
\end{equation}%
where the nontrivial N--adapted coefficients of curvature $\widehat{\mathbf{R%
}}_{\ \beta \gamma \tau }^{\alpha }$ of $\widehat{\mathbf{D}}$ are (such
formulas can be proven also from (\ref{acurv}) written for $\widehat{\mathbf{%
D}}_{\alpha }$ and $\mathbf{\check{e}}_{\alpha })$
\begin{eqnarray}
\widehat{R}_{\ hjk}^{i} &=&\mathbf{\check{e}}_{k}\widehat{L}_{\ hj}^{i}-%
\mathbf{\check{e}}_{j}\widehat{L}_{\ hk}^{i}+\widehat{L}_{\ hj}^{m}\widehat{L%
}_{\ mk}^{i}-\widehat{L}_{\ hk}^{m}\widehat{L}_{\ mj}^{i}-\widehat{C}_{\
ha}^{i}\check{\Omega}_{\ kj}^{a},  \label{cdcurv} \\
\widehat{P}_{\ jka}^{i} &=&e_{a}\widehat{L}_{\ jk}^{i}-\widehat{\mathbf{D}}%
_{k}\widehat{C}_{\ ja}^{i},\ \widehat{S}_{\ bcd}^{a}=e_{d}\widehat{C}_{\
bc}^{a}-e_{c}\widehat{C}_{\ bd}^{a}+\widehat{C}_{\ bc}^{e}\widehat{C}_{\
ed}^{a}-\widehat{C}_{\ bd}^{e}\widehat{C}_{\ ec}^{a}.  \notag
\end{eqnarray}%
If instead of an effective Lagrange function, one considers a Finsler
generating fundamental function $\mathcal{F}^{2},$ similar formulas for the
torsion and curvature of the normal d--connection can also be found.

There is another very important property of the normal d--connection:

\begin{theorem}
The normal d--connection $\widehat{\mathbf{D}}$ defines a unique almost
symplectic d--connection, $\widehat{\mathbf{D}}\equiv \ _{\theta }\widehat{%
\mathbf{D}},$ see Definition \ref{defasstr}, which is N--adapted, i.e. it
preserves under parallelism the splitting (\ref{whitney}), $_{\theta }%
\widehat{\mathbf{D}}_{\mathbf{X}}\check{\theta}\mathbf{=}0$ and $\widehat{T}%
_{jk}^{i}=\widehat{T}_{bc}^{a}=0,$ i.e. the torsion is of type (\ref{cdtors}%
).
\end{theorem}

\begin{proof}
Applying the conditions of the theorem to the coefficients (\ref{cdcc}), the
proof follows in a straightforward manner. $\square $
\end{proof}

\vskip3pt

In this section, we proved that a N--adapted and almost symplectic $\widehat{%
\mathbf{\Gamma }}_{\beta \gamma }^{\alpha }$ can be uniquely defined by a
(pseudo) Riemannian metric $\mathbf{g}$ if we prescribe an effective
Lagrange, or Finsler, function $\mathcal{L},$ or $\mathcal{F}$ on $V^{2n}.$
This allows us to construct an analogous Lagrange model for gravity and, at
the next step, to transform it equivalently in an almost K\"{a}hler
structure adapted to a corresponding $n+n$ spacetime splitting. For the
Einstein metrics, we get a canonical $2+2$ decomposition for which we can
apply the Fedosov's quantization if the geometric objects and operators are
adapted to the associated N--connection.

\begin{definition}
A (pseudo) Riemannian space is described in Lagrange--Finsler variables if
its vielbein, metric and linear connection structures are equivalently
transformed into corresponding canonical N---connection,
La\-gran\-ge--Finsler metric and normal / almost symplectic d--connection
structures.
\end{definition}

It should be noted that former approaches to the canonical and quantum loop
quantization of gravity were elaborated for $3+1$ fibrations and
corresponding ADM and Ashtekar variables with further modifications. On the
other hand, in order to elaborate certain approaches to deformation
quantization of gravity, it is crucial to work with nonholonomic $2+2$
structures, which is more convenient for certain Lagrange geometrized
constructions and their almost symplectic variants. For other models, the $%
3+1$ splitting preserves a number of similarities to Hamilton mechanics. In
our approach, the spacetime decompositions are defined by corresponding
N--connection structures, which can be induced canonically by effective
Lagrange, or Finsler, generating functions. One works both with N--adapted
metric coefficients and nonholonomic frame coefficients, the last ones being
defined by generic off--diagonal metric coefficients and related
N--connection coefficients. In the models related to $3+1$ fibrations, one
works with shift functions and frame variables which contain all dynamical
information, instead of metrics.

We also discuss here the similarities and differences of preferred classes
of linear connections used for $3+1$ and $2+2$ structures. In the first
case, the Ashtekar variables (and further modifications) were proved to
simplify the constraint structure of a gauge like theory to which the
Einstein theory was transformed in order to develop a background independent
quantization of gravity. In the second case, the analogs of Ashtekar
variables are generated by a canonical Lagrange--Finsler type metric and/or
corresponding almost symplectic structure, both adapted to the N--connection
structure. It is also involved the normal d--connection which is compatible
with the almost symplectic structure and completely defined by the metric
structure, alternatively to the Levi--Civita connection (the last one is not
adapted to the N--connection and induced almost symplectic structure). In
fact, all constructions for the normal d--connection can be redefined in an
equivalent form to the Levi--Civita connection, or in Ashtekar variables,
but in such cases the canonical $2+2$ splitting and almost K\"{a}hler
structure are mixed by general frame and linear connection deformations.

\subsubsection{Distinguished Fedosov's Operators}

The Fedosov's approach to deformation quantization will be extended for
(pseudo) Riemannian manifolds $V^{2n}$ endowed with an effective Lagrange
function $\mathcal{L}.$ The constructions elaborated by A. Karabegov and M.
Schlichenmeier will be adapted to the canonical N--connection structure by
considering decompositions with respect to $\mathbf{\breve{e}}_{\nu }=(%
\mathbf{\breve{e}}_{i},e_{a^{\prime }})$ and $\mathbf{\breve{e}}^{\mu
}=(e^{i},\mathbf{\breve{e}}^{a^{\prime }})$ defined by a metric $\mathbf{g}$
(\ref{m1}). For simplicity, we shall work only with the normal/ almost
symplectic d--connection, $\widehat{\mathbf{D}}\equiv \ _{\theta }\widehat{%
\mathbf{D}}$ (\ref{ndc}), see Definition \ref{defasstr}, but it should be
emphasized here that we can use any d--connection from the family (\ref{fsdc}%
) and develop a corresponding deformation quantization. In this work, the
formulas are redefined on nonholonomic (pseudo) Riemannian manifolds
modeling effective regular mechanical systems and corresponding almost K\"{a}%
hler structures.

We introduce the tensor $\ \mathbf{\check{\Lambda}}^{\alpha \beta }\doteqdot
\check{\theta}^{\alpha \beta }-i\ \mathbf{\check{g}}^{\alpha \beta },$ where
$\check{\theta}^{\alpha \beta }$ is the form (\ref{asymstr}) with ''up''
indices and $\ \mathbf{\check{g}}^{\alpha \beta }$ is the inverse to $%
\mathbf{\check{g}}_{\alpha \beta }$ stated by coefficients of (\ref{hvmetr1}%
). The local coordinates on $\mathbf{V}^{2n}$ are parametrized as $%
u=\{u^{\alpha }\}$ and the local coordinates on $T_{u}\mathbf{V}^{2n}$ are
labeled $(u,z)=(u^{\alpha },z^{\beta }),$ where $z^{\beta }$ are fiber
coordinates.

The formalism of deformation quantization can be developed by using $%
C^{\infty }(V)[[v]]$, the space of formal series of variable $v$ with
coefficients from $C^{\infty }(V)$ on a Poisson manifold $(V,\{\cdot ,\cdot
\})$ (in this work, we deal with an almost Poisson structure defined by the
canonical almost symplectic structure). One defines an associative algebra
structure on $C^{\infty }(V)[[v]]$ with a $v$--linear and $v$--adically
continuous star product
\begin{equation}
\ ^{1}f\ast \ ^{2}f=\sum\limits_{r=0}^{\infty }\ _{r}C(\ ^{1}f,\ ^{2}f)\
v^{r},  \label{starp}
\end{equation}%
where $\ _{r}C,r\geq 0,$ are bilinear operators on $C^{\infty }(V)$ with $\
_{0}C(\ ^{1}f,\ ^{2}f)=\ ^{1}f\ ^{2}f$ and $\ _{1}C(\ ^{1}f,\ ^{2}f)-\
_{1}C(\ ^{2}f,\ ^{1}f)=i\{\ ^{1}f,\ ^{2}f\};$\ $i$ being the complex unity.
Constructions of type (\ref{starp}) are used for stating a formal Wick
product
\begin{equation}
a\circ b\ (z)\doteqdot \exp \left( i\frac{v}{2}\ \mathbf{\check{\Lambda}}%
^{\alpha \beta }\frac{\partial ^{2}}{\partial z^{\alpha }\partial
z_{[1]}^{\beta }}\right) a(z)b(z_{[1]})\mid _{z=z_{[1]}},  \label{fpr}
\end{equation}%
for two elements $a$ and $b$ defined by series of type
\begin{equation}
a(v,z)=\sum\limits_{r\geq 0,|\{\alpha \}|\geq 0}\ a_{r,\{\alpha
\}}(u)z^{\{\alpha \}}\ v^{r},  \label{formser}
\end{equation}%
where by $\{\alpha \}$ we label a multi--index. This way, we define a formal
Wick algebra $\mathbf{\check{W}}_{u}$ associated with the tangent space $%
T_{u}\mathbf{V}^{2n},$ for $u\in \mathbf{V}^{2n}.$ It should be noted that
the fibre product (\ref{fpr}) can be trivially extended to the space of $%
\mathbf{\check{W}}$--valued N--adapted differential forms $\mathcal{\check{W}%
}\otimes \Lambda $ by means of the usual exterior product of the scalar
forms $\Lambda ,$ where $\ \mathcal{\check{W}}$ denotes the sheaf of smooth
sections of $\mathbf{\check{W}.}$ There is a standard grading on $\Lambda $
denoted $\deg _{a}.$ One also introduces gradings $\deg _{v},\deg _{s},\deg
_{a}$ on $\ \mathcal{W}\otimes \Lambda $ defined on homogeneous elements $%
v,z^{\alpha },\mathbf{\check{e}}^{\alpha }$ as follows: $\deg _{v}(v)=1,$ $%
\deg _{s}(z^{\alpha })=1,$ $\deg _{a}(\mathbf{\check{e}}^{\alpha })=1,$ and
all other gradings of the elements $v,z^{\alpha },\mathbf{\check{e}}^{\alpha
}$ are set to zero. In this case, the product $\circ $ from (\ref{fpr}) on $%
\ \mathcal{\check{W}}\otimes \mathbf{\Lambda }$ is bigraded. This is written
w.r.t the grading $Deg=2\deg _{v}+\deg _{s}$ and the grading $\deg _{a}.$

\subsubsection{Normal Fedosov's d--operators}

The normal d--connection $\widehat{\mathbf{D}}\mathbf{=\{}\widehat{\mathbf{%
\Gamma }}\mathbf{_{\alpha \beta }^{\gamma }\}}$ (\ref{ndc}) can be extended
to operators
\begin{equation}
\widehat{\mathbf{D}}\left( a\otimes \lambda \right) \doteqdot \left( \mathbf{%
\check{e}}_{\alpha }(a)-u^{\beta }\ \widehat{\mathbf{\Gamma }}\mathbf{%
_{\alpha \beta }^{\gamma }\ }^{z}\mathbf{\check{e}}_{\alpha }(a)\right)
\otimes (\mathbf{\check{e}}^{\alpha }\wedge \lambda )+a\otimes d\lambda ,
\label{cdcop}
\end{equation}%
on $\mathcal{\check{W}}\otimes \Lambda ,$ where $^{z}\mathbf{\check{e}}%
_{\alpha }$ is $\mathbf{\check{e}}_{\alpha }$ redefined in $z$--variables.
This operator $\widehat{\mathbf{D}}$ is a N--adapted $\deg _{a}$--graded
derivation of the distinguished algebra $\left( \mathcal{\check{W}}\otimes
\mathbf{\Lambda ,\circ }\right) ,$ called d--algebra. Such a property
follows from (\ref{fpr}) and (\ref{cdcop})).

\begin{definition}
The Fedosov distinguished operators (d--operators) $\check{\delta}$ and $%
\check{\delta}^{-1}$ on$\ \ \mathcal{\check{W}}\otimes \mathbf{\Lambda ,}$
are defined%
\begin{equation}
\check{\delta}(a)=\ \mathbf{\check{e}}^{\alpha }\wedge \mathbf{\ }^{z}%
\mathbf{\check{e}}_{\alpha }(a),\ \mbox{and\ } \check{\delta}%
^{-1}(a)=\left\{
\begin{array}{c}
\frac{i}{p+q}z^{\alpha }\ \mathbf{\check{e}}_{\alpha }(a),\mbox{ if }p+q>0,
\\
{\qquad 0},\mbox{ if }p=q=0,%
\end{array}%
\right.  \label{feddop}
\end{equation}%
where any $a\in \mathcal{\check{W}}\otimes \mathbf{\Lambda }$ is homogeneous
w.r.t. the grading $\deg _{s}$ and $\deg _{a}$ with $\deg _{s}(a)=p$ and $%
\deg _{a}(a)=q.$
\end{definition}

The d--operators (\ref{feddop}) define the formula $a=(\check{\delta}\
\check{\delta}^{-1}+\check{\delta}^{-1}\ \check{\delta}+\sigma )(a),$ where $%
a\longmapsto \sigma (a)$ is the projection on the $(\deg _{s},\deg _{a})$%
--bihomogeneous part of $a$ of degree zero, $\deg _{s}(a)=\deg _{a}(a)=0;$ $%
\check{\delta}$ is also a $\deg _{a}$--graded derivation of the d--algebra $%
\left( \mathcal{\check{W}}\otimes \mathbf{\Lambda ,\circ }\right) .$ In
order to emphasize the almost K\"{a}hler structure, we used the canonical
almost symplectic geometric objects defined by a fixed $\mathcal{L}.$
Nevertheless, we can always change the ''Lagrangian mechanics variables''
and redefine $\mathbf{\check{\theta},}$ $\mathbf{\check{e}}_{\alpha }$ and $%
\widehat{\mathbf{\Gamma }}\mathbf{_{\alpha \beta }^{\gamma }}$ with respect
to arbitrary frame and co--frame bases using vielbeins $\mathbf{e}_{\
\underline{\alpha }}^{\alpha }$ and their duals $\mathbf{e}_{\alpha \ }^{\
\underline{\alpha }},$ defined by $e_{\ \underline{i}}^{i},e_{\ \underline{i}%
}^{a}$ and $e_{\ \underline{a}}^{a}$.

\begin{proposition}
\label{prthprfo}The torsion and curvature canonical d--operators of the
extension of $\widehat{\mathbf{D}}$ to $\mathcal{\check{W}}\otimes \mathbf{%
\Lambda ,}$ are computed
\begin{equation}
^{z}\widehat{\mathcal{T}}\ \doteqdot \frac{z^{\gamma }}{2}\ \check{\theta}%
_{\gamma \tau }\ \widehat{\mathbf{T}}_{\alpha \beta }^{\tau }(u)\ \mathbf{%
\check{e}}^{\alpha }\wedge \mathbf{\check{e}}^{\beta },  \label{at1}
\end{equation}%
and%
\begin{equation}
\ ^{z}\widehat{\mathcal{R}}\doteqdot \frac{z^{\gamma }z^{\varphi }}{4}\
\check{\theta}_{\gamma \tau }\ \widehat{\mathbf{R}}_{\ \varphi \alpha \beta
}^{\tau }(u)\ \mathbf{\check{e}}^{\alpha }\wedge \mathbf{\check{e}}^{\beta },
\label{ac1}
\end{equation}%
where the nontrivial coefficients of $\ \widehat{\mathbf{T}}_{\alpha \beta
}^{\tau }$ and $\ \widehat{\mathbf{R}}_{\ \varphi \alpha \beta }^{\tau }$
are defined respectively by formulas (\ref{cdtors}) and (\ref{cdcurv}).
\end{proposition}

By straightforward verifications, it follows the proof of
\begin{theorem}
\label{thprfo}The properties
 $\left[ \widehat{\mathbf{D}},\check{\delta}\right] =\frac{i}{v}ad_{Wick}(^{z}%
\widehat{\mathcal{T}})$ and  \newline $\widehat{\mathbf{D}}^{2}=-\frac{i}{v}%
ad_{Wick}(\ ^{z}\widehat{\mathcal{R}})$,
hold for the above operators, where $[\cdot ,\cdot ]$ is the $\deg _{a}$%
--graded commutator of endomorphisms of $\mathcal{\check{W}}\otimes \mathbf{%
\Lambda }$ and $ad_{Wick}$ is defined via the $\deg _{a}$--graded commutator
in $\left( \mathcal{\check{W}}\otimes \mathbf{\Lambda ,\circ }\right) .$
\end{theorem}

The above  formulas  can be redefined for any linear connection
structure on $\mathbf{V}^{2n}.$ For example, we consider how similar
formulas can be provided for the Levi--Civita connection.

\subsubsection{Fedosov's d--operators and the Levi--Civita connection}

For any metric structure $\mathbf{g}$ on a manifold $\mathbf{V}^{2n}\mathbf{,%
}$ the Levi--Civita connection $\bigtriangledown =\{\ _{\shortmid }\Gamma
_{\beta \gamma }^{\alpha }\}$ is by definition the unique linear connection
that is metric compatible $(\bigtriangledown g=0)$ and torsionless $( \
_{\shortmid }\mathcal{T}=0 )$. It is not a d--connection because it does not
preserve the N--connection splitting under parallel transports (\ref{whitney}%
). Let us parametrize its coefficients in the form
\begin{eqnarray*}
_{\shortmid }\Gamma _{\beta \gamma }^{\alpha } &=&\left( _{\shortmid
}L_{jk}^{i},_{\shortmid }L_{jk}^{a},_{\shortmid }L_{bk}^{i},\ _{\shortmid
}L_{bk}^{a},_{\shortmid }C_{jb}^{i},_{\shortmid }C_{jb}^{a},_{\shortmid
}C_{bc}^{i},_{\shortmid }C_{bc}^{a}\right) ,\mbox{\ where} \\
\bigtriangledown _{\mathbf{\check{e}}_{k}}(\mathbf{\check{e}}_{j}) &=&\
_{\shortmid }L_{jk}^{i}\mathbf{\check{e}}_{i}+\ _{\shortmid
}L_{jk}^{a}e_{a},\ \bigtriangledown _{\mathbf{\check{e}}_{k}}(e_{b})=\
_{\shortmid }L_{bk}^{i}\mathbf{\check{e}}_{i}+\ _{\shortmid }L_{bk}^{a}e_{a},
\\
\bigtriangledown _{e_{b}}(\mathbf{\check{e}}_{j}) &=&\ _{\shortmid
}C_{jb}^{i}\mathbf{\check{e}}_{i}+\ _{\shortmid }C_{jb}^{a}e_{a},\
\bigtriangledown _{e_{c}}(e_{b})=\ _{\shortmid }C_{bc}^{i}\mathbf{\check{e}}%
_{i}+\ _{\shortmid }C_{bc}^{a}e_{a}.
\end{eqnarray*}%
A straightforward calculation shows that the coefficients of the
Levi--Civita connection can be expressed as {\small
\begin{eqnarray}
\ _{\shortmid }L_{jk}^{a} &=&-\widehat{C}_{jb}^{i}\check{g}_{ik}\check{g}%
^{ab}-\frac{1}{2}\check{\Omega}_{jk}^{a},\ \ _{\shortmid }L_{bk}^{i}=\frac{1%
}{2}\check{\Omega}_{jk}^{c}\check{g}_{cb}\check{g}^{ji}-\Xi _{jk}^{ih}%
\widehat{C}_{hb}^{j},  \label{lccon} \\
\ _{\shortmid }L_{jk}^{i} &=&\widehat{L}_{jk}^{i},\ _{\shortmid }L_{bk}^{a}=%
\widehat{L}_{bk}^{a}+~^{+}\Xi _{cd}^{ab}\ ^{\circ }L_{bk}^{c},\ \
_{\shortmid }C_{kb}^{i}=\widehat{C}_{kb}^{i}+\frac{1}{2}\check{\Omega}%
_{jk}^{a}\check{g}_{cb}\check{g}^{ji}+\Xi _{jk}^{ih}\widehat{C}_{hb}^{j},
\notag \\
\ _{\shortmid }C_{jb}^{a} &=&-~^{+}\Xi _{cb}^{ad}\ ^{\circ }L_{dj}^{c},\
_{\shortmid }C_{bc}^{a}=\widehat{C}_{bc}^{a},\ _{\shortmid }C_{ab}^{i}=-%
\frac{\check{g}^{ij}}{2}\left\{ \ ^{\circ }L_{aj}^{c}\check{g}_{cb}+\
^{\circ }L_{bj}^{c}\check{g}_{ca}\right\} ,  \notag
\end{eqnarray}%
} where $e_{b}=\partial /\partial y^{a},$ $\check{\Omega}_{jk}^{a}$ are
computed as in (\ref{nccurv}) but for the canonical N--connection $\mathbf{%
\check{N}}$ (\ref{ncel}), $\Xi _{jk}^{ih}=\frac{1}{2}(\delta _{j}^{i}\delta
_{k}^{h}-\check{g}_{jk}\check{g}^{ih}),~^{\pm }\Xi _{cd}^{ab}=\frac{1}{2}%
(\delta _{c}^{a}\delta _{d}^{b}\pm \check{g}_{cd}\check{g}^{ab}),\ \ ^{\circ
}L_{aj}^{c}=\widehat{L}_{aj}^{c}-e_{a}(\check{N}_{j}^{c}), $ $\check{g}_{ik}$
and $\check{g}^{ab}$ are defined for the representation of the metric in
Lagrange--Finsler variables (\ref{hvmetr1}) and the normal d--connection $%
\widehat{\mathbf{\Gamma }}_{\beta \gamma }^{\alpha }=(\widehat{L}_{jk}^{i},\
^{v}\widehat{C}_{bc}^{a})$ (\ref{ndc}) is given by coefficients (\ref{cdcc}).

Let introduce the distortion d--tensor $\ _{\shortmid }Z_{\ \alpha \beta
}^{\gamma }$ with N--adapted coefficients
\begin{eqnarray}
\ _{\shortmid }Z_{jk}^{a} &=&-\widehat{C}_{jb}^{i}\check{g}_{ik}\check{g}%
^{ab}-\frac{1}{2}\check{\Omega}_{jk}^{a},~_{\shortmid }Z_{bk}^{i}=\frac{1}{2}%
\check{\Omega}_{jk}^{c}\check{g}_{cb}\check{g}^{ji}-\Xi _{jk}^{ih}~\widehat{C%
}_{hb}^{j},  \notag \\
\ _{\shortmid }Z_{jk}^{i} &=&0,\ _{\shortmid }Z_{bk}^{a}=~^{+}\Xi
_{cd}^{ab}~~^{\circ }L_{bk}^{c},_{\shortmid }Z_{kb}^{i}=\frac{1}{2}\check{%
\Omega}_{jk}^{a}\check{g}_{cb}\check{g}^{ji}+\Xi _{jk}^{ih}~\widehat{C}%
_{hb}^{j},  \label{cdeftc} \\
\ _{\shortmid }Z_{jb}^{a} &=&-~^{-}\Xi _{cb}^{ad}~~^{\circ }L_{dj}^{c},\
_{\shortmid }Z_{bc}^{a}=0,_{\shortmid }Z_{ab}^{i}=-\frac{g^{ij}}{2}\left[
~^{\circ }L_{aj}^{c}\check{g}_{cb}+~^{\circ }L_{bj}^{c}\check{g}_{ca}\right]
,  \notag
\end{eqnarray}

The next result follows from the above arguments.

\begin{proposition}
The N--adapted coefficients, of the normal d--connection and of the
distortion d--tensors define the Levi--Civita connection as
\begin{equation}
\ _{\shortmid }\Gamma _{\ \alpha \beta }^{\gamma }=\widehat{\mathbf{\Gamma }}%
_{\ \alpha \beta }^{\gamma }+\ _{\shortmid }Z_{\ \alpha \beta }^{\gamma },
\label{cdeft}
\end{equation}%
where $\ _{\shortmid }Z_{\ \alpha \beta }^{\gamma }$ are given by formulas (%
\ref{cdeft}) and h-- and v--components of $\widehat{\mathbf{\Gamma }}_{\beta
\gamma }^{\alpha }$ are given by (\ref{cdcc}).
\end{proposition}

We emphasize that all components of $\ _{\shortmid }\Gamma _{\ \alpha
\beta}^{\gamma }, \widehat{\mathbf{\Gamma }}_{\ \alpha \beta }^{\gamma }$
and $\ _{\shortmid }Z_{\ \alpha \beta }^{\gamma }$ are uniquely defined by
the coefficients of d--metric (\ref{m1}), or (equivalently) by (\ref{hvmetr1}%
) and (\ref{ncel}). The constructions can be obtained for any $n+n$
splitting on $V^{2n},$ which for suitable $\mathcal{L},$ or $\mathcal{F},$
admit a Lagrange, or Finsler, like representation of geometric objects.

By proposition \ref{prthprfo}, the expressions for the curvature and torsion
of canonical d--operators of the extension of $\bigtriangledown $ to $%
\mathcal{\check{W}}\otimes \mathbf{\Lambda ,}$ are
\begin{eqnarray}
\ _{\shortmid }^{z}\mathcal{R} &\doteqdot &\frac{z^{\gamma }z^{\varphi }}{4}%
\ \check{\theta}_{\gamma \tau }\ \ _{\shortmid }R_{\ \varphi \alpha \beta
}^{\tau }(u)\ \mathbf{\check{e}}^{\alpha }\wedge \mathbf{\check{e}}^{\beta },
\label{ac1cl} \\
\ _{\shortmid }^{z}\mathcal{T}\ &\doteqdot &\frac{z^{\gamma }}{2}\ \check{%
\theta}_{\gamma \tau }\ \ _{\shortmid }T_{\alpha \beta }^{\tau }(u)\ \mathbf{%
\check{e}}^{\alpha }\wedge \mathbf{\check{e}}^{\beta }\equiv 0,  \notag
\end{eqnarray}%
where $\ _{\shortmid }T_{\alpha \beta }^{\tau }$ $=0,$ by definition, and$\
\ _{\shortmid }R_{\ \varphi \alpha \beta }^{\tau }$ is computed with respect
to the N--adapted Lagange--Finsler canonical bases by introducing $\widehat{%
\mathbf{\Gamma }}_{\ \alpha \beta }^{\gamma }=-\ _{\shortmid }\Gamma _{\
\alpha \beta }^{\gamma }+\ _{\shortmid }Z_{\ \alpha \beta }^{\gamma },$ see (%
\ref{cdeft}), into (\ref{cdcurv}). To the N--adapted d--operator (\ref{cdcop}%
), we can associate
\begin{equation}
\widehat{\bigtriangledown }\left( a\otimes \lambda \right) \doteqdot \left(
\mathbf{\check{e}}_{\alpha }(a)-u^{\beta }\ _{\shortmid }\Gamma \mathbf{%
_{\alpha \beta }^{\gamma }\ }^{z}\mathbf{\check{e}}_{\alpha }(a)\right)
\otimes (\mathbf{\check{e}}^{\alpha }\wedge \lambda )+a\otimes d\lambda ,
\label{lcexop}
\end{equation}%
on $\mathcal{\check{W}}\otimes \Lambda ,$ where $^{z}\mathbf{\check{e}}%
_{\alpha }$ is $\mathbf{\check{e}}_{\alpha }$ redefined in $z$--variables.
This almost symplectic connection $\widehat{\bigtriangledown }$ is
torsionles and, in general, is not adapted to the N--connection structures.

\begin{corollary}
For the Levi--Civita connection $\bigtriangledown =\{\ _{\shortmid }\Gamma
_{\beta \gamma }^{\alpha }\}$ on a N--anholo\-no\-mic manifold $\mathbf{V}%
^{2n},$ we have $\left[ \widehat{\bigtriangledown },\check{\delta}\right] =0$
and $\widehat{\bigtriangledown }^{2}=-\frac{i}{v}ad_{Wick}(\ _{\shortmid
}^{z}\mathcal{R})$, where $\widehat{\bigtriangledown }$ is defined by
formula (\ref{lcexop}), $\ _{\shortmid }^{z}\mathcal{R}$ is given by (\ref%
{ac1cl}), $[\cdot ,\cdot ]$ is the $\deg _{a}$--graded commutator of
endomorphisms of $\mathcal{\check{W}}\otimes \mathbf{\Lambda }$ and $%
ad_{Wick}$ is defined via the $\deg _{a}$--graded commutator in $\left(
\mathcal{\check{W}}\otimes \mathbf{\Lambda ,\circ }\right) .$
\end{corollary}

\begin{proof}
It is a straightforward consequence of the Theorem \ref{thprfo} for the
Levi--Civita and curvature operators extended on $\mathcal{\check{W}}\otimes
\Lambda .$ $\square $
\end{proof}

\vskip5pt

Prescribing a $n+n$ splitting on $\mathbf{V}^{2n}$, we can work equivalently
with any metric compatible linear connection structure which is N--adapted,
or not, if such a connection is completely defined by the (pseudo)
Riemannian metric structure. It is preferable to use the approach with the
normal d--connection because this way we have both an almost symplectic
analogy and Lagrange, or Finsler, like interpretation of geometric objects.
In standard classical gravity, in order to solve some physical problems, it
is more convenient to work with the Levi--Civita connection or its spin like
representations (for instance, in the Einstein--Dirac theory). The
self--dual and further generalizations to Ashtekar variables are more
convenient, respectively, in canonical ADN classical and quantum gravity
and/or loop quantum gravity.

It should be noted that the formulas for Fedosov's d--operators and their
properties do not depend in explicit form on generating functions $\mathcal{L%
},$ or $\mathcal{F}.$ Such a function may be formally introduced for
elaborating a Lagrange mechanics, or Finsler, modeling for a (pseudo)
Riemannian space with a general $n+n$ nonholonomic splitting. This way, we
emphasize that the Fedosov's approach is valid for various type of (pseudo)
Riemann, Riemann--Cartan, Lagrange--Finsler, almost K\"{a}hler and other
types of holonomic and nonholonic manifolds used for geometrization of
mechanical and field models. Nevertheless, the constructions are performed
in a general form and the final results do not depend on any ''background''
structures. We conclude that $3+1$ fibration approaches are more natural for
loop quantum gravity, but the models with nonholonomic $2+2$ splitting
result in almost K\"{a}hler quantum models; althought both types of
quantization, loop and deformation, provide background independent
constructions.

\subsubsection{Deformation Quantization of Einstein and Lagrange Spaces}

Formulating a (pseudo) Riemannian geometry in Lagrange--Finsler variables,
we can quantize the metric, frame and linear connection structures following
standard methods for deformation quantization of almost K\"{a}hler
manifolds. The goal of this section is to provide the main Fedosov type
results for such constructions and to show how the Einstein manifolds can be
encoded into the topological structure of such quantized nonholonomic spaces.

\paragraph{Fedosov's theorems for normal d--connections:}

The third main result of this work will be stated below by three theorems
for the normal d--connection (equivalently, canonical almost symplectic
structure) $\widehat{\mathbf{D}}\equiv \ _{\theta }\widehat{\mathbf{D}}$ (%
\ref{ndc}).

\begin{theorem}
\label{th3a}Any (pseudo) Riemanian metric $\mathbf{g}$ (\ref{m1})
(equivalently, $\mathbf{g=\check{g}}$ (\ref{hvmetr1})) defines a flat normal
Fedosov d--connection $\widehat{\mathcal{D}} := -\ \check{\delta}+\widehat{%
\mathbf{D}}-\frac{i}{v}ad_{Wick}(r)$ satisfying the condition $\widehat{%
\mathcal{D}}^{2}=0,$ where the unique element $r\in $ $\mathcal{\check{W}}%
\otimes \mathbf{\Lambda ,}$ $\deg _{a}(r)=1,$ $\check{\delta}^{-1}r=0,$
solves the equation
 $ \check{\delta}r=\widehat{\mathcal{T}}\ +\widehat{\mathcal{R}}+\widehat{%
\mathbf{D}}r-\frac{i}{v}r\circ r$ and this element can be computed recursively with respect to the total
degree $Deg$ as follows:%
\begin{eqnarray*}
r^{(0)} &=&r^{(1)}=0, r^{(2)}=\check{\delta}^{-1}\widehat{\mathcal{T}},
r^{(3)}=\ \check{\delta}^{-1}\left( \widehat{\mathcal{R}}+\widehat{\mathbf{D}%
}r^{(2)}-\frac{i}{v}r^{(2)}\circ r^{(2)}\right) , \\
r^{(k+3)} &=&\ \ \check{\delta}^{-1}\left( \widehat{\mathbf{D}}r^{(k+2)}-%
\frac{i}{v}\sum\limits_{l=0}^{k}r^{(l+2)}\circ r^{(l+2)}\right) ,k\geq 1,
\end{eqnarray*}%
where by $a^{(k)}$ we denoted the $Deg$--homogeneous component of degree $k$
of an element $a\in $ $\ \mathcal{\check{W}}\otimes \mathbf{\Lambda }.$
\end{theorem}

\begin{proof}
It follows from straightforward verifications of the property $\widehat{%
\mathcal{D}}^{2}=0$ using for $r$ formal series of type (\ref{formser}) and
the formulas for N--adapted coefficients: (\ref{cdcc}) for $\widehat{\mathbf{%
D}},$ (\ref{cdtors}) for $\widehat{\mathcal{T}},$ (\ref{cdcurv}) for $%
\widehat{\mathcal{R}},$ and the properties of Fedosov's d--operators (\ref%
{feddop}) stated by Theorem \ref{thprfo}. The length of this paper does not
allow us to present such a tedious calculation which is a N--adapted version
for corresponding ''hat'' operators.$\square $
\end{proof}

\vskip5pt

The procedure of deformation quantization is related to the definition of a
star--product which in our approach can be defined canonically because the
normal d--connection $\widehat{\mathbf{D}}$ is a N--adapted variant of the
affine and almost symplectic connection considered in that work. This
provides a proof for

\begin{theorem}
\label{th3b}A star--product on the almost K\"{a}hler model of a (pseudo)
Riemannian space in Lagrange--Finsler variables is defined on $C^{\infty }(%
\mathbf{V}^{2n})[[v]]$ by formula
 $\ ^{1}f\ast \ ^{2}f\doteqdot \sigma (\tau (\ ^{1}f))\circ \sigma (\tau (\
^{2}f))$,
 where the projection $\sigma :\mathcal{\check{W}}_{\widehat{\mathcal{D}}%
}\rightarrow C^{\infty }(\mathbf{V}^{2n})[[v]]$ onto the part of $\deg _{s}$%
--degree zero is a bijection and the inverse map $\tau :C^{\infty }(\mathbf{V%
}^{2n})[[v]]\rightarrow \mathcal{\check{W}}_{\widehat{\mathcal{D}}}$ can be
calculated recursively w.r..t the total degree $Deg,$%
\begin{eqnarray*}
\tau (f)^{(0)} &=&f\mbox{\ and, for \ }k\geq 0, \\
\tau (f)^{(k+1)} &=&\ \check{\delta}^{-1}\left( \widehat{\mathbf{D}}\tau
(f)^{(k)}-\frac{i}{v}\sum\limits_{l=0}^{k}ad_{Wick}(r^{(l+2)})(\tau
(f)^{(k-l)})\right) .
\end{eqnarray*}
\end{theorem}

We denote by $\ ^{f}\xi $ the Hamiltonian vector field corresponding to a
function $f\in C^{\infty }(\mathbf{V}^{2n})$ on space $(\mathbf{V}^{2n},%
\check{\theta})$ and consider the antisymmetric part $\ ^{-}C(\ ^{1}f,\
^{2}f)\ \doteqdot \frac{1}{2}\left( C(\ ^{1}f,\ ^{2}f)-C(\ ^{2}f,\
^{1}f)\right) $ of bilinear operator $C(\ ^{1}f,\ ^{2}f).$ We say that a
star--product (\ref{starp}) is normalized if $\ _{1}C(\ ^{1}f,\ ^{2}f)=\frac{%
i}{2}\{\ ^{1}f,\ ^{2}f\},$ where $\{\cdot ,\cdot \}$ is the Poisson bracket.
For the normalized $\ast ,$ the bilinear operator $\ _{2}^{-}C$ defines a de
Rham--Chevalley 2--cocycle, when there is a unique closed 2--form $\ \check{%
\varkappa}$ such that$\ _{2}C(\ ^{1}f,\ ^{2}f)=\frac{1}{2}\ \check{\varkappa}%
(\ ^{f_{1}}\xi ,\ ^{f_{2}}\xi )$ for all $\ ^{1}f,\ ^{2}f\in C^{\infty }(%
\mathbf{V}^{2n}).$ This is used to introduce $c_{0}(\ast )\doteqdot \lbrack
\check{\varkappa}]$ as the equivalence class.
A straightforward computation of $\ _{2}C$ and the results of Theorem \ref%
{th3b} provide the proof of
\begin{lemma}
The unique 2--form defined by the normal d--connection can be computed as $%
\check{\varkappa}=-\frac{i}{8}\mathbf{\check{J}}_{\tau }^{\ \alpha ^{\prime
}}\widehat{\mathcal{R}}_{\ \alpha ^{\prime }}^{\tau }-\frac{i}{6}d\left(
\mathbf{\check{J}}_{\tau }^{\ \alpha ^{\prime }}\widehat{\mathbf{T}}_{\
\alpha ^{\prime }\beta }^{\tau }\ \mathbf{\check{e}}^{\beta }\right)$, where
the coefficients of the curvature and torsion 2--forms of the normal
d--connection 1--form are given respectively by formulas (\ref{cform}) and (%
\ref{tform}).
\end{lemma}

We now define another canonical class $\check{\varepsilon},$ for $\ ^{%
\check{N}}T\mathbf{V}^{2n}=h\mathbf{V}^{2n}\oplus v\mathbf{V}^{2n},$ where
the left label indicates that the tangent bundle is split nonholonomically
by the canonical N--connection structure $\mathbf{\check{N}}.$ We can
perform a distinguished complexification of such second order tangent
bundles in the form $T_{\mathbb{C}}\left( \ ^{\check{N}}T\mathbf{V}%
^{2n}\right) =T_{\mathbb{C}}\left( h\mathbf{V}^{2n}\right) \oplus T_{\mathbb{%
C}}\left( v\mathbf{V}^{2n}\right) $ and introduce $\ \check{\varepsilon}$ as
the first Chern class of the distributions $T_{\mathbb{C}}^{\prime }\left( \
^{N}T\mathbf{V}^{2n}\right) =T_{\mathbb{C}}^{\prime }\left( h\mathbf{V}%
^{2n}\right) \oplus T_{\mathbb{C}}^{\prime }\left( v\mathbf{V}^{2n}\right) $
of couples of vectors of type $(1,0)$ both for the h-- and v--parts. In
explicit form, we can calculate $\check{\varepsilon}$ by using the
d--connection $\widehat{\mathbf{D}}$ and the h- and v--projections $h\Pi =%
\frac{1}{2}(Id_{h}-iJ_{h})$ and $v\Pi =\frac{1}{2}(Id_{v}-iJ_{v}),$ where $%
Id_{h}$ and $Id_{v}$ are respective identity operators and $J_{h}$ and $%
J_{v} $ are almost complex operators, which are projection operators onto
corresponding $(1,0)$--subspaces. Introducing the matrix $\left( h\Pi ,v\Pi
\right) \ \widehat{\mathcal{R}}\left( h\Pi ,v\Pi \right) ^{T},$ where $%
(...)^{T}$ means transposition, as the curvature matrix of the N--adapted
restriction of $\ $of the normal d--connection$\ \widehat{\mathbf{D}}$ to $%
T_{\mathbb{C}}^{\prime }\left( \ ^{\check{N}}T\mathbf{V}^{2n}\right) ,$ we
compute the closed Chern--Weyl form
\begin{equation}
\check{\gamma}=-iTr\left[ \left( h\Pi ,v\Pi \right) \widehat{\mathcal{R}}%
\left( h\Pi ,v\Pi \right) ^{T}\right] =-iTr\left[ \left( h\Pi ,v\Pi \right)
\widehat{\mathcal{R}}\right] =-\frac{1}{4}\mathbf{\check{J}}_{\tau }^{\
\alpha ^{\prime }}\widehat{\mathcal{R}}_{\ \alpha ^{\prime }}^{\tau }.
\label{aux4}
\end{equation}%
We get that the canonical class is $\check{\varepsilon}\doteqdot \lbrack
\check{\gamma}],$ which proves the

\begin{theorem}
\label{th3c}The zero--degree cohomology coefficient $c_{0}(\ast )$ for the
almost K\"{a}hler model of a (pseudo) Riemannian space defined by d--tensor $%
\mathbf{g}$ (\ref{m1}) (equivalently, by $\mathbf{\check{g}}$ (\ref{hvmetr1}%
)) is computed $c_{0}(\ast )=-(1/2i)\ \check{\varepsilon}.$
\end{theorem}

The coefficient $c_{0}(\ast )$ can be similarly computed for the case when a
metric of type (\ref{m1}) is a solution of the Einstein equations and this
zero--degree coefficient defines certain quantum properties of the
gravitational field. A more rich geometric structure should be considered if
we define a value similar to $c_{0}(\ast )$ encoding the information about
Einstein manifolds deformed into corresponding quantum configurations.

\paragraph{The zero--degree cohomology coefficient for Einstein ma\-nifolds}

\label{ssensp}The priority of deformation quantization is that we can
elaborate quantization schemes when metric, vielbein and connection fields
are not obligatory subjected to satisfy certain field equations and/or
derived by a variational procedure. On the other hand, in certain canonical
and loop quantization models, the gravitational field equations are
considered as the starting point for deriving a quantization formalism. In
such cases, the Einstein equations are expressed into ''lapse'' and
''shift'' (and/or generalized Ashtekar) variables and the quantum variant of
the gravitational field equations is prescribed to be in the form of Wheeler
De Witt equations (or corresponding systems of constraints in complex/real
generalized connection and dreibein variables). In this section, we analyze
the problem of encoding the Einstein equations into a geometric formalism of
 deformation quantization.

\paragraph{Gravitational field equations:}

For any d--connection $\mathbf{D=\{\Gamma \},}$ we can define the Ricci
tensor $Ric(\mathbf{D})=\{\mathbf{R}_{\ \beta \gamma }\doteqdot \mathbf{R}%
_{\ \beta \gamma \alpha }^{\alpha }\}$ and the scalar curvature $\
^{s}R\doteqdot \mathbf{g}^{\alpha \beta }\mathbf{R}_{\alpha \beta }$ ($%
\mathbf{g}^{\alpha \beta }$ being the inverse matrix to $\mathbf{g}_{\alpha
\beta }$ (\ref{m1})). If a d--connection is uniquely determined by a metric
in a unique metric compatible form, $\mathbf{Dg}=0,$ (in general, the
torsion of $\mathbf{D}$ is not zero, but induced canonically by the
coefficients of $\mathbf{g),}$ we can postulate in straightforward form the
field equations
\begin{equation}
\mathbf{R}_{\ \beta }^{\underline{\alpha }}-\frac{1}{2}(\ ^{s}R+\lambda )%
\mathbf{e}_{\ \beta }^{\underline{\alpha }}=8\pi G\mathbf{T}_{\ \beta }^{%
\underline{\alpha }},  \label{deinsteq}
\end{equation}%
where $\mathbf{T}_{\ \beta }^{\underline{\alpha }}$ is the effective
energy--momentum tensor, $\lambda $ is the cosmological constant, $G$ is the
Newton constant in the units when the light velocity $c=1,$ and $\mathbf{e}%
_{\ \beta }=\mathbf{e}_{\ \beta }^{\underline{\alpha }}\partial /\partial u^{%
\underline{\alpha }}$ is the N--elongated operator (\ref{2dder}).

Let us consider the absolute antisymmetric tensor $\epsilon _{\alpha \beta
\gamma \delta }$ and effective source 3--form
 $\mathcal{T}_{\ \beta }=\mathbf{T}_{\ \beta }^{\underline{\alpha }}\ \epsilon
_{\underline{\alpha }\underline{\beta }\underline{\gamma }\underline{\delta }%
}du^{\underline{\beta }}\wedge du^{\underline{\gamma }}\wedge du^{\underline{%
\delta }}$ and express the curvature tensor $\mathcal{R}_{\ \gamma }^{\tau }=\mathbf{R}%
_{\ \gamma \alpha \beta }^{\tau }\ \mathbf{e}^{\alpha }\wedge \ \mathbf{e}%
^{\beta }$ of $\mathbf{\Gamma }_{\ \beta \gamma }^{\alpha }=\ _{\shortmid
}\Gamma _{\ \beta \gamma }^{\alpha }-\ Z_{\ \beta \gamma }^{\alpha }$ as $%
\mathcal{R}_{\ \gamma }^{\tau }=\ _{\shortmid }\mathcal{R}_{\ \gamma }^{\tau
}-\mathcal{Z}_{\ \gamma }^{\tau },$ where $\ _{\shortmid }\mathcal{R}_{\
\gamma }^{\tau }$ $=\ _{\shortmid }R_{\ \gamma \alpha \beta }^{\tau }\
\mathbf{e}^{\alpha }\wedge \ \mathbf{e}^{\beta }$ is the curvature 2--form
of the Levi--Civita connection $\nabla $ and the distortion of curvature
2--form $\mathcal{Z}_{\ \gamma }^{\tau }$ is defined by $\ Z_{\ \beta \gamma
}^{\alpha }.$ For the gravitational $\left( \mathbf{e,\Gamma }\right) $ and
matter $\mathbf{\phi }$ fields, we consider the effective action
 $S[\mathbf{e,\Gamma ,\phi }]=\ ^{gr}S[\mathbf{e,\Gamma }]+\ ^{matter}S[%
\mathbf{e,\Gamma ,\phi }]$.
\begin{theorem}
\label{theq}The equations (\ref{deinsteq}) can be represented as 3--form
equations%
\begin{equation}
\epsilon _{\alpha \beta \gamma \tau }\left( \mathbf{e}^{\alpha }\wedge
\mathcal{R}^{\beta \gamma }+\lambda \mathbf{e}^{\alpha }\wedge \ \mathbf{e}%
^{\beta }\wedge \ \mathbf{e}^{\gamma }\right) =8\pi G\mathcal{T}_{\ \tau }
\label{einsteq}
\end{equation}%
following from the action by varying the components of $\mathbf{e}_{\ \beta
},$ when%
{\small
\begin{eqnarray*}
\mathcal{T}_{\ \tau }&=&\ ^{m}\mathcal{T}_{\ \tau }+\ ^{Z}\mathcal{T}_{\
\tau }, \\
\ ^{m}\mathcal{T}_{\ \tau}&=&\ ^{m}\mathbf{T}_{\ \tau }^{\underline{\alpha
}}\epsilon _{\underline{\alpha }\underline{\beta }\underline{\gamma }%
\underline{\delta }}du^{\underline{\beta }}\wedge du^{\underline{\gamma }%
}\wedge du^{\underline{\delta }},
\ ^{Z}\mathcal{T}_{\ \tau}=(8\pi G) ^{-1}\mathcal{Z}_{\tau
}^{\underline{\alpha }}\epsilon _{\underline{\alpha }\underline{\beta }%
\underline{\gamma }\underline{\delta }}du^{\underline{\beta }}\wedge du^{%
\underline{\gamma }}\wedge du^{\underline{\delta }},
\end{eqnarray*}%
}
where $\ ^{m}\mathbf{T}_{\ \tau }^{\underline{\alpha }}=\delta \
^{matter}S/\delta \mathbf{e}_{\underline{\alpha }}^{\ \tau }$ are equivalent
to the usual Einstein equations for the Levi--Civita connection $\nabla ,$\ $%
\ _{\shortmid }\mathbf{R}_{\ \beta }^{\underline{\alpha }}-\frac{1}{2}(\
_{\shortmid }^{s}R+\lambda )\mathbf{e}_{\ \beta }^{\underline{\alpha }}=8\pi
G\ ^{m}\mathbf{T}_{\ \beta }^{\underline{\alpha }}$.
\end{theorem}
For the Einstein gravity in Lagrange--Finsler
variables, we obtain:
\begin{corollary}
The vacuum Einstein eqs  with cosmological constant in terms of the
canonical N--adapted vierbeins and normal d--connection are%
\begin{equation}
\epsilon _{\alpha \beta \gamma \tau }\left( \mathbf{\check{e}}^{\alpha
}\wedge \widehat{\mathcal{R}}^{\beta \gamma }+\lambda \mathbf{\check{e}}%
^{\alpha }\wedge \mathbf{\check{e}}^{\beta }\wedge \ \mathbf{\check{e}}%
^{\gamma }\right) =8\pi G\ ^{Z}\widehat{\mathcal{T}}_{\ \tau },
\label{veinst1}
\end{equation}%
or, for the Levi--Civita connection,
 $\epsilon _{\alpha \beta \gamma \tau }\left( \mathbf{\check{e}}^{\alpha
}\wedge \ _{\shortmid }\mathcal{R}^{\beta \gamma }+\lambda \mathbf{\check{e}}%
^{\alpha }\wedge \mathbf{\check{e}}^{\beta }\wedge \ \mathbf{\check{e}}%
^{\gamma }\right) =0$.
\end{corollary}

\begin{proof}
The conditions of the mentioned Theorem \ref{theq} are redefined for the
co--frames $\mathbf{\check{e}}^{\alpha }$ elongated by the canonical
N--connection (\ref{ncel}), deformation of linear connections (\ref{cdeft})
and curvature (\ref{cdcurv}) with deformation of curvature 2--form of type
\begin{equation}
\widehat{\mathcal{R}}_{\ \gamma }^{\tau }=\ _{\shortmid }\mathcal{R}_{\
\gamma }^{\tau }-\widehat{\mathcal{Z}}_{\ \gamma }^{\tau }.  \label{def2}
\end{equation}%
We put ''hat'' on $\ ^{Z}\widehat{\mathcal{T}}_{\ \tau }$ because this value
is computed using the normal d--connection. $\square $
\end{proof}

Using formulas (\ref{veinst1}) and (\ref{def2}), we can write
\begin{equation}
\widehat{\mathcal{R}}^{\beta \gamma }=-\lambda \mathbf{\check{e}}^{\beta
}\wedge \ \mathbf{\check{e}}^{\gamma }-\widehat{\mathcal{Z}}^{\beta \gamma }%
\mbox{ and }\ _{\shortmid }\mathcal{R}^{\beta \gamma }=-\lambda \mathbf{%
\check{e}}^{\beta }\wedge \ \mathbf{\check{e}}^{\gamma }  \label{aux3}
\end{equation}%
which is  necessary for encoding the vacuum field equations into the
cohomological structure of the quantum almost K\"{a}hler model of Einstein
gravity.

\paragraph{The Chern--Weyl form and Einstein equations:}

Introducing the formulas (\ref{veinst1}) and (\ref{aux3}) into the
conditions of Theorem \ref{th3c}, we obtain the forth main result in this
subsection:
\begin{theorem}
\label{th4r}The zero--degree cohomology coefficient $c_{0}(\ast )$ for the
almost K\"{a}hler model of an Einstein space defined by a d--tensor $\mathbf{%
g}$ (\ref{m1}) (equivalently, by $\mathbf{\check{g}}$ (\ref{hvmetr1})) as a
solution of \ (\ref{veinst1}) is $c_{0}(\ast )=-(1/2i)\ \check{\varepsilon},$
for $\check{\varepsilon}\doteqdot \lbrack \check{\gamma}],$ where
 $\check{\gamma}=\frac{1}{4}\mathbf{\check{J}}_{\tau \alpha }^{\ }\left(
-\lambda \mathbf{\check{e}}^{\tau }\wedge \ \mathbf{\check{e}}^{\alpha }+%
\widehat{\mathcal{Z}}^{\tau \alpha }\right)$.
\end{theorem}
\begin{proof}
We sketch the key points of the proof which follows from (\ref{aux4}) and (%
\ref{aux3}). It should be noted that for $\lambda =0$ the 2--form $\widehat{%
\mathcal{Z}}^{\tau \alpha }$ is defined by the deformation d--tensor from
the Levi--Civita connection to the normal d--connection (\ref{cdeft}), see
formulas (\ref{cdeftc}). Such objects are defined by classical vacuum
solutions of the Einstein equations. We conclude that $c_{0}(\ast )$ encodes
the vacuum Einstein configurations, in general, with nontrivial constants
and their quantum deformations. $\square $
\end{proof}

If the Wheeler De Witt equations represent a quantum version of the Einstein
equations for loop quantum gravity, the Chern--Weyl 2--form
can be used to define the quantum version of Einstein equations (\ref%
{einsteq}) in the deformation quantization approach:

\begin{corollary}
In Lagrange--Finsler variables, the quantum field equations corresponding to
Einstein's general relativity are
\begin{equation}
\mathbf{\check{e}}^{\alpha }\wedge \check{\gamma}=\epsilon ^{\alpha \beta
\gamma \tau }2\pi G\mathbf{\check{J}}_{\beta \gamma }\widehat{\mathcal{T}}%
_{\ \tau }\ -\frac{\lambda }{4}\mathbf{\check{J}}_{\beta \gamma }\mathbf{%
\check{e}}^{\alpha }\wedge \mathbf{\check{e}}^{\beta }\wedge \ \mathbf{%
\check{e}}^{\gamma }.  \label{aseq}
\end{equation}
\end{corollary}

\begin{proof}
Multiplying $\mathbf{\check{e}}^{\alpha }\wedge $ to the above 2--from  written
in Lagrange--Finsler variables and taking into account (\ref{einsteq}),
re--written also in the form adapted to the canonical N--connection, and
introducing the almost complex operator $\mathbf{\check{J}}_{\beta \gamma },$
we get the almost symplectic form of Einstein's equations (\ref{aseq}). $%
\square $
\end{proof}

It should be noted that even in the vacuum case, when $\lambda =0,$ the
2--form $\check{\gamma}$  from (\ref{aseq}) is not zero but
defined by $\widehat{\mathcal{T}}_{\ \tau }=\ ^{Z}\widehat{\mathcal{T}}_{\
\tau }.$

Finally, we emphasize that an explicit computation of $\check{\gamma}$ for
nontrivial matter fields has yet to be performed for a deformation
quantization model in which interacting gravitational and matter fields are
geometrized in terms of an almost K\"{a}hler model defined for spinor and
fiber bundles on spacetime. This is a subject for further investigations.

\chapter{Further Perspectives}

It is possible to elaborate a quite exact research program for the next three years in relation to the fact that the applicant
won recently a Romanian Government Grant IDEI, PN-II-ID-PCE-2011-3-0256, till October 2014. Section \ref{ch2s1}
is devoted to some important ideas and plans on future applicant's research activity using the Proposal for that Grant and
other ones. In section \ref{ch2s2}, we speculate on possible teaching and advanced pedagogical activity.

\section{Future Research Activity and Collaborations}

\label{ch2s1}

\subsection{Scientific context and motivation}

The elaboration of new geometric models and methods and their applications
in physics have a number of motivations coming form modern high energy
physics, gravity and geometric mechanics together with a general very
promising framework to construct a "modern geometry of physics". In a more
restricted context, but not less important, the geometric methods were
recently applied as effective tools for generating exact solutions for
fundamental physics and evolution equations.

Today, physical theories and a number of multi-disciplinary research
directions are so complex that it is often very difficult to formulate,
investigate and elaborate any applications without a corresponding
especially closed mathematical background and inter-disciplinary approaches
and methods. If in the past the physicists tried traditionally to not attack
problems of "pure" mathematics, the situation has substantially changed
during last 20 years. A number of mathematical ideas, notions and objects
were proposed and formulated in terms of general relativity, statistics and
quantum filed theory and strings. In modern fundamental physical theories,
mathematics provides not only tools and methods of solution of physical
problems, but governs the physicists' intuition.

The future applicant's research activity is planned in the line of the
mentioned unification of mathematics and physics being stated as a present
days program of developing new mathematical ideas and methods with
applications in modern classical and quantum gravity, Ricci flow theory and
noncommutative generalizations, geometric mechanics, stochastic processes
etc. Our general research goals are related to five main directions: 1) to
develop a new nonholonomic approach to geometric and deformation
quantization of gravity and nonlinear systems; 2) to construct and study
exact solutions in gravity and Ricci flow theory with generic local
anisotropy, non-trivial topology and/or hidden noncommutative structure
having motivation from string/brane and extra dimension gravity theories; 3)
to elaborate a corresponding formalism of nonholonomic Dirac operators and
generalizations to noncommutative geometry and evolution models of
fundamental geometric objects, exact solutions and quantum deformations; 4)
modified theories of gravity, exact solutions and quantization methods; 5)
anisotropic and nonholonomic configurations in modern cosmology and
astrophysics related to dark energy/matter problems.

We shall focus on new
features of the geometry of nonlinear connections and nonholonomic and
quantum deformations and study new aspects related to solitonic hierarchies,
bi-Hamilton formalism,  (non) commutative almost symplectic
structures and analogous models of Lagrange-Finsler and Hamilton-Cartan
geometries, differential geometry of superspaces, fractional derivatives and
 dimensions, stochastic anisotropic processes etc.

\subsection{Objectives}

There are formulated five main objectives with respective exploratory
importance, novelty, interdisciplinary character and possible applications:

\begin{enumerate}
\item Objective 1. Geometric and Deformation Quantization of Gravity and
Matter Field Interactions and Nonholonomic Mechanical Systems.
\begin{itemize}
\item In a general geometric approach, theories of classical and quantum
interactions with gravitational field equations for the Levi--Civita
connection can be re-formulated equivalently in almost K\"{a}hler and/or
Lagrange-Finsler variables on nonholonomic manifolds and bundle spaces. Such
nonholonomic configurations and/ or dynamical systems are determined by
corresponding non-inte\-grable distributions on space/-time manifolds. Our
goal and novelty are oriented to geometric quantization and renormalization
schemes for nonlinear theories following the nonlinear connection formalism
and techniques originally elaborated for quantum K\"{a}hler and almost
symplectic geometries.

\item We shall compare the new approach to geometric quantization of
nonholonomic almost K\"{a}hler spaces with former our constructions performed
for deformation and A-brane quantization. We shall analyze possible
connections and find key differences between geometric schemes and
physically important perturbative models with renormalization of
analogous/emergent commutative and noncommutative gravitational and gauge
like theories.

\item There will be provided a series of applications of geometric and
deformation quantization methods to theories with nonlinear dispersions,
local anisotropy and Lorenz violation induced from quantum gravity and/or
string/brane theories. Such models can  be described  as analogous
classical/quantum Lagrange-Hamilton and Finsler-Cartan geometries which
will extend the research with applications in modern mechanics and nonlinear
dynamics.
\end{itemize}

\item Objective 2. Exact Solutions in Gravity and Ricci Flow Theory

\begin{itemize}
\item Via nonholonomic (equivalently, anholonomic) deformations of the frame
and connection structures, the Einstein field equations can be decoupled and
solved in very general forms. This allows us to generate various classes of
off-diagonal solutions with associated solitonic hierarchies, characterized
by stochastic, fractional, fractal behavior and nonholonomic dynamical
multipole moments.

\item Following new geometric techniques, we shall derive new classes of
locally anisotropic black holes, ellipsoids, wormholes and cosmological
spacetimes. This requests new ideas and methods and generalizations for
black hole uniqueness and non-hair theorems in general relativity and
modified gravity theories. Various examples of vacuum and non-vacuum metrics
with noncommutative, supersymmetric and/or nonsymmetric variables will be
constructed and analyzed in explicit form.

\item Nonholonomically constrained Ricci flows of (semi) Riemannian
geometries result, in general, in various classes of commutative and
noncommutative geometries. A special interest presents the research related
to geometric evolution of Einstein metrics and possible connections to beta
functions and renormalization in quantum gravity. Encoding geometric and
physical data in terms of almost K\"{a}hler geometry, the classical and quantum
evolution scenarios can be performed and studied following our former
approach elaborated for commutative and noncommutative evolution of Einstein
and/or Finsler geometries.
\end{itemize}

\item Objective 3. Nonholonomic Clifford Structures and Dirac Operators

\begin{itemize}
\item The geometry of nonholonomic Clifford and spinor bundles enabled with
nonlinear connections was elaborated in our works following methods of
classical and quantum Lagrange--Finsler geometry. It is important to extend
such constructions to nonholonomic spinor and twistor structures derived via
almost K\"{a}hler and/or nonholonomic variables for certain important models of
classical and quantum gravity. The concept of nonholonomic Dirac operators
for almost symplectic classical and quantum systems will be developed in
connection to new spinor and twistor methods in gravity and gauge models and
exact solutions for Einstein--Yang-Mills--Dirac systems.

\item Quantization of nonholonomic Clifford and related almost K\"{a}hler
structures will be performed following geometric and deformation
quantization techniques. Possible connections to twistor diagrams formalism
and quantization will be analyzed.

\item The theory of nonholonomic Dirac operators and spectral triples and
functionals consists a fundamental mathematical background for various
approaches to noncommutative geometry, particle phy\-sics and Ricci flow
evolution models. Our new idea is to study nonholonomic and noncommutative
Clifford structures and their geometric evolution scenarios using almost K%
\"{a}hler and spinor variables. A comparative study with noncommutative
models derived via Seiberg-Witten transforms and deformation quantization
will be performed.
\end{itemize}

\item Objective 4. Modified theories of gravity, exact solutions and
quantization methods

\begin{itemize}
\item We shall extend our methods of constructing generic off--diagonal
exact solutions in Einstein gravity and (non) commutative
Ein\-stein--Finsler gravity theories to modified theories of gravity with
aniso\-tropic scaling and nonlinear dependence on scalar curvature, torsion
components, variation of constants etc. The conditions when such effective
theories can be modeled by nonholonomic constraints and nonlinear
off--diagonal interactions will be formulated in a geometric form. There
will be elaborated analystic and computer modeling programs for solitonic
interactions, \ black hole configurations, and other classes of exact
solutions. The criteria when analogous Dirac operators for almost symplectic
classical and quantum systems can be considered and exact solutions for
Einstein-Yang-Mills-Dirac systems in modified gravity will be constructed.

\item There are known perturbative theories of gravity, in covariant and
generalize non--perturbative forms, which seem to provide physically
important scenaria for quanum gravity, Lorenz violations, super-luminal
effects etc. We shall be interested to develope certain geometric schemes
for quantization of nonholonomic Clifford and related almost K\"{a}hler
structures in modified gravity and theories with off--diagonal analogous
configurations.
\end{itemize}

\item Objective 5. Anisotropic and nonholonomic configurations in modern
cosmology and astrophysics related to dark energy/matter problems

\begin{itemize}
\item This direction is strongly related to ''changing of paradigms'' in
modern fundamental physics. In some sence, it depends on experimental and
observational data and phenomenology in gravity and particle physics. We
suppose to apply our experience on geometric methods and mathematical
physics extended to computer modeling and graphics.

\item Possible tests and phenomenology for (non) commutative and quantum gravity models, with
almost K\"{a}hler and spinor variables, generalize Seiberg--Witten transforms
and deformation quantization, will be proposed and analyzed in details.
\end{itemize}
\end{enumerate}

\subsection{Methods and approaches}

\subsubsection{Techniques, Milestones and Objectives:}

The research methods span pure geometry, partial differential equations and
analytic methods of constructing of exact solutions and important issues in
mathematical physics and theoretical particle physics in equal measure.
Certain problems of geometric mechanics, diffusion theory, fractional
differential geometry related to nontrivial solutions in gravity and
quantization will be also concerned. Although the bulk motivation of the
tasks comes from fundamental gravity and particle physics, the approach to
the Project objectives is a patient and systematic development of what one
believes to be the necessary and inevitable application of geometrical tools
from almost K\"{a}hler geometry, generalized Finsler geometry and nonholonomic
manifolds. In the process, it is planned to contribute significantly to
geometric and deformation quantization and renormalization of nonholonomic
Einstein and generalized Lagrange--Finsler gravity models elaborated on
Lorenz spacetimes and, respectively, on (co) tangent bundles to such
manifolds.

We shall also consider both rigorous mathematical issues on uniqueness of
solutions, the simplest examples of exact solutions of fundamental field and
evolution equations and their encoding as bi-Hamilton structures and
solitonic hierarchies. There will be investigated the fundamental relations
between nonholonomic structures and quantum geometries and duality and
deformation quantization of almost K\"{a}hler models of nonholonomic
pseudo-Riemann. As phase space constructions they will provide a good
challenge for noncommutative Lagrange and Hamilton geometry. The methodology
will consist broadly in looking at such almost symplectic geometry methods
applied both on (semi) Riemannian and Riemann-Cartan manifolds and (co)
tangent bundle in order to elaborate a Fedosov type formalism related to
Lagrange and Hamilton geometries. We shall combine the approach with our
previous results and methods on generalized Finsler (super) geometries and
the anholonomic frame method in various models of gravity and strings. Let
us state the specific particularities with respect to assigned number of
Objectives (Obj.) in previous section:

For Obj. 1 oriented to geometric quantization of nonlinear physical systems,
gravity and mater fields and mechanics: Our approach is supposed to be a
synthesis of quantization schemes with nonholonomic distributions elaborated
in our recent papers oriented applications to quantum gravity, geometric
mechanics and almost symplectic geometries. There will be involved new
issues connected to geometric quantization of almost K\"{a}hler geometries and
quantization of Einstein and Einstein--Finsler spaces. As intermediate
milestones there will be considered and developed certain explicit
computations for perturbative models, diagrammatic techniques and
nonholonomic geometric renormalization of analogous/emergent commutative and
noncommutative gravitational and gauge like theories. In explicit form, we
shall compute observable effects with nonlinear dispersions, local
anisotropy and Lorenz violation determined from quantum gravity models on
pseudo-Riemannian spacetimes and their tangent bundle extensions. Possible
effective corrections for quantum Lagrange-Hamilton spaces, quantum solitons
and solitonic hierarchies will be computed using analytic methods.

For Obj. 2 on exact solutions in gravity and Ricci flow theory: Explicit
study of solutions for evolution equations and systems of nonlinear partial
differential equations (NPDE) modeling field interactions subjected to
nonholonomic constraints positively impose a coordinate / index style for
geometric and functional analysis constructions. Such coordinate-type
formalism, tensor-index formulas and corresponding denotations are typical
in Hamilton's and Grisha Perelman's works. Additional geometric studies are
necessary to state constructions in global form, for instance, with the aim
to derive global symmetries and study of certain nontrivial topological
configurations. We shall elaborate a distinguished tensor calculus, with
respect to frames adapted to the nonlinear connection and generalized
nonholonomic structures in order to be able to compute the evolution of such
objects under Ricci flows. As a first intermediate milestone step we shall
elaborate coordinate free criteria stating the conditions when nonholonomic
deformations of the frame and connection structures result in decoupling of
the Einstein field equations and formulating of general solutions in
abstract/global forms. We shall use for such constructions our former
results on associated solitonic hierarchies and multipole moment formalism.
The next step will be oriented to geometric methods and generalizations of
sigma models and nonholonomic deformation methods for generalized black hole
uniqueness and non-hair theorems in general relativity and modified gravity
theories. The third step (intermediate milestone) will be related to
computations for certain types of commutative and noncommutative and/or
nonholonomic variables. The geometry of hypersurfaces and possible relations
to global solutions of evolution of exact solutions of Einstein equations,
possible connections to beta functions and renormalization in quantum
gravity will be applied for classification purposes and study of possible
implications in modern cosmology and astrophysics.

For Obj. 3-5: We shall develop and apply an abstract index spinor techniques
adapted to nonlinear connections in Clifford bundles. The theory of
nonholonomic Dirac operators will be formulated in a form admitting
straightforward extensions to complex manifolds, almost symplectic spinors,
nonholonomic spinor spaces and noncommutative generalizations. Methods of
twistor geometry and applications to self-dual Einstein and Yang-Mills
systems will generalized for nonholonomic configurations and Pfaff systems
associated to twistor equations. The approach with twistor diagrams and
quantization will be extended to nonholonomic gravitational gauge
interactions and applied to almost symplectic manifolds and bundle spaces. A
recent techniques of distinguished spectral triples and nonlonomic Dirac
operators will be applied for generating nonholonomic and/or noncommutative
Clifford structures. We shall compute generalized series decompositions for
Seiberg-Witten transforms of gauge like reformulated Einstein equations,
deformation quantization of noncommutative spacetimes and their geometric
evolution.

 Finally, it is noted that there not presented comments in explicit form on
techniques and milestones for Objs. 4 and 5 because such issues are on
constant modification depending on observational and experimental data.
There are planned some important International Conferences on Gravity and
Cosmology for the second part of 2012 which will allow to formulate more
exact plans on activity in such directions.

\subsubsection{Travels and human and material resources}

The applicant have a more than 15 year experience as a senior researcher
(CS1) and administrative charge as an expert in the fields of mathematical
and theoretical physics, geometry and physics, evolution equations in
physics, deformation quantization, quantization, analogous gravity,
application in cosmology and astrophysics etc.

He is the researcher planned
to have most travels with lectures and talks at conferences in Western
Countries and Romania, all related to the Project IDEI and other funds. He
is assisted by a technician (employed for 36 months) with specific skills on
performing mathematical works, schemes, posters, latex and beamer
arrangements of manuscripts, posters, slides etc.

For a successful
realization of this multi-disciplinary research program (based on geometric
methods and new directions in modern mathematics, computer methods etc), he
has a very important professional support from  members of traditionally
strong school of differential geometry and applications existing at the
Department of Mathematics of UAIC, Institute of Mathematics at Ia\c{s}i etc.
It is planned to involve post-graduate students in the project at least with
a half charge during 36 months (depending on financial sources). It is also
supposed that the applicant as a leader of project may invite some
researchers outside Romania for a period up till 3 months.

Finally, in this section, it should be emphasized that the Project IDEI allows the
applicant to get 3 very desk tops/ laptops and computer macros necessary
for advanced research on mathematics, physics, astronomy and mechanics.

\section{Supervision and Pedagogical Activity}

\label{ch2s2}Applicant's activity after obtaining PhD in 1994 is a typical
research one (beginning 1996, as a senior researcher) for mathematical physics
scientists originating from former URSS and with ''high mobility'' in
Western Countries determined by a number of research grants and fellowships
and certain human rights issues. Nevertheless, it has a pluralistic
pedagogical activity and experience which can be quite important for his possible future
senior research positions in Romania.

\subsection{Teaching and supervision experience}

In brief, one should be mentioned such activities:

\begin{enumerate}
\item {\it University teaching in English, Romanian and Russian:} During 1996-1997, 2001,
2006, the applicant delivered lectures with seminars and labs activities
respectively at two Universities in R. Moldova (Free University of Moldova
and Academy of Economic Studies at Chi\c sin\v au; it was an attempt to introduce teaching in English for students at some departments), California State University at
Fresno, USA, and Brock University, Ontario Canada in such directions: a)
higher mathematics, mathematical programming, statistics and probability for
economists; b) physics lab; c) partial differential equations and discrete
optimization. A typical course of lectures, with problems and computer lab
elaborated by the applicant can be found in the Web, see \cite{vb4}.

\item In 2000, the applicant supervised a {\it Republican Seminar on ''Geometric
Methods in String Theory and Gravity''} at the Institute of Applied Physics,
Academy of Sciences, R. Moldova. He and the bulk of that participants moved
their activities as (post--graduate) students and researchers in Western Countries, after 2001. Let us consider some
 explicit examples of common research, publications in high
influence score journals, local journals and International Conferences: a)
papers on twistors and conservation laws in modified gravity, in
collaboration with S. Ostaf during 1993-1996, \cite{vcb2,vl2}; research on
gauge like Finsler--Lagrange gravity, together with Yu. Goncharenko (1995), %
\cite{ve1}; research and a series of publications and collaboration with E.
Gaburov and D. Gon\c ta (2000-2001), and with Prof. P. Stavrinos (Athens,
Greece), see a series of works in monograph \cite{vb1}, on metric--affine
Lagrange--Finsler gravity, exact solutions in such theories and brane
physics; a series of important papers was published together with Nadejda A.
Vicol  [one paper together with I. Chiosa, and other two students from
Chisinau, and Prof. D.\ Singleton, USA, and Profs. P. Stavrinos and G.\
Tsagas, Greece], see \cite{ve21a,vcb3,vcp4,vcp6,ve8}, on spinors in
generalized Finsler--Lagrange (super) spaces, wormholes, noncommutative
geometry etc.

\item During his fellowships and visits in Europe and North America, the
applicant collaborated and published papers with young researchers (master
students, post-graduates and post-docs): from Spain (2004-2007),  J. F.
Gonzalez--Hernandes \cite{ve22d} and R. Santamaria (also with Prof. F. Etayo) %
\cite{ve18}; from Romania (2001-2002), F. C. Popa \cite{ve9,vl2} and O.
Tintareanu--Mircea \cite{ve11,vl1},\ (post-graduates of Prof. M. Visinescu),
on locally anisotropic Taub NUT spining spaces, Einstein--Dirac solitonic
waves, locally anisotropic superspaces etc.

\item In R. Moldova, the applicant had the right to supervise PhD, master
and diploma theses, with different competencies, during 1993-2001.
\end{enumerate}

\subsection{Future plans}

The Habilitation Thesis would allow the applicant to supervise PhD thesis
in Romania. He may use his former experience (more than 5 years of
pluralistic activity) in such directions:

\begin{enumerate}
\item  involve young researchers in activities related to his grant
IDEI etc

\item organize a seminar (similarly to point 2 in previous subsection) on
"Geometric Methods in Modern Physics"  at UAIC and other universities and research centers

\item perform PhD supervision and research collaborations with post-docs, students
etc from various countries

\item organize advance teaching on math and physics in Romanian and English, with possible
visits and collaborations with researches from various places
\end{enumerate}

It should be concluded that main activity of the applicant, for the future, is supposed to be  a senior research one (a Romanian equivalent to Western "research professor") with certain additional advanced  pedagogical activity.

\chapter{Publications, Conferences and Talks}

\label{ch3} In this Chapter, there are  listed a series of "most important" applicant's publications and
last 7 years conference/seminar activity (see additional information in his
complete Publication List included in the file for Habilitation Thesis\footnote{ see also applicant's papers  in
http://inspirehep.net and/ or arXiv.org }).  In brief, the Bibliography is presented in this form:

\begin{itemize}
\item ten selected most important papers,  \cite{v1,v2,v3,v4,v5,v6,v7,v8,v9,v10};

\item important ISI and high influence absolute score papers relevant to the
first ones are with numbers \cite{ve1}-\cite{ve49};

\item works published in  Romania,  \cite{vl1}-\cite{vl5};

\item books and reviews of books and encyclopedia are with numbers \cite{vb1}%
-\cite{vb6}, chapters and  sections in books and collections are \cite{vcb1}-%
\cite{vcb3};

\item publications in R. Moldova, \cite{vrm1}-\cite{vrm5};

\item papers published in proceedings of conferences, \cite{vcp1}-\cite{vcp8};

\item communications and participation at conferences and seminars with
support of organizers/hosts, \cite{vcs1}-\cite{vcs40};

\item two electronic preprints \cite{veprep1,veprep2}  (from more than 100
ones) are listed because they contain some important references and
computations.
\end{itemize}

Additional references and citation of works by other authors can be found in
the mentioned works.

\end{document}